\documentclass[12pt,reqno]{amsart}
\usepackage{amsthm,amsfonts,amssymb,euscript}

\def\emph#1{{\it #1}}
\def\textbf#1{{\bf #1}}
\def\Gab{\underline{\Ga}}
\newcommand{\bea}{\begin{eqnarray}}
\newcommand{\eea}{\end{eqnarray}}
\def\beaa{\begin{eqnarray*}}
\def\eeaa{\end{eqnarray*}}

\def\dual{{\,\,^*}}
\def\Div{\mbox{Div\,}}
\def\Lieh{\widehat{\Lie}}

\def\pihO{\,^{  (O)  }     \pih\,}
\def\pO{\,^{(O)} p}
\def\qO{\,^{(O)} q}

\def\piL  {\, ^{(4)}\pi}
\def\piLb {\, ^{(3)}\pi}

\def\pihLb {\, ^{(3)}\hat \pi}

\def\pil  {\, ^{(L)}\pi}
\def\pilb {\, ^{(\Lb)}\pi}

\newcommand{\Llie}{\mbox{$\Lie \mkern-10mu /$\,}}

\def\ub{{\underline{u}}}

\def\dcal{{\mathcal D}_1}
\def\dcall{{\mathcal D}_2}
\def\dcalll{\,^\star{\mathcal D}_1}
\def\dcallll{\,^\star{\mathcal D}_2}

\def\Dcal{{\mathcal D}}

\def\II{{\mathcal I}}
\def\ba{\begin{array}}
\def\ea{\end{array}}
\def\be#1{\begin{equation} \label{#1}}
\def \eeq{\end{equation}}

\newcommand{\nn}{\nonumber}

\def\pih{\hat{\pi}}
\def\nn{\nonumber}

\newcommand{\lapp}{\mbox{$\bigtriangleup  \mkern-13mu / \,$}}

\def\Lie{{\mathcal L}}
\def\tr{\mbox{tr}}

\def\Err{\mbox{Err}}

\def\piL  {\, ^{(4)}\pi}
\def\piLb {\, ^{(3)}\pi}

\def\RR{{\mathcal R}}

\def \NI{\noindent}

\renewcommand{\div}{\mbox{div }}
\newcommand{\curl}{\mbox{curl }}

\def\a{\alpha}
\def\alphab{{\underline\alpha}}
\def\b{\beta}
\def\betab{{\underline\beta}}
\def\thb{\underline{\th}}
\def\Gab{\underline{\Ga}}
\def\aa{\alphab}
\def\bb{\betab}

\def\ga{\gamma}
\def\Ga{\Gamma}
\def\de{\delta}
\def\De{\Delta}
\def\ep{\epsilon}

\def\la{\lambda}

\def\si{\sigma}

\def\om{\omega}
\def\omegab{{\underline\omega}}
\def\Om{\Omega}

\def\rhoc{\check{\rho}}

\def\th{\theta}
\def\Th{\Theta}
\def\ze{\zeta}
\def\nab{\nabla}

\def\bb{\underline{\b}}
\def\lap{\Delta}
\def\mub{{\underline\mu}}

\def\phib{\underline{\phi}}
\def \phit{\, ^{(3)} \phi}
\def\phibt{\,^{(3)}\phib}
\def\phif{\, ^{(4)} \phi}
\def\phibf{\,^{(4)} \phib}

\def\mubv{\mub \ \mkern-16mu/}

\def\hot{\widehat{\otimes}}
\newcommand{\trchb}{\tr \chib}
\newcommand{\wtrchb}{\widetilde{\tr \chib}}
\newcommand{\chih}{\hat{\chi}}
\newcommand{\chib}{\underline{\chi}}

\newcommand{\etab}{\underline{\eta}}
\newcommand{\chibh}{\underline{\hat{\chi}}\,}
\newcommand{\omb}{{\underline{\om}}}
\def\f14{\frac{1}{4}}
\def\f12{{\frac{1}{2}}}

\def\c{\cdot}

\newcommand{\les}{\lesssim}

\def \piY{\,^{  (Y)  }     \pi\,}
\def \piZ{\,^{  (Z)  }     \pi\,}

\def \piO{\,^{  (O)  }     \pi\,}

\def\dual{{\,\,^*}}
\def\Div{\mbox{Div\,}}
\def\Lieh{\widehat{\Lie}}

\def\pihO{\,^{  (O)  }     \pih\,}
\def\pO{\,^{(O)} p}
\def\qO{\,^{(O)} q}

\newcommand{\Tr}{\text{Tr}\,}

\newcommand{\QQ}{\mathcal{Q}}
\newcommand{\HH}{{\mathcal H\,}}

\newcommand{\DD}{{\mathcal D}}
\newcommand{\EE}{{\mathcal E}}

\def\HHb{\underline{\HH}\,}
\def\RRb{\underline{\mathcal R}}
\def\QQb{\underline{\QQ}}
\def\OO{{\mathcal O}}
\def\OS{\,\,^{(S)}\OO}
\def\OH{\,^{(H)}\OO}
\def\OHb{\,^{(\Hb)}\OO}

\def\Lsc{{\mathcal L}_{(sc)}   }

\def\piX{\,^{(X)}\pi}

\def\ub{\underline{u}}

\def\Lb{{\underline{L}}}

\def\sc{{\text sc}}
\def\QQ{{\mathcal Q}}

\def\Hb{{\underline{ H}}}
\def\pr{\partial}

\def\chih{\hat{\chi}}
\def\trch{\mbox{tr}\chi}

\def\trchbt{\widetilde{\trchb}}

\def\ombtild{<\omb>}
\def\omtild{<\om>}

\def\ombt{\omb^\dagger}
\def\omt{\om^\dagger}
\def\nabd{\,^*\nab}

\def\sgn{{\text sgn}}

\def\kapb{\underline{\kappa}}
\def\kap{\kappa}
 \def\Trsc{Tr_{(sc)}}

\usepackage{graphics}
\usepackage{graphicx}
\begin{document}
\theoremstyle{plain}
  \newtheorem{theorem}[subsection]{Theorem}
  \newtheorem{conjecture}[subsection]{Conjecture}
  \newtheorem{proposition}[subsection]{Proposition}
  \newtheorem{lemma}[subsection]{Lemma}
  \newtheorem{corollary}[subsection]{Corollary}

\theoremstyle{remark}
  \newtheorem{remark}[subsection]{Remark}
  \newtheorem{remarks}[subsection]{Remarks}

\theoremstyle{definition}
  \newtheorem{definition}[subsection]{Definition}

\include{psfig}

\author{Sergiu Klainerman}
\address{Department of Mathematics, Princeton University,
 Princeton NJ 08544}
\email{ seri@@math.princeton.edu}
\title[Trapped surfaces]{On the formation of trapped surfaces}

\author{Igor Rodnianski}
\address{Department of Mathematics, Princeton University, 
Princeton NJ 08544}
\email{ irod@@math.princeton.edu}
\vspace{-0.3in}
\maketitle

\section{Introduction}
\subsection{ Main Goals} In a recent important breakthrough D. Christodoulou  \cite{Chr:book}  has solved a long standing problem of General Relativity of evolutionary formation of trapped surfaces in the Einstein-vacuum 
space-times. He has identified an open  set of regular  initial conditions  on a finite outgoing null hypersurface  
leading to a formation a trapped surface in the corresponding  vacuum space-time 
to the future of the initial outgoing hypersurface and another incoming null hypersurface with the 
prescribed Minkowskian data.  
He also gave a version of the same result for data  given on part of past null 
infinity. His proof, which we outline below, is based on     an inspired   choice 
of  the initial condition, an ansatz  which he calls \textit{short pulse},  and a  complex argument of propagation of estimates,  consistent with   the ansatz, based, largely, on the methods used in the global stability of the Minkowski space \cite{Chr-Kl}. Once such estimates are established in a sufficiently large region of  the space-time the actual  proof of the formation  of a  trapped surface is  quite straightforward.

  The goal of the present  paper  is to give a simpler 
proof  by   enlarging the admissible set of initial conditions and,
  consistent with  this,   relaxing the corresponding propagation
    estimates  just enough   that a trapped surface still forms.
  We also reduce the number of derivatives  needed in the argument 
  from two derivatives of the curvature to just one. More importantly, 
  the proof, which can be easily localized  with respect to angular sectors,  has the potential  for further developments. We prove in fact  another result, concerning the formation of \textit{pre-scarred}  surfaces, i.e surfaces whose   outgoing expansion is negative in an open angular sector.   We only concentrate here on the finite problem,
  the problem from  past null infinity can be treated  in the same fashion as in
  \cite{Chr:book} once the finite problem is well understood. The problem 
  from past null infinity  has been subsequently considered in a recent
   preprint by Reiterer and Trubowitz, \cite{R-T}. 
  
  We start by providing the framework of double
  null foliations in which Christodoulou's   result is formulated. We then
  present, in subsection  \ref{subs:heuristic},  the heuristic argument  for the formation of a trapped surface. In subsection \ref{subs:short-pulse} we then introduce Christodolou's \textit{short pulse} ansatz and discuss the propagation estimates
   which it entails.

  \subsection{Double null foliations}
    We consider a region $\DD=\DD(u_*,\ub_*)$ of a vacuum spacetime $(M,g)$
    spanned by a double null foliation  generated by the optical functions  $(u,\ub)$  increasing towards the future,   $0\le u\le u_*$ and $0\le\ub\le  \ub_*$.  We denote by $H_u$ the outgoing  null hypersurfaces generated by the  level surfaces of $u$  and     by  $\underline{H}_{\ub}$ the incoming  null hypersurfaces generated  level hypersurfaces  of  $\ub$.  We write 
$S_{u,\ub}=H_u\cap \underline{H}_{\ub}$ and   denote by $H_{u}^{(\ub_1,\ub_2)}$,  and $\underline{H}_{\ub}^{(u_1,u_2)}$ the regions of these null hypersurfaces   defined by $\ub_1\le\ub\le\ub_2$ and respectively $u_1\le u\le u_2$.
    Let $L,\Lb $ be the geodesic vectorfields associated to the two foliations and  define, 
\bea
\frac 1 2 \Omega^2=-g(L,\Lb)^{-1}\label{eq:def.omega}
\eea
Observe that the flat value\footnote{Note that our normalization for 
 $\Om$ differ from that of  \cite{KNI:book} } of $\Om$ is $1 $.  
 As well known,  our  space-time slab  $\DD(u_*, \ub_*)$  is completely 
determined  (for small values of $u_*, \ub_*$)   by  data along the null, characteristic,  hypersurfaces $H_0$,  $\Hb_0$ corresponding to
 $\ub=0$, respectively $u=0$.
Following \cite{Chr:book} we assume that   our  data is trivial along 
$\Hb_0$, i.e. assume that $H_0$ extends for $\ub<0$   and  the  spacetime    $(M, g)$ is Minkowskian for   $\ub <0$  and all values of  $u\ge 0$. Moreover we can construct our double null foliation such that
 $\Om=1$ along $H_0$, i.e.,
\bea
\Om(0,\ub)=1, \qquad 0\le \ub\le\ub_*.
\eea

  Throughout this paper we   work with the normalized null pair $(e_3,e_4)$,
$$
e_3=\Omega\Lb,\quad e_4=\Omega L,\qquad g(e_3,e_4)=-2.
$$
Given a   2-surfaces  $S(u,\ub)$  and   $(e_a)_{a=1,2}$   an arbitrary  frame tangent to it  we define  define the Ricci coefficients,
 \bea
\Ga_{(\la)(\mu)(\nu)}=g(e_{(\la)}, D_{e_{(\nu)}} e_{(\mu)} ),\quad \la,\mu ,\nu =1,2,3,4
 \eea
 These coefficients are completely determined by the following components,
 \begin{equation}
\begin{split}
&\chi_{ab}=g(D_a e_4,e_b),\, \,\, \quad \chib_{ab}=g(D_a e_3,e_b),\\
&\eta_a=-\frac 12 g(D_3 e_a,e_4),\quad \etab_a=-\frac 12 g(D_4 e_a,e_3)\\
&\omega=-\frac 14 g(D_4 e_3,e_4),\quad\,\,\, \omegab=-\frac 14 g(D_3 e_4,e_3),\\
&\ze_a=\frac 1 2 g(D_a e_4,e_3)
\end{split}
\end{equation}
where $D_a=D_{e_{(a)}}$. We also introduce the  null curvature components,
 \begin{equation}
\begin{split}
\a_{ab}&=R(e_a, e_4, e_b, e_4),\quad \, \,\,   \aa_{ab}=R(e_a, e_3, e_b, e_3),\\
\b_a&= \frac 1 2 R(e_a,  e_4, e_3, e_4) ,\quad \bb_a =\frac 1 2 R(e_a,  e_3,  e_3, e_4),\\
\rho&=\frac 1 4 R(Le_4,e_3, e_4,  e_3),\quad \si=\frac 1 4  \,^*R(Le_4,e_3, e_4,  e_3)
\end{split}
\end{equation}
Here $\, ^*R$ denotes the Hodge dual of $R$.  We denote by $\nab$ the 
induced covariant derivative operator on $S(u,\ub)$ and by $\nab_3$, $\nab_4$
the projections to $S(u,\ub)$ of the covariant derivatives $D_3$, $D_4$, see
precise definitions in \cite{KNI:book}.
Observe that,
\begin{equation}
\begin{split}
&\omega=-\frac 12 \nab_4 (\log\Omega),\qquad \omegab=-\frac 12 \nab_3 (\log\Omega),\\
&\eta_a=\zeta_a +\nab_a (\log\Omega),\quad \etab_a=-\zeta_a+\nab_a (\log\Omega)
\end{split}
\end{equation}
The connection coefficients $\Ga$ verify  
 equations which have, very  roughly,  the form,
\begin{equation}
\label{eq:nullstrGa}
\begin{split}
\nab_4  \Ga&=R+\nab\Ga+\Ga\c \Ga\\
\nab_3\Ga &=R+\nab\Ga+\Ga\c \Ga
\end{split}
\end{equation}
Similarly the Bianchi identities for  the null curvature components verify, also  very roughly,
\begin{equation}
\label{eq:nullBianchi}
\begin{split}
\nab_4  R&=\nab R+\Ga\c R\\
\nab_3R  &=R+\Ga\c R
\end{split}
\end{equation}

The precise form of these equations is given in 
the next section, see  \eqref{null.str1}--\eqref{null.str4}.
Among these equations we note the following two, which play an essential
role in  Christodoulou's argument for  the formation of trapped surfaces. 
\bea
\nab_4 \trch+\frac 12 (\trch)^2&=&-|\chih|^2-2\om \trch  \label{eq:intr.nab4trchi}\\
\nab_3\chih+\frac 1 2 \trchb \chih&=&\nab\widehat{\otimes} \eta+2\omb \chih-\frac 12 \trch \chibh +\eta\widehat{\otimes} \eta\label{eq:intr.nab3chi}
\eea
\subsection{Heuristic argument}
\label{subs:heuristic}
We start by  making  some important  simplifying assumptions.  As mentioned above we   assume  that our  data is trivial along 
$\Hb_0$, i.e. assume that $H_0$ extends for $\ub<0$   and  the  spacetime    $(M, g)$ is Minkowskian for   $\ub <0$  and all values of  $u\ge 0$.  We introduce 
a small parameter $\de>0$ and     restrict the values of $\ub $  to  $0\le\ub\le \de$,  i.e. $\ub_*=\de$.

 \begin{minipage}[!t]{0.4\textwidth}
    \includegraphics[width=4.2in]{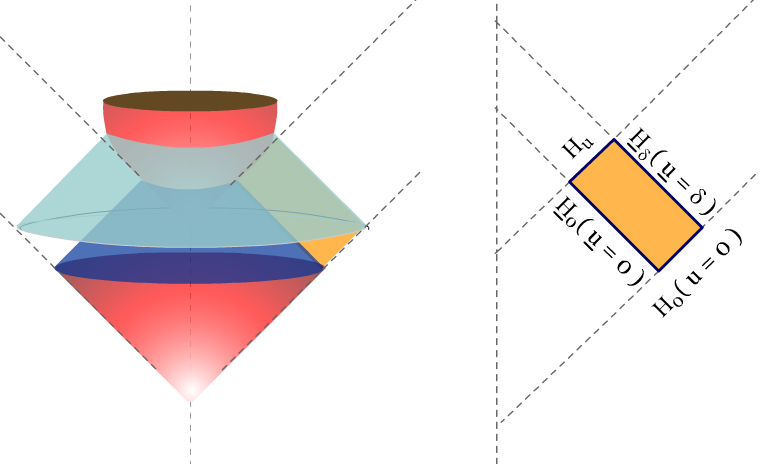}
\end{minipage}
\hspace{0.2\textwidth} 
\begin{minipage}[!t]{0.4\textwidth}
The colored region on the right represents the domain $\DD(u,\ub)$,
$0\le \ub\le \de$.  The same picture is represented, more realistically   on the left  
 The   lower  red region on the left is  the   flat  portion of $H_0$, $u=0$,   while   the   upper red region, corresponding to a  large values of $u$,
   is trapped starting with $\ub=\de$. 
\end{minipage}

 We also make  the following additional assumptions,
assumed to hold 
in \textit{the entire slab}  $\DD(u,\de)$.  We denote  by  $r=r(u,\ub)$  the radius of the $2$-surfaces $S=S(u,\ub)$,  i.e. $|S(u,\ub)|=4\pi r^2$.

\begin{itemize}
\item For small $\de$, $u, \ub$ are comparable with their standard  values in flat space, i.e.
$
u\approx\frac{t-r+r_0}{2},\,\,\quad  \ub\approx\frac{t+r-r_0}{2}
$. We also assume that $\Om\approx 1$,   $\frac{dr}{du}\approx -1$. 
\item  Assume that $\trchb$ is close to its value in flat space, i.e. 
$\trchb\approx-\frac{2}{r}$.  
\item Assume that the term $E=\nab\widehat{\otimes} \eta+2\omb \chih-\frac 12 \trch \chibh +\eta\widehat{\otimes} \eta$  on  the    right hand side of equation 
 \eqref{eq:intr.nab3chi} is sufficiently small and 
can be neglected in a first approximation. Assume  also that
we can neglect the term $\trch\om $ on the right hand side
of \eqref{eq:intr.nab4trchi}.

\end{itemize} 
Given these assumptions we can rewrite  \eqref{eq:intr.nab4trchi},
\beaa
\frac{d}{d\ub}\trch&\les &-|\chih|^2
\eeaa
or, integrating,
\bea
\trch(u,\ub)&\les &\trch(u,0)-\int_0^{\ub}   |\chih|(u,\ub')^2 d\ub'\label{eq.ineq.intrch}\\
&=&\frac{2}{r(u,0)}-
\int_0^{\ub} |\chih(u,\ub')|^2 d  \ub'\nn
\eea
Multiplying  \eqref{eq:intr.nab3chi} by $\chih$ we deduce,
\beaa
\frac{d}{du}|\chih|^2+\trchb|\chih|^2&=&\chih\c E
\eeaa
or, in view of our assumptions for $\trchb$,
and $\frac{dr}{du} $
\beaa
\frac{d}{du} (r^2 |\chih|^2)&=&r^2\frac {d}{du}|\chih|^2+2 r \frac {dr}{du} |\chih|^2
=r^2|\chih|^2\big(-\trchb +\frac{2}{r} \frac{dr}{du}\big)+r^2\chih\c  E\\
&=&r^2|\chih|^2\big(-(\trchb+\frac{2}{r})+\frac{2}{r}(1+ \frac{dr}{du}) \big)+r^2\chih\c  E
:=F
\eeaa
i.e.
\beaa
r^2 |\chih|^2(u,\ub)=r^2(0,\ub)  |\chih|^2(0 ,\ub)+\int_0^u F(u', \ub)  du' 
\eeaa
Therefore, as   $\int_0^\de  |F|$  is negligible  in $\DD$,  we deduce
 \beaa
r^2 |\chih|^2(u,\ub)\approx r^2(0,\ub)  |\chih|^2(0 ,\ub)
\eeaa
We now freely prescribe $\chih$ along  the initial hypersurface $H_0^{(0,\de)}$, 
i.e.
\bea
\chih(0,\ub)=\chih_0(\ub)
\eea
for some traceless $2$ tensor $\chih_0$. 
We  deduce,
\beaa
|\chih|^2(u,\ub)\approx\frac{ r^2(0,\ub)}{r^2(u,\ub)}   |\chih_0|^2(\ub)
\eeaa
or, since $|\ub|\le \de$ and $r(u,\ub)=r_0+\ub-u$,
\beaa
|\chih|^2(u,\ub)\approx\frac{ r_0^2}{(r_0-u)^2}   |\chih_0|^2(\ub)
 \eeaa
Thus, returning to  \eqref{eq.ineq.intrch},
\beaa
\trch(u,\ub)&\le &\frac{2}{r_0-u} -
\frac{ r_0^2}{(r_0-u)^2}\int_0^{\ub}  |\chih_0|^2(\ub') d\ub'+ \mbox{ error}
\eeaa
Hence, for small $\de$, the necessary condition to have $\trch(u,\ub)\le 0$
is,
\beaa
\frac{2(r_0-u)}{r_0^2}<  \int_0^{\de}  |\chih_0|^2
\eeaa
Analyzing equation \eqref{eq:intr.nab4trchi} along $H_0$ we easily deduce that
the condition  for the initial hypersurface $H_0$ not to contain trapped hypersurfaces is,
\beaa
\int_0^{\de}  |\chih_0|^2< \frac{2}{r_0}
\eeaa
i.e. we  are led to prescribe $\chih_0$ such that,
\bea
\frac{2(r_0-u)}{r_0^2}<  \int_0^{\de}  |\chih_0|^2< \frac{2}{r_0}
\label{eq:initial.trapped}
\eea
We thus expect, following Christodoulou, that trapped surfaces
may  form  if \eqref{eq:initial.trapped} is verified.

\subsection{Short pulse data}
\label{subs:short-pulse}
 To prove such a result however
we need to check that all the assumptions   we made above can
be verified.  To start with, the assumption \eqref{eq:initial.trapped} 
requires, in particular, an  $L^\infty$ upper  bound of the form,
\beaa
|\chih_0|\les \de^{-1/2} 
\eeaa
If we can show that such a bound persist in $\DD$ then,
in order to control the error terms $F$ we need, for some
$c>0$,
\begin{align}
&\trchb+\frac 2 r=O(\de^{c}),\qquad 
\frac{dr}{du}+1=O(\de^{c} ), \qquad \eta=O(\de^{-1/2+c}),\nn\\ 
&\om=O(\de^{-1+c}),\qquad 
\nab\eta=O(\de^{-1/2+c}).\label{eq:boun} 
\end{align}
Other bounds will be however needed as we have  to take into account  all null
structure equations. We face, in particular,  the difficulty that  most null
structure equations have curvature components as sources.  Thus we are
obliged to derive   bounds  not just for all  Ricci coefficients $\chi,\om, \eta,\etab, \chib,\omb$ but also  for all null curvature components $\a, \b,\rho,\si,
\bb, \aa$.   
In his work   \cite{Chr:book} Christodoulou  has been able to derive such estimates
starting with an ansatz (which he calls short pulse)  for  the initial data $\chih_0$. More precisely
he assumes, in addition to the triviality of  the initial data along $\Hb_0$,
that $\chih_0$ verifies, relative to coordinates $\ub $ and  transported 
coordinates  $\om$  along $H_0$,  (i.e. transported with respect to $\frac{d}{d\ub}$),
\bea
\chih_0(\ub,\om)= \de^{-1/2}f_0( \de^{-1} \ub, \, \om)\label{short.pulse}
\eea
where $f_0$ is a fixed traceless, symmetric  $S$-tangent  two tensor  
along $H_0$.   This ansatz is consistent with the following
more general condition,    for sufficiently large number
of derivatives $N$ and sufficiently small $\de>0$.
\bea
\de^{1/2+k}\|\nab_4^k\nab ^m \chih_0\|_{L^2(0,\ub )}<\infty,\qquad   0\le k+m\le N,\,\,\, 0\le\ub\le \de.\label{assumption.Chr}
\eea
\textit{Notation}.\,\, 
Here $\|\c\|_{L^2(u,\ub)}$ denotes the standard
$L^2$ norm  for tensorfields   on $S(u,\ub)$. Whenever
there is no possible confusion we will also denote
these norms by $\|\c\|_{L^2(S)}$. We shall also
denote by  $\|\c\|_{L^2(H)}$ and $\|\c\|_{L^2(\Hb)}$
the standard $L^2 $ norms along the null hypersurfaces
$H=H_u$ and $\Hb=\Hb_{\ub}$.
\begin{remark}
In  \cite{Chr:book} Christodoulou also includes  weights, depending  on $|u|$,
in his estimates. These allow him to derive not only a local result
but also one with data at past null infinity.  In our work here we only concentrate
on the local result, for $|u|\les  1$,  and thus drop the weights.  
\end{remark}
 
Assumption \eqref{assumption.Chr}, together  with  the null structure equations \eqref{eq:nullstrGa} and null Bianchi equations  \eqref{eq:nullBianchi}
   leads  to the following  estimates for the null  curvature  components,
   along  the initial null hypersurface $H_0$,
\bea
\de\|\a\|_{L^2(H_0)}+\|\b\|_{L^2(H_0)}+\de^{-1/2}\|(\rho,\si)\|_{L^2(H_0)} +\de^{-3/2}\|\bb\|_{L^2(H_0)}&<&\infty\nn\\
\label{Christ.curv.initial}
\eea
Consistent  with \eqref{assumption.Chr}, the angular derivatives 
of $ \a, \b,\rho,\si,\b$ obey the same scaling as in  \eqref{Christ.curv.initial} while  each  $\nab_4$
derivative costs an additional power of $\de$.  
\bea
\de \|\nab \a\|_{L^2(H_0)}+ \|\nab \b\|_{L^2(H_0)}+\de^{-\frac 12} 
\|\nab(\rho,\si)\|_{L^2(H_0)} +\de^{-3/2}\|\nab\bb\|_{L^2(H_0)}&<\infty,\nn
\\
\de^{2} \|\nab_4\a\|_{L^2(H_0)}+\de\|\nab_4\b\|_{L^2(H_0)}+\de^{1/2}\|\nab_4(\rho,\si)\|_{L^2(H_0)} +
\de^{-1/2} \|\nab_4\bb\|_{L^2(H_0)}&<\infty\nn\\\label{Christ.curv.initial'}
\eea
Moreover  one  can derive estimates for
the Ricci coefficients, in various norms, weighted by appropriated powers of $\de$.
Note that  if one were to neglect the quadratic terms in  \eqref{eq:nullBianchi} 
than the expected scaling behavior in $\de$ would have been,
\beaa
\de\|\a\|_{L^2(H_0)}+\|\b\|_{L^2(H_0)}+\de^{-1}\|(\rho,\si)\|_{L^2(H_0)} +\de^{-2}\|\bb\|_{L^2(H_0)}&<&\infty
\eeaa

 Most of the   
 body of work in \cite{Chr:book} is to prove that these estimates 
 can be propagated  in the entire  space-time region $\DD(u_*, \de)$,
  with $u_*$ of size one and $\de$ sufficiently small, and thus fulfill the necessary 
  conditions for the formation of a trapped surface along the lines
  of the  heuristic argument presented above. The proof of such estimates,
  which follows the main outline of the proof of  stability of Minkowski space,
  as  in \cite{Chr-Kl} and \cite{KNI:book},  requires a step by step analysis
   to make sure that all estimates are consistent  with the assigned 
    powers of $\de$. This task is made particularly taxing  in view of the fact
  that there are many  nonlinear interferences which have to be tracked precisely.
  \subsection{Outline of Christodoulou's  propagation estimates}
   To see what this entails it pays to say a few words about the strategy of the proof.  As in \cite{Chr-Kl} and \cite{KNI:book}  the centerpiece of the entire 
   proof consists in  proving   spacetime curvature  estimates consistent
   with \eqref{Christ.curv.initial}. In this case however the primary attention has to be given to the stratification 
   of the estimates for different curvature components based on their $\de$-weights. 
   This is done using the Bianchi identities,
   \beaa
   D_{[\ep} R_{\a\b]\ga\de}=0,
   \eeaa 
    the associated   Bel-Robinson tensor 
   $Q$ and carefully chosen  vectorfields $X$ whose deformation tensors
   $\piX$ depend only on the Ricci coefficients $\chi,\om,\eta,\etab,\chib, \omb$.
   These vectorfields can be used either as commutation vectorfields or
   multipliers. In the latter case  we would have, 
   \bea
   D^\de( Q_{\a\b\ga\de}  X^{\a} Y^{\b} Z^{\de}   )=Q ( \piX ,  Y,  Z)+\ldots
   \label{eq:diverg.Bel}
   \eea
   As multipliers $X,Y, Z$ we can chose  the vectorfields $e_3, e_4$.
   The choice   $X=Y=Z=e_4$  leads to,
   after  integration  on $\DD(u,\ub)$, 
   \bea
  \|\a\|_{L^2(H^{(0,\ub)}_u)}^2+\|\b\|_{L^2(\Hb^{(0,u)}_{\ub})}^2&=&   \|\a\|_{L^2(H^{(0,\ub)}_0)}^2+\int\int_{\DD(u,\ub)}   3Q(\piL, e_4, e_4)\qquad\quad 
  \label{eq:intr.curv.first}
   \eea
   where $\pi$ is the  deformation tensor of $e_4$.    Since the initial data  at $H_0$ verifies \eqref{Christ.curv.initial} we write,
   \beaa
   \de^2\big( \|\a\|_{L^2(H^{(0,\ub)}_u)}^2+\|\b\|_{L^2(\Hb^{(0,u)}_{\ub})}^2\big)&=& \de^2  \|\a\|_{L^2(H^{(0,\ub)}_0)}^2+3\de^2\int\int_{\DD(u,\ub)}  
    Q(\piL, e_4, e_4)
   \eeaa   
   and expect to bound the double integral term on the right. 
   One can derive similar identities for all  other possible choices 
   of $X, Y, Z$ among the set $\{ e_3, e_4\}$.  This allows one
   to estimate   both    the $L^2(H)$ norms of $\a,\b, \rho,\si, \bb$ 
   and  the $L^2(\Hb)$  of $\b,\rho,\si,\bb,\aa$, with appropriate
    $\de$ weights,  in terms
   of  corresponding $\de$-weighted   $L^2(H_0)$ norms of $\a,\b,  \rho,\si, \bb$  and spacetime
   integrals  of  $Q(\piL, e_\mu, e_\nu )$ and     $Q(\piLb, e_\mu, e_\nu)$
   with $\mu,\nu=3,4$.  
         We can thus extend the initial  estimates  \eqref{Christ.curv.initial}
   to every null hypersurface $H_u$ in our slab provided that we 
   can bound all the double integrals on the right hand side 
   of our integral  identities.  Now,  both deformation tensors  $\piL$ and $\piLb$    can be expressed in terms of our connection coefficients 
   $\chi,\om,\eta,\etab,\omb,\chib$.  Since $Q$ is quadratic 
   in $R$,  to be able to close  estimates   for our  null curvature components 
   we need to derive sup-norm estimates for  all our Ricci coefficients.
    This leads us to the second pillar of the
   construction which is to derive estimates for  Ricci coefficients
   in terms of the null curvature components,  with the help of  the null structure 
   equations \eqref{eq:nullstrGa}. Combining these equations with
   the constrained equations, on fixed $2$ surfaces $S(u,\ub)$, and
   the null Bianchi identities we are lead  to precise $\de$-
   weighted  estimates  of 
   all Ricci coefficients in terms of  $\de$-  weighted $L^2(H)$ and $L^2(\Hb)$     norms  of all  null curvature  components    and their derivatives. Thus, in  a first approximation, the error terms
   in the above integral identities  are quadratic in $R$ and linear
   in their first derivatives. Therefore to be able to close one needs:
   \begin{enumerate}
   \item    Derive higher   derivative estimates for the curvature components. 
   \item Make sure that all error terms can be controlled in terms of the 
   principal terms, in the corresponding  energy inequality,   or terms
   which have already been estimated at previous steps.   
   \end{enumerate}
   Note that  2) here seems counterintuitive in view of the large data character 
   of the problem under consideration. Indeed, typically,  in such situations one cannot expect to
   control the  nonlinear error terms by  the principal energy terms.
   The miracle here is that  the error terms
   are either linear (in the main energy terms), or they contain factors which have been already estimated in previous
   steps, or  are truly nonlinear, in which case they are small in powers of $\de$ relative to 
   the principal energy  terms. This is due to the structure  of the error terms,  reminiscent of the null condition,  in which the factors combine in  such a way  that the total weight in powers
   of $\de$ is positive.

   In his work Christodoulou derives estimates for the first two derivatives 
   of the curvature tensor by   commuting the 
   Bianchi identities with the vectorfields $L$,
   $S=\frac{1}{2}(ue_3+\ub e_4) $ and rotation vectorfields $O$.
   This process leads to a proliferation of  error terms. Moreover
   not all error terms which are generated this way verify the following  essential  requirement, alluded above;  \textit{that they lead to an overall factor of $\de^c$,
    with a positive exponent $c$,  and thus  can be absorbed on the left, for sufficiently small $\de$}.
   Due to nonlinear interactions,   Christodoulou has to tackle  anomalous  error terms  which   are $O(1)$ in $\de$.    
    Yet he  is able to show, by a careful step by step  analysis,  that all such terms are, indeed, linear relative to terms which have
   already been estimated and thus only quadratic (i.e. linear in the principal energy norm)
    relative to the remaining
   components.  They can therefore be absorbed by a standard Gronwall
    inequality.  A similar phenomenon helps him to  estimate, step by step,  all Ricci
     coefficients. 
   
   \subsection{New initial conditions} As explained above the main purpose  of this paper  is to embed   the short-pulse ansatz of Christodoulou into  a more general
   set of initial conditions, based on a   different underlying  scaling.  The new scaling,
   which we incorporate into our basic norms,  allows us to conceptualize the  separation between
the linear and nonlinear terms in the null  Bianchi and null structure equations and explain 
the favorable  appearance of additional positive powers of $\de$ in the nonlinear error terms mentioned  above. Though the initial conditions required to  include
Christodoulou's data do not quite satisfy this scaling, the generated anomalies 
are fewer and thus much  easier to track.  

 We start with the observation that a natural alternative to 
  \eqref{short.pulse} which  comes to mind, related  to the familiar  
  parabolic scaling on null hyperplanes  in Minkowski space,   
  is
   \bea
\chih_0(\ub,\om)= \de^{-1/2}f_0( \de^{-1} \ub,\, \de^{-1/2} \om), \label{eq:ansatz0}
\eea
This does not quite make sense in our   framework of compact $2$-surfaces $S(u,\ub)$, unless of course one is willing 
to consider the initial data $\chih_0(\ub,\om)$ supported in the angular sector $\om$ of size $\de^{\frac 12}$. 
Such a support assumption would be however in contradiction with the lower bound in \eqref{eq:initial.trapped}
required to be satisfied for {\it each} $\om\in {\Bbb S}^2$.
 
 The following interpretation
of \eqref{eq:ansatz0}   (compare with \eqref{assumption.Chr})
makes sense however.
\bea
\de^{k+\frac m 2  }\|\nab_4^k\nab ^m \chih_0\|_{L^2(0,\de)}<\infty,\qquad    \, 0\le k+m\le N
\label{eq:ansatz1.gen}
\eea
Just as in the derivation of \eqref{Christ.curv.initial} we can use 
  null structure equations \eqref{eq:nullstrGa} and null Bianchi equations  \eqref{eq:nullBianchi}
  to derive, from \eqref{eq:ansatz1.gen},
  \begin{equation}
  \label{Scale.curv.initial0}
  \begin{split}
\de^{1/2} \|\a\|_{L^2(H_0)}+\|\b\|_{L^2(H_0)}+\de^{-1/2}\|(\rho,\si)\|_{L^2(H_0)} +\de^{-1}\|\bb\|_{L^2(H_0)}&<\infty
\\
\de \|\nab \a\|_{L^2(H_0)}+\de^{1/2} \|\nab \b\|_{L^2(H_0)}+\|\nab(\rho,\si)\|_{L^2(H_0)} +\de^{-1/2}\|\nab\bb\|_{L^2(H_0)}&<\infty,
\\
\de^{3/2} \|\nab_4\a\|_{L^2(H_0)}+\de\|\nab_4\b\|_{L^2(H_0)}+\de^{1/2}\|\nab_4(\rho,\si)\|_{L^2(H_0)} +\|\nab_4\bb\|_{L^2(H_0)}&<\infty
\end{split}
\end{equation}

We refer to these  conditions, consistent with  the  null parabolic scaling,  
as \textit{$\de$-coherent}  assumptions.
Observe that, unlike in the Christodoulou's case, each $\nab$ derivative costs 
a $\de^{-1/2}$.  It turns out that proving the propagation of such estimates can be done easily
and systematically without the need of  the step  by step procedure mentioned earlier.
 In fact  one 
can  show, in this case, that all error terms, generated in the process
of the energy estimates  are either  quadratic 
in the curvature and can be  easily  taken care by Gronwall or, if cubic, they must come with a factor of $\de^{1/2}$
and therefore can be all absorbed for small values of $\de$. 

The  main problem with the  ansatz \eqref{eq:ansatz0}, as with
 initial conditions \eqref{eq:ansatz1.gen}, however, is that it is   
 inconsistent  with the  formation of trapped surfaces requirements discussed above.  One can  only hope  to show that
 the expansion scalar $\trch$  along $H_u$, at  $S(u,\ub)$, for some  $u\approx 1$, will become  negative\footnote{We could call such a region  locally trapped, or a pre-scar}  only in a  small angular
 sector of size $\de^{1/2}$. 
   This is because, consistent  with  
 \eqref{Scale.curv.initial0},  condition  \eqref{eq:initial.trapped}  may only be satisfied  in such a sector. 
   
 At this point we abandon the ansatz formulation of the characteristic initial data problem for the Einstein-vacuum
 equations and replace with an hierarchy of bounds, which ``interpolate" between the regular $\de$-coherent 
 assumptions \eqref{Scale.curv.initial0} and the estimates \eqref{Christ.curv.initial}-\eqref{Christ.curv.initial'} following 
 from Christodoulou's  short pulse ansatz.
 
  At the level of curvature the new assumptions correspond to:
   \begin{equation}
   \begin{split}
  \label{Scale.curv.initial.new}
\de \|\a\|_{L^2(H_0)}+\|\b\|_{L^2(H_0)}+\de^{-1/2}\|(\rho,\si)\|_{L^2(H_0)} +\de^{-1}\|\bb\|_{L^2(H_0)}&<\infty\\
\de \|\nab \a\|_{L^2(H_0)}+\de^{1/2} \|\nab \b\|_{L^2(H_0)}+\|\nab(\rho,\si)\|_{L^2(H_0)} +\de^{-1/2}\|\nab\bb\|_{L^2(H_0)}&<\infty,\\
\de^2 \|\nab_4\a\|_{L^2(H_0)}+\de\|\nab_4\b\|_{L^2(H_0)}+\de^{1/2}\|(\nab_4\rho,\nab_4\si)\|_{L^2(H_0)} +\|\nab_4\bb\|_{L^2(H_0)}
&<\infty
\end{split}
\end{equation}
Observe that, by comparison with  \eqref{Scale.curv.initial0},
 the only anomalous terms are $\|\a\|_{L^2(H_0)}$ and  $  \|\nab_4\a\|_{L^2(H_0)}$.
 
In the next section we make precise our initial data assumptions, state the main results and explain the strategy of
the proof. We close the discussion here with a summary of our approach
\begin{enumerate}
\item Replace the short pulse ansatz of Christodoulou with a larger class of data satisfying \eqref{Scale.curv.initial.new}
\item Prove propagation of the curvature estimates consistent with  \eqref{Scale.curv.initial.new} through the domain of existence
and show that these (weaker) estimates are sufficient for the existence result
\item The propagation estimates involve only the $L^2$ based norms of curvature and its first derivatives but generate 
nonlinear terms involving both the Ricci coefficients and its first derivatives. To close such estimates requires addressing 
two major difficulties

 \begin{itemize}
\item Regularity problem: show that the $L^2$ propagation curvature estimates are sufficient to control the Ricci coefficients (in $L^\infty$) and its first and even second derivatives in appropriate norms required by the nonlinear
terms in the curvature estimates
\item  $\de$-consistency problem: show that the nonlinear terms are either effectively linear in (curvature and its derivatives), 
and thus can be handled by the Gronwall inequality,  or contain a smallness coefficient generated by 
an additional power of the parameter $\de$. Our approach, based on the  weaker propagation estimates  \eqref{Scale.curv.initial.new}, is particularly 
suitable for dealing with  this problem in that 
a) it generates fewer borderline terms of the first kind and 
b) it naturally lends itself to the introduction of a notion of {\it scale-invariant} norms relative to which the structure 
of the nonlinear terms and their $\de$-smallness become apparent and nearly universal.

\end{itemize}
\item The propagation estimates  consistent with  \eqref{Scale.curv.initial.new}, and the corresponding
Ricci coefficient  estimates which it generate, are not strong  enough to prove the formation of a trapped surface.  However, once such estimates have been proved in the entire domain 
$\DD(u\approx 1, \ub=\de)$ it is straightforward to impose slightly stronger  conditions
on the initial data  and show that they lead to    spacetimes which satisfy all the necessary
conditions  to implement, rigorously, the   informal  argument  presented  above.

\end{enumerate}

\section{Main Results}
\subsection{Initial data assumptions}    We define the initial
data quantity,
          \bea
           \II^{(0)}&=&\sup_{0\le \ub\le \de}\II^{(0)}(\ub)\label{eq:initial.condition.quantity}
           \eea
where, with the notation convention in \eqref{assumption.Chr},
 \beaa
 \II^{(0)}(\ub)&=& \de^{1/2}  \|\chih_0\|_{L^\infty}+ \sum_{0\le k\le 2}   \de^{1/2}  \|(\de\nab_4) ^k \chih_0\|_{L^2(0,\ub)}\\
 &+&
  \sum_{0\le k\le 1} \, \sum_{1\le m\le 4} \de^{1/2} \|(\de^{1/2}\nab)^{m-1} (\de\nab_4 )^k\, \nab \chih_0\|_{L^2(0,\ub)}
          \eeaa       
          Our main assumption, replacing Christodoulou's ansatz, is
        \bea
\II^{(0)}<\infty\label{eq:initial.condition}
\eea

     We      show  that, under this assumption  and for sufficiently small
        $\de>0$, the spacetime slab $\DD(u, \de)$ can be extended 
        for values of $u\ge 1 $, with precise estimates 
        for all Ricci coefficients of the double null foliation
         and null components of the curvature tensor. We can then show, by a slight modification        of this assumption   together with   Christodoulou's   lower bound  assumption on $\int_0^\de | \chih_0|^2 $  (see 
 equations 14, 15       in \cite{Chr:book}),  that a trapped surface must form in $\DD(u\approx1,\de)$.    As in the case of \cite{Chr:book}) most  of the work is required to prove the semi global result  concerning the double null foliation.
  Once this   is established    the actual formation of trapped surfaces result is proved by making a slight modification of the  main assumption  \eqref{eq:initial.condition} and  following the heuristic argument outlined below. 
  In addition we  show that a small modification of
  the regular $\de$-coherence  assumption  leads to the formation of a pre-scar. 
\subsection{Curvature norms}
To give a precise formulation of our result we need to introduce 
the following norms.
 \bea
 \RR_0(u,\ub):&=&\de\|\a\|_{H^{(0,\ub)} _u}+ \|\b\|_{H^{(0,\ub)} _u}+\de^{-1/2}\|(\rho,\si)\|_{H^{(0,\ub)} _u}+\de^{-1}  \|\bb\|_{H^{(0,\ub)} _u}\nn\\
  \RR_1(u,\ub):&=&\de\|\nab \a\|_{H^{(0,\ub)} _u}+\de^{1/2} \|\nab\b\|_{H^{(0,\ub)} _u}+\|\nab(\rho,\si)\|_{H^{(0,\ub)} _u}+\de^{-1/2}  \|\nab \bb\|_{H^{(0,\ub)} _u}\nn\\
  &+&\de\|\nab_4\a\|_{H^{(0,\ub)} _u}
 \label{eq:curv.norms}\\
  \RRb_0(u,\ub):&=&\de\|\b\|_{\Hb ^{(0,u)}_{\ub}}+ \|(\rho,\si)\|_{\Hb ^{(0,u)}_{\ub}}+\de^{-1/2} \|\bb\|_{\Hb ^{(0,u)} _{\ub}}+\de^{-1} \|\aa\|_{\Hb ^{(0,u)} _{\ub}}\nn\\
   \RRb_1(u,\ub):&=&\de\|\nab \b\|_{\Hb ^{(0,u)} _{\ub}}+ \de^{1/2}\|\nab(\rho,\si)\|_{\Hb ^{(0,u)} _{\ub}}+ \|\nab \bb\|_{\Hb ^{(0,u)} _{\ub}}+\de^{-1/2} \|\nab \aa\|_{\Hb ^{(0,u)} _{\ub}}\nn\\
   &+&\de^{-1}\|\nab_3\aa\|_{\Hb ^{(0,u)} _{\ub}}\nn\
 \eea
  
   We also set $\RR_0, \RR_1$ to be the supremum over $u,\ub$ in our spacetime slab
   of $\RR_0(u,\ub)$ and respectively  $\RR_1(u,\ub)$ 
    and similarly for the norms  $\RRb$. Also we write $\RR=\RR_0+\RR_1$ and
    $\RRb=\RRb_0+\RRb_1$.  Finally,
   $\RR^{(0)}$ denotes the initial value for  the norm $\RR$ i.e.,
   \beaa
   \RR^{(0)}&=&\sup_{0\le \ub\le \de}\big(\RR_0(0,\ub)+\RR_1(0,\ub)\big)
   \eeaa
Remark that the only $\nab_4$ derivative appearing 
in the norms above is that of $\a$. All other $\nab_4$ derivatives  can be deduced from the null Bianchi equations and thus do not need to be incorporated in our norms. We denote   the norms of a specific curvature component $\psi$ by  $\RR_0[\psi]$ and $\RR_1[\psi]$.  
 
    \subsection{Ricci coefficient norms}:  We introduce norms for the Ricci coefficients   $\chih,\trch,\om,\eta,\etab,\omb, \chibh$    and 
   $   \trchbt=\trchb-\trchb_0$, with $ \trchb_0=-\frac{4}{\ub-u+2r_0}    $
    the  flat value of $\trchb$ along the initial hypersurface $\Hb_0$.
    
    For any $S=S(u,\ub)$ we  introduce  norms $\OS(s,p)(u,\ub)$,
    \bea
    \OS_{0,\infty}(u,\ub)&=& \de^{1/2}\big(\|\chih\|_{L^\infty(S)}   + \|\om\|_{L^\infty(S)}\big)+  \|\eta\|_{L^\infty(S)}+\|\etab\|_{L^\infty(S)}\nn\\
&+&\de^{-1/2}\big( \|\chibh\|_{L^\infty(S) }  + \|\trchbt\|_{L^\infty(S)    }     +
\|\omb\|_{L^\infty(S)}  \big)\nn\\
\OS_{0,4}(u,\ub)&=&\de^{1/2}\|\chih\|_{L^4(S) }+\de^{1/4}|\om\|_{L^4(S) }+
 \de^{-1/4} \big(\|\eta\|_{L^4(S)}+\|\etab\|_{L^4(S)}    \big)\nn\\
 &+&\de^{-1/2}\|\chibh\|_{L^4(S)}+\de^{-3/4}\big(\|\trchbt \|_{L^4(S)}+\|\omb\|_{L^4(S)}\big)\nn\\
  \label{eq:Ricci.norms}\\
 \OS_{1,4}(u,\ub)&=&\de^{3/4}\big(\|\nab\chi\|_{L^4(S) }+|\om\|_{L^4(S) }\big)+
 \de^{1/4} \big(\|\nab\eta\|_{L^4(S)}+\|\nab\etab\|_{L^4(S)}    \big)\nn\\
 &+&\de^{-1/4}\big(\|\nab \chib\|_{L^4(S)}+\|\omb\|_{L^4(S)}\big)\nn\\
 \OS_{1,2}(u,\ub)&=&\de^{1/2}\big(\|\nab\chi\|_{L^2(S) }+|\om\|_{L^4(S) }\big)+
  \|\nab\eta\|_{L^2(S)}+\|\nab\etab\|_{L^2(S)}    \nn\\
 &+&\de^{-1/2}\big(\|\nab \chib\|_{L^2(S)}+\|\omb\|_{L^2(S)}\big)\nn
    \eea
    Also,
    \beaa
    \OH(u,\ub)&=&\de^{1/2}\big(\|\nab^2\chi\|_{L^2(H_u^{(0,\ub)}  )}+
    \|\nab^2\om\|_{L^2(H_u^{(0,\ub)}) }   \big)\\
    &+&
    \big(\|\nab^2\eta\|_{L^2(H_u^{(0,\ub)} )}+
    \|\nab^2\etab\|_{L^2(H_u^{(0,\ub)} )}\big)\\
    &+&\de^{-1/2} \big(\|\nab^2\chibh\|_{L^2(H_u^{(0,\ub)} ) }+
    \|\nab^2\omb\|_{L^2(H_u^{(0,\ub)} ) }\big)\\
    \eeaa
    and,
    \beaa
    \OHb(u,\ub)&=&\de^{1/2}\big(\|\nab^2\chi\|_{L^2(\Hb_{\ub}^{(0,u)} )}+
    \|\nab^2\om\|_{L^2(\Hb_{\ub}^{(0,u)})}   \big)\\
    &+&
    \big(\|\nab^2\eta\|_{L^2(\Hb_{\ub}^{(0,u)})}+
    \|\nab^2\etab\|_{L^2(\Hb_{\ub}^{(0,u)})}\big)\\
    &+&\de^{-1/2} \big(\|\nab^2\chibh\|_{L^2(\Hb_{\ub}^{(0,u)})}+
    \|\nab^2\omb\|_{L^2(\Hb_{\ub}^{(0,u)})} \big)
    \eeaa
    We define the norms $\OS_{0,4}, \OS_{1,2}, \OS_{1,4}, \OH, \OHb$ to 
    be the  supremum over all values of $u,\ub$    in our slab  of 
    the corresponding  norms.
        Finally we set set  total Ricci norm  $\OO$,    \beaa
    \OO=\OS_{0,\infty}+\OS_{0,4}+\OS_{1,2}+\OS_{1,4}+\OH+\OHb
    \eeaa
      and by $\OO^{(0)}$ the corresponding  norm
      of the initial hypersurface $H_0$. We further differentiate
      between the first order norms  $
      \OO_{[1]}=\OS_{0,4}+\OS_{1,2}$
      and second order ones,
      $
      \OO_{[2]} =\OS_{1,4}.
      $
\subsection{Main Theorems}
   We are now ready to state our main result. The first
    result follows from analyzing   assumption 
     \eqref{eq:initial.condition.quantity} 
     on the initial hypersurface $H_0$.    
    \begin{proposition}
    \label{thm.main.initial}
    In view of our initial assumption \eqref{eq:initial.condition.quantity} we have,
    for sufficiently small $\de>0$, along $H_0$,
    \bea
    \RR^{(0)}+\OO^{(0)} \les \II^{(0)}
    \label{eq:initialOR}
    \eea
    \end{proposition}
    The proof of the proposition  follows by analyzing the
    null structure and null Bianchi  equations restricted to 
    the initial hypersurface $H_0$, as in chapter 2 of \cite{Chr:book}. 
 In view of this  result we may   replace assumption
     \eqref{eq:initial.condition.quantity} with \eqref{eq:initialOR}, as 
     an initial data assumption.
      Alternatively  we may assume only that $\RR^{(0)}\les \II^{(0)}$.
     It is not too hard to see, following roughly the same steps as in the proof
     of proposition \ref{thm.main.initial}, that, for small $\de$,  we would 
     also have $\OO^{(0)}\les \II^{(0)}$. 

\begin{theorem}[Main Theorem]
      \label{thm.main}
    Assume that 
    $\RR^{(0)}\les \II^{(0)}$ for an arbitrary constant $\II^{(0)}$. 
Then, there exists a sufficiently small $\de>0$ such that,
\bea
\RR+\RRb+\OO&\les& \II^{(0)}.\label{eq:estimates.thm.main}
\eea
       \end{theorem}
       \begin{theorem}\label{thm.add}
       Assume that , in addition to \eqref{eq:initial.condition.quantity},
        we also have, for $2\le k\le 4$
       \bea
       \de^{\frac 12} \|(\de^{\frac 12} \nab)^k\chih_0\|_{L^2(0,\ub)} \le \ep \label{eq:additional.cond}
       \eea
       for a sufficiently small parameter $\ep$ such that $0<\de\ll\ep$. Assume also that $\chih_0$
       verifies \eqref{eq:initial.trapped}.
        Then, for $\de>0$ sufficiently small, a trapped surface must form in 
        the slab $\DD(u\approx 1, \de)$.
      \end{theorem}
      \begin{proof}
    We sketch below the proof of theorem \ref{thm.add}.
      
{\it Step 1.}\,\,
 We  reinterpret  \eqref{eq:additional.cond} in terms of the curvature norms  according to the following:
 \begin{proposition}
 \label{prop:improved.in.curv}
 Under the smallness condition  \eqref{eq:additional.cond} the initial curvature norms  satisfy,
 in addition to  the estimates of proposition \ref{thm.main.initial},
 \begin{equation}
  \label{eq:improved.in.curv}
 \begin{split}
 \de^{1/2} \|\nab\b\|_{H^{(0,\de)} _0}+\|\nab(\rho,\si)\|_{H^{(0,\de)} _0}+\de^{-1/2}  \|\nab \bb\|_{H^{(0,\de)} _0}& \le \ep. 
 \end{split}
 \end{equation}
 \end{proposition}
 The proof is standard and will be omitted. 
 
{\it Step 2.}\,\, We show, see the end of section 15,  that this condition   can be propagated  in the entire slab $\DD(u\approx 1, \de)$,
 \begin{proposition}
 \label{prop:improved.curv}
  Under the assumptions 
  \eqref{eq:additional.cond} we have, uniformly in $u\les  1, \ub\le \de$, for $\de$ sufficiently small,
   \begin{equation}
  \label{eq:improved.curv}
 \begin{split}
 \de^{1/2} \|\nab\b\|_{H^{(0,\ub)} _u}+\|\nab(\rho,\si)\|_{H^{(0,\ub)} _u}+\de^{-1/2}  \|\nab \bb\|_{H^{(0,\ub)} _u}& \le \ep. \\
  \de^{1/2} \|\nab(\rho,\si)\|_{\Hb^{(0,u)} _{\ub}}+\|\nab\bb\|_{\Hb^{(0,u)} _{\ub}}+\de^{-1/2}  \|\nab \aa\|_{\Hb^{(0,u)} _{\ub}}   & \le \ep.
 \end{split}
 \end{equation}

 \end{proposition}
 
{\it  Step 3.}\,\,
 We return  to the system   \eqref{eq:intr.nab4trchi}-
       \eqref{eq:intr.nab3chi},
       \beaa
\nab_4 \trch+\frac 12 (\trch)^2&=&-|\chih|^2-2\om \trch\\
\nab_3\chih+\frac 1 2 \trchb \chih&=&\nab\widehat{\otimes} \eta+2\omb \chih-\frac 12 \trch \chibh +\eta\widehat{\otimes} \eta
\eeaa
responsible, as we have seen,  for the formation of a trapped surface. Theorem \ref{thm.main} implies that the terms ignored in 
our heuristic derivation are negligible. Specifically, the bounds
$
|\om\trch|\les \de^{-\frac 12},$\, 
$|\omb \chih|+| \trch \chibh|+|\eta\widehat{\otimes} \eta|\les 1
$
should be compared to the principle terms of size $\de^{-1}$ and $\de^{-\frac 12}$ in the first and the second equation
respectively. We can also easily verify the other bounds in \eqref{eq:boun} with the exception of that for 
$\nab\widehat{\otimes} \eta$.   The additional  condition \eqref{eq:additional.cond}
       is imposed  in fact precisely in order to assure that the linear  term $\nab\hot \eta$ in
       \eqref{eq:intr.nab3chi} is sufficiently small.
To control this  term   we rely on  the following
proposition.
\begin{proposition}\label{prop:better} Under the assumptions of Theorem \ref{thm.add}
the solution $^{(3)} \phi$ of the problem $\nab_3 ^{(3)} \phi=\nab\hot\eta$, with trivial initial data on $H_0$, verifies,
\bea
|^{(3)}\phi |\le C\de^{-1/2} \ep^{\frac 14}
\eea
\end{proposition}
The proof of proposition \ref{prop:better}, which  appear is section \ref{subs:last}, depends
on the arguments of section 11, in particular proposition  \ref{prop:betterphi}.
The argument for the formation of a trapped surface then proceeds as above with a renormalized 
quantity $(\chih-^{(3)}\phi)$ in place of $\chih$. Note that in view of the estimate on $^{(3)}\phi$ 
the size of  $(\chih-^{(3)}\phi)$ is comparable to that $\chih$. An important comment in this regard,   is that our curvature
propagation estimates does not allow us to control the $L^\infty$ norm of $\nab\widehat{\otimes} \eta$,  let alone
prove the bound stated in \eqref{eq:boun}. This regularity problem, which is  discussed in the  two remarks below,  is resolved with the help of 
the renormalized estimates for the Ricci coefficients in section 11,     of which  Proposition
\ref{prop:better} is an important example.
\end{proof}

         \noindent   {\bf Remark 1.}\,   We remark  that while a loss of derivatives 
       occurs when passing from assumption \eqref{eq:initial.condition}
       to  assumption  $\RR^{(0)}\les \II^{(0)}$ in the main theorem,  
       no further derivative losses  occurs in \eqref{eq:estimates.thm.main}.      
 \newline
         \noindent   {\bf Remark 2.}\,    By contrast with \cite{Chr:book}, where 
       two derivatives of the curvature and up to three
       derivatives of the Ricci coefficients are needed,  here we need 
       only one derivative of the curvature and two of 
       the Ricci coefficients. This is due to our new
       refined estimates  for the  deformation tensor of the angular
       momentum vectorfields  $O$. As mentioned above these  vectorfields are needed to derive estimates for the angular derivatives of the null curvature components. These new estimates for the deformation tensor of the angular
       momentum vectorfields  $O$ are based on the {\it renormalized} estimates for the Ricci coefficients developed
       in Section \ref{sec:ren}. Together with the trace estimates for the curvature components, which serve as a 
       replacement for the failed $H^1(S)\subset L^\infty(S)$ embedding on a 2-dimensional surface $S$, proved 
       in Section \ref{sec:trace}, they allow us to limit the degree of differentiability required in the proof to the 
       $L^2$ norms of curvature and its first derivatives.
       Similar ideas related to the gain of differentiability via renormalization and trace estimates were exploited in
       our earlier work \cite{KR:causal}.
   \vskip 1pc

 Our next and final  result concerns the formation of a
  {\it pre-scar} in an  angular sector of size $\de^{\frac 12}$.
 \begin{theorem}\label{thm:scar}
 Let $\ep$ be  a small parameter such that $0<\de\ll\ep$. Assume that the initial data $\chih_0$ satisfies
 \beaa
  \de^{1/2}  \|\chih_0\|_{L^\infty} +
  \sum_{0\le k\le 1} \, \sum_{0\le m\le 4} \ep\|(\ep^{-1}\de^{\frac 12}\nab)^{m} (\de\nab_4 )^k\, \chih_0\|_{L^2(0,\ub)}<\infty
  \eeaa
  and that the lower bound in \eqref{eq:initial.trapped} is verified in angular sector 
  $\om\in \Lambda$ of size $\de^{\frac 12}$.
  Then, for $\de>0$ sufficiently small, a pre-scar must form in 
   the slab $\DD(u\approx 1, \de)$, i.e.  the expansion  scalar $\trch(u,\ub,\om)$  becomes stricly negative for some values of $u\approx 1$,
   $\ub=\de$ and all $\om \in \Lambda$.
 \end{theorem}

         \noindent   {\bf Remark. }\,             
       Theorem \ref{thm:scar} corresponds to the initial data consistent with the ansatz 
       $$
       \chih_0(\ub,\om)= \de^{-\frac 12}f_0( \de^{-1} \ub,\, \de^{-1/2} \ep\,\om) 
       $$
       and localized in an angular sector of size $\de^{\frac 12} \ep^{-1}$. This should be compared with the data discussed in
       \eqref{eq:ansatz0}. As in Theorem \ref{thm.add} additional smallness provided by the parameter $\ep$ is only 
       needed to guarantee the formation of a pre-scar but not required for the proof of the existence result. A direct comparison shows that the data of Theorem  \ref{thm:scar} is significantly more regular than that of Theorems 
       \ref{thm.main} and \ref{thm.add}. In particular, it essentially corresponds to the $\de$-coherent assumptions,
       consistent with the natural null parabolic scaling discussed in \eqref{Scale.curv.initial0}. Thus the  proof of Theorem
       \ref{thm:scar} is significantly easier than that of our main result  and will be omitted.

  \subsection{Strategy of the proof}
       We divide  proof of the main theorem   in three parts. In the first part
       we  derive estimates for the Ricci coefficients norms   $\OO$ in terms of 
       the initial data $\II^{(0)}$ and the curvature norms $\RR$.  More precisely
       we prove:
       \begin{theorem}[Theorem A] 
       \label{theoremA} Assume that $\OO^{(0)}<\infty $  and  $\RR<\infty$.
       There exists a constant $C$ depending only on $ \OO^{(0)}$ and $\RR,\RRb$
       such that,
       \bea
       \OO\les C( \OO^{(0)}, \RR,\RRb). \label{eq:mainestimateOO}
       \eea
       Moreover, 
       \bea
       \OS_{0,4}[\chibh]\les\OO^{(0)}+ C( \II^{(0)}, \RR,\,  \RRb)\,  \de^{1/4} 
       \label{precise.estim.chib}
       \eea
       \end{theorem}
       We prove the theorem by a bootstrap argument. We start by assuming
  that    there exists a  sufficiently large constant $\De_0$ such that,
    \bea
  \OS_{0,\infty}\le \De_0.\label{bootstrap:Linfty}
  \eea 
  Based on this assumption we  show  that, if $\de$ is sufficiently small, estimate   
  \eqref{eq:mainestimateOO}  also holds.   This allows us to derive a better estimate
  than  \eqref{bootstrap:Linfty}.

       In the second part we need to  define  angular momentum operators $O$
       and show that their deformation tensors verify   compatible  estimates,
        stated in Theorem B, at the end of section \ref{sect:deformation} .
       
            Finally in the last and main part we need to use the estimates
       of Theorems  A and B  to  derive estimates for the 
       curvature norms $\RR$  and thus  end the proof of the main theorem.  
       
           \begin{theorem}[Theorem C]   There exists $\de$ sufficiently small
           such that,
           \bea
           \RR+\RRb\les \II_0
           \eea
           
       \end{theorem}
       
       Theorem C is proved in sections \ref{sect:CurvI} and \ref{sect:CurvII}.

   \subsection{Signature and Scaling}
   Our  norms are  intimately tied with a natural 
   scaling which we introduce below.  
      
\noindent {\sl Signature}.\,    To every  null curvature
component $\a,\b,\rho,\si, \bb,\aa$,   null Ricci coefficients
components $\chi,\ze,\eta, \etab, \om,\omb$, and metric $\ga$  we assign a signature  according to the following rule:
\bea
\sgn(\phi)=1\c  N_4(\phi)+\frac 12 \c N_a(\phi)+0\c N_3(\phi)    - 1
\eea
where $N_4(\phi), N_3(\phi), N_a(\phi)$ denote  the number of  times $e_4$,
respectively $e_3$ and $(e_a)_{a=1,2}$, which  appears in the
definition of $\phi$.
 Thus,
\beaa
\sgn(\a)=2,\quad  \sgn (\b)=1+1/2, \quad  \sgn(\rho, \si)=1,\quad  \sgn(\bb)=1/2,\quad 
\sgn(\aa)=0.
\eeaa
Also,
\beaa
\sgn(\chi)=\sgn( \om)=1,\quad \sgn(\ze, \eta, \etab)=1/2,\quad \sgn(\chib)=\sgn(\omb)=\sgn(\ga)= 0.
\eeaa
Consistent with this definition we have, for any given null component $\phi$,
\beaa
\sgn(\nab_4\phi)=1+\sgn(\phi),\quad \sgn(\nab\phi)=\frac 1 2 +\sgn(\phi),\quad
\sgn(\nab_{3}\phi)=\sgn(\phi).
\eeaa
Also, based on our convention,
$$
\sgn(\phi_1\cdot\phi_2)=\sgn (\phi_1)+\sgn(\phi_2).
$$
{\bf Remark. }\,\, All terms in a given   null structure or  null Bianchi identity
 (see equations  \eqref{null.str1}--\eqref{eq:null.Bianchi})
have the same overall signature.  

We   now   introduce  a  notion of  scale for   any  quantity $\phi$
which has a signature $\sgn(\phi)$, in particular for  
 our basic null curvature  quantities  $\a,\b,\rho,\si, \bb,\aa$ and   null Ricci coefficients  components $\chi,\ze,\eta, \etab, \om,\omb$.  This scaling plays a fundamental role in our work.
\begin{definition}
For an arbitrary  horizontal tensor-field  $\phi$, with a well defined 
signature $\sgn(\phi)$,     we set:
\bea
\sc(\phi)=-\sgn(\phi)+\frac 12 \label{def:scaling1}
\eea
\end{definition}
Observe that
$
\sc(\nab_L\phi)=\sc(\phi)-1,\, 
\sc(\nab\phi)=\sc(\phi)-\frac 1 2,\,  \sc(\nab_{\Lb}\phi)=\sc(\phi)
$.
For a given product of two horizontal tensor-fields we have,
\bea
\sc(\phi_1\cdot\phi_2)=\sc(\phi_1)+\sc(\phi_2)-\frac 12 \label{product.scale}
\eea
\subsection{Scale invariant   norms}
 For any horizontal tensor-field $\psi$  with scale $\sc(\psi)$ we define the following 
 scale invariant norms along the null hypersurfaces $H=H_u^{(0,\de)}$ and $\underline{H}=\underline{H}_{\ub}^{(0,1)}$.
 \bea
\|\psi\|_{\Lsc^2(H)}&=&\de^{-\sc(\psi)-1} \|\psi\|_{L^2(H)},\quad 
\|\psi\|_{\Lsc^2(\underline{H})  }=\de^{-\sc(\psi)-\frac 1 2 } \|\psi\|_{L^2(\underline{H})}\nn\\
\eea
We also define  the scale invariant norms on the $2$ surfaces $S=S_{u,\ub}$,
\bea
\|\psi\|_{\Lsc^p(S)}&=&\de^{-\sc(\psi)-\frac 1 p } \|\psi\|_{L^p(S)}
\eea
In particular,
\beaa
\|\psi\|_{\Lsc^2(S)}&=&\de^{-\sc(\psi)-\frac 1 2} \|\psi\|_{L^2(S)},\qquad 
\|\psi\|_{\Lsc^\infty(S) }=\de^{-\sc(\psi) }\|\psi\|_{L^\infty(S)  }
\eeaa
Observe that we have,
\bea
\|\psi\|_{\Lsc^2(H_u^{(0,\ub)})}^2   &=&\de^{-1} \int_0^{\ub}\|\psi\|_{\Lsc^2(u,\ub')} ^2d\ub',\qquad 
\|\psi\|_{\Lsc^2(\Hb_{\,\ub}^{(0,u)})}^2=\int_0^{u}\|\psi\|_{\Lsc^2(u',\ub)}^2 du'
\eea
We  denote  the scale invariant  $L^\infty$  norm in $\DD$   by  $ \|\psi\|_{\Lsc^\infty}$.  

\NI{\bf Remark.}\quad Observe that the noms above are scale invariant if we take into account  the scales of the $L^2$ noms along $H$ and $\underline{H}$,  given by,
\beaa
\sc(\,\|\,\,\|_{L^2(H^{0,\de}_u)}\,)=1,\quad \sc(\,\|\,\,\|_{L^2(\underline{H}_{\ub}^{0,1})}\,)=\frac 1 2, 
\quad \sc( \,\|\,\,\|_{L^p(S)}\,)=\frac 1 p.
\eeaa
Moreover  they  are consistent    to the  following convention, 
$$
\nab_4\sim \de^{-1},\quad \nabla \sim \de^{-\frac 12}, \quad \nab_3\sim 1
$$
  
In view of  \eqref{product.scale} all standard  product estimates  in the usual $L^p$
spaces translate into product estimates in $\Lsc$ spaces with a gain of $\de^{1/2}$.
Thus, for example,
\bea
\|\psi_1\c\psi_2\|_{\Lsc^2(S)} &\les&\de^{1/2} \|\psi_1\|_{\Lsc^\infty(S)}\c \|\psi_2\|_{\Lsc^2(S)}
\eea
or,
\bea
\|\psi_1\c\psi_2\|_{\Lsc^2(H)} &\les&\de^{1/2} \|\psi_1\|_{\Lsc^\infty(H)}\c \|\psi_2\|_{\Lsc^2(H)}\label{product.inv.estim}
\eea
\begin{remark}
\label{Rem.consts}   If $f $ is a scalar function constant along the surfaces $S(u,\ub)\subset \DD$,
we have  
\beaa
\| f \c  \psi\|_{\Lsc^p(S)} \les    \|\psi\|_{\Lsc^p(S)}
\eeaa
or,  if $ f$ is also  bounded on $H$,
\beaa
\| f \c  \psi\|_{\Lsc^2(H)} \les    \|\psi\|_{\Lsc^2(H)}
\eeaa
This remark applies  in particular to the constant   $\trchb_0=\frac{4}{2r_0+\ub-u}$. 
\end{remark}

We can reinterpret  our main curvature and Ricci
 coefficient norms in light of the scale invariant norms.
 Thus \eqref{eq:curv.norms} can be rewritten in the 
 form\footnote{We use the short hand notation     $\|(\b,\rho,\si,\bb)\|_{\Lsc^2(H_u^{(0,\ub)})}=  \|\b\|_{\Lsc^2(H_u^{(0,\ub)})}+\|\rho\|_{\Lsc^2(H_u^{(0,\ub)})} +\|\si\|_{\Lsc^2(H_u^{(0,\ub)})}+\ldots$},
 \beaa
  \RR_0(u,\ub):&=&\de^{1/2}\|\a\|_{\Lsc^2(H_u^{(0,\ub)})}+\|(\b,\rho,\si,\bb)\|_{\Lsc^2(H_u^{(0,\ub)})}\\
   \RR_1(u,\ub):&=&\de^{1/2}\|\nab_4\a\|_{\Lsc^2(H_u^{(0,\ub)})}+\|\nab( \a, \b,\rho,\si, \bb)\|_{\Lsc^2(H_u^{(0,\ub))})}\\
   \\
    \RRb_0(u,\ub):&=&\de^{1/2}\|\b\|_{\Lsc^2(\Hb_{\ub}^{(u,0)})}+\|(\rho,\si,\bb,\aa)\|_{\Lsc^2(\Hb_{\ub}^{(0,u)})}\\
   \RR_1(u,\ub):&=& \|\nab_3\aa\|_{\Lsc^2(\Hb_{\ub}^{(u,0)})}+         \|\nab(  \b,\rho,\si, \bb,\aa)\|_{\Lsc^2(\Hb_{\ub}^{(0,u)})}
    \eeaa
  \begin{remark}
    All  curvature norms are scale invariant  except 
    for the anomalous  $\|\a\|_{\Lsc^2(H_u^{(0,\ub)})}$, $\|\nab_4\a\|_{\Lsc^2(H_u^{(0,\ub)})}$  and
     $\|\b\|_{\Lsc^2(\Hb_{\ub}^{(u,0)})}$.   By abuse of language,   in a given context, we   refer to $\a$, respectively $\b$, as anomalous.
    \end{remark}  
    To rectify the anomaly of $\a$ we introduce an additional scale-invariant norm 
$$
\RR_0^\de[\a] := \sup_{^\de H\subset H} \|\a\|_{\Lsc^2(^\de H)},
$$
where $^\de H$ is a piece of the hypersurface $H=H_u^{0,\de}$ obtained by evolving a disc
$S_\de\subset S_{u,0}$ of radius $\de^{\frac 12}$ along the integral curves of the vectorfield $e_4$.

 The Ricci
 coefficient norms \eqref{eq:Ricci.norms}   can be written,
 \beaa
 \OS_{0,\infty}(u,\ub)&=&\|(\chih,\om ,\eta,\etab,\trchbt,\chibh, \omb)\|_{\Lsc^\infty(S)}\\
 \OS_{0,4}(u,\ub)&=&\de^{1/4}\big(\|\chih\|_{\Lsc^4(S) }
 +\|\chibh\|_{\Lsc^4(S)}\big)\\
 &+&\|(\trch, \om,\eta,\etab,\trchbt, \omb)\|_{\Lsc^2(S)}\\
  \OS_{1,4}(u,\ub)&=&\|\nab(\chi, \om,\eta,\etab,\trchbt, \chibh, \omb)\|_{\Lsc^4(S)}\\
   \OS_{1,2}(u,\ub)&=&\|\nab(\chi, \om,\eta,\etab,\trchbt, \chibh, \omb)\|_{\Lsc^2(S)}\\
  \OH(u,\ub)&=&\|\nab^2(\chi, \om,\eta,\etab,\trchbt, \chibh, \omb)\|_{\Lsc^2(H_u^{(0,\ub)})}  
 \eeaa
\begin{remark}
All quantities are scale invariant except for 
 $\chih, \chibh$ in the $\Lsc^4(S)$
 norm.
 \end{remark} 
 As  before   we  complement  the anomalous norms  for $\chih, \chibh$ by the local, non-anomalous, scale-invariant norms 
\beaa
\OO^\de_0[\chih](u,\ub)=\sup_{^\de S\subset S}\|\chih\|_{\Lsc^4(^\de S)},\quad 
\OO^\de_0[\chibh](u,\ub)=\sup_{^\de S\subset S}\|\chibh\|_{\Lsc^4(^\de S)},
\eeaa
where $^\de S$ is a disk of radius $\de^{\frac 12}$ obtained by transporting from the initial data 
embedded in $S_{u,0}$.

   \section{Main equations. Preliminaries}
   \subsection{Null structure equations}
   \label{sub:nullstr}
We recall  the  null structure equations (see section 3.1 in  \cite{KNI:book}  or \cite{Chr:book}.)
\begin{equation}
\label{null.str1}
\begin{split}
\nab_4\chi&=-\chi\c\chi-2\om \chi-\a\\
\nab_3\chib&=-\chib\c\chib-2\omb \chib -\aa\\
\nab_4\eta&=-\chi\c(\eta-\etab)-\b\\
\nab_3\etab &=-\chib\c (\etab-\eta)+\bb\\
\nab_4\omb&=2\omega\omegab+\frac 34 |\eta-\etab|^2-\frac 14 (\eta-\etab)\cdot (\eta+\etab)-
\frac 18 |\eta+\etab|^2+\frac 12 \rho\\
\nab_3\om&=2\omega\omegab+\frac 34 |\eta-\etab|^2+\frac 14 (\eta-\etab)\cdot (\eta+\etab)- \frac 18 |\eta+\etab|^2+\frac 12 \rho\\
\end{split}
\end{equation}
and the constraint equations
\begin{equation}
\label{null.str2}
\begin{split}
\div\chih&=\frac 12 \nab\trch - \frac 12 (\eta-\etab)\cdot (\chih -\frac 1 2 \trch) -\beta,\\
\div\chibh&=\frac 12 \nab\trchb + \frac 12 (\eta-\etab)\cdot (\chibh-\frac 1 2   \trchb) +\betab\\
\curl\eta &=-\curl\etab=\si +\chibh \wedge\chih\\
K&=-\rho+\frac 1 2 \chih\c\chibh-\frac 1 4 \trch\c\trchb
\end{split}
\end{equation}
with $K$ the Gauss curvature of the surfaces $S$.
The first two equation in \eqref{null.str1} can also be written in the form,
\begin{equation}
\label{null.str3}
\begin{split}
\nab_4 \trch+\frac 12 (\trch)^2&=-|\chih|^2-2\om \trch\\
\nab_4\chih+\trch \chih&=-2 \omega \chih-\alpha\\
\nab_3\trchb+\frac 12 (\trchb)^2&=-2\omegab \trchb-|\chibh|^2\\
\nab_3\chibh + \trchb\,  \chibh&= -2\omegab \chibh -\alphab
\end{split}
\end{equation}
Also, with $\rhoc=\rho-\frac 1 2 \chih\c\chibh$,
\begin{equation}
\label{null.str4}
\begin{split}
\nab_4\trchb+\frac1 2 \trch\trchb &=2\om\trchb +2\rhoc +2\div \etab +2|\etab|^2\\
\nab_3\trch+\frac1 2 \trchb\trch &=2\omb \trch+2\rhoc +2\div \eta+2|\eta|^2
\end{split}
\end{equation}
and\footnote{Recall  the notation  $(u\hot v)_{ ab}=u_a v_b+u_b v_a- 
(u\c v) \de_{ab}$.},  
\begin{equation}
\label{null.str5}
\begin{split}
\nab_3\chih+\frac 1 2 \trchb \chih&=\nab\widehat{\otimes} \eta+2\omb \chih-\frac 12 \trch \chibh +\eta\widehat{\otimes} \eta\\
\nab_4\chibh +\frac 1 2 \trch \chibh&=\nab\widehat{\otimes} \etab+2\om \chibh-\frac 12 \trchb \chih +\etab\widehat{\otimes} \etab
\end{split}
\end{equation}
{\bf Remark}.
The transport equations for $\om $ and $\omb$ 
in \eqref{null.str1}  are obtained from  the null structure 
equation,   
\beaa
\nab_4 \omb+\nab_3 \om =\ze\c (\eta-\etab)-\eta\c\etab +4 \om\omb +\rho
\eeaa 
and the commutation relation, for a scalar $f$ (see proposition 4.8.1
in \cite{KNI:book})
\bea
[\nab_3, \nab_4] f =-2\om \nab_3 f+2\omb \nab_4 f +4 \ze\c  \nab f \label{eq:34}
\eea
applied to $f=\log\Om$.
\subsection{Null Bianchi}
We  record below the null  Bianchi identities (Observe that we  can eliminate  $\ze=\frac 1 2 (\eta-\etab) $ in the equations below), 
\begin{equation}
\label{eq:null.Bianchi}
\begin{split}
&\nab_3\alpha+\frac 12 \trchb \alpha=\nabla\hot \beta+ 4\omegab\alpha-3(\chih\rho+^*\chih\sigma)+
(\zeta+4\eta)\hot\beta,\\
&\nab_4\beta+2\trch\beta = \div\alpha - 2\omega\beta +  \eta \alpha,\\
&\nab_3\beta+\trchb\beta=\nabla\rho + 2\omegab \beta +^*\nabla\sigma +2\chih\cdot\betab+3(\eta\rho+^*\eta\sigma),\\
&\nab_4\sigma+\frac 32\trch\sigma=-\div^*\beta+\frac 12\chibh\cdot ^*\alpha-\zeta\cdot^*\beta-2\etab\cdot
^*\beta,\\
&\nab_3\sigma+\frac 32\trchb\sigma=-\div ^*\betab+\frac 12\chih\cdot ^*\alphab-\zeta\cdot ^*\betab-2\eta\cdot 
^*\betab,\\
&\nab_4\rho+\frac 32\trch\rho=\div\beta-\frac 12\chibh\cdot\alpha+\zeta\cdot\beta+2\etab\cdot\beta,\\
&\nab_3\rho+\frac 32\trchb\rho=-\div\betab- \frac 12\chih\cdot\alphab+\zeta\cdot\betab-2\eta\cdot\betab,\\
&\nab_4\betab+\trch\betab=-\nabla\rho +^*\nabla\sigma+ 2\omega\betab +2\chibh\cdot\beta-3(\etab\rho-^*\etab\sigma),\\
&\nab_3\betab+2\trchb\betab=-\div\alphab-2\omegab\betab+\etab \cdot\alphab,\\
&\nab_4\alphab+\frac 12 \trch\alphab=-\nabla\hot \betab+ 4\omega\alphab-3(\chibh\rho-^*\chibh\sigma)+
(\zeta-4\etab)\hot \betab
\end{split}
\end{equation}

We record   below  commutation  formulae between $\nab$ and $\nab_4, \nab_3$:
\begin{lemma}
\label{le:comm}
For a scalar function $f$:
\bea
\,[\nab_4,\nab] f&=&\frac 12 (\eta+\etab) D_4 f -\chi\cdot \nab f\label{comm:nabnab4}\\
\,[\nab_3,\nab] f&=&\frac 12 (\eta+\etab) D_3 f -\chib\cdot \nab f,
\label{comm:nabnab3}
\eea

For a 1-form  tangent to $S$:
\beaa
\,[D_4,\nab_a] U_b&=& -\chi_{ac} \nab_c U_b +\in_{ac}\, ^*\beta_b U_c   
+\frac 12(\eta_a+\etab_a) D_4 U_b \\
&-&-\chi_{ac} \,\etab_b\,  U_c +\chi_{ab} \,\etab\cdot U\\
\,[D_3,\nab_a] U_b&=&-\chib_{ac} \nab_c U_b    +\in_{ac} *\betab_b U_c +
\frac 12(\eta_a+\etab_a) D_3 U_b   \\
 &-& \chib_{ac}  \eta_b\,U_c+\chib_{ab} \, \eta\cdot U \\
\eeaa
In particular,
\beaa
\,[\nab_4,\div] U&=&-\frac 1 2 \trch  \, \div U-\chih\c \nab U   -\b\c U +\frac 1 2 (\eta+\etab)\c \nab_4 U\\
&-&   \etab\c \chih\c U -\frac 1 2 \trch \etab\c U +\trch \etab\c U\\
\,[\nab_3,\div] U&=&-\frac 1 2 \trchb  \, \div U-\chibh\c \nab U +  \bb\c U +\frac 1 2 (\eta+\etab)\c \nab_4 U\\
&-&   \eta\c \chih\c U -\frac 1 2 \trchb \eta\c U +\trchb \eta\c U
\eeaa

\end{lemma}
\subsection{Integral formulas}
Given a scalar function $f$ in $\DD$ we have\footnote{see for example  Lemma 3.1.3 in \cite{KNI:book}},
\beaa
\frac{d}{ d\ub} \int_{S(u,\ub)} f &=& \int_{S(u,\ub)} \big(\frac{df}{d\ub}+ \Om \trch  f\big)= \int_{S(u,\ub)} \Om \big(e_4(f)+ \trch  f\big)\\
\frac{d}{ du } \int_{S(u,\ub)} f &=&  \int_{S(u,\ub)} \big(\frac{df}{du}+\Om \trchb  f\big)= \int_{S(u,\ub)}\Om \big(e_3(f)+ \trchb  f\big)
\eeaa

As a consequence of these  we deduce, for  any horizontal  tensorfield $\psi$,
\begin{equation}
\label{eq:integr.form}
\begin{split}
\|\ \psi\|^2_{L^2(S(u, \ub))   }&=\| \psi\|^2_{L^2(S(u,0) )}+\int_{H^{(0,\ub)}_u} \,\, 2\Om \big( \psi\c \nab_4\psi +\frac 1 2  \trch  |\psi|^2\big)\\
\|\ \psi\|^2_{L^2(S(u, \ub))   }&=\| \psi\|^2_{L^2(S(0,\ub) )}+\int_{H^{(u,0)}_{\ub}} \,\, 2\Om \big(\psi\c\nab_3\psi +\frac 1 2  \trchb  |\psi|^2\big)
\end{split}
\end{equation}
\begin{proof}
The first formula in \eqref{eq:integr.form}  is derived as follows,
\beaa
\|\ \psi\|^2_{L^2(S(u, \ub))   }&=&\| \psi\|^2_{L^2(S(u,0) )}+\int_0^{\ub}\frac{d}{ d\ub} \big(\int_{S(u,\ub)} |\psi|^2\big)\\
&=&  \| \psi\|^2_{L^2(S(u,0) )}+      \int_{H^{(0,\ub)}_u} \,\,  2\Om \big(\psi\c\nab_4\psi +\frac 1 2  \trch  |\psi|^2\big)
\eeaa
The second formula is proved in the same manner.
\end{proof}

\subsection{Hodge systems}\label{sect:hodge}
We work with  the following Hodge operators acting on  the leaves $S=S(u,\ub)$ of our double null foliation.
\begin{enumerate}
\item The operator $\dcal $ takes any $1$-form $  F  $ into the pairs of
functions $(\div  F  \,,\, \curl  F  )$
\item The operator $\dcall$ takes any $S$ tangent symmetric, traceless
tensor $  F  $ into the $S$ tangent one form $\div  F  $.
\item The operator $\dcalll$ takes the pair of scalar functions 
$(\rho, \si)$ into the $S$-tangent 1-form\footnote{Here $(\nabd \,\si)_a=\in_{ab}\nab_b \si$. } $-\nab \rho+\nabd \,\si$.
\item The operator $\dcallll$ takes 1-forms $F$  on $S$ into  the 2-covariant, symmetric,
traceless tensors $-\f12 \widehat{\Lie_F\ga }$ with $\Lie_F\ga$ the traceless part of the 
Lie derivative of the metric $\ga$ relative to $F$, i.e.
$$\widehat{(\Lie_F\ga)}_{ab}=\nab_b F_a+\nab_a F_b-(\div F)\ga_{ab}.$$

\end{enumerate}
The kernels of both $\dcal$ and $\dcall$ in $L^2(S)$ are 
trivial and that $\dcalll$, resp. $\dcallll$ are  the $L^2$ adjoints of
 $\dcal$, respectively $\dcall$. The
kernel of $\dcalll$ consists of pairs of  constant functions $(\rho, \si)$ while
that of $\dcallll$ consists 
of the set of all conformal Killing vectorfields on $S$. In particular  the  $L^2$- range 
of $\dcal$ consists of all pairs of functions $\rho, \si$ on $S$ with
vanishing mean. The $L^2$ range of $\dcall$ consists of all
$L^2$ integrable 1-forms on $S$ which are orthogonal
to the Lie algebra of  all conformal Killing vectorfields
on $S$. Accordingly we shall
consider the inverse operators $\dcal^{-1}$ and $\dcall^{-1}$ and
implicitly assume that they are defined on the  $L^2$ subspaces 
identified above.

 Finally we record the following simple
identities,
\bea
\dcalll\c\dcal&=&-\lap+K,\qquad \dcal\c\dcalll=-\lap\label{eq:dcalident}\\
\dcallll\c\dcall&=&-\f12\lap+K,\qquad \dcall\c\dcallll=-\f12(\lap+K)\label{eq:dcallident}
\eea
\begin{proposition} 
\label{prop:hodge}
Let $(S,\ga)$ be
a compact manifold with Gauss curvature $K$.

{\bf i.)}\quad The following identity holds for vectorfields  $  \psi  $ 
 on $S$:
\be{eq:hodgeident1}
\int_S\big(|\nab   \psi |^2+K|  \psi  |^2\big)=\int_S\big( |\div   \psi  |^2+|\curl  \psi  |^2\big)=\int_S|\dcal  \psi  |^2
\end{equation} 

{\bf ii.)}\quad The following identity holds for symmetric, traceless,
 2-tensorfields   $  \psi  $ 
 on $S$:
\be{eq:hodgeident2}
\int_S\big(|\nab   \psi  |^2+2K|  \psi  |^2\big)=2\int_S |\div   \psi |^2=2\int_S |\dcall   \psi  |^2
\end{equation} 

{\bf iii.)}\quad The following identity holds for pairs of functions $(\rho,\si)$
 on $S$:
\be{eq:hodgeident3}
\int_S\big(|\nab \rho|^2+|\nab\si|^2\big)=\int_S |-\nab\rho+(\nab\si)^\star |^2=\int_S|\dcalll
(\rho,\si)|^2
\end{equation} 

{\bf iv.)}\quad  The following identity holds for vectors $  \psi  $ on $S$,
\be{eq:hodgeident3*}
\int_S \big(|\nab  \psi  |^2-K|  \psi  |^2\big)=2\int_S|\dcallll   \psi  |^2
\end{equation}

\label{prop:hodgeident}
\end{proposition}
\section{Preliminary estimates}
As explained in the introduction the proof  of Theorem A is based on the bootstrap assumption
\eqref{bootstrap:Linfty}, i.e.
\beaa
 \OS_{0,\infty}\le \De_0.
\eeaa
 In this section we use this bootstrap  to prove  various preliminary results.
 In the following three sections we then derive 
estimates for the Ricci coefficient  norms $\OS_{0, 4}$, 
 $\OS_{1, 2}$ and $\OS_{1,4}$ respectively.

\subsection{Preliminary results} We prove here  results
which follows easily from our bootstrap assumption.
 $\OS_{0,\infty}\le \De_0$.
We  first derive an estimate for $\Om$.  To do this we 
use the definition of $\omb=-\frac 1 2 \nab_3\log  \Om=\frac 1 2  \Om \nab_3(\Om)^{-1}=\frac 1 2 \frac{d}{du}(\Om)^{-1} $. Thus,
since $\Om^{-1}=2$ on  $H_0$,
\beaa
\|\Om^{-1}- 2\|_{L^\infty(u,\ub) }&\les& \int_0^u \|\omb\|_{L^\infty(u',\ub)} du'
\les \de^{1/2}\, \OS_{0,\infty}[\omb]\les \de^{1/2} \De_0
\eeaa
Thus, if $\de $ is sufficiently small  we deduce  that $|\Om-\frac 1 2| $ is small and therefore,
\bea
\frac 1 4 \le \Om\le 4.\label{estim-Om}
\eea
 We now prove the following proposition.
\begin{proposition}
\label{prop:transp}
Under  assumption \eqref{bootstrap:Linfty} we have the following estimates 
for an arbitrary horizontal tensor-field $\psi$,
\begin{equation}
\label{eq:transp.Hu}
\begin{split}
\|\ \psi\|_{L^2(u, \ub)   }&\les \| \psi\|_{L^2(u,0 )}+\int_0^{\ub}  \|   \nab_4\psi   \|_{L^2(u,\ub'  ) }\,d\ub'\\
\|\ \psi\|_{L^2(u, \ub)   }&\les \| \psi\|_{L^2(0,\ub  )}+\int_0^{u}  \|   \nab_3\psi  \|_{L^2(u',\ub  ) } \, du'
\end{split}
\end{equation}
More generally the same estimates hold  in $L^p(S)$ norms.

Also,
\begin{equation}
\label{eq:transp.Hu-qr}
\begin{split}
\|\psi\|_{L^2(u, \ub)}^2&\les \| \psi\|_{L^2(u,0 )}^2+\|\psi\|_{L^2(H_u^{(0,\ub) }) }\|\nab_4 \psi\|_{L^2(H_u^{(0,\ub)})}\\
\|\psi\|_{L^2(u, \ub)}^2&\les \| \psi\|_{L^2(0,\ub  )}^2+\|\psi\|_{L^2(\Hb_{\,\ub}^{(0,u )}) }\|\nab_4 \psi\|_{L^2(\Hb_{\, \ub}^{(0, u)}  )}
\end{split}
\end{equation}
\end{proposition}
\begin{corollary} Under the same hypothesis,
\begin{equation}
\label{eq:transp.Hu.sc}
\begin{split}
\|\ \psi\|_{\Lsc^2(u, \ub)   }&\les \| \psi\|_{\Lsc^2(u,0 )}+\int_0^{\ub} \de^{-1} \|   \nab_4\psi   \|_{\Lsc^2(u,\ub'  ) }\,d\ub'\\
\|\ \psi\|_{\Lsc^2(u, \ub)   }&\les \| \psi\|_{\Lsc^2(0,\ub  )}+\int_0^{u}  \|   \nab_3\psi  \|_{\Lsc^2(u',\ub  ) } \, du'
\end{split}
\end{equation}
and,
\begin{equation}
\label{eq:transp.Hu-qr.sc}
\begin{split}
\|\psi\|_{\Lsc^2(u, \ub)}^2&\les \| \psi\|_{\Lsc^2(u,0 )}^2+\|\psi\|_{\Lsc^2(H_u^{(0,\ub) }) }\|\nab_4 \psi\|_{\Lsc^2(H_u^{(0,\ub)})}\\
\|\psi\|_{\Lsc^2(u, \ub)}^2&\les \| \psi\|_{\Lsc^2(0,\ub  )}^2+\|\psi\|_{\Lsc^2(\Hb_{\,\ub}^{(0,u )}) }\|\nab_4 \psi\|_{\Lsc^2(\Hb_{\, \ub}^{(0, u)}  )}
\end{split}
\end{equation}
More generally, let $S'\subset S_{u,\ub}$ and $S'_{u',\ub}$, $S'_{u,\ub'}$ are obtained by evolving $S'$ 
along the null generators of $\underline H_{\underline u}$, $H_u$ respectively. Then 
\begin{equation}
\label{eq:Hu.sc-loc}
\begin{split}
\|\ \psi\|_{\Lsc^p(S')   }&\les \| \psi\|_{\Lsc^p(S'_{u,0} )}+\int_0^{\ub} \de^{-1} \|   \nab_4\psi +\frac 1p\trch \psi  \|_{\Lsc^p(S'_{u,\ub'}  ) }\,d\ub'\\
\|\ \psi\|_{\Lsc^p(S')   }&\les \| \psi\|_{\Lsc^p(S'_{0,\ub}  )}+\int_0^{u}  \|   \nab_3\psi +\frac 1p\trchb\psi \|_{\Lsc^p(S'_{u',\ub}  ) } \, du'
\end{split}
\end{equation}

\end{corollary}
\begin{proof} The corollary follows immediately from the proposition
and definition of the scale invariant norms. The last statement of the corollary follows by applying 
\eqref{eq:transp.Hu.sc} to the function $\chi \psi$, where the cut-off function $\chi$ is first defined on 
$S_{u,\ub}$ as the characteristic function of $S'$ and then extended by solving the transport equations
$\nabla_4\chi=0$ and $\nabla_3\chi=0$.

\NI To prove the proposition
we first make use of \eqref{estim-Om} and
  \eqref{bootstrap:Linfty},
\beaa
\|\trch\|_{L^\infty}&\les &\De_0 \de^{-\frac 12}
\eeaa
and deduce from the first equation in \eqref{eq:integr.form},
\beaa
\|\ \psi\|^2_{L^2(S(u, \ub))   }&\les &\| \psi\|^2_{L^2(S(u,0) )}+\int_0^{\ub} \int_{S(u,\ub') } \,  | \psi\,|\,|   \nab_4\psi   + \frac 1 2  \trch \psi | \\
&\les&\| \psi\|^2_{L^2(S(u,0) )}+\int_0^{\ub}  \|\psi\|_{L^2(S)}\big( \|\nab_4\psi\|_{L^2(S)}+\De_0\de^{-\frac 12} \|\psi\|_{L^2(S)\big)}\\
&\les&\| \psi\|^2_{L^2(S(u,0) )}+\int_0^{\ub}   \|\psi\|_{L^2(S)} \|\nab_4\psi\|_{L^2(S)}+ \De_0\de^{-1/2}   \int_0^{\ub} \|\psi\|_{L^2(S)}^2
\eeaa
Thus, by Gronwall,  since $\ub\le \de$, 
\bea
\| \psi\|^2_{L^2(S(u, \ub))   }&\les &\| \psi\|^2_{L^2(S(u,0) )}+\int_0^{\ub}  \|   \nab_4\psi  \|_{L^2(u,\ub' ) }\c   \|   \psi  \|_{L^2(u,\ub' ) }\,   d\ub'
\eea
from which we easily derive  the $\nab_4$  equations in  both   \eqref{eq:transp.Hu}  and \eqref{eq:transp.Hu-qr}.

To prove the $\nab_3$  estimates  we need to take into account the
anomalous character of  $ \trchb$.
From our bootstrap assumption we deduce (recall that $\trchb_0=-\frac{4}{\ub-u+2r_0}$ is the flat value of $\trchb$) ,
\beaa
\|\trchb -\trchb_0\|_{L^\infty}\les \De_0\de^{1/2}
\eeaa
Thus,
\beaa
\|\ \psi\|^2_{L^2(S(u, \ub))   }&\les &\| \psi\|^2_{L^2(S(0,\ub) )}+\int_0^{u} \int_{S(u',\ub) } \,  | \psi\,|\,|   \nab_3\psi   + \frac 1 2  \trchb \psi | \\
&\les &\| \psi\|^2_{L^2(S(0,\ub) )}+\int_0^u  
\|\psi\|_{L^2(S)}\big(\| \nab_3\psi\|_{L^2(S)}+\De_0\de^{1/2}\|\psi\|_{L^2(S)}\big)\\
&+&  \int_0^u  \|\trchb_0\|_{L^\infty} \|\psi\|_{L^2(S)}^2\\
&\les& | \psi\|^2_{L^2(S(0,\ub) )}+\int_0^u 
\|\psi\|_{L^2(S)}\big(\| \nab_3\psi\|_{L^2(S)}+ (1+\De_0\de^{1/2})\|\psi\|_{L^2(S)}\big)
\eeaa
Thus, using Gronwall and  smallness of   $\de^{1/2}\De_0$  we deduce,
\bea
\|\ \psi\|^2_{L^2(S(u, \ub))   }&\les &\| \psi\|^2_{L^2(S(0,\ub) )}+
\int_0^u \|\psi\|_{L^2(S)}\| \nab_3\psi\|_{L^2(S)}
\eea
from which both  \eqref{eq:transp.Hu}  and \eqref{eq:transp.Hu-qr}
follow. \end{proof}

We next prove an improved estimate
for $\trch$.
\begin{proposition}
For $\de^{1/2}\De_0$ sufficiently small we have for all $S=S(u,\ub)$,
\bea
\|\trch\|_{L^\infty(S)}&\les&\De_0^2  \label{eq:imprtrch}
\eea
\end{proposition}
\begin{proof}
We recall that $\trch$ verifies the transport equation,
\beaa
\nab_4 \trch=-\frac 12 (\trch)^2-|\chih|^2-2\om \trch
\eeaa
or,
\beaa
\frac{d}{d\ub}\trch =-\Om( \frac 12 (\trch)^2+|\chih|^2+2\om \trch\big)
\eeaa
Thus, since $\|\chi, \om \|_{L^\infty}\les  \de^{-1/2}\De_0$, 
\beaa
\|\trch\|_{L^\infty(u,\ub)}&\les&\int_0^{\ub}\|\chi\|_{L^\infty(u,\ub')}\big(\|\chi\|_{L^\infty(u,\ub')}+\|\om\|_{L^\infty(u,\ub')}\big) d\ub'\\
&\les&\De_0^2+\de^{1/2}\De_0.
\eeaa
\end{proof}
\subsection{Transported coordinates}
We define systems of, local, transported coordinates along  the null 
 hypersurfaces $H$ and $\Hb$.  Staring with a local coordinate system $\th=(\th^1, \th^2)$
on   $U\subset  S(u,0)\subset H_u$ we parametrize any point along the null geodesics starting
 in $ U$ by the  the corresponding coordinate $\th$ and affine parameter $\ub$. Similarly,
 starting with a local coordinate system $\thb=(\thb ^1, \thb^2)$  on 
 $V\subset S(0,\ub)\subset \Hb_{\ub}$
 we parametrize  any point   along the null geodesics starting
 in $ V$ by the  the corresponding coordinate $\thb$ and affine parameter $u$.
 We denote 
the respective metric components by $\gamma_{ab}$ and  $\underline\gamma_{ab}$. 
\begin{proposition}\label{prop:gamma}
Let $\gamma^0_{ab}$ denote the standard metric on ${\Bbb S}^2$. Then, 
for any $0\le u\le 1$ and $0\le\ub\le\de$ and sufficiently small $\de^{\frac 12}\Delta_0$

$$
|\gamma_{ab}-\gamma^0_{ab}|\le \de^{\frac 12} \Delta_0,\qquad 
|\underline\gamma_{ab}-\gamma^0_{ab}|\le \de^{\frac 12} \Delta_0.
$$
In addition, the transported coordinates verify
\beaa
\label{eq:3a}
|\nab_3\theta^a| &\les& \de\Delta_0,\qquad |\nab\theta^a|\les 1\\
\|\nab_4\thb^a|&\les&\de\De_0,\qquad |\nab \thb^a|\les 1 
\eeaa
for $a=1,2$. 
The Christoffel symbols $\Gamma_{abc}$ and $\Gab_{ab}$,  obey the scale invariant  estimates\footnote{we can attach signature to $\Gamma$
and $\Gab$
$
sgn(\Gamma)=\frac 12,\qquad sgn(\Gab )=\frac 12
$
} 
\begin{align}
&\|\Gamma_{abc}\|_{\Lsc^2(S)}\les \OO_{[1]},\qquad \|\pr_d \Gamma_{abc}\|_{\Lsc^2(S)}\les\OO_{[2]},\label{eq:crist-sc}\\
&\|\Gab_{abc}\|_{\Lsc^2(S)}\les \OO_{[1]},\qquad \|\pr_d \Gab_{abc}\|_{\Lsc^2(S)}\les\OO_{[2]},\label{eq:cristb-sc}
\end{align}
\end{proposition}
\begin{proof}
We will only show the argument in the case of $\gamma_{ab}$. In the transported coordinate system 
the metric $\gamma_{ab}$ verifies
$$
\frac d{d\underline u} \gamma_{ab}=2\Omega \chi_{ab}.
$$
Therefore,
$$
|\gamma_{ab}-\gamma_{ab}^0|\le 2\int_0^{\ub} |\chi_{ab}| \le \de^{\frac 12} \Delta_0,
$$
where in the last inequality we used that  $|\chi_{ab}|\le |\chi| |\gamma^{-1}|$ and ran 
a simple bootstrap argument.

The transported system of coordinates $\theta^a$ satisfies the system of equations 
$$
\nab_4\theta^a=0.
$$ 
Commuting these equations with $\nab_3$ and taking into account the commutation formula  \eqref{eq:34} we obtain
$$
\nab_4 (\nab_3 \theta^a)=2\om \nab_3 \theta^a-4 \ze\c  \nab \theta^a
$$
Using the bootstrap assumptions \eqref{bootstrap:Linfty}, the inequality $|\nab\theta^a|\les 1$
and the triviality of the data for $\nab_3\theta^a$ we obtain that 
$$
|\nab_3\theta^a|\les \de \De_0.
$$
To verify that $|\nab\theta^a|\les 1$ we commute the transport equation for $\theta^a$ with $\nab$
to obtain according to \eqref{comm:nabnab4}
$$
\nab_4 (\nab \theta^a)=-\chi\cdot \nab \theta^a,
$$
which together with the bootstrap assumption  \eqref{bootstrap:Linfty} gives the desired result.

To prove \eqref{eq:crist-sc} we differentiate the transport equation for $\gamma_{ab}$ to obtain
$$
\frac {d}{d\ub} \left (\pr_{c} \ga_{ab}\right)=2\pr_{c}\Omega \chi_{ab}+
2\Omega \pr_{c}\chi_{ab}.
$$
Taking into account that 
$$
|\pr_{c}\Omega|\les |\nab\Omega|\le |\eta|+|\etab|,\qquad
|\pr_{c}\chi_{ab}|\les |\nab\chi|+|\Gamma||\chi|
$$
we derive 
\beaa
\|\pr_c \ga_{ab}\|_{L^2(u,\ub)} &\les& \int_0^{\ub} \left (\|\eta\|_{L^4((u,\ub')}+\|\etab\|_{L^4(u,\ub')}\right)
\|\chi\|_{L^4(u,\ub')} d\ub'\\ &+& \int_0^{\ub}  \left(\|\nab\chi\|_{L^2(u,\ub')}+\|\Gamma\|_{L^2(u,\ub')}  \right) 
 \|\chi\|_{L^\infty(u,\ub')} d\ub'\\ &\les& \de^{\frac 34} \OS_{0,4}[\chi]\OS_{0,4}[\eta,\etab]+\OS_{1,2}[\chi] +\de^{-\frac 12} \De_0
\int_0^{\ub} \|\Gamma\|_{L^2(u,\ub')} d\ub'.
\eeaa
Thus, by Gronwall,
\beaa
\|\Ga\|_{L^2(u,\ub)}\les  \OS_{1,2} +\de^{3/4}\OS_{0,4}^2
\eeaa
The desired estimate for $\Gamma$ follows by Gronwall. The second estimate of
 \eqref{eq:crist-sc}can be derived 
by an additional differentiation of the transport equation. The estimates \eqref{eq:cristb-sc}
 are proved in the same manner.  We omit the details.
\end{proof}
\subsection{Estimates for $\RR_0^\de[\alpha]$}
Using the transported coordinates of the 
previous subsection  we now derive estimates 
 for  $\RR_0^\de[\alpha]$ norm of the anomalous curvature component $\a$.
\begin{proposition}
$$
\RR_0^\de[\alpha](u)\les \RR_0^\de[\alpha](0)+\RR
$$
\end{proposition}
\begin{proof} Recall that,
$\RR_0^\de[\a] := \sup_{^\de H\subset H} \|\a\|_{\Lsc^2(^\de H)}$, where  
 $^\de H$ is the   subset of $H_u$ generated by transporting a disk $^\de S$ of radius $\de^{\frac 12}$,
embedded in the sphere $S_{u,0}$,  along the integral curves of the vectorfield $e_4$.  We denote by $^\de S_{\ub}$ the intersection between $^\de H$ and
the level hypersurfaces of $\ub$ and by $^\de S_{u',\ub}$ the sets obtained by transporting $^\de S_{\ub}$
along the integral curves of $e_3$
According to \eqref{eq:Hu.sc-loc}
$$
\|\ \a\|_{\Lsc^2(^\de S_{\ub})   }\les \| \a\|_{\Lsc^2(^\de S_{0,\ub}  )}+
\int_0^{u}  \|   \nab_3\a+\frac 12 \trchb\a \|_{\Lsc^2(^\de S_{u',\ub}  ) } \, du'
$$
We note that \eqref{eq:3a} implies that  $^\de S_{u',\ub}$ are contained in the intersection 
of $^{2\de} H_{u'}$ and the level hypersurface of $\ub$. Therefore,
$$
\|\ \a\|_{\Lsc^2(^\de H_u)   }\les \| \a\|_{\Lsc^2(^{2\de} H_0)}+
\int_0^{u}  \|   \nab_3\a+\frac 12\trchb\a \|_{\Lsc^2(^{2\de} H_{u'}  ) } \, du'
$$
Using the equation for $\a$
$$
\nab_3\alpha+\frac 12 \trchb \alpha=\nabla\hot \beta+ 4\omegab\alpha-3(\chih\rho+^*\chih\sigma)+
(\zeta+4\eta)\hot\beta
$$
and the bootstrap assumptions  \eqref{bootstrap:Linfty} we obtain
\begin{align*}
 \|   \nab_3\a+\frac12\trchb\a \|&_{\Lsc^2(^{2\de} H_{u'}  ) }\le  \|   \nab_3\a+\frac 12\trchb\a \|_{\Lsc^2(H_{u'}  ) }\\ &\le
  \|  \nab\b \|_{\Lsc^2(H_{u'}  ) }+\de^{\frac 12} \OS_{0, \infty} \c \RR_0\les \RR+\de^{\frac 12} \De_0\RR_0
\end{align*}
It remains to observe that 
$$
\| \a\|_{\Lsc^2(^{2\de} H_0)}\les \RR_0^\de[\a](u=0),
$$
which follows from a simple covering argument.
\end{proof}
\subsection{Calculus inequalities}
\label{subs.calculus.ineq}
\begin{proposition}\label{prop:interpol}
Let $(S,\gamma)$ be a compact 2-dimensional surface covered by local charts (disks) $U_i$ in which  
the metric $\gamma$ satisfies
$$
|\gamma_{ij}-\de_{ij}|\le \frac 12.
$$
Let $d$ denote the minimum between $1$ and the smallest radius of the disks $U_i$. Then for any $p>2$
\begin{align}
&\|\psi\|_{L^4(S)}\les \|\psi\|_{L^2(S)}^{\frac 12}  \|\nab\psi\|_{L^2(S)}^{\frac 12} +
d^{-\frac 1{2}}  \|\psi\|_{L^2(S)},\label{eq:L4-glob}\\
&\|\psi\|_{L^\infty(S)}\les \|\psi\|_{L^p(S)}^{\frac p{p+4}}  \|\nab\psi\|_{L^p(S)}^{\frac 4{p+4}} +
d^{-\frac 4{p+4}}  \|\psi\|_{L^p(S)}.\label{eq:Linf-glob}
\end{align}
More generally,
\begin{align}
&\|\psi\|_{L^4(U_i)}\les \|\psi\|_{L^2(U_i)}^{\frac 12}  \|\nab\psi\|_{L^2(U_i')}^{\frac 12} +
d^{-\frac 1{2}}  \|\psi\|_{L^2(U_i')},\label{eq:L4-loc}\\
&\|\psi\|_{L^\infty(S)}\les \sup_{U_i} \left (\|\psi\|_{L^p(U_i)}^{\frac p{p+4}}  \|\nab\psi\|_{L^p(U'_i)}^{\frac 4{p+4}} +
d^{-\frac 4{p+4}}  \|\psi\|_{L^p(U'_i)}\right).\label{eq:Linf-loc}
\end{align}
The disk $U_i'$ is a doubled version of $U_i$.
\end{proposition}
We can combine the above proposition with Proposition \ref{prop:gamma} to obtain
\begin{corollary}
Let $S=S_{u,\ub}$ and $S_\de\subset S$ denote a disk of radius $\de^{\frac 12}$ relative to either 
$\theta$ or $\underline\theta$ coordinate system. Then for any horizontal tensor $\psi$
\begin{align}
&\|\psi\|_{L^4(S)}\les \|\psi\|_{L^2(S)}^{\frac 12}  \|\nab\psi\|_{L^2(S)}^{\frac 12} +
 \|\psi\|_{L^2(S)},\label{eq:L4-glob'}\\
&\|\psi\|_{L^\infty(S)}\les \|\psi\|_{L^p(S)}^{\frac p{p+4}}  \|\nab\psi\|_{L^p(S)}^{\frac 4{p+4}} + 
 \|\psi\|_{L^p(S)}.\label{eq:Linf-glob'}
\end{align}
and 
\begin{align}
&\|\psi\|_{L^4(S_\de)}\les \de^{\frac 14}\|\nab\psi\|_{L^2(S_{2\de})}+
\de^{-\frac 1{4}}  \|\psi\|_{L^2(S_{2\de})},\label{eq:L4-loc'}\\
&\|\psi\|_{L^\infty(S)}\les \sup_{S_\de\subset S} \left (\de^{\frac 14} \|\nab\psi\|_{L^4(S_{2\de})}+
\de^{-\frac 14}  \|\psi\|_{L^4(S_{2\de})}\right).\label{eq:Linf-loc'}
\end{align}

\end{corollary} 
Also, in the scale invariant norms
\begin{corollary}\label{cor:interpol}
Let $S=S_{u,\ub}$ and $S_\de\subset S$ denote a disk of radius $\de^{\frac 12}$ relative to either 
$\theta$ or $\underline\theta$ coordinate system. Then for any horizontal tensor $\psi$
\begin{align}
&\|\psi\|_{\Lsc^4(S)}\les \|\psi\|_{\Lsc^2(S)}^{\frac 12}  \|\nab\psi\|_{\Lsc^2(S)}^{\frac 12} +
 \de^{\frac 14}\|\psi\|_{\Lsc^2(S)},\label{eq:L4-glob-sc}\\
&\|\psi\|_{\Lsc^\infty(S)}\les \|\psi\|_{\Lsc^p(S)}^{\frac p{p+4}}  \|\nab\psi\|_{\Lsc^p(S)}^{\frac 4{p+4}} + 
 \de^{\frac 1p}\|\psi\|_{\Lsc^p(S)}.\label{eq:Linf-glob-sc}
\end{align}
and 
\begin{align}
&\|\psi\|_{\Lsc^4(S_\de)}\les \|\nab\psi\|_{\Lsc^2(S_{2\de})}+
\|\psi\|_{\Lsc^2(S_{2\de})},\label{eq:L4-loc-sc}\\
&\|\psi\|_{\Lsc^\infty(S)}\les \sup_{S_\de\subset S} \left (\|\nab\psi\|_{\Lsc^4(S_{2\de})}+
\|\psi\|_{\Lsc^4(S_{2\de})}\right).\label{eq:Linf-loc-sc}
\end{align}
\end{corollary}  
\subsection{ Codimension $1$ trace formulas.}
We will use the  $L^4(S)$ trace formulas\footnote{Our bootstrap assumption are more than enough   to  verify the conditions of validity of these estimates. } along the null hypersurfaces $H$
and $\Hb$, see  \cite{Chr-Kl}, \ \cite{KNI:book}, \cite{KR:LP}.
\begin{lemma}
\label{le:trace.form}
The following formulas hold true for any two sphere $S=S(u,\ub)=H(u)\cup \Hb(\ub) $ and
any horizontal tensor $\psi$
\beaa
\|\psi\|_{L^4(S)}&\les&\big(\|\psi\|_{L^2(H)}+\|\nab\psi\|_{L^2(H)}\big)^{1/2}\big(\|\psi\|_{L^2(H)}+\|\nab_4\psi\|_{L^2(H)}\big)^{1/2}\\
\|\psi\|_{L^4(S)}&\les&\big(\|\psi\|_{L^2(\Hb)}+\|\nab\psi\|_{L^2(\Hb)}\big)^{1/2}\big(\|\psi\|_{L^2(\Hb)}+\|\nab_3\psi\|_{L^2(\Hb)}\big)^{1/2}\
\eeaa
\end{lemma}
 Also, in scale invariant norms,
 \begin{proposition}
 \label{prop.trace.sc}
 The following formulas hold true for  a fixed  $S=S(u,\ub)=H(u)\cap \Hb(\ub)\subset\DD $ and
any horizontal tensor $\psi$
\beaa
\|\psi\|_{\Lsc^4(S)}&\les&\big(\de^{1/2} \|\psi\|_{\Lsc^2(H)}+\|\nab\psi\|_{\Lsc^2(H)}\big)^{1/2}\big(\de^{1/2}\|\psi\|_{\Lsc^2(H)}+\|\nab_4\psi\|_{\Lsc^2(H)}\big)^{1/2}\\
\|\psi\|_{\Lsc^4(S)}&\les&\big(\de^{1/2}\|\psi\|_{\Lsc^2(\Hb)}+\|\nab\psi\|_{\Lsc^2(\Hb)}\big)^{1/2}\big(\de^{1/2} \|\psi\|_{\Lsc^2(\Hb)}+\|\nab_3\psi\|_{\Lsc^2(\Hb)}\big)^{1/2}\
\eeaa
 \end{proposition}
 \subsection{Estimates for  Hodge systems}
 \label{section.Hodge}
Consider a Hodge  system,
\beaa
\Dcal \psi=F
\eeaa
with $\Dcal$ one of the operators
  in section \ref{sect:hodge}.
In view of proposition \ref{prop:hodge},
\beaa
\int_S|\nab\psi |^2+\int_S K |\psi |^2 \les\|F\|_{L^2(S)}^2
\eeaa
where,
\beaa
K&=&-\rho+\frac{1}{2}\chih\c\chibh -\frac 1 4 \trch\trchb
\eeaa
is the Gauss curvature of $S$.
Hence,
\beaa
\|\nab\psi\|_{L^2(S)}^2&\les& \|K\|_{L^2(S)}\|\psi\|_{L^4(S)}^2+\|F\|_{L^2(S)}^2
\eeaa
Making use  of the calculus inequality on $S$,
\beaa
\|\psi\|_{L^4(S)}^2&\les&\|\nab\psi\|_{L^2(S)}\, \|\psi\|_{L^2(S)}
\eeaa
we deduce,
\beaa
\|\nab\psi\|_{L^2(S)}^2&\les& \|K\|_{L^2(S)}\|\nab\psi\|_{L^2(S)}\, \|\psi\|_{L^2(S)}+\|F\|_{L^2(S)}^2
\eeaa
and consequently,
\beaa
\|\nab\psi\|_{L^2(S)}&\les& \|K\|_{L^2(S)} \|\psi\|_{L^2(S)}+\|F\|_{L^2(S)}
\eeaa
We state below the same result in scale invariant norms
\begin{proposition}
Let $\psi$ verify the Hodge system
\bea
\DD \psi=F
\eea
Then,
\bea
\|\nab\psi\|_{\Lsc^2(S)}&\les&\de^{1/2} \|K\|_{\Lsc^2(S)} \|\psi\|_{\Lsc^2(S)}+\|F\|_{\Lsc^2(S)}
\eea
\label{prop:Hodge.estim}
\end{proposition}
To obtain the second derivative estimates for the Hodge system $\DD\psi=F$ we apply the operator 
$\DD^*$ and write the resulting equation schematically in the form
$$
\Delta \psi=K\psi + D^* F.
$$
Multiplying the equation by $\Delta\psi$, integrating over $S$ and using that 
$\|D^*F\|_{L^2(S)}\les \|\nab\psi\|_{L^2(S)}$ we obtain 
$$
\|\Delta\psi\|_{L^2(S)} \les \|K\|_{L^2(S)} \|\psi\|_{L^\infty(S)} + \|\nab F\|_{L^2(S)}
$$
Using B\"ochner's identity, see e.g. \cite{KR:LP}, 
\begin{equation}\label{ea:boch}
\|\nab^2\psi\|_{L^2(S)} \les \|K\|_{L^2(S)}\|\psi\|_{L^\infty(S)} + \|K\|^{\frac 12}_{L^2(S)}\|\nab\psi\|_{L^4(S)} + 
\|\Delta\psi\|_{L^2(S)}.
\end{equation}
we then obtain
\begin{proposition}
Let $\psi$ verify the Hodge system
\bea
\DD \psi=F
\eea
Then,
\bea
\|\nab^2\psi\|_{\Lsc^2(S)}&\les&\de^{\frac 12} \|K\|_{\Lsc^2(S)} \|\psi\|_{\Lsc^\infty(S)}+\de^{\frac 14} \|K\|^{\frac 12}_{\Lsc^2(S)} 
\|\nab\psi\|_{\Lsc^4(S)}\nn\\
&+&\|\nab F\|_{\Lsc^2(S)}
\eea
\label{prop:Hodge.estim-2}
\end{proposition}

 \section{$\OS_{0,4}$  and $\OS_{0,2}$  estimates }
 \subsection{Estimates for  $\chi,\eta, \omb$ }
 The null Ricci coefficients  $\chi, \eta$ and $\omb$
verify transport equations of the form,
\bea
\nab_4 \psi^{(s)}=\sum_{s_1+s_2=s+1}\psi^{(s_1)}\c\psi^{(s_2)}+\Psi^{(s+1)}\label{eq:transp.symb4}
\eea
Here $\psi^{(s)}$ denotes  an arbitrary  Ricci  coefficient component 
of signature  $s$ while $\Psi^{(s)}$ denotes a null curvature component
of  signature $s$.
 In view of proposition \ref{prop:transp} 
we have 
\beaa
\|\psi^{(s)}\|_{\Lsc^4(u,\ub)}&\les& \|\psi^{(s)}\|_{\Lsc^4(u,0)}+
\int_0^{\ub} \de^{-1}  \|    \nab_4\psi^{(s)}   \|_{\Lsc^4(u,\ub'  )}
\eeaa
To estimate $ \|    \nab_4\psi^{(s)}   \|_{\Lsc^4(u,\ub'  )}$ we make us 
of   the scale invariant  estimates 
\beaa
\|\phi\cdot\psi\|_{\Lsc^4(S)}\les \de^{1/2} \|\phi\|_{\Lsc^\infty(S)} \|\psi\|_{\Lsc^4(S)}
\eeaa
Hence,
\beaa
 \|    \nab_4\psi^{(s)}   \|_{\Lsc^4(S )}&\les& \| \Psi^{(s+1)}   \|_{\Lsc^4(S )}+\de^{\frac 12}\sum_{s_1+s_2=s+1}\|\psi^{(s_1)} \|_{\Lsc^\infty(S)}\|\psi^{(s_2)} \|_{\Lsc^4(S)}
 \eeaa
At this point we remark that if all  Ricci coefficient and curvature    norms $\OS_{0,4},  \RR_0$
were scale invariant  we would proceed in a straightforward manner 
as follows,
\beaa
 \|    \nab_4\psi^{(s)}   \|_{\Lsc^4(S )}&\les& \| \Psi^{(s+1)}   \|_{\Lsc^4(S )} +
 \de^{1/2}\OS_{0, \infty}\c \OS_{0, 4}\\
 &\les&\| \Psi^{(s+1)}   \|_{\Lsc^4(S )} + \de^{1/2}\De_0\c \OS_{0, 4}
 \eeaa
 Hence
 \beaa
 \|\psi^{(s)}\|_{\Lsc^4(u,\ub)}&\les& \|\psi^{(s)}\|_{\Lsc^4(u,0)}+
\int_0^{\ub} \de^{-1} \| \Psi^{(s+1)}   \|_{\Lsc^4(u,\ub' )} + \de^{1/2}\De_0\c \OS_{0,4}\\
&\les&\|\psi^{(s)}\|_{\Lsc^4(u,0)}+ \RR^{\frac 12}_0\RR^{\frac 12}_1+\de^{\frac 14}\RR_0+ \de^{1/2}\De_0\c \OS_{0,4},
 \eeaa
 where in the last step we used the interpolation inequality \eqref{eq:L4-glob-sc} for the 
 curvature $\Psi^{s+1}$.
 Thus,
 since the initial data is trivial along $\ub=0$,
 \beaa
 \|(\omb,\eta)\|_{\Lsc^4(u,\ub)}&\les&  \RR^{\frac 12}_0\RR^{\frac 12}_1+\de^{\frac 14}\RR_0+  \de^{1/2}\De_0\c \OS_{0,4}
 \eeaa

We only  have to be more careful with the cases when $ \| \Psi^{(s+1)}   \|_{\Lsc^4(S )}$ is anomalous, i.e. $\Psi=\a$, and both $\psi^{(s_1)}$, $\psi^{(s_2)}$  are anomalous.  The first situation ( but not second)
appear  only in the case of the transport equation for $\chih$ while
the  second  appear only in the transport equation  for $\trch$.
\beaa
\nab_4\chih+\trch \chih&=&-2 \omega \chih-\alpha\\
\nab_4 \trch+\frac 12 (\trch)^2&=&-|\chih|^2-2\om \trch
\eeaa
Thus,  for fixed $u$,  we estimate with $^\de S_{\ub}$ denoting a disc of radius $\de^{\frac 12}$ transported 
from the data at $S_{u,0}$( recall also the triviality of the initial data on $H_0$), 
\beaa
\|\chih\|_{\Lsc^4(^\de S_{\ub})}
&\les&\int_0^{\ub} \de^{-1}\| \a   \|_{\Lsc^4(^\de S_{u,\ub'})} 
d\ub '+\de^{1/2}\De_0\c\OS_{0,4}
\eeaa
Using \eqref{eq:L4-loc-sc} we obtain
$$
\int_0^{\ub} \de^{-1}\| \a   \|_{\Lsc^4(^\de S_{u,\ub'})} d\ub'\les \| \a   \|_{\Lsc^2(^{2\de} H_u^{(0,\ub)}  )}+ 
\| \nab\a   \|_{\Lsc^2(^{2\de} H_u^{(0,\ub)}  )}\les \RR_0^\de[\a]+\RR_1[\b]
$$
Therefore,
\beaa
\|\chih\|_{\Lsc^4(^\de S_{\ub})}\les  \|\chih\|_{\Lsc^4(^\de S_0)}+ \RR_0^\de[\a]+\RR_1[\a]+\de^{1/2}\De_0\c\OS_{0,4}
\eeaa
from which we derive both the scale invariant $\de$ estimate  for $\chih$,
\bea
\OS^\de_{0,4}[\chih]&\les& \RR_0^\de[\a]+\RR_1[\a]+\de^{1/2}\De_0\c\OS_{0,4}.
\label{eq:special.chih}
\eea
We  can also estimate  directly  the anomalous $\OS_{0, 4}[\chih]$  from,
\beaa
\|\chih\|_{\Lsc^4( S_{\ub})}&\les& \int_0^{\ub} \de^{-1}\| \a   \|_{\Lsc^4(^\de S_{u,\ub'})} 
d\ub '+\de^{1/2}\De_0\c\OS_{0,4}
\eeaa
Using the scale invariant interpolation inequality  \eqref{eq:L4-glob} we deduce,
\beaa
\|\chih\|_{\Lsc^4( S_{\ub})}&\les& \|\a\|_{\Lsc^2(H_u^{(0,\ub)} }   ^{1/2}\c\|\nab\a\|_{\Lsc^2(H_u^{(0,\ub)}}^{1/2}+
\de^{1/4}  \|\a\|_{\Lsc^2(H_u^{(0,\ub)} } +\de^{1/2}\De_0\c\OS_{0,4}
\eeaa
Taking into account  the anomalous character
of $\RR_0[\a]$ and the definition of  $\OS_{0,4}[\chih]$, we deduce,
\bea
\OS_{0,4}[\chih]&\les&\RR_0[\a]^{1/2}\big(\RR_1[ \a]+\RR_0[\a]\big)^{1/2}+\de^{1/4}\De_0\c\OS_{0,4}
\label{estim.chihL4}
\eea

On the other hand,
\beaa
\|\trch\|_{\Lsc^4(u,\ub)}&\les& \|\trch\|_{\Lsc^4(u,0)}+\int_0^{\ub} \de^{-\frac 12} \De_0 \|\trch\|_{\Lsc^4(u,\ub')} d\ub'\\
&+&
\de^{-\frac 12} \De_0\int_0^{\ub} \|\chih\|_{\Lsc^4(u,\ub')} d\ub'+\de^{1/2}\De_0\c\OS_{0,4}\\
&\les&\|\trch\|_{\Lsc^4(u,0)}+ \de^{\frac 14} \De_0\OS_{0,4}
\eeaa
We summarize the results of the section in the following\footnote{Recall the triviality
of our initial conditions at $\ub=0$.}.
\begin{proposition}
\label{prop:Ricciestim0.4}
Under    the bootstrap assumption $\OS_{0,\infty}\le \De_0$ and  assuming that $\de^{1/2}\De_0$ is sufficiently small we derive,
\beaa
\OS_{0,4}[\omb,\eta]&\les& \RR_0+ \RR^{\frac 12}_0\RR^{\frac 12}_1+\de^{\frac 14} \RR_0+ \de^{1/2}\De_0\c \OS_{0,4}\\
\OS_{0,4}[\trch]&\les&1+\de^{\frac 14}\De_0 \c\OS_{0,4},\\
\OS_{0,4}[\chih] &\les&\RR_0[\a]^{1/2}\big(\RR_1[ \a]+\RR_0[\a]\big)^{1/2}+\de^{1/4}\De_0\c\OS_{0,4}
\eeaa
Also,
\beaa
\OO^\de_0[\chih]&\les& \RR_{[1]}+\de^{1/2}\De_0\c \OS_{0,4}
\eeaa
\end{proposition}
\subsection{Estimates for  $\chib,\etab, \om$ }
The Ricci coefficients  $\etab, \chib$ and $\omb$ verify equations of the form,
\beaa
\nab_3 \psi^{(s)}=-\frac 1 2 k\,  \trchb \psi^{(s)}+\sum_{s_1+s_2=s}\psi^{(s_1)}\c\psi^{(s_2)}+\Psi^{(s)}
\eeaa
with $k$ a positive  integer. 
Writing $\trchb=\trchb_0+\trchbt$, with $\trchb_0=-\frac{4}{\ub-u+2 r_0}$,
we derive 
\bea
\nab_3 \psi^{(s)}=-\frac 1 2 k \, \trchb_0 \psi^{(s)}+\sum_{s_1+s_2=s}\psi^{(s_1)}\c\psi^{(s_2)}+\Psi^{(s)}
\label{eq:transp.symb3}
\eea
In this case we observe that  the curvature term $\Psi^{(s)}$ is never anomalous 
and the only time when both $\psi^{(s_1)}$ and  $\psi^{(s_2)}$ are
anomalous is in the case of  the transport  equations for $\chibh$ and $\trchb$. 
In all other cases we can write, proceeding exactly as before,
\beaa
\|\psi^{(s)}\|_{\Lsc^4(u,\ub)}&\les& \|\psi^{(s)}\|_{\Lsc^4(0,\ub)}+\int_0^u 
 \|    \nab_3\psi^{(s)}   \|_{\Lsc^4(u',\ub  )}
 \eeaa
and,
\beaa
 \|    \nab_3\psi^{(s)}   \|_{\Lsc^4(u,\ub)}&\les& \|\psi^{(s)}\|_{\Lsc^4(u,\ub)}+ \|\Psi^{(s)}\|_{\Lsc^4(u,\ub)}+   \de^{1/2}\OS_{0,\infty}\c \OS_{0,4}
\eeaa
Thus, in these cases,
\bea
\|\psi^{(s)}\|_{\Lsc^4(u,\ub)}&\les& \|\psi^{(s)}\|_{\Lsc^4(0,\ub)}+\int_0^u  \|\Psi^{(s)}\|_{\Lsc^4(u',\ub)}+   \de^{1/2}\OS_{0,\infty}\c\OS_{0,4}\nn\\
&\les& \|\psi^{(s)}\|_{\Lsc^4(0,\ub)}+\RRb_0^{\frac 12}\RRb_1^{\frac 12} + \de^{\frac 14}\RRb_0+ \de^{1/2}\De_0\c\OS_{0,4}
  \label{eq:ricci0.13}
\eea
Similarly,
\bea
\|\psi^{(s)}\|_{\Lsc^2(u,\ub)}\les  \|\psi^{(s)}\|_{\Lsc^2(0,\ub)}+\RRb_0+ \de^{1/2}\De_0\c\OS_{0,4}
  \label{eq:ricci0.13'}
\eea

It thus only remains to estimate $\trchb, \chibh$.
We first estimate  $\OO^\de_0[\chib]$ from the equation,
\beaa
\nab_3\chibh&=&-\aa    +  \trchb_0\, \chibh      -\wtrchb\,  \chibh       -2\omegab \chibh
\eeaa
Clearly, for fixed $\ub$
\beaa
\|\nab_3\chibh+\frac 12\trchb\chibh\|_{\Lsc^4(^\de S_u)}&\les& \|\aa\|_{\Lsc^4(^\de S_u)}+
\|\chibh\|_{\Lsc^4(^\de S_u)}+\de^{1/2} \OS_{0,\infty}\c\OS_{0,4}
\eeaa
and thus, after a standard application of the Gronwall inequality,
\beaa
\|\chibh\|_{\Lsc^4(^\de S_u)}&\les&\|\chibh\|_{\Lsc^4(^\de S_0)}+
\int_0^u \|\aa\|_{\Lsc^4(^\de S_{u'})}
\eeaa
Taking into account the scale invariant interpolation inequality 
\eqref{eq:L4-glob-sc} we deduce,
\beaa
\|\chibh\|_{\Lsc^4(^\de S_u)}&\les&\|\chibh\|_{\Lsc^4(^\de S_0)}+\RRb^{\frac 12}_0[\aa]\c \RRb^{\frac 12}_1[\aa]+\de^{\frac 14} \RRb_0[\aa]
+\de^{1/2}\De_0 \OS_{0,4}
\eeaa
or, since  $\|\chibh\|_{\Lsc^4(^\de S_0)}\les \OO^{(0)}$,
\bea
\|\chibh\|_{\Lsc^4(^\de S_u)}&\les&\OO^{(0)}+\RRb^{\frac 12}_0[\aa]\big(\RRb^{\frac 12}_1[\aa]+\de^{\frac 14} \RRb_0[\aa]\big)
+\de^{1/2}\De_0 \OS_{0,4}
\eea
Proceeding in the same fashion, 
\beaa
\|\chibh\|_{\Lsc^4(S_u)}&\les&\|\chibh\|_{\Lsc^4( S_0)}+\RRb^{\frac 12}_0[\aa]\c \RRb^{\frac 12}_1[\aa]+\de^{\frac 14} \RRb_0[\aa]
+\de^{1/2}\De_0 \OS_{0,4}
\eeaa
Now,  observe that the only anomaly on the right hand side is due to $\|\chibh\|_{\Lsc^4( S_0)}$. In fact
\bea
\|\chibh\|_{\Lsc^4( S_0)}&\les&\de^{-1/4} \OO^{(0)}\eea
Thus,
\bea
\OS_{0,4}[\chibh]&\les&  \OO^{(0)}+\de^{1/4} \RRb^{\frac 12}_0[\aa] \c\RRb^{\frac 12}_1[\aa]+\de^{\frac 12} \RRb_0
+\de^{3/4}\De_0 \OS_{0,4}
\eea

To estimate $\trchbt=\trchb-\trchb_0$ we start with the equation
 \beaa
 D_3\trchb+\frac 12 (\trchb)^2=-2\omegab \trchb-|\chibh|^2.
 \eeaa
Since,
 $
 D_3 u =\Omega^{-1},\,\,D_3\ub=0
 $ 
 we have,  since  $\trchb_0=-\frac 4{\ub-u+2r_0}$,
  \beaa
 D_3\trchb_0=-\Om^{-1} \frac 1 4 \trchb_0^2
 \eeaa
 Hence, using  $ \widetilde{\trchb}=\trchb-\trchb_0$,
\bea
 \nab_3 \widetilde{\trchb}+\trchb_0\c \widetilde{\trchb}&=&  -\frac{1}{2\Om}(\Om-\frac 1 2 )\trchb_0^2 +  2\omb \trchb_0 -2\omegab\widetilde{\trchb}     -|\chibh|^2
\eea
Now, taking  into account  the anomalous scaling of $\OS_{0,4}[\chibh]$
and estimate,
 $
 \|\Om-\frac 1 2 \|_{\Lsc^2(S)}\les  \|\omb\|_{\Lsc^2(S)}$ (which can be easily  derived   using the transport equation 
 $\nab_3\Om=\omb$)
 we derive,
\beaa
\|\nab_3 \widetilde{\trchb}\|_{\Lsc^4(S)}&\les & \|\widetilde{\trchb}\|_{\Lsc^4(S)}
+\|\omb\|_{\Lsc^4(S)}+\de^{\frac 14}\OS_{0,\infty}\c\OS_{0,4}.
\eeaa
 from which,
\beaa
 \|\widetilde{\trchb}\|_{\Lsc^4(u,\ub)}&\les &\| \widetilde{\trchb}\|_{\Lsc^4(0,\ub)}+\int_0^u\|\nab_3 \widetilde{\trchb}\|_{\Lsc^4(u', \ub)}\\
 &\les&\| \widetilde{\trchb}\|_{\Lsc^4(0,\ub)}+\int_0^u \|\widetilde{\trchb}\|_{\Lsc^4(u',\ub)} du' \\
 &+&\int_0^u\|\omb\|_{\Lsc^4(u',\ub)} du' 
 +\de^{\frac 14}\De_0\c\OS_{0,4}.
\eeaa
By Gronwall, and using the estimate for $\omb$ derived in the previous section,
\beaa
 \|\widetilde{\trchb}\|_{\Lsc^4(u,\ub)}&\les &\|\widetilde{\trchb}\|_{\Lsc^4(0,\ub)}
 +\RR_0^{\frac 12}\RR_1^{\frac 12}+\de^{\frac 14}\RR_0+\de^{\frac 14}\De_0\c\OS_{0,4}. \nn\\
\eeaa
  Thus,
  \bea
  \OS_{0,4}[\trchbt]&\les& \OO^{(0)}+\RR_0\RR_1+\de^{\frac 14}\RR_0+\de^{\frac 14}\De_0\c\OS_{0,4}
    \eea
We summarize the result of this subsection in the following 
\begin{proposition}
\label{prop:L4estim.bar}
We have, for sufficiently small $\de$,
\beaa
\OS_{0,4}[\etab,\om]&\les&  \OO^{(0)}+\RRb_0+\RRb_0^{\frac 12}\RRb_1^{\frac 12}+\de^{\frac 14}\RRb_0+\de^{\frac 12}\De_0\c\OS_{0,4}\\
\OS_{0,4}[\chibh]&\les&  \OO^{(0)}+\de^{1/4} \RRb^{\frac 12}_0 \c\RRb^{\frac 12}_1+\de^{\frac 12} \RRb_0
+\de^{3/4}\De_0 \OS_{0,4}\\
 \OS_{0,4}[\trchbt]&\les&\OO^{(0)}+\RR_0\RR_1+\de^{\frac 14}\RR_0+\de^{\frac 14}\De_0\c\OS_{0,4}
 \eeaa
 Also,
 \beaa
\OO^\de[\chibh]&\les& \OO^{(0)}+\RRb_0^{\frac 12}\RRb_1^{\frac 12}+\de^{\frac 14}\RRb_0+\de^{\frac 12}\De_0\c\OS_{0,4} 
\eeaa
\end{proposition}

\subsection{ Summary of $\OS_{0,4}$ estimates }
\label{section:gen.str}
Putting together the results of the last two propositions we deduce the following.
\begin{proposition}
\label{prop.finalOO0}
There exists a constant $C$ depending only on $\OO^{(0)}$ and $\RR$ such that,
if  $\de^{1/2}\De_0$ is sufficiently small, we have,
\bea
\OS_{0,4}&\les& C\label{main.estim.O04}
\eea
Moreover,
\bea
\OS_{0,4}[\chih] &\les&\RR_0[\a]^{1/2}\big(\RR_1[ \a]+\RR_0[\a]\big)^{1/2}+\de^{1/4} C\\
\OS_{0,4}[\chibh]&\les& \OO^{(0)}+\de^{1/4} C\label{final.estim. O04chibh}
\eea

\end{proposition}
\subsection{$\OS_{0,2}$ estimates}  The following 
estimates will also be needed.
\begin{proposition}
\label{prop.finalOS02}
There exists a constant $C$ depending only on $\OO^{(0)}$ and $\RR$ such that,
if  $\de^{1/2}\De_0$ is sufficiently small, we have,
\bea
\OS_{0,2}&\les& C\label{main.estim.OS02}
\eea
\end{proposition}
\begin{proof}

These are similar  but somewhat simpler, 
once  we already have the $\OS_{0,4}$ estimates.  Indeed, starting
with \eqref{eq:transp.symb4},  (dropping  indices for simplicity)  we write as before,
\beaa
\|\psi\|_{\Lsc^2(u,\ub)}&\les& \|\psi\|_{\Lsc^2(u,0)}+
\int_0^{\ub} \de^{-1}  \|    \nab_4\psi \|_{\Lsc^2(u,\ub'  )}
\eeaa
and, assuming the worst case scenario when  both terms in $\psi\c\psi$ are anomalous,
i.e.  both satisfy $\|\psi \|_{\Lsc^4(S)}\les C\de^{-1/4}$,
\beaa
 \|    \nab_4\psi   \|_{\Lsc^2(S )}&\les& \| \Psi   \|_{\Lsc^2(S )}+\de^{\frac 12}\|\psi \|_{\Lsc^4(S)}\|\psi \|_{\Lsc^4(S)}\\
 &\les& \| \Psi   \|_{\Lsc^2(S )}+\OS_{0, 4}^2\\
 &\les&\| \Psi   \|_{\Lsc^2(S )}+C^2.
\eeaa
Thus,
\beaa
\|\psi\|_{\Lsc^2(u,\ub)}&\les& +
\int_0^{\ub} \de^{-1}  \|    \Psi \|_{\Lsc^2(u,\ub'  )}+C^2\\
&\les&  \|    \Psi \|_{\Lsc^2(H_u)}+C^2
\eeaa
$\Psi$ can only be the anomalous  $\a$ in the case of  the transport equation
for $\chih$. Thus,
  \beaa
 \|(\omb,\eta)\|_{\Lsc^2(u,\ub)}&\les&  \RR_0+C^2\\
 \|\chih\|_{\Lsc^2(u,\ub)} &\les & \de^{-1/2} \RR_0[\a]
 \eeaa
 or, with a constant $C=C(\OO^{(0)},\RR, \RRb)$,
 \beaa
 \OS_{0,2}[\trch,\chih, \omb, \eta]&\les & C
 \eeaa
 The estimates for $\trchb, \chibh, \om, \etab$ are proved
 in the same manner. 

\end{proof}
\section{$\OO_1$ estimates}
\subsection{General Strategy}
To get the  first and second derivative estimates for the Ricci coefficients
we cannot proceed as we did in the previous section. Following a path
first pursued in \cite{Chr-Kl} and continued  in \cite{KNI:book},   \cite{KR:causal} and \cite{Chr:book}   we  introduce new 
quantities\footnote{Different components
$\Th$ appear in \eqref{eq:nab4Th} and \eqref{eq:nab3Th}. It may in fact be more appropriate  to call  $\Th$ the components which appear on the left of the 
$\nab_4$ equation and by $\underline{\Th} $ those appearing on the left of the $\nab_3$ equations. } $\Theta^{(s)}$, with signature $s$,  depending on first derivative of the Ricci coefficients and which verify transport equations of the form\footnote{We neglect to write possible  constants in front of each term on the right
of our equations}
\bea
\nab_4 \Th^{(s)} &=&\trch\big( \Th^{(s)}+\nab\psi^{(s-\frac 1 2 )}   \big)  
+ \sum_{s_1+s_2+\frac 1 2 =s+1}     \psi^{(s_1)}\big(   \nab  \psi^{(s_2)}+\Psi^{(s_2)} \big)\nn\\
&+&\sum_{s_1+s_2=s+1}\trchb_0 \c \psi^{(s_1)}\c\psi^{(s_2)}+     \sum_{s_1+s_2+s_3=s+1}  \psi^{(s_1)}\c\psi^{(s_2)}\c \psi^{(s_3)}\nn\\
\label{eq:nab4Th}
\eea
\bea
\nab_3 \Th^{(s)} &=&\trchb\big( \Th^{(s)}+\nab\psi^{(s-\frac 1 2 )}   \big)  
+ \sum_{s_1+s_2+\frac 1 2 =s}     \psi^{(s_1)}\big(   \nab  \psi^{(s_2)}+\Psi^{(s_2)} \big)\nn\\
&+&\sum_{s_1+s_2=s}\trchb_0 \c \psi^{(s_1)}\c\psi^{(s_2)}+     \sum_{s_1+s_2+s_3=s}  \psi^{(s_1)}\c\psi^{(s_2)}\c \psi^{(s_3)}\nn\\
\label{eq:nab3Th}
\eea
Here $\psi^{(s)}$ are components of all  the Ricci coefficients  $(\trch, \chih,\omb, \eta, \etab, \trchbt, \chibh$) with signature $s$,  while $\Psi^{(s)}$ are curvature components with signature $s$.

The main idea behind our   strategy is to show that once  we control
the $\Lsc^2(S)$ norms of these  quantities $\Th$ we derive all
$\OO_1$ estimates by using the  elliptic  Hodge systems. 
The most general form of such systems is given  by 
\bea
\DD \psi^{(s)}&=&\Th^{(s+\frac 1 2)} +   \Psi^{(s+\frac 1 2)}+\trchb_0\c \psi^{(s+\frac 1 2 )}+
\sum_{s_1+s_2=s+\frac 1 2 }  \psi^{(s_1)} \psi^{(s_2)}. 
\label{hodge.abstract}
\eea
where $\DD$ is one of the Hodge systems of section  \ref{sect:hodge}.
Observe also that both Hodge systems 
have non- anomalous  curvature source terms, 
$\b$, respectively $\bb$ and no quadratic anomalies in $\psi$ (relative to the $\OO_0$ norm).

\subsection{ Explicit    $\Th$     variables and Hodge systems} In this section we introduce explicit variables $\Th^{(s)}$ and derive transport equations of the type \eqref{eq:nab4Th}, \eqref{eq:nab3Th}.

\medskip
\noindent
\textit{Transport-Hodge systems for $\chi,\chib$.}
First observe that  the Codazzi
equations 
\bea
\div\chih&=&\frac 12 \nab\trch - \frac 12 (\eta-\etab)\cdot (\chih -\frac 1 2 \trch) -\beta,\label{eq:codazzi.chih}\\
\div\chibh&=&\frac 12 \nab\trchb + \frac 12 (\eta-\etab)\cdot (\chibh-\frac 1 2   \trchb) +\betab\label{eq:codazzi.chbh}
\eea
can be written as Hodge systems of type  \eqref{hodge.abstract}.
with $\DD$ the Hodge operator  $\dcall$, discussed in section
\ref{sect:hodge},  and  $\Th=\nab\trch$, resp.  $\Th=\nab\trchb$.

We now  derive a $\nab_4$ transport equation for $\nab\trch$.   Using  
commutation formula,
$
[\nab_4,\nab] f=\frac 12 (\eta+\etab) D_4 f -\chi\cdot \nab f,
$
we obtain,
\bea
\label{eq:nabtrch}
\nab_4 \nab\trch &=&-\nab\trch \trch -2   \trch \nab\omega -2\omega\nab\trch -2\nab\chih\cdot\chih
\\
&+&\frac 12(\eta+\etab) \big(-\frac 12(\trch)^2-2\omega\trch-|\chih|^2\big) -\chi\cdot \nab\trch\nn
\eea
which is clearly of the form \eqref{eq:nab4Th} with no curvature terms 
present and no  triple  anomalies (relative to the $\OO_0$ norm, i.e. 
among the cubic terms at least one of the factors are not anomalous).

To derive a transport equation for $\nab\trchb$ we
start  with the transport equation,
 \beaa
 \nab_3\trchb&=& -\frac 12 (\trchb)^2+F,\quad
  F=   -2\omegab \trchb-|\chibh|^2 = -2\omegab \trchb_0-2\omegab \trchbt-|\chibh|^2
  \eeaa
 Using the commutator formula,
$ [\nab_3,\nab ] f= -\chib\cdot \nab f+\frac 12 (\eta+\etab) D_3 f$
 we deduce,
 \beaa
 \nab_3(\nab\trchb)&=& -\chibh\cdot \nab \trchb-\frac 3 2 \trchb \nab\trchb-\big(\nab+\frac 12 (\eta+\etab) \big) F
 \eeaa
Or, writing $\trchb=\trchb_0+\trchbt $, we deduce,
 \bea
 \nab_3(\nab\trchb)&=& -\chibh\cdot \nab \trchb-\frac 3 2 \trchb_0 \nab\trchb -\frac 3 2 \trchbt \nab\trchb -\big(\nab+\frac 12 (\eta+\etab) \big) F  \label{eq:transport.nabtrch.3}
 \eea
This is clearly a system of the form \eqref{eq:nab3Th} with  no curvature terms present and  no anomalous  cubic terms.

\medskip
\noindent
\textit{Transport- Hodge systems  for $\mu,\mub, \nab\eta, \nab\etab $.}
We start with equation
\beaa
\curl \eta&=&\curl \etab=\si+\chibh\wedge\chih\\
\eeaa
We derive  equations for $\div\eta $ and $\div \etab$
by taking he divergence of the transport equations
\beaa
\nab_4\eta&=- \frac 1 2\trch(\eta-\etab)-      \chih\c(\eta-\etab)-\b\\
\nab_3\etab &=-    \frac 1 2\trchb(\etab-\eta)    -   \chibh\c (\etab-\eta)+\bb
\eeaa
Using commutation  lemma \eqref{le:comm} we derive,
\beaa
\nab_4(\div \eta)&=&\div (- \frac 1 2\trch(\eta-\etab)-      \chih\c(\eta-\etab)-\b)\\
&-&\frac 1 2 \trch  \, \div \eta -\chih\c \nab \eta-\eta\c\b 
+\frac 1 2 (\eta+\etab)\c \nab_4 \eta\\
&=&-\div\b -\frac 1 2 \trch (2\div\eta-\div \etab)  - (\eta-\etab)\c\big(\frac 1 2 \nab\trch +\div  \chih)\\
&-&\chih\c\nab  (2\eta-\etab)-\eta\c\b +\frac 1 2 (\eta+\etab)\c \nab_4 \eta
\eeaa
Using the null Codazzi equation,
\beaa
\frac 1 2 \nab\trch +\div \chih&=& \nab \trch +\frac 1 2 \ze\trch-\b
\eeaa
we derive,
\beaa
\nab_4(\div \eta)&=&-\div\b -\frac 1 2 \trch (2\div\eta-\div \etab)-\chih\c\nab  (2\eta-\etab) -(\eta-\etab)\c\nab\trch\\
&-&\etab \c\b- \frac 1 4\trch (\eta-\etab)^2+\frac 1 2 (\eta+\etab)\big(-\frac 1 2 \trch (\eta-\etab)-\chih\c (\eta-\etab)-\b\big)\\
&=&-\div\b -\frac 1 2 \trch (2\div\eta-\div \etab)-\chih\c\nab  (2\eta-\etab) -(\eta-\etab)\c\nab\trch\\
&-&\frac{1}{2} (3\etab+\eta)\c\b-\frac 1 2 \trch(|\eta|^2-\eta\c\etab)-\frac 1 2 
(\eta+\etab)\c\chih\c (\eta-\etab)
\eeaa
or,
\beaa
\nab_4(\div \eta)+\trch\div \eta&=&
-\div\b +\frac 1 2 \trch\div \etab-\chih\c\nab  (2\eta-\etab) -(\eta-\etab)\c\nab\trch\\
&-&\frac{1}{2} (3\etab+\eta)\c\b-\frac 1 2 \trch(|\eta|^2-\eta\c\etab)-\frac 1 2 
(\eta+\etab)\c\chih\c (\eta-\etab)
\eeaa

On the other hand, 
\beaa
\nab_4\rho+\frac 32\trch\rho=\div\beta-\frac 12\chibh\cdot\alpha+\zeta\cdot\beta+2\etab\cdot\beta
\eeaa
Adding  the two equations and setting,
\beaa
\mu=-\div \eta-\rho
\eeaa
we derive,
\beaa
\nab_4\mu+\trch\mu&=&-\frac 1 2\trch  \div\etab+ (\eta-\etab)\nab\trch  +\chih\c \nab  (2\eta-\etab)
+\frac 1 2\,  \chibh\c \a -(\eta-3\etab)\c\b +\frac 1 2 \trch \rho\\
&+&\frac 1 2 \trch(|\eta|^2-\eta\c\etab)+\frac 1 2 (\eta+\etab)\c\chih\c (\eta-\etab)
\eeaa
Similarly,
setting
\beaa
\mub=-\div \etab-\rho
\eeaa
we derive,
\beaa
\nab_3\mub+\trchb\mub&=&-\frac 1 2\trchb  \div\eta+ (\etab-\eta)\nab\trchb  +\chibh\c \nab  (2\etab-\eta)
+\frac 1 2\,  \chih\c \aa -(\etab-3\eta)\c\bb +\frac 1 2 \trchb \rho\\
&+&\frac 1 2 \trchb(|\etab|^2-\etab\c\eta)+\frac 1 2 (\eta+\etab)\c\chibh\c (\etab-\eta)
\eeaa
We summarize the results above in the following.
\begin{lemma}\label{lem:tr-mu} The reduced mass aspect functions,
\beaa
\mu&=&-\div\eta-\rho\\
\mub&=&-\div\etab -\rho
\eeaa
verify the transport equations,
\bea
\nab_4\mu+\trch\mu&=&-\frac 1 2\trch  \div\etab+ (\eta-\etab)\nab\trch  +\chih\c \nab  (2\eta-\etab)
+\frac 1 2\,  \chibh\c \a -(\eta-3\etab)\c\b +\frac 1 2 \trch \rho\nn\\
&+&\frac 1 2 \trch(|\eta|^2-\eta\c\etab)+\frac 1 2 (\eta+\etab)\c\chih\c (\eta-\etab)
\label{eq:transport-mu}
\eea
\bea
\nab_3\mub+\trchb\mub&=&-\frac 1 2\trchb  \div\eta+ (\etab-\eta)\nab\trchb  +\chibh\c \nab  (2\etab-\eta)
+\frac 1 2\,  \chih\c \aa -(\etab-3\eta)\c\bb +\frac 1 2 \trchb \rho\nn\\
&+&\frac 1 2 \trchb(|\etab|^2-\etab\c\eta)+\frac 1 2 (\eta+\etab)\c\chibh\c (\etab-\eta)\label{eq:transport-mub}
\eea
\end{lemma}
\begin{remark}
Observe that our mass aspect functions differ from
those of \cite{Chr-Kl}  or  \cite{KNI:book}. Thus,  in \cite{KNI:book}),
$\mu=-\div\eta-\rho+\frac 1 2 \chih\c\chibh $ verifies  (see equation 4.3.32 in \cite{KNI:book}),
\beaa
\nab_4 \mu+\trch \mu &=& \chih\c (\nab\hot \eta)+
 (\eta-\etab)\c(\nab\trch+\trch \ze)+\frac 1 2 \trch \big(\mu +\div(\eta-\etab)\big) \nn\\
 &-&\frac 1 4 \trchb |\chih|^2+\frac 1 2 \trch (\chih\c\chibh+2\rho-|\etab|^2)
 +2( \eta\c\chih\c \etab -\eta\c\b)
\eeaa
The reason we prefer  our definition here  is to avoid  the presence of  triple anomalous terms on the right hand side of the transport equations for $\mu, \mub$.
\end{remark}
We write  \eqref{eq:transport-mu} symbolically in the form,
\bea
\nab_4 \mu&=&  \psi \c(\nab \psi+\Psi_g)+\chibh\c\a +\psi\c\psi\c\psi_g
\label{eq:transpmu}
\eea
which is of the form  \eqref{eq:nab4Th}, with $\psi_g\in\{\trch,\chih,\eta,\etab, \om,\omb,\trchb\}$ and $\Psi_g\in\{\b, \rho,\si\,\bb\}$.
We can also write, in shorter form,
\beaa
\nab_4 \mu&=&  \psi \c(\nab \psi+\Psi )+\psi\c\psi\c\psi_g
\eeaa
and recall that $\psi\c \Psi$ contains  the more  difficult  term
$ \chibh\c\a$ anomalous in both $\psi$ and $\Psi$. 

We also rewrite  \eqref{eq:transport-mub}  symbolically. 
In  this case we  have to keep track of the terms proportional to $\trchb=\trchb_0+\trchbt$. 
We thus write symbolically,
\bea
\nab_3 \mub&=& \trchb_0(\nab \psi+\mub)+ \psi \c(\nab \psi+\Psi_g)+\psi_g\c \bb
+\trchb_0\psi\c\psi_g+\psi\c\psi\c\psi_g\nn\\\label{eq:transpmub}
\eea
Here $\Psi_g\in\{\rho,\si,\bb,\aa\}$.
Observe that at least one of the factors $\psi$ in   $\trchb_0\psi\c\psi_g$
 and    $\psi\c\psi\c\psi_g  $  can be anomalous.   Unlike
 in the case of $\nab_4 \mu$ equation, there are no terms of the form 
 $\psi\c\b$ with $\psi$ also anomalous (recall that $\b$ is anomalous for $\RRb_0$).

 We combine the transport equations \eqref{eq:transpmu} and \eqref{eq:transpmub} with the Hodge systems,
\bea
\div\eta&=&-\mu-\rho\label{eq:div.curleta}\\
\curl\eta&=& \si-\frac 1 2 \chih\wedge\chibh\nn
\eea
and,
\bea
\div\etab&=&-\mub-\rho \label{eq:div.curletab}\\
\curl\etab&=& \si-\frac 1 2 \chih\wedge\chibh
\eea
They are both systems of type  \eqref{hodge.abstract}. Note 
that the  quadratic term $\chih\c\chibh $ is anomalous with respect to both factors.

\medskip
\noindent
\textit{Transport-Hodge systems   for $\kapb,\kap,   \nab\omb, \nab \om$.}
\noindent
We look for transport equations for quantities connected to
$\nab\om$ and $\nab \omb$. Recall that 
\bea
\nab_4\omegab&=&\frac 12 \rho+F\\
F&=& 2\omega\omegab+\frac 34 |\eta-\etab|^2-\frac 14 (\eta-\etab)\cdot (\eta+\etab)-
\frac 18 |\eta+\etab|^2\nn
\eea
and,
\bea
\nab_3\omega&=&\frac 12 \rho+\underline{F}\label{nab3.omega.second}\\
 \underline{F}&=& 2\omega\omegab+\frac 34 |\eta-\etab|^2+\frac 14 (\eta-\etab)\cdot (\eta+\etab)-
\frac 18 |\eta+\etab|^2\nn
\eea
We introduce the auxiliary quantities $\ombt$ and $\omt$ as follows.
\bea
\nab_4\ombt&=&\frac 12 \si \label{eq:ombt4}\\
\nab_3\omt&=&\frac 12 \si \label{eq:omt3}
\eea
with zero boundary conditions  along $\Hb_0$, respectively $H_0$.
We introduce  the pair of scalars $\ombtild=(\omb,\ombt)$  and $\omtild=(-\om,\omt)$
and apply   the Hodge operator $\dcalll$   ( see subsection
 \ref{sect:hodge}),
\beaa
\dcalll\, \ombtild&=&-\nab \omb+\,^*\nab \ombt, \quad \dcalll\, \omtild=
\nab\om+\,^*\nab\omt.
\eeaa
Next we derive  a $\nab_4 $ equation  for $\ombtild$  and a $\nab_3$
equation for $\omtild$. To do this we write the commutation relation
\eqref{comm:nabnab3} in the form,
\beaa
\,[\nab_4,\nab] f&=&-\frac 1 2 \trch \nab f -\chih\cdot \nab f +\frac 12 (\eta+\etab) D_4 f \\
\,[\nab_4,\nabd] g&=&-\frac 1 2 \trch \nabd g+ \chih\cdot \nabd g +\frac 12 (\eta^*+\etab^*) D_4 g
\eeaa
Thus, for a pair of scalars $(f,g)$,
\beaa
\,[\nab_4,\dcalll] (f,g)&=&-\frac 1 2\trch \dcalll(f,g)+\chih\c (\nab f+\nabd g)
-\frac 1 2(\eta+\etab)\nab_4 f +\frac 12 (\eta^*+\etab^*) D_4 g
\eeaa
Therefore,
\beaa
\nab_4\dcalll \ombtild&=&\dcalll(\rho, \si)-\nab F+\,[\nab_4,\dcalll] \ombtild\\
&=&\dcalll(\rho, \si)-\nab F-\frac 1 2\trch\dcalll \ombtild+\chih\c (\nab \omb+\nabd \ombt)\\
&-&\frac 1 2(\eta+\etab)(\rho+F)+\frac 12 (\eta^*+\etab^*)\si
\eeaa
On the other hand,
 we have the Bianchi equation,
\beaa
D_4\betab+\trch\betab={\mathcal D}_1^*(\rho,\sigma)+ 2\omega\betab +
2\chibh\cdot\beta-3(\etab\rho-^*\etab\sigma),
\eeaa
Thus,introducing the new horizontal vector,
\bea
\kapb:=\dcalll\ombtild-\frac 1 2 \bb= \dcalll (\omb,\,\ombt)-\frac 12 \bb=-\nab\omb+\nabd\ombt-\frac 1 2 \bb
\eea
we deduce,
\bea
\nab_4\kapb&=&-\trch \c \kapb-\om \bb-\chibh\c \b +\frac 3 2(\etab \rho-\,^* \etab \si)-\frac 1 2(\eta+\etab)\rho+\frac 12 (\eta^*+\etab^*)\si\nn
\\
&+& \chih\c (\nab \omb+\nabd \ombt) - \nab F-\frac 1 2(\eta+\etab) \,F
\label{eq.nab4kapb}
\eea

Similarly we set,
\bea
\kap:=\dcalll\omtild-\frac 1 2 \b= \dcalll (-\om,\,\omt)-\frac 12 \b=\nab\om+\nabd\omt-\frac 1 2 \b
\eea
and,
using the Bianchi equations,
\beaa
D_3\b+\trchb\b={\mathcal D}_1^*(-\rho,\sigma)+ 2\omb\b +
2\chih\cdot\bb+3(\eta\rho+^*\eta\sigma),
\eeaa
we derive,
\bea
\nab_3\kap&=&-\trchb \c \kap-\omb \b-\chih\c \bb +\frac 3 2(\eta \rho+\,^* \eta \si) -\frac 1 2(\etab+\eta)\rho+\frac 12 (\etab^*+\eta^*)\si
\nn\\
&+&\chibh\c (-\nab \om+\nabd \omt) +\nab \underline{F}+\frac 1 2(\etab+\eta)\underline{F}
\label{eq.nab3kap}
\eea
To estimate $\nab\omb$ we combine
the $\nab_4$ equation \eqref{eq.nab4kapb} with
the Hodge system,
\bea
\dcalll (\omb,\, \ombt)&=&\kapb+\frac 1 2 \bb
\eea
To estimate $\nab\om$ we combine  the $\nab_3$ equation \eqref{eq.nab3kap} with
the Hodge system,
\bea
\dcalll (-\om,\,\omt)&=&\kap+\frac 1 2 \b
\eea

Clearly transport equations for $\kapb$ and $\kap$
are of the form \eqref{eq:nab4Th} and \eqref{eq:nab3Th}
provided that we extend the set of Ricci coefficients 
$\psi$ to also include  the new scalars $\ombt$ and $\omt$.
We observe that $\ombt$ has the same signature as $\omb$ and 
$\omt$ has the same signature as $\om$. Moreover   $\ombt$, $\omt$
they satisfy equations similar  to those satisfied by
$\omb, \om$. Thus,
for example, we can easily derive  both $\Lsc^2$ and $\Lsc^4$ estimates
for them. Indeed,  from \eqref{eq:ombt4} we easily derive,
$$
\|\ombt\|_{\Lsc^2(u,\ub)}\les\int_0^{\ub}\de^{-1} \|\si\|_{\Lsc^2(u,\ub')} d\ub' \les\RR_0[\si].
$$
Similarly,   from \eqref{eq:omt3},
\beaa
\|\ombt\|_{\Lsc^2(u,\ub)}\les\int_0^{u} \|\si\|_{\Lsc^2(u',\ub)} d\ub' \les\RRb_0[\si]
\eeaa
It thus make perfect sense to extend the definition
of  the set of Ricci coefficients  as well as the definition
 of the  norms  $\OS_{\infty}$, $\OS_{0,4}, \OS_{1,2} ,\OS_{1,4}$   to  include them.
We thus  also assume, from now on, that the main  bootstrap assumption
\eqref{bootstrap:Linfty} includes $\ombt, \omt$.

  Finally we observe that equations \eqref{eq.nab4kapb},   \eqref{eq.nab3kap}
   can be written in the form,
\beaa
\nab_4\kapb&=&-\trch \c\kapb +\psi\c (\Psi_g+\nab \psi)+\psi\c\psi\c\psi_g\\
\nab_3\kap&=&-\trchb \c\kap +\psi\c (\Psi_g+\nab \psi)+\psi\c\psi\c\psi_g
\eeaa
with $\Psi_g\in\{\b, \rho,\si, \bb\}$ and $\psi_g\in\{\trch,\omb,\ombt, \eta,\etab,\om,\omt, \trchbt\}$.
Since  $\kapb$ can be expressed in terms of  $\nab\om, \nab\omt$  and   $\bb$
we can also write the first equation in the form
\beaa
\nab_4\kapb&=&\psi\c (\Psi_g+\nab \psi)+\psi\c\psi\c\psi_g
\eeaa
The second equation can be written in the form,
\bea
\nab_3\kap&=&-\trchb_0 \c\kap +\psi\c (\Psi_g+\nab \psi)+\psi\c\psi\c\psi_g\label{eq:transport.kappa.3}
\eea

\subsection{Main $\OO_1$ estimates}
We start by  rewriting  systems  \eqref{eq:nab4Th},   \eqref{eq:nab3Th} and
\eqref{hodge.abstract} in  short form, dropping the reference to signature.
\bea
\nab_4 \Th &=&     \psi\c \big(   \nab  \psi+\Psi \big) +\trchb_0 \c \psi  \c\psi_g+    \psi\c\psi\c\psi_g \label{eq:nab4Th.short} \\
\nab_3 \Th &=&\trchb_0\c \nab\psi     +     \psi \c\big(   \nab  \psi+\Psi \big) +\trchb_0 \c \psi\c\psi_g+    \psi\c\psi\c\psi_g  \label{eq:nab3Th.short}
\eea
where   $\psi_g$ denotes an extended         Ricci coefficient term (i.e.  including $\ombt, \omt$ defined
 below.)  which is not anomalous 
in the $\OS_{0,4}$-norm.). Also, 
\bea
\DD \psi&=&\Th +   \Psi+\trchb_0\c \psi_g+
 \psi\c\psi. \label{hodge.abstract.short}
\eea
\noindent
{\bf Remark 1.}  In reality equation \eqref{eq:nab3Th.short} should also contain a term of
the form $\trchb_0 \Th$ as  seen in \eqref{eq:transport.nabtrch.3}, \eqref{eq:transpmub} and 
\eqref{eq:transport.kappa.3}.  We observe however that such terms can be easily eliminated
by  a standard Gronwall inequality. 

\noindent 
{\bf Remark 2.} The curvature terms  $\Psi$  appearing on the right hand side of  \eqref{eq:nab4Th.short} belong to the admissible\footnote{This
are the curvature components appearing in the main curvature norms  $\RR_0, \RR_1$.}  set $\{\a,\b,\rho,\si,\bb\}$.  Special attention needs to be given to  terms of the form\footnote{such a term appear 
in the transport equation for $\mu$.}   $\chibh\c\a$.

\noindent
{\bf Remark 3.}    The curvature terms  $\Psi$  appearing on the right hand side of  \eqref{eq:nab3Th.short} belong to the admissible\footnote{This
are the curvature components appearing in the main curvature norms  $\RRb_0, \RRb_1$.}  set $\{\b,\rho,\si,\bb, \aa\}$.    Special attention needs to be given to  terms of the form  $\psi\c\b$, since $\RRb_0[\b]$ is
anomalous. We observe however that among         all possible   terms  of the form  $\psi\c\b$,     $\psi$  is never  anomalous.

 \noindent
{\bf Remark 4.}    The curvature terms  $\Psi$  appearing on the right hand side of  \eqref{hodge.abstract.short}  belong to the set  $\{\b,\rho,\si,\bb\}$.

 \noindent 
 {\bf Remark 5.} $\psi_g$ denotes an extended  Ricci coefficient  which is not anomalous in  the $\OO_0$ norm. Whenever we write simply $\psi$
 we allow for the possibility that it may be anomalous.  For example
 the terms of the form $\psi\c\psi$ in  \eqref{hodge.abstract.short}
 may be both anomalous (as happens to be the case  for the div -curl systems 
 for $\eta, \etab$,  due to $\chih\c\chibh$).
 
 \noindent 
 {\bf Remark 6.} Due to the triviality of our initial data at $\ub=0$
 we have
   \beaa\|\Th\|_{\Lsc^2(u,0)}=0.
   \eeaa
    In view
 of the definition of the  $\Th$ we have,
 \bea
   \|\Th\|_{\Lsc^2(0,\ub)}&\les&\OO^{(0)}+\RR^{(0)}.
 \eea

 We start deriving estimates for 
    \eqref{eq:nab4Th.short}.
     As in the proof of the $\OO_0$ estimates,
\beaa
\| \Th \|_{\Lsc^2(u,\ub)}&\les&\| \Th \|_{\Lsc^2(u,0)}+
\int_0^{\ub} \de^{-1}\|\nab_4\Th \|_{\Lsc^2(u,\ub')}
\eeaa
Recall that none of the  $\Lsc^\infty(S)$ norms of the  Ricci coefficients $\psi$ or the  $\Lsc^2(S)$ norms  of their derivatives $\nab\psi$
are  anomalous.  Moreover,   \beaa
\|\psi_g\|_{\Lsc^4(S)}+\de^{1/4}  \|\psi_g\|_{\Lsc^4(S)}   \les \OS_{0,4}(S)\les C
\eeaa
where $C$ is  the constant in  proposition \ref{prop.finalOO0}. Also,
\beaa
\|\psi\|_{\Lsc^\infty(S)} &\les& \de^{1/2}\De_0,\,\, \quad 
\|\nab\psi\|_{\Lsc^2(S)}\les\OS_{1,2},\,\, 
\eeaa

Now, according to  \eqref{eq:nab4Th.short}, for  $\de^{1/2}\De_0\les 1$,
\beaa
\|\nab_4\Th \|_{\Lsc^2(S)}&\les&\|\psi\c \Psi\|_{\Lsc^2(S)}+
\de^{1/2}\|\psi\|_{\Lsc^\infty(S)}\c \| \nab\psi\|_{\Lsc^2(S)} \\ 
&+& \de^{1/2}\|\psi\|_{\Lsc^4(S)} \|\psi_g\|_{\Lsc^4(S)} +   \de\|\psi\|_{\Lsc^\infty(S)}\|\psi\|_{\Lsc^4(S)}   \|\psi_g\|_{\Lsc^4(S)}\\
&\les&\|\psi\c \Psi \|_{\Lsc^2(S)}+\de^{1/2}\De_0   \|\nab\psi\|_{\Lsc^2(S)}+  
 \de^{1/4} C^2  
\eeaa
Recalling the triviality of the initial conditions at $\ub=0$, we deduce,
 \beaa
\| \Th \|_{\Lsc^2(u,\ub)}&\les&
   \int_0^{\ub}\de^{-1}  \|\Th\| _{\Lsc^2(u,\ub') } \,d\ub'   \\
&\les&\de^{-1}  \int_0^{\ub} \| \psi\c \Psi  \|_{\Lsc^2(u,\ub') } \, d\ub' + \De_0\, \de^{1/2} \OS_{1,2} +\de^{1/4} C^2
\eeaa
Among the terms of the form $\psi\c \Psi $ the most 
dangerous\footnote{This is the case for the $\nab_4 $ equation for $\mu$.} is $\chibh\c \a$  which
is anomalous in both $\psi$ and $\Psi$. In this case, recalling  estimate \eqref{main.estim.O04}, 
 \beaa
 \|\chibh\|_{\Lsc^4(S)}&\les& \de^{-1/4} C
 \eeaa
 we deduce,
 \beaa
  \|\chibh\c \a \|_{\Lsc^2(S)}&\les&\de^{1/2} \|\chibh\|_{\Lsc^4(S)}\c\|\a\|_{\Lsc^4(S)}\\
&   \les& \de^{1/4}C
 \left (\|\nab\a\|_{\Lsc^2(S)}^{1/2}\c\|\a\|_{\Lsc^2(S)}^{1/2}+\de^{\frac 14} \|\a\|_{\Lsc^2(S)}\right)
 \eeaa
 All other terms are better in powers of $\de$, i.e.,
  \beaa
  \|\psi \c \Psi  \|_{\Lsc^2(S)}&\les& \de^{1/4}C
 \left (\|\Psi\|_{\Lsc^2(S)}^{1/2}\c\|\nab\Psi \|_{\Lsc^2(S)}^{1/2}+\de^{\frac 14} \|\Psi\|_{\Lsc^2(S)}\right)
 \eeaa
 Therefore, 
 recalling Remark 2 and the definition of the scale
 invariant norms  $\Lsc^2(H_u)$, 
 \beaa
 \de^{-1} \int_0^{\ub} \|\psi\c \Psi \|_{\Lsc^2(u,\ub')}\,  d\ub'
  &\les&C\de^{-3/4}\int_0^{\ub} \|\Psi\|_{\Lsc^2(u,\ub') }^{1/2}  \|\nab \Psi\|_{\Lsc^2(u,\ub') }^{1/2}
  \\
  &\les& C\de^{-1/4} \bigg(  \int_0^{\ub}\|\Psi\|^2_{\Lsc^2(u,\ub')}  d\ub' \,\cdot 
  \int_0^{\ub}\|\nab \Psi\|^2_{\Lsc^2(u,\ub') }   d\ub'\bigg)^{1/2}\\
 &\les& C\, \RR_0^{1/2}\cdot\big(   \RR_0+ \RR_1\big) ^{1/2}
 \eeaa
We have thus established,
\bea
\| \Th \|_{\Lsc^2(u,\ub)}&\les&\de^{1/2}\De_0 \c \OS_{1,2}+C\, \RR_0^{1/2}\cdot\big(   \RR_0+ \RR_1\big) ^{1/2}+\de^{1/4} C^2  \label{eq:estim.Th1}
\eea

We next  estimate the $\Th$ components which
verify the $\nab_3$ equation   \eqref{eq:nab3Th.short}.
The only terms which do not
appear  in \eqref{eq:nab4Th.short} are  of the form, 
  $\trchb_0\nab\psi$.  Thus, 
  exactly as before,
  \beaa
\|\nab_3 \Th \|_{\Lsc^2(S)}&\les&
 \|\psi\c \Psi \|_{\Lsc^2(S)}+(1+\de^{1/2}\De_0 )  \|\nab\psi\|_{\Lsc^2(S)}+   \de^{1/4} C^2  
\eeaa
and,
\beaa
\eeaa
In view of Remark 3 $\Psi\in  \{\b,\rho,\si,\bb, \aa\}$ and there  are  no double anomalous terms $\psi\cdot \Psi$.  Thus, proceeding exactly as above,
\beaa
 \|\Th \|_{\Lsc^2(u,\ub)}&\les& \|\Th \|_{\Lsc^2(u,0)}+\int_0^u \|\nab_3 \Th \|_{\Lsc^2(u',\ub)}
  du'\\
  &\les& \int_0^u  \|\nab\psi\|_{\Lsc^2(u', \ub )}\, du'+ \de^{1/2}\De_0 \c \OS_{1,2}\\
  &+&
  C\de^{1/4} \RRb_0^{1/2}( \RRb_1 +\RRb_0) ^{1/2}+C^2 \de^{1/4}
  \eeaa
 Combining with \eqref{eq:estim.Th1}  we deduce,
 for a constant $C=C( \OO^{(0)},  \RR,  \RRb )$  and
   sufficiently small $\de$,
 \bea
 \|\Th \|_{\Lsc^2(u,\ub)}&\les&C+  \int_0^u
  \|\nab\psi\|_{{\Lsc^2}(u',\ub) }du' +\de^{\frac 12} \De_0 \OO_1 \label{eq:estim.Th3}
 \eea

It remains to discuss estimates  for the Hodge systems
  \eqref{hodge.abstract.short}. The following proposition
  will be needed.  
  \begin{proposition}
\label{prop:curv.estim}
There exists a constant $C=C(\OO^{(0)},  \RR, \RRb)$
 such that if $\de $ is sufficiently small,  the
 following estimates hold true:
 \bea
 \|\b,\rho,\si,\bb\|_{\Lsc^2(S)}&\les&C\\
 \|K\|_{\Lsc^2(S)} &\les&  C
 \eea
\end{proposition}
In view of proposition \ref{prop:Hodge.estim} we derive from \eqref{hodge.abstract.short},
\beaa
\|\nab\psi\|_{\Lsc^2(S)}&\les&\de^{\frac 14} \|K\|^{\frac 12}_{\Lsc^2(S)} \|\psi\|_{\Lsc^4(S)}+
\|\Th\|_{\Lsc^2(S)}\\
&+&\|\Psi\|_{\Lsc^2(S)}+\| \psi_g \|_{\Lsc^2(S)}+
 \|\psi\c\psi\|_{\Lsc^2(S)}.
\eeaa
According to proposition \ref{prop:curv.estim}, 
$ \|K\|_{\Lsc^2(S)} \les C$. Thus even if  the term$ \|\psi\|_{\Lsc^2(S)}$   multiplying $  \|K\|_{\Lsc^2(S)}$ is anomalous\footnote{This situation
 occur only for the Hodge system $\div\chih$, see  \eqref{eq:codazzi.chih},
 since $\OO_0[\chih]$ is anomalous. },   i.e.    
  $  \|\psi\|_{\Lsc^4(u,\ub)}\les\de^{-1/4} \OS_{0,4} \les C\de^{-1/4}  $ 
  we deduce, for some $C=C(\OO^{(0)},\,  \RR, \RRb)$,
  \beaa
  \de^{\frac 14} \|K\|^{\frac 12}_{\Lsc^2(S)} \|\psi\|_{\Lsc^4(S)}&\les& C
  \eeaa
  Also, since $ \|\Psi\|_{\Lsc^2(S)}\les C$ for $\Psi\in\{\b,\rho,\sigma,\betab\}$ and 
    $\|\psi_g\|_{\Lsc^2(S)} \les \OO_0[\psi_g] \les C$  we   deduce,
\beaa
\|\nab\psi\|_{\Lsc^2(S)}&\les&C+\|\Th\|_{\Lsc^2(S)}+ \|\psi\c\psi\|_{\Lsc^2(S)}.
\eeaa
 Among  the remaining quadratic terms   $ \|\psi\c\psi\|_{\Lsc^2(S)}$  we  can have terms  such as 
  $\chih\c\chibh$ , in which both factors are anomalous\footnote{In fact  $\chih\c\chibh $ appears in the Hodge systems for $\eta$ and $\etab$, see formulas \eqref{eq:div.curleta} and  \eqref{eq:div.curletab}. } .
  For such terms  
\beaa
\|  \psi \c\psi  \|_{\Lsc^2(S)}&\les&\de^{\frac 12}\|  \psi  \|_{\Lsc^4(S)}\c \|  \psi  \|_{\Lsc^4(S)}\les  C^2
\eeaa
Henceforth,
\beaa
\|\nab\psi\|_{\Lsc^2(S)}&\les&C^2+ \|\Th\|_{\Lsc^2(S)}
\eeaa
Combining this with \eqref{eq:estim.Th3} we deduce,
\beaa
\|\nab\psi\|_{\Lsc^2(S_{u,\ub})}&\les&C^2+\int_0^u \|\nab\psi\|_{\Lsc^2(S_{u',\ub})} du'+\de^{\frac 12}\De_0\OO_1
\eeaa
from which, by Gronwall,
\beaa
\|\nab\psi\|_{\Lsc^2(S_{u,\ub})}&\les&C^2+\de^{\frac 12}\De_0\OS_{1,2}.
\eeaa
and thus 
$$
\OS_{1,2}+ \|\Th\|_{\Lsc^2(S)}\les C^2
$$
as desired.
We summarize the results in the following 
\begin{proposition}\label{prop:1lev}
Consider systems of the form \eqref{eq:nab4Th},
 \eqref{eq:nab3Th}, \eqref{hodge.abstract} verifying the properties
 discussed in the Remarks 1-5 below.  There exists a constant
 $C=C(\OO^{(0)},\RR,\RRb)$ such that,
 \bea
\|\Th\|_{\Lsc^2(S)}+\OS_{1,2}&\les&C.
\eea
\end{proposition}

\subsection{Curvature Estimates}
In this subsection we prove proposition
\ref{prop:curv.estim} concerning 
 $\Lsc^2(S)$ estimates for the curvature components
$\b,\rho,\si, \bb$. We also provide estimates 
for $\a,\aa$ which will be needed later.
  Recall  the  Bianchi identities,
\beaa
\nab_4\beta+2\trch\beta &=& \div\alpha - 2\omega\beta -(2\zeta+\etab)\alpha\\
\nab_4\rho+\frac 32\trch\rho&=&-\div\beta+\frac 12\chibh\cdot\alpha-\zeta\cdot\beta-2\etab\cdot\beta,\\
\nab_4\sigma+\frac 32\trch\sigma&=&-\div^*\beta+\frac 12\chibh\cdot ^*\alpha-\zeta\cdot^*\beta-2\etab\cdot
^*\beta,\\
\nab_4\betab+\trch\betab&=&-\nabla\rho +^*\nabla\sigma+ 2\omega\betab +2\chibh\cdot\beta-3(\etab\rho-^*\etab\sigma)
\eeaa
 Thus  $\b, \rho, \si, \bb$ verify equations 
 of the form:
 \beaa
 \nab_4\Psi^{(s)}&=&\nab \Psi^{(s+\frac 1 2 )}+\sum_{s_1+s_2=s+1} \psi^{(s_1)}\c\Psi^{(s_2)}
 \eeaa
 Among the curvature terms on the right we have to play special
 attention to multiples of the curvature term $\a$ with signature
 $2$. 
 We write schematically,
 \bea
 \nab_4 \Psi_g =\nab\Psi +\psi\c\Psi 
 \eea
 with $\Psi_g\in \{\b,\rho,\si, \bb\}$  while  $\Psi\in \{\a, \b,\rho,\si, \bb\}$.

 Thus,
 \beaa
 \|\nab_4\Psi_g\|_{\Lsc^2(S)}&\les&  \|\nab \Psi\|_{\Lsc^2(S)}+ \|\a \c\psi \|_{\Lsc^2(S)}+
\de^{1/2}\|\psi\|_{\Lsc^\infty}\|\Psi_g \|_{\Lsc^2(S)}
 \eeaa
  Now,  as in  the estimates for $\Th$ in the previous
  section  the worst case scenario estimate for $\|\a \c\psi \|_{\Lsc^2(S)}$, for 
  anomalous $\psi$, 
  has the form
  $$
   \|\psi\c \a \|_{\Lsc^2(S)}\les C\de^{\frac 14}
 \left (\|\nab\a\|_{\Lsc^2(S)}^{1/2}\c\|\a\|_{\Lsc^2(S)}^{1/2}+\de^{\frac 14} \|\a\|_{\Lsc^2(S)}\right)
$$
 We deduce,
  \beaa
 \|\nab_4\Psi_g\|_{\Lsc^2(S)}&\les&  \|\nab \Psi\|_{\Lsc^2(S)}+\de^{\frac12} \De_0\|\Psi_g\|_{\Lsc^2(S)}\\
 &+ &
 C\de^{\frac 14}
 \left (\|\nab\a\|_{\Lsc^2(S)}^{1/2}\c\|\a\|_{\Lsc^2(S)}^{1/2}+\de^{\frac 14} \|\a\|_{\Lsc^2(S)}\right)
 \eeaa
 from which, 
\beaa
\|\Psi_g\|_{\Lsc^2(u,\ub)}
&\les&\|\Psi_g\|_{\Lsc^2(u,0)}+\RR_1+
 \de^{\frac 12}\De_0\RR_0+C\RR_0^{\frac 12}[\a]\c \RR_1^{\frac 12}[\a]+ C\RR_0[\a]
\eeaa
Thus, since the initial data $\|\Psi_g\|_{\Lsc^2(u,0)}$ is trivial
\beaa
\|\Psi_g\|_{\Lsc^2(u,\ub)}&\les&\RR_1+
 \de^{\frac 12}\De_0\RR_0+C\RR_0^{\frac 12}[\a]\RR_1^{\frac 12}[\a]+ C\RR_0[\a]
 \eeaa
or, with a new constant $C=C(   \OO^{(0)} , \RR, \RRb)$,
\bea
\|\Psi_g\|_{\Lsc^2(u,\ub)}&\les& C
\eea
as desired.

It remains to   estimate the $\Lsc^2(S)$ norm
of the Gauss curvature
 \beaa
K=-\rho+\frac 1 2 \chih\c\chibh-\frac 1 4 \trch\c\trchb=\rho+\frac 1 2 \chih\c\chibh-\frac 1 4 \trch\c\trchb_0-\frac 1 4 \trch\c\trchbt
\eeaa
Thus, 
\beaa
\|K\|_{\Lsc^2(S)}&\les&\|\rho\|_{\Lsc^2(S)}+\de^{1/2} \|\chih\|_{\Lsc^4(S)}\c  \|\chibh\|_{\Lsc^4(S)}\\
&+&\|\trch\|_{\Lsc^2}+\de^{1/2}\De_0 \|\trchbt\|_{\Lsc^2}\\
&\les& C+\de^{1/2}\De_0\RR_0
\eeaa
from which the desired estimate follows. 
\beaa
\|K\|_{\Lsc^2(S)} \les  C(   \OO^{(0)} , \RR, \RRb)
\eeaa
as desired.

 In the next proposition we derive estimates for the remaining 
 curvature components. 
 \begin{proposition}\label{prop:aproh}
There  exists a constant  $C=C(\OO^{(0)}, \RR, \RRb) $
such that for $\de^{\frac 12} \De_0$ sufficiently small
$$
\|\a\|_{\Lsc^2(S)}\le C \de^{-\frac 12},\qquad \|\alphab\|_{\Lsc^2(S)}\le C
$$
\end{proposition}
\begin{proof}
To prove the estimate for $\a$ we use the Bianchi equation for $\nab_3\a$, which can be written 
schematically in the form
$$
\nab_3\a=\trchb_0\c\a+\psi\c\a+\nab\Psi + \psi\c\Psi
$$
with $\Psi$ from the set not containing $\a$. We therefore obtain 
\beaa
\|\a\|_{{\Lsc^2}(S_{u,\ub})}&\les& \|\a\|_{{\Lsc^2}(S_{0,\ub})}+(1+\de^{\frac 12} \De_0)\int_0^u \|\a\|_{{\Lsc^2}(S_{u',\ub})} du'\\
&+&
\RRb_1+\de^{\frac 12} \De_0 \|\Psi\|_{\Lsc^2(\Hb_{\ub})}.
\eeaa
In the worst   case when $\Psi=\b$, which is anomalous,  we have, \, 
$\|\Psi\|_{\Lsc^2(\Hb_{\ub})}\les \de^{-\frac 12}\RRb_0$. 
Thus, by Gronwall, 
\beaa
\|\a\|_{{\Lsc^2}(S_{u,\ub})}&\les&  \|\a\|_{{\Lsc^2}(S_{0,\ub})}+\RRb
\eeaa

Similarly, the equation for $\nab_4\alphab$ has  the form
$$
\nab_4\alphab=\nab\Psi+\psi\c\Psi,
$$
where the curvature term in $\nab\Psi$ is not $\alphab$ and $\Psi\ne\a$ in the nonlinear term.
Therefore, using the triviality of initial data 
\beaa
\|\alphab\|_{\Lsc^2(u,\ub)}\les \RR_1+\de^{\frac 12} \De_0\, \int_0^u\big(  \|\aa\|_{\Lsc^2(u',\ub)}+  \|\Psi_g\|_{\Lsc^2(u',\ub)} )du'
\eeaa
with $\Psi_g\in{\rho,\si,\bb}$.
 The result   follows then easily  by Gronwall and the $\Lsc^2(H)$ curvature bounds for $\Psi_g$.
\end{proof}

\section{Second angular  derivative   estimates for the Ricci coefficients  }

To derive second  angular derivative estimates for the Ricci coefficients
we  differentiate 
\eqref{eq:nab4Th},  \eqref{eq:nab3Th} 
 and \eqref{hodge.abstract} with respect to $\nab$.

 \subsection{Basic equations}
Based on the experience with the first derivative estimates we expect that
the $\nab_3$ equation for $\nab\Th$ is  slightly more challenging as it contains a lot more $\trchb$ terms. Thus, differentiating  
\eqref{eq:nab3Th.short} we derive,
\beaa
\nab_3 \nab \Th &=&\trchb_0 \big( \nab \Th+\nab\Psi+\nab^2\psi) +
\psi \c \big( \nab \Th+\nab\Psi+\nab^2\psi) +\nab\psi\c \big( \Th+\Psi+\nab\psi  \big)\\
&+&\trchb_0 \psi\c\nab\psi+ \psi\c\psi\c\nab \psi
 +[\nab_3,\nab]\Th
\eeaa
According to commutation formulae  of  lemma  \eqref{le:comm} we write
symbolically,
\beaa
[\nab_3,\nab]\Th&=&\trchb \c\nab\Th+\chibh \c\nab \Th+\Psi\c\Th+\trchb\c\psi\c \Th+      \psi\c\psi\c\Th+\psi\c\nab_3\Th\\
&=&\trchb_0\nab\Th+\psi\c\nab\Th +\Psi\c\Th+\trchb_0\c\psi\c \Th+      \psi\c\psi\c\Th+\psi\c\nab_3\Th
\eeaa
Hence,
\beaa
\nab_3 \nab \Th &=\trchb_0 \big( \nab \Th+\nab\Psi+\nab^2\psi) +
\psi \c \big( \nab \Th+\nab\Psi+\nab^2\psi) +\nab\psi\c \big( \Th+\Psi+\nab\psi  \big)\\
&+\Th\c\Psi+ \trchb_0\big( \psi\c\nab\psi  +\psi\c\Th   )+ \psi\c\big(\psi\c\nab \psi+\psi\c\Th+\nab_3\Th)
\eeaa
Ignoring the term  of the form $\trchb_0  \nab \Th$ which can be easily
eliminated by Gronwall,  and observing that  $\Th$  and $\nab \Th$ on the left
 can be  can be expressed in terms 
of $\nab \psi$ and $\Psi$, respectively, $\nab^2 \psi$ and $\nab \Psi$, we
write,
\begin{equation}
\label{eq:nabth-3}
\begin{split}
\nab_3 \nab \Th &=(\trchb_0 +\psi)(\nab\Psi+\nab^2\psi)+ \trchb_0\big( \psi\c\nab\psi  +\psi\c\Th)
 +\nab\psi\c \big( \Psi+\nab\psi  \big)\\
&+\Th\c\Psi+ \psi\c\big(\psi\c\nab \psi+\psi\c\Th+\nab_3\Th)\\
&= F_1+F_2+F_3+F_4+F_5 
\end{split}
\end{equation}
Similarly,
\beaa
\nab_4 \nab \Th =
\psi \c \big( \nab\Psi+\nab^2\psi) +\nab\psi\c \big(\Psi+\nab\psi  \big)+\Th\c\Psi
+ \psi\c\psi\c\nab \psi
 +[\nab_4,\nab]\Th
\eeaa
and
\beaa
[\nab_4,\nab]\Th
=\psi\c\nab\Th +\Psi\c\Th+      \psi\c\psi\c\Th+\psi\c\nab_4\Th
\eeaa
so that\footnote{Observe that the structure of } 
\begin{equation}
\label{eq:nabth-4}
\begin{split}
\nab_4 \nab \Th &=
\psi \c \big( \nab\Psi+\nab^2\psi) +\nab\psi\c \big(\Psi+\nab\psi  \big)
+ \psi\c\left(\psi\c\nab \psi+\psi\c\Th+\nab_4\Th\right)\\&=G_1+G_3+G_4+G_5
\end{split}
\end{equation}
Equations \eqref{eq:nabth-3},\eqref{eq:nabth-4} will be combined with the differentiated 
Hodge system for $\psi$ in \eqref{hodge.abstract.short}:
\bea
\DD^*\DD \psi&=&\DD^* \left(\Th +  \Psi+\psi \c \psi +      \trchb_0\c \psi\right),
\eea
which can be schematically written in the form 
$$
\Delta\psi=K\psi +\nab\Th+\nab\Psi+\nab\psi\c\psi+\trchb_0\c\nab\psi 
$$
\subsection{ Estimates for $\nab \Th$, $\nab^2 \psi$}
We now collect estimates for the terms on the right hand side of the transport equations
\eqref{eq:nabth-3},\eqref{eq:nabth-4}:
\beaa
\| F_1\|_{\Lsc^2(S)}&\les &(1+ \de^{1/2} \De_0)\big(\|\nab^2  \psi\|_{\Lsc^2(S)} +\|\nab \Psi\|_{\Lsc^2(S)} \big)\\
\| F_2\|_{\Lsc^2(S)}&\les&\de^{1/2}\De_0(
 \|\Th\|_{\Lsc^2(S)}+\|\nab\psi\|_{\Lsc^2(S)})\\
\| F_3\|_{\Lsc^2(S)}&\les&\de^{1/2} \|\nab\psi \|_{\Lsc^4(S)}\c \big(\|\nab\psi\|_{\Lsc^4(S)}+\|\Psi\|_{\Lsc^4(S)} \big)\\
\| F_4\|_{\Lsc^2(S)}&\les&\de^{1/2}\|\Th\|_{\Lsc^4(S)}\c \|\Psi\|_{\Lsc^4(S)}\\
\| F_5\|_{\Lsc^2(S)}&\les&\de^{1/2}\De_0\big(\de^{\frac 12}\De_0 \|(\nab\psi, \Th) \|_{\Lsc^2(S)}+      \|\nab_3\Th\|_{\Lsc^2(S)}\big)
\eeaa
Similarly,
\beaa
\| G_1\|_{\Lsc^2(S)}&\les &\de^{1/2} \De_0 \, \big( \|\nab  \Th\|_{\Lsc^2(S)}+ \|\nab^2  \psi\|_{\Lsc^2(S)} +\|\nab \Psi\|_{\Lsc^2(S)} \big)\\
\| G_3\|_{\Lsc^2(S)}&\les&\de^{1/2} \|\nab\psi \|_{\Lsc^4(S)}\c \big(\|\nab\psi\|_{\Lsc^4(S)}+\|\Psi\|_{\Lsc^4(S)} \big)\\
\| G_4\|_{\Lsc^2(S)}&\les&\de^{1/2}\|\Th\|_{\Lsc^4(S)}\c \|\Psi\|_{\Lsc^4(S)}\\
\| G_6\|_{\Lsc^2(S)}&\les&\de^{1/2}\De_0\c\big(\de^{\frac 12}\De_0\c   \|(\nab\psi, \Th) \|_{\Lsc^2(S)}+      \|\nab_4\Th\|_{\Lsc^2(S)}\big)
\eeaa
We note that the curvature terms $\Psi$ present in the $F$ terms belong to the admissible set
$\{\b,\rho,\sigma,\betab,\underline\a\}$ while the curvature terms $\Psi$ appearing in the $G$ terms 
belong to the set $\{\a,\b,\rho,\sigma,\betab\}$. We also recall that according to the $\OS_{1,2}$ estimates and their 
consequences proved in the previous section
$$
\|\nab\psi\|_{L^2(S)}+\|\Th\|_{L^2(S)}+\|\nab_4\Th\|_{L^2(H)}+\|\nab_3\Th\|_{L^2(\Hb)}\le C
$$
Using the $L^4$ interpolation estimate from \eqref{eq:L4-glob-sc} which imply that 
\begin{align*}
&\|\nab\psi\|_{\Lsc^4(S)}\les \|\nab\psi\|^{\frac 12}_{\Lsc^2(S)}\|\nab^2\psi\|^{\frac 12}_{\Lsc^2(S)}+
\de^{\frac 14} \|\nab\psi\|_{\Lsc^2(S)}\les C\|\nab^2\psi\|^{\frac 12}_{\Lsc^2(S)}+
\de^{\frac 14} C,\\
&\|\Th\|_{\Lsc^4(S)}\les \|\Th\|^{\frac 12}_{\Lsc^2(S)}\|\nab\Th\|^{\frac 12}_{\Lsc^2(S)}+
\de^{\frac 14} \|\Th\|_{\Lsc^2(S)}\les C\|\nab\Th\|^{\frac 12}_{\Lsc^2(S)}+
\de^{\frac 14} C,\\
&\|\Psi\|_{\Lsc^4(S)}\les \|\Psi\|^{\frac 12}_{\Lsc^2(S)}\|\nab\Psi\|^{\frac 12}_{\Lsc^2(S)}+
\de^{\frac 14} \|\Psi\|_{\Lsc^2(S)}
\end{align*}
we obtain for $\de^{\frac 12} \De_0$ sufficiently small
\begin{align*}
\|\nab\Th\|&_{\Lsc^2(u,\ub)}\les\|\nab\Th\|_{\Lsc^2(0,\ub)}+
\int_0^u\|\nab_3\nab\Th\|_{\Lsc^2(u',\ub)} du'\\
&\les \|\nab\Th\|_{\Lsc^2(0,\ub)}+C\int_0^u \big(  \|\nab^2  \psi\|_{\Lsc^2} 
+\|\nab \Psi\|_{\Lsc^2}\big) du'\\ &+ \de^{\frac 12}C \int_0^u\left (\|\nab^2\psi\|^{\frac 12}_{\Lsc^2(S)}\|\Psi\|^{\frac 12}_{\Lsc^2(S)}\|\nab\Psi\|^{\frac 12}_{\Lsc^2(S)}+\de^{\frac 14} \|\nab^2\psi\|^{\frac 12}_{\Lsc^2(S)}
\|\Psi\|^{\frac 12}_{\Lsc^2(S)}\right) du'\\ &+\de^{\frac 12} C\int_0^u\left ( \|\nab\Th\|^{\frac 12}_{\Lsc^2(S)} 
\|\Psi\|^{\frac 12}_{\Lsc^2(S)}\|\nab\Psi\|^{\frac 12}_{\Lsc^2(S)}+\de^{\frac 14}\|\nab\Th\|^{\frac 12}_{\Lsc^2(S)} 
\|\Psi\|_{\Lsc^2(S)}\right) du'+\de^{\frac 12}C
\end{align*}
We kept track of the terms containing $\|\Psi\|_{\Lsc^2(S)}$ as they may lead to the potentially anomalous norm
$\|\Psi\|_{\Lsc^2(\Hb)}$ in the case of $\Psi=\b$. However, even in that case 
$$
\|\Psi\|_{\Lsc^2(\Hb)}\les \de^{-\frac 12} \RRb_0 
$$
By Gronwall, and recalling the definition\footnote{note again that 
$\aa$ does not appear among the $\Psi$'s}  of $\RRb_1$
\begin{equation}\label{eq:nab3th}
\|\nab\Th\|_{\Lsc^2(u,\ub)}\les\|\nab\Th\|_{\Lsc^2(0,\ub)}+C\int_0^u  \|\nab^2  \psi\|_{\Lsc^2(u', \ub)} du'+ C\RRb_1.
\end{equation}
In view of the estimates for the $G$ terms we similarly obtain
 \begin{equation}\label{eq:nab4th}
\|\nab\Th\|_{\Lsc^2(u,\ub)}\les\|\nab\Th\|_{\Lsc^2(u,0)}+C\int_0^{\ub}  \|\nab^2  \psi\|_{\Lsc^2(u, \ub')} d\ub' + C\RR_1.
\end{equation}
We now couple this with the second derivative estimates for the Hodge system 
$$
D\psi=\Th+\Psi+\trchb_0\psi+\psi\c\psi.
$$
Using Proposition \ref{prop:Hodge.estim-2} we deduce
\beaa
\|\nab^2\psi\|_{\Lsc^2(S)}&\les&\de^{\frac 12} \|K\|_{\Lsc^2(S)} \|\psi\|_{\Lsc^\infty(S)}+\de^{\frac 14} \|K\|^{\frac 12}_{\Lsc^2(S)} 
\|\nab\psi\|_{\Lsc^4(S)}\\ &+&\|\nab \Th\|_{\Lsc^2(S)}+\|\nab \Psi\|_{\Lsc^2(S)}+
\|\trchb_0\nab \psi\|_{\Lsc^2(S)}+\|\psi\c\nab \psi\|_{\Lsc^2(S)}
\eeaa
By Proposition \ref{prop:curv.estim}, $\|K\|_{\Lsc^2(S)}\les C$ with a constant $C=C(\OO^{(0)} ,\RR,\RRb)$.
Therefore,
\beaa
\|\nab^2\psi\|_{\Lsc^2(S)}&\les&\de^{\frac 12} C\De_0+\de^{\frac 14} C
\left (\|\nab^2\psi\|^{\frac12}_{\Lsc^2(S)}\|\nab\psi\|^{\frac12}_{\Lsc^2(S)}+
\de^{\frac 14}\|\nab\psi\|_{\Lsc^2(S)}\right)\\ &+&\|\nab \Th\|_{\Lsc^2(S)}+\|\nab \Psi\|_{\Lsc^2(S)}+
\|\nab \psi\|_{\Lsc^2(S)}+\de^{\frac 12} C\De_0\|\nab \psi\|_{\Lsc^2(S)}
\eeaa
Using Cauchy-Schwarz and the boundedness of the $\OS_{1,2}$ norm we then obtain
\begin{equation}\label{eq:nab2}
\|\nab^2\psi\|_{\Lsc^2(S)}\les C + \|\nab \Th\|_{\Lsc^2(S)}+\|\nab \Psi\|_{\Lsc^2(S)}.
\end{equation}
We note that the curvature terms $\Psi$ involved in the above inequality belong to the set
$\{\b,\rho,\sigma,\betab\}$. In particular,
$$
\|\nab \Psi\|_{\Lsc^2(H)}\les \RR_1,\qquad \|\nab \Psi\|_{\Lsc^2(\Hb)}\les \RRb_1.
$$
Thus, substituting the estimate for $\|\nab^2\psi\|_{\Lsc^2(S)}$ into \eqref{eq:nab3th} and \eqref{eq:nab4th}
and using Gronwall we obtain
\begin{align*}
&\|\nab\Th\|_{\Lsc^2(u,\ub)}\les\|\nab\Th\|_{\Lsc^2(0,\ub)}+ C\RRb_1,\\
&\|\nab\Th\|_{\Lsc^2(u,\ub)}\les\|\nab\Th\|_{\Lsc^2(u,0)}+ C\RR_1.
\end{align*}
This, together with \eqref{eq:nab2}, in turn, implies 
\begin{proposition}\label{prop:2lev}
There exists a constant $C=C(\OO^{(0)} ,\RR, \RRb)$
such that all second derivatives  $\nab^2\psi$ of the Ricci coefficients
$\psi\in \{\trch,\chih,\eta, \etab, \om, \omb, \chibh, \trchbt\}$ and the first derivatives of 
the quantities $\Th\in\{\nab\trch,\, \div\eta+\rho,\, \div\etab+\rho,\, \nab\om+\nabd\omt-\frac 1 2 \b,\,
-\nab\om+\nabd\omt-\frac 1 2 \betab,\, \nab\trchb\}$
verify,
\beaa
\|\nab\Th\|_{\Lsc^2(u,\ub)}+
\|\nab^2\psi\|_{\Lsc^2(H_u)}+\|\nab^2\psi\|_{\Lsc^2(H_{\ub})}&\les& C.
\eeaa
\end{proposition}
\subsection{$\OS_{1,4}$ estimates}
As a corollary of proposition \ref{prop:2lev}, together with
corollary \ref{cor:interpol} we also have,
\begin{corollary}
\label{co:OS14}
There exists a constant $C=C(\OO^{(0)} ,\RR, \RRb)$
such that, for $\de^{1/2}\De_0$ sufficiently small,
\bea
\OS_{1,4}&\les& C.
\eea

\end{corollary}
We end this section by deriving a slightly more refined estimate on the second angular derivatives of 
$\eta$. These estimates are needed in the application to the problem of formation of a trapped surface.
We review the system of equations for $\eta$, written schematically it has the form
\begin{align*}
&\curl \eta=\sigma+\psi\c\psi,\qquad \div\eta=-\mu-\rho,\\
&\nab_4\mu=\psi\c(\nab\psi+\Th+\Psi+\psi\c\psi).
\end{align*}
We note the absence of $\trchb_0$ terms in this system. Applying $\DD^*$ to the Hodge system for $\eta$ 
and commuting the equation for $\mu$ with $\nab$ we obtain
\begin{align*}
&\Delta\eta=\nab\sigma+\nab\rho+\nab\mu+\nab\psi\c\psi+K \eta,\\
&\nab_4\nab\mu=\nab\psi\c(\nab\psi+\Th+\Psi+\psi\c\psi)+\psi\c(\nab^2\psi+\nab\Th+\nab\Psi+\nab\psi\c\psi)
\end{align*}
The absence of $\trchb_0$ terms allows us to estimate $\nab\mu$ in terms of its (trivial) data on $H_0$
and an error term of size $\de^{\frac 12}$. To show that we bound
\begin{align*}
\|&\nab\psi\c(\nab\psi+\Th+\Psi+\psi\c\psi)\|_{\Lsc^2(H_u)}\\ &\les \de^{\frac 12} \|\nab\psi\|_{\Lsc^4(S)}
\left ( \|\nab\psi\|_{\Lsc^4(S)}+ \|\Th\|_{\Lsc^4(S)}+ \|\Psi\|_{\Lsc^4(S)}+\de^{\frac 12} \|\psi\|_{\Lsc^\infty}
\|\nab\psi\|_{\Lsc^4(S)}\right) \\ &+ \de^{\frac 12} \|\psi\|_{\Lsc^\infty}\left ( \|\nab^2\psi\|_{\Lsc^2(H_u)}+ \|\nab\Th\|_{\Lsc^2(H_u)}+
 \|\nab\Psi\|_{\Lsc^2(H_u)}+ \de^{\frac 12}\|\psi\|_{\Lsc^\infty}\|\nab\psi\|_{\Lsc^2(H_u)}\right)\\&\les \de^{\frac 12} C
\end{align*}
In the final estimate the only dangerous term is $ \|\Psi\|_{\Lsc^4(S)}$, which may be $\de^{-\frac 14}$ anomalous
in the case of $\Psi=\alpha$. It is not difficult to check however that $\Psi=\a$ does not appear in this system but
even if it did the size of the error term would have been $\de^{\frac 14}$ instead of $\de^{\frac 12}$. 
As a result of this estimate and the trivial data for $\nab\mu$ we obtain
$$
\|\nab\mu\|_{\Lsc^2(S)}\les \de^{\frac 12} C.
$$
To estimate $\eta$ we remember that $K=\rho +\trchb_0\c\psi_g+\psi\c\psi$. Therefore,
\begin{align*}
\|&\Delta\eta\|_{\Lsc^2(H_u)} \les \|\nab\rho\|_{\Lsc^2(H_u)}+\|\nab\si\|_{\Lsc^2(H_u)}+\|\nab\mu\|_{\Lsc^2(H_u)}
\\ &+\de^{\frac 12} \|\psi\|_{\Lsc^\infty} \left (\|\nab\psi\|_{\Lsc^2(H_u)} + \|\rho\|_{\Lsc^2(H_u)}+\|\psi_g\|_{\Lsc^2(H_u)}
+\de^{\frac 12} \|\psi\|_{\Lsc^\infty}\c\|\psi\|_{\Lsc^2(H_u)}\right)\\&\les 
\|\nab\rho\|_{\Lsc^2(H_u)}+\|\nab\si\|_{\Lsc^2(H_u)}+\de^{\frac 12} C.
\end{align*}
Using the B\"ochner identity we obtain 
\begin{align*}
\|\nab^2\eta\|_{\Lsc^2(H_u)}&\les \|\Delta\eta\|_{\Lsc^2(H_u)}+\de^{\frac 12} \|K\|_{\Lsc^2(H_u)}\|\psi\|_{\Lsc^\infty}+
\de^{\frac 14} \|K\|_{\Lsc^2(H_u)}^{\frac 12} \|\nab\psi\|_{\Lsc^4(S)}\\ &\les \|\nab\rho\|_{\Lsc^2(H_u)}+\|\nab\si\|_{\Lsc^2(H_u)}+\de^{\frac 14} C.
\end{align*}
The same estimates also hold along the $\Hb_{\ub}$ hypersurfaces.

We summarize this in a proposition.
\begin{proposition}\label{prop:bettereta}
The Ricci coefficient $\eta$ verifies the estimate
\begin{align*}
&\|\nab^2\eta\|_{\Lsc^2(H_u)}\les \|\nab\rho\|_{\Lsc^2(H_u)}+\|\nab\si\|_{\Lsc^2(H_u)}+\de^{\frac 14} C,\\
&\|\nab^2\eta\|_{\Lsc^2(\Hb_{\ub})}\les \|\nab\rho\|_{\Lsc^2(\Hb_{\ub})}+\|\nab\si\|_{\Lsc^2(\Hb_{\ub})}+\de^{\frac 14} C.
\end{align*}
\end{proposition}

\section{Remaining first and second derivative estimates}
In the previous sections we have derived estimates on the first and second angular derivatives of the 
Ricci coefficients. In this section examine 
their $\nab_3$, $\nab_4$, $\nab\nab_4$ and $\nab\nab_3$ derivatives. 
\subsection{Direct $\nab_3,\nab_4 $ estimates.}
These are   derived directly  from  the null structure equations (see section \ref{sub:nullstr}). 
\begin{proposition}\label{prop:1der}
There exists a constant $C=C(\OO^{(0)} ,\RR,\RRb)$ such that for $\de^{\frac 12}\De_0$ 
sufficiently small and any $S=S_{u,\ub}$
\begin{align*}
\|\nab_4\trch\|_{\Lsc^2(S)}+ \|\nab_4\eta\|_{\Lsc^2(S)}
+ \|\nab_4\omb\|_{\Lsc^2(S)}+ \|\nab_4\trchb\|_{\Lsc^2(S)}&\le C,\\
\|\nab_3{\widetilde{\trchb}}\|_{\Lsc^2(S)}+\|\nab_3\etab\|_{\Lsc^2(S)}
+ \|\nab_3\om\|_{\Lsc^2(S)}+ \|\nab_3\trch\|_{\Lsc^2(S)}&\le C,\\
 \|\nab_4\chih\|_{\Lsc^2(S)}\,+\|\nab_4\chibh\|_{\Lsc^2(S)}+\,
  \|\nab_3\chih\|_{\Lsc^2(S)}+\,   \|\nab_3\chibh\|_{\Lsc^2(S)}&\le C\,\de^{-\frac 12}.
\end{align*}
\end{proposition}
\noindent  {\bf Remark}.  Note the anomalous estimates of the last line.   The anomaly  of
  $\nab_4\chih$   is due to  the curvature term $\a$ in the second equation in  \eqref{null.str3}. The anomaly   of   $\nab_3\chibh$ 
is due to  the term $\trchb\c\chibh$ in the fourth equation   in  \eqref{null.str3} . The anomalies  for $\nab_3\chih$
 and  $\nab_4\chibh$  are  explained  by the presence of $\trchb\chih$ in  both equations of \eqref{null.str4}.

\begin{proof}
The claimed estimates  follow directly from   all the estimates derived so far. 
We need the  full set of $\|\Psi\|_{L^2(S)}$ estimates for all null  curvature components $\Psi$ which  were derived in propositions  \ref{prop:curv.estim} and \ref{prop:aproh}. We also need to make use of the $\OS_{0,2}$ estimates of  proposition \ref{prop.finalOS02}. 
 As an example we  prove the estimate  for $\nab_4\chih$ in more detail.
 We start with $\nab_4\chih =-\trch\chih-2\om\chih-\a$ which we write in the form,
$$
\nab_4\chih= \psi_g\c\chih +\a
$$
As a result,
\beaa
\|\nab_4\chih\|_{\Lsc^2(S)} &\les& \|\psi_g\c\chih \|_{\Lsc^2(S)} +\|\a\|_{\Lsc^2(S)}\\
&\les& \de^{1/2} \|\psi_g\|_{\Lsc^4(S)}\c\|\chih\|_{\Lsc^4}+\|\a\|_{\Lsc^2(S)}\\
&\les& \de^{1/4} \OS_{0,4}^2 +C\de^{-1/2}\les C\de^{-1/2} 
\eeaa
as desired. 
Similarly we write,
\beaa
\nab_3 \chih&=&\trchb_0\c\psi_b+\psi_g\c\psi_b+\nab\psi+\Psi_g,
\eeaa
with $\psi_g$,  $\Psi_g$  non- anomalous and $\psi_b$ anomalous.
Hence,
\beaa
\|\nab_3\chih\|_{\Lsc^2(S)} &\les&\|\psi_b\|_{\Lsc^2(S)}+ \|\psi_g \|_{\Lsc^4(S)} \c \|\psi \|_{\Lsc^4(S)}+\|\nab\psi\|_{\Lsc^2(S)}+\|\Psi_g\|_{\Lsc^2(S)}\\
&\les& \de^{-1/2} C+\de^{-1/4} C^2 +C
\eeaa

More generally, all of our null structure equations have the form 
\begin{align*}
&\nab_4\psi=\trchb_0\c\psi+\psi\c\psi+\nab\psi+\Psi,\\
&\nab_3\psi=\trchb_0\c\psi+\psi\c\psi+\nab\psi+\Psi,
\end{align*}
and one can easily see that the only anomalies occur 
for  $\nab_{3},\nab_{4} $ of  $\chi, \chibh$ .
\end{proof}
\subsection{Estimates for $\nab_3\eta,\nab_4\etab,\nab_3\omb, \nab_4\om$.}\label{sec:nodirect}
The above proposition does not address the fate of $\nab_3\eta,\nab_4\etab,\nab_3\omb$ and 
$\nab_4\om$ derivatives which do not appear in the null structure equations. These can be estimated by commuting  the valid  transport equations for these quantities 
with the desired derivative.
\begin{proposition}\label{prop:1der-eta}
There exists a constant $C=C(\OO^{(0)}, \RR, \RRb)$ 
such that for $\de^{\frac 12} \De_0$ sufficiently small
$$
 \|\nab_4\etab\|_{\Lsc^2(S)}
+ \|\nab_4\om\|_{\Lsc^2(S)}+ \|\nab_3\eta\|_{\Lsc^2(S)}
+ \|\nab_3\omb\|_{\Lsc^2(S)}\le C.
$$
\end{proposition}
\begin{proof}
As all the arguments are similar we will only derive the estimate for $\nab_4\etab$. Commuting the 
transport equation 
$$
\nab_3\etab=-\frac 1 2 \trchb(\etab-\eta) -\chibh\c(\etab-\eta) +\betab
$$
with $\nab_4$ (according to Lemma \ref{le:comm}) we obtain
\beaa
\nab_3(\nab_4\etab)&=& -\frac 1 2\nab_4 \trchb(\etab-\eta) -
\frac 1 2 \trchb\nab_4(\etab-\eta)  \\
&-&\nab_4\chib\c(\etab-\eta) -\chib\c\nab_4(\etab-\eta)+\nab_4\betab\\
&-&2(\etab-\eta)\c \nab \etab +2\omb \nab_4 \etab -2\om\nab_3\etab-2
(\eta_a\etab_b-\etab_b\etab_a-\in_{ab} \sigma) \etab_b
\eeaa
which we write symbolically,
\beaa
\nab_3(\nab_4\etab)&=&\trchb_0\c(\nab_4\psi_g+\nab_4\etab+\psi\c\psi_g) +\psi\c (\nab_4\psi+\nab_4\etab)\\&+&
\psi\c(\nab\psi+\Psi_g+\psi\c\psi_g)
+\nab_4\betab
\eeaa
{\bf Remark}. In the above expression,   $\nab_4\psi$ denotes  quantities already controlled  according to the
previous proposition and,  among them,   $\nab_4\psi_g$ denote  those 
 which are not anomalous. Also $\Psi_g$  is a  curvature component different
 from $\a$.
  Furthermore we can eliminate $\nab_4 \b$ according to the null Bianchi equations
$$
\nab_4\betab+\trch\betab=-\nabla\rho +^*\nabla\sigma+ 2\omega\betab +2\chibh\cdot\beta-3(\etab\rho-^*\etab\sigma)
$$
Thus,
\beaa
\nab_3(\nab_4\etab)&=&\trchb_0\c(\nab_4\psi_g+\nab_4\etab+\psi\c\psi_g) +\psi\c (\nab_4\psi+\nab_4\etab)\\
&+&\psi\c(\nab\psi+\Psi_g+\psi\c\psi_g)
+\nab\Psi_g.
\eeaa 
Therefore,
\beaa
\|\nab_3(\nab_4\etab)\|_{\Lsc^2(S)}&\les&(1+\de^{1/2}\De_0)\|\nab_4\etab\|_{\Lsc^2(S)}
 + \|\nab_4\psi_g\|_{\Lsc^2(S)} +\de^{1/2}\De_0|\nab_4\psi\|_{\Lsc^2(S)}    \\
&+&\|\psi\|_{\Lsc^\infty(S)}\big(\|\psi_g\|_{\Lsc^2(S)}+
\|\nab\psi\|_{\Lsc^2(S)}+\|\Psi_g\|_{\Lsc^2(S)}\big)+\|\nab\Psi_g\|_{\Lsc^2(S)}\\
&\les&(1+\de^{1/2}\De_0)\|\nab_4\etab\|_{\Lsc^2(S)}+\|\nab\Psi_g\|_{\Lsc^2(S)}+C
\eeaa
Therefore, 
\beaa
\|\nab_4\etab\|_{\Lsc^2(u,\ub)}&\les& \|\nab_4\etab\|_{\Lsc^2(S(0,\ub)}+
\int_0^u \|\nab_3\nab_4\etab\|_{\Lsc^2(u',\ub)} du'\\
&\les& \|\nab_4\etab\|_{\Lsc^2(0,\ub)}+(1+\de^{\frac 12} \De_0) \int_0^u 
\|\nab_4\etab\|_{\Lsc^2(u',\ub)} du'\\
&+&\int_0^u\|\nab\Psi_g\|_{\Lsc(u', \ub)} du' +C\\
&\les& \OO^{(0)} +(1+\de^{\frac 12} \De_0) \int_0^u 
\|\nab_4\etab\|_{\Lsc^2(u',\ub)} du'   +   \RRb_1+C
\eeaa
Thus by Gronwall,
\beaa
\|\nab_4\etab\|_{\Lsc^2(u,\ub)}&\les& \OO^{(0)}+C.
\eeaa
\end{proof}
\subsection{Direct angular derivative estimates. }
Here we derive angular derivative estimates for all 
the quantities which appear in proposition \ref{prop:1der}.
We shall first prove the following:
\begin{lemma} If $\de^{1/2}\De_0$ is small we have with a constant
$C=C(\OO^{(0)},\RR, \RRb)$, for all  Ricci coefficients $\psi$,
\beaa
\|[\nab_4,\nab]\psi\|_{\Lsc^2(S)}&\les&  C\\
\|[\nab_4,\nab]\psi\|_{\Lsc^2(S)}&\les& C
\eeaa
As a corollary we also have,
\begin{align*}
&\|[\nab_4,\nab]\psi\|_{\Lsc^2(H)}+\|[\nab_4,\nab]\psi\|_{\Lsc^2(\Hb)}\les  C
,\\
&\|[\nab_3,\nab]\psi\|_{\Lsc^2(H)}+\|[\nab_3,\nab]\psi\|_{\Lsc^2(\Hb)}\les 
C
\end{align*}
\end{lemma}
\label{le:comm.nab4nab3psi}
\begin{proof}
We write, 
\begin{align*}
&[\nab_4,\nab]\psi=\psi\c\nab\psi+\b\c\psi+\psi_g\nab_4\psi,\\
&[\nab_3,\nab]\psi=\trchb_0\c\nab\psi+\psi\c\nab\psi+\bb\c\psi+\psi_g\nab_3\psi,
\end{align*}
Hence, in view of  the previous estimates   $\OS_{1,2}\les C$, \, $\|\b\|_{\Lsc^2(S)}\les C$ and  
  the possibly anomalous estimate $\|\nab_4\psi \|_{\Lsc^2(S)}\les C\de^{-1/2} $,
 we derive,
\beaa
\|[\nab_4,\nab]\psi\|_{\Lsc^2(S)}&\les& \de^{\frac 12} \De_0 \left (\|\nab\psi\|_{\Lsc^2(S)}+
\|\b\|_{\Lsc^2(S)}+\|\nab_4\psi \|_{\Lsc^2(S)}   \right )\les  C
\eeaa
Similarly,
\beaa
\|[\nab_3,\nab]\psi\|_{\Lsc^2(S)}&\les &(1+\de^{1/2}\De_0) \|\nab\psi\|_{\Lsc^2(S)}\\&+&
\de^{\frac 12} \De_0 \left (
\|\bb\|_{\Lsc^2(S)}+\|\nab_3\psi\|_{\Lsc^2(S)}\right)\les C.
\eeaa
from which  the estimates of the lemma quickly follow by integration. 
\end{proof}

\begin{proposition}\label{prop:2der}
There exists a constant $C=C(\OO^{(0)}, \,\RR,\, \RRb)$ such that for $\de^{\frac 12}\De_0$
sufficiently small
\begin{align*}
&\|\nab\nab_4\chi\|_{\Lsc^2(H)}+ \|\nab\nab_4\eta\|_{\Lsc^2(H)}
+ \|\nab\nab_4\omb\|_{\Lsc^2(H)}+ \|\nab\nab_4\chib\|_{\Lsc^2(H)}\le C,\\
&\|\nab\nab_4\trch\|_{\Lsc^2(\Hb)}+ \|\nab\nab_4\eta\|_{\Lsc^2(\Hb)}
+ \|\nab\nab_4\omb\|_{\Lsc^2(\Hb)}+ \|\nab\nab_4\chib\|_{\Lsc^2(\Hb)}\le C,\\
&\|\nab\nab_3\chib\|_{\Lsc^2(\Hb)}+\|\nab\nab_3\etab\|_{\Lsc^2(\Hb)}
+ \|\nab\nab_3\om\|_{\Lsc^2(\Hb)}+ \|\nab\nab_3\chi\|_{\Lsc^2(\Hb)}\le C,\\
&\|\nab\nab_3{\widetilde{\trchb}}\|_{\Lsc^2(H)}+\|\nab\nab_3\etab\|_{\Lsc^2(H)}
+ \|\nab\nab_3\om\|_{\Lsc^2(H)}+ \|\nab\nab_3\chi\|_{\Lsc^2(H)}\le C
\end{align*}
\end{proposition} 
\begin{remark}
Note the absence of anomalies. This is analogous to the situation with $\OS_{1,2}$  
estimates: additional $\nab$ derivatives eliminate   the anomalies due  $\a$ and Ricci coefficients $\chih,\chibh$. 
\end{remark}
\begin{remark}
The quantities $\nab\nab_4\chih$ 
and $\nab\nab_3\chibh$ are controlled only along $H$ and $\Hb$ respectively. This is 
due to the absence of the corresponding estimates for $\nab\a$ and $\nab\alphab$ along 
$\Hb$ and $H$ respectively.
\end{remark}
\begin{remark}
As a consequence of the Lemma above the same
estimates hold  true if we reverse the order of differentiation.
\end{remark}
\begin{proof}
Consider the  $\nab_4$  transport equations verified by  $\psi\in\{\trch, \chih, \omb,   \eta,    \trchbt, \chibh  \} $
\beaa
\nab_4\psi=\trchb_0\c\psi+\psi\c\psi+\nab\psi+\Psi_4,
\eeaa
with   curvature  components $\Psi_4\in \{\a,\b,\rho,\si\}$.
Clearly,
\beaa
\|\nab\nab_4\psi\|_{\Lsc^2(H)}&\les& \big(\|\nab^2\psi\|_{\Lsc^2(H)} 
+\|\nab\Psi_4\|_{\Lsc^2(H)}  \big)  +  (1+\de^{\frac 12})   \|\nab\psi\|_{\Lsc^2(H)} \\
&\les& C.
\eeaa
Also,  along $\Hb$, 
\beaa
\|\nab\nab_4\psi\|_{\Lsc^2(\Hb)}&\les& \big(\|\nab^2\psi\|_{\Lsc^2(\Hb)} 
+\|\nab\Psi_4\|_{\Lsc^2(\Hb)}  \big)  +  (1+\de^{\frac 12})   \|\nab\psi\||_{\Lsc^2(\Hb)}\\
&\les& C
\eeaa
provided that  $\Psi_4\ne \a$,  (i.e.    the original $\psi$ on the left  is not  $ \chih$).

On the other hand  the  $\nab_3$  transport equations verified by  
$\psi\in\{\trch, \chih,    \etab,    \trchbt, \chibh, \om   \} $
are of the form,
\beaa
\nab_3\psi=\trchb_0\c\psi+\psi\c\psi+\nab\psi+\Psi_3,
\eeaa
with  the curvature  components  $\Psi_3\in\{\rho,\si,\bb,\alphab\}$. 
The corresponding estimates follow precisely in the same manner.
\end{proof}

\subsection{Estimates for $\nab\nab_3\eta, \nab\nab_4\etab, \nab\nab_3\omb$, $\nab\nab_4\om$ }\label{sec:nabnodir}
In this subsection we prove the following:
\begin{proposition}\label{prop:2der-eta}
There exists a constant $C=C(\OO^{(0)}, \RR, \RRb )$ such that 
$$
\|\nab\nab_4\etab\|_{\Lsc^2(H)}+ \|\nab\nab_3\eta\|_{\Lsc^2(H)}
+ \|\nab\nab_4\etab\|_{\Lsc^2(\Hb)}+ \|\nab\nab_3\eta\|_{\Lsc^2(\Hb)}\le C,
$$
\end{proposition}
\begin{remark}
Together with the previous proposition,   this proposition allows
us to control  all angular derivatives of all $\nab_3, \nab_4$
derivatives of all the Ricci coefficients $\trchb, \chih,\omb,\eta, \etab,\trchbt,\chibh,\om$
(in some $\Lsc^2(H)$ or $\Lsc^2(\Hb)$ or both)
 except  for $\nab\nab_4\om$ and $\nab\nab_3\omb$.
\end{remark}
\begin{proof}
To control $\nab\nab_3\eta, \nab\nab_4\etab$ we make use of lemma
\ref{lem:tr-mu}. Recall that 
 reduced mass aspect functions  $\mu$ and $\mub$  verify equations of the form,
\bea
\nab_4 \mu &=&     \psi\c \big(\nab\trch+   \nab  \psi+\Psi_4 \big) +  \psi\c\psi\c\psi_g \notag \\
\nab_3 \mub &=&\trchb_0\c\big( \nab\trchb+\nab\psi  \big)   +     \psi \c\big( \nab\trchb+  \nab  \psi+\Psi_3 \big)
\label{funny.mub3}\\
& +&\trchb_0 \c \psi\c\psi_g+    \psi\c\psi\c\psi_g  \nn
\eea
which  are to be coupled with the Hodge systems of the form
\bea
\DD (\eta,\etab)&=&(\mu,\mub) +   \rho+\si+
 \psi\c\psi. \label{eq:Hodge2}
\eea
Here $\Psi_4=\{\a\,\b,\rho,\si\}$ and $\Psi_3=\{\aa, \bb,\rho,\si\}$.

\NI
{\bf Remark}.  
We note absence of the Ricci coefficients $\om,\omb$ among 
the $\psi$ variables in the above equations, in particular among the terms of the form $\nab\psi$. 
This fact is very  important  in view of the lack of estimates for $\nab\nab_4\om$ and $\nab\nab_3\omb$.
Equally important is the absence of the terms $\trchb_0\c\psi$ with $\psi=\{\chih,\chibh\}$ in 
equation \eqref{funny.mub3}. Such terms would lead to an unmanageable  double anomaly. 

To estimate $\nab\nab_4\etab$ we need to  commute the  above equations for $\etab,\mub$ 
with $\nab_4$. Making use of lemma \ref{le:comm} we derive, 
\beaa
\nab_3 (\nab_4\mub) &=&\nab_4\trchb_0\c\big( \nab\trchb+\nab\psi  \big)  +\trchb_0\c\big( \nab_4\nab\trchb+
\nab_4\nab\psi  \big)  +\psi\c\nab\mub \\ &+&    \nab_4 \psi \c\big( \nab\trchb+  \nab  \psi+\Psi_3 \big) +\psi \c\big( \nab_4\nab\trchb+ 
\nab_4 \nab  \psi+\nab_4\Psi_3+\nab_4\etab  \big)\notag\\ &+&\nab_4\trchb_0 \c \psi\c\psi_g+\trchb_0 \c \nab_4\psi\c\psi+   
 \nab_4\psi\c\psi\c\psi+\omb\nab_4\mub+\om\nab_3\mub\\
\DD (\nab_4\etab)&=&\nab_4\mub + \nab_4(\rho,\si)+\psi\c( \nab_4\psi +
\nab\etab+\Psi_4 )
\eeaa
Proceeding as many times before, we write,
\beaa
\|\nab_4\mub\|_{\Lsc^2(u,\ub)}&\les&\|\nab_4\mub\|_{\Lsc^2(0,\ub)}+\int_0^u\|\nab_3\nab_4\mub\|_{\Lsc^2(u',\ub)}\eeaa
and (with $\Hb(u,\ub)=H_u^{0,\ub}$)
\beaa
 \int_0^u\|\nab_3\nab_4\mub\|_{\Lsc^2(u',\,\ub)}du' &\les&\int_0^u\|\nab_4\nab\psi\|_{\Lsc^2(u',\,\ub)} du'+
 \int_0^u\|\omb \nab_4\mub\|_{\Lsc^2(u',\ub))} du' \\
 &+&
\|\nab_4\psi\c \Psi_3\|_{\Lsc^2(\Hb(u,\ub))}+\|\nab_4\psi\c \nab\psi \|_{\Lsc^2(\Hb(u,\ub))}\\
&+&\|\psi\c \nab_4\psi\|_{\Lsc^2(\Hb(u,\ub))}
+\|\psi\c \nab_4\Psi_3\|_{\Lsc^2(\Hb)(u,\ub)}
\\
&+&\|\psi\c\nab\mub\|_{\Lsc^2(\Hb(u,\ub))}+\|\om \nab_3\mub\|_{\Lsc^2(\Hb(u,\ub))} ....
\eeaa
We have kept on the right  only the most  problematic  terms.
We now  write,
\beaa
\|\nab_4\psi\c \Psi_3\|_{\Lsc^2(\Hb)}&\les&\de^{1/2}\|\nab_4\psi\|_{\Lsc^4(\Hb)}\c \|\Psi_3\|_{\Lsc^4(\Hb)}\\
\eeaa
Using the interpolation estimates of corollary \ref{cor:interpol},
\beaa
\|\nab_4\psi\|_{\Lsc^4(\Hb)}&\les& \|\nab\nab_4\psi\|^{\frac 12}_{\Lsc^2(\Hb)}\|\nab_4\psi\|^{\frac 12}_{\Lsc^4(\Hb)}
+\de^{\frac 14} \|\nab_4\psi\|_{\Lsc^2(\Hb)}\\
\|\Psi_3 \|_{\Lsc^4(\Hb)}&\les& \|\nab\Psi_3\|^{\frac 12}_{\Lsc^2(\Hb)}\|\Psi_3\|^{\frac 12}_{\Lsc^2(\Hb)}
+\de^{\frac 14} \|\Psi_3\|_{\Lsc^2(\Hb)}
\eeaa
Taking into account the possible anomaly of  $\|\nab_4\psi\|_{\Lsc^2(S)}$  (recalling  also  that $\psi$  here differs 
 from $\om, \omb$ !) we deduce,
 \beaa
 \|\nab_4\psi\|_{\Lsc^4(\Hb)}&\les& C\de^{-1/4}\qquad 
  \|\Psi_3\|_{\Lsc^4(\Hb)}\les C
 \eeaa
 Therefore,
 \beaa
\|\nab_4\psi\c \Psi_3\|_{\Lsc^2(\Hb)}&\les&C\de^{1/4}.
 \eeaa
 Similarly, taking into account the estimates for $\OS_{1,4}$ of 
 corollary \ref{co:OS14},
  \beaa
\|\nab_4\psi\c \nab\psi \|_{\Lsc^2(\Hb)}&\les&\de^{1/2} \|\nab_4\psi\|_{\Lsc^4(\Hb)}\c\|\nab\psi \|_{\Lsc^4(\Hb)}\les C\de^{1/4}
 \eeaa
 
 To estimate  $\|\psi\nab_4\Psi_3\|_{\Lsc^2(\Hb)}$ we write, using the Bianchi equations,
 \beaa
\nab_4\Psi_3=\nab\Psi_g+\psi\c\Psi+\om\c\Psi,
 \eeaa
 where $\nab\Psi_g\in\{\nab\b,\nab\rho,\nab\si,\nab\bb\}$. Recalling the estimate
 $ \|\Psi\|_{\Lsc^2(\Hb)}\les C\de^{-1/2}$ encountered before and  $\|\nab\Psi_g\|_{\Lsc^2(\Hb)}\les \RRb$,
 \begin{align*}
\|\psi\c \nab_4\Psi_3\|_{\Lsc^2(\Hb)}&\les\de^{\frac 12} \|\psi\|_{\Lsc^\infty(\Hb)}\left (\|\nab\Psi_g\|_{\Lsc^2(\Hb)}+
\|\Psi\|_{\Lsc^2(\Hb)}\right)\les C
\end{align*}
The term $\|\nab_4\psi\c\psi\|_{\Lsc^2(\Hb)}$ may contain a double anomaly. We estimate it as follows:
\begin{align*}
\|\nab_4\psi\c\psi\|_{\Lsc^2(\Hb)}&\les \de^{1/2}  \|\psi\|_{\Lsc^\infty} 
\|\nab_4\psi\|_{\Lsc^2(\Hb)} \les C 
\end{align*}
All other terms in   $\Lsc^2(\Hb)$  can be estimated in the same manner to derive,
\beaa
\|\nab_4\mub\|_{\Lsc^2(u,\ub)}&\les&\|\nab_4\mub\|_{\Lsc^2(0,\ub)} +   \int_0^u \|\nab_4\mub\|_{\Lsc^2(u',\ub)} du'  \\&+&
\int_0^u \|\nab_4\nab\psi \|_{\Lsc^2(u',\ub)} du'   +        C
\eeaa
or, by Gronwall,
\beaa
\|\nab_4\mub\|_{\Lsc^2(u,\ub)}&\les&\|\nab_4\mub\|_{\Lsc^2(0,\ub)}+ \int_0^u \|\nab_4\nab\psi \|_{\Lsc^2(u',\ub)} du'   +        C
\eeaa
Now,
\beaa
\int_0^u \|\nab_4\nab\psi \|_{\Lsc^2(u',\ub)} du' 
  \les \int_0^u \|\nab_4\nab\etab \|_{\Lsc^2(u',\ub)} du'+\|\nab_4\nab\psi_g\|_{\Lsc^2(\Hb(u,\ub))}
\eeaa
where $\psi_g\in \{\trch, \chih, \eta, \chibh, \trchbt\}$. Thus, in view 
of the estimates of proposition \ref{prop:2der}  and commutator  lemma 
\ref{le:comm.nab4nab3psi},
\beaa
\int_0^u \|\nab_4\nab\psi \|_{\Lsc^2(u',\ub)} du' \les   \int_0^u \|\nab\nab_4\etab \|_{\Lsc^2(u',\ub)} du'+C
\eeaa
and therefore,
\bea
\|\nab_4\mub\|_{\Lsc^2(u,\ub)}&\les&\|\nab_4\mub\|_{\Lsc^2(0,\ub)}+ \int_0^u \|\nab\nab_4\etab \|_{\Lsc^2(u',\ub)} du'   +        C\label{final.nab4mubL2}
\eea

 Using the elliptic estimates of proposition \ref{prop:Hodge.estim}   applied to  the Hodge system 
 for $\nab_4\psi$ we derive,
\beaa
\|\nab\nab_4\etab\|_{\Lsc^2(S)}&\les& \|\nab_4\mub\|_{\Lsc^2(S)} +  \|\nab_4(\rho,\si)\|_{\Lsc^2(S)}\\
&+&
\de^{\frac 12}\De_0 \big(  \|\nab_4\psi\|_{\Lsc^2(S)}+
 \|\nab\psi\|_{\Lsc^2(S)} +
\|\Psi_4\|_{\Lsc^2(S)} \big)
\eeaa
Now, 
$$
\nab_4(\rho,\si)=\nab\b+\psi\c\Psi_4+\om\c\Psi_4,
$$
with $\Psi_4\in\{\a,\b,\rho,\si\}$,
Now,
\beaa
\|\nab_4(\rho,\si)\|_{\Lsc^2(S)}\les \|\nab\b\|_{\Lsc^2(S)}+\de^{\frac 12} \De_0
\|\Psi_4\|_{\Lsc^2(S)},
\eeaa
In the particular case when $\Psi_4=\a$,  (recall that $\a$ component
 is not allowed in the definition of the curvature norms $\RRb$)  we  recall (see proposition \ref{prop:aproh})  the  estimate $\|\a\|_{\Lsc^2(S)}\les \de^{-1/2} C$.  Therefore, in all cases,
 \beaa
 \|\Psi\|_{\Lsc^2(S)}\les C\de^{-1/2}
 \eeaa
  and consequently,
 \beaa
\|\nab_4(\rho,\si)\|_{\Lsc^2(S)}\les   \|\nab\b\|_{\Lsc^2(S)}+   C
\eeaa
with $C=C(\OO^{(0)}, \RR, \,\RRb)$.
Therefore,  
\bea
 \|\nab\nab_4\etab\|_{\Lsc^2(S)}&\les&\|\nab_4\mub\|_{\Lsc^2(S)}+       \|\nab\b\|_{\Lsc^2(S)}+C
 \label{interm.nab4etab}
\eea
Integrating,
\beaa
\int_0^u \|\nab\nab_4\etab\|_{\Lsc^2(u',\ub)} du' &\les& \int_0^u \|\nab_4\mub\|_{\Lsc^2(u',\ub)} du' +
 \int_0^u \|\nab \b\|_{\Lsc^2(u',\ub)} du' +C\\
 &\les& \int_0^u \|\nab_4\mub\|_{\Lsc^2(u',\ub)} du' +\RRb+C.
\eeaa
i.e.,
\bea
\int_0^u \|\nab\nab_4\etab\|_{\Lsc^2(u',\ub)} du' \les \int_0^u \|\nab_4\mub\|_{\Lsc^2(u',\ub)} du' +C.\label{eq:nabnab4etab}
\eea
Therefore, combining with \eqref{final.nab4mubL2}
and applying Gronwall again,
we deduce,
\beaa
\|\nab_4\mub\|_{\Lsc^2(u,\ub)}&\les&\|\nab_4\mub\|_{\Lsc^2(0,\ub)} + C
\eeaa
It is easy to check on the initial hypersurface  $H_0$,
\beaa
\|\nab_4\mub\|_{\Lsc^2(0,\ub)}\les \OO^{(0)}.
\eeaa
On the other hand,  returning to   \eqref{interm.nab4etab}, we deduce
\beaa
 \|\nab\nab_4\etab\|_{\Lsc^2(S)}&\les& C+  \|\nab\b\|_{\Lsc^2(S)}.
 \eeaa
 Hence,
 \beaa
  \|\nab\nab_4\etab\|_{\Lsc^2(H)}+   \|\nab\nab_4\etab\|_{\Lsc^2(\Hb)}  &\les& C
 \eeaa
 as desired.
 
 The   remaining estimate 
 \beaa
   \|\nab\nab_3\eta\|_{\Lsc^2(H)}+   \|\nab\nab_3\eta\|_{\Lsc^2(\Hb)} 
   \eeaa
   is proved in exactly the same manner. 
\end{proof}

\section{$\OO_\infty$ estimates and proof of Theorem A }
In this section we combine the estimates obtained so far  to derive   $L^\infty$ estimates for all  our  Ricci coefficients and thus verify the bootstrap assumption \eqref{bootstrap:Linfty}.  This would also
allow us to conclude the proof of theorem  A \ref{theoremA}.   To achieve this 
 we combine the $\OS_{0,4},\OO_{1,2} ,\OH, \OHb $ and 
the remaining second derivative  estimates with the interpolation results of Proposition \ref{prop.trace.sc}. 
We will only require results before and culminating with Proposition \ref{prop:2der}. In particular it does need
the estimates of Proposition \ref{prop:2der-eta}.

For the Ricci  coefficients $\psi\in\{\trch,\chih,\eta,\om\}$ we make use of the interpolation estimate 
of  Proposition \ref{prop.trace.sc} together with $\OS_{1,2}+\OH\les C$ and  
$\|\nab_4\nab \psi\|_{\Lsc^2(H)}\les C$ of    Proposition \ref{prop:2der}  in the previous section, to derive
\beaa
\|\nab\psi\|_{\Lsc^4(S)}&\les&\big(\de^{1/2} \|\nab\psi\|_{\Lsc^2(H)}+\|\nab^2\psi\|_{\Lsc^2(H)}\big)^{1/2}\\
&\c&\big(\de^{1/2}
\|\nab\psi\|_{\Lsc^2(H)}+\|\nab_4\nab\psi\|_{\Lsc^2(H)}\big)^{1/2}\\
&\les& C
\eeaa
Similarly, for  $\psi\in\{\widetilde\trchb,\chibh,\etab,\omb\}$, using  the estimates $\OS_{1,2}+\OHb\les C$ and  
 estimate $\|\nab_3\nab \psi\|_{\Lsc^2(H)}\les C$  of  Proposition \ref{prop:2der} in   the previous section
\beaa
\|\nab\psi\|_{\Lsc^4(S)}&\les&\big(\de^{1/2}\|\nab\psi\|_{\Lsc^2(\Hb)}+\|\nab^2\psi\|_{\Lsc^2(\Hb)}\big)^{1/2} \\
&\c&\big(\de^{1/2} \|\nab\psi\|_{\Lsc^2(\Hb)}+\|\nab_3\nab\psi\|_{\Lsc^2(\Hb)}\big)^{1/2}\\
&\les& C
\eeaa

Next, for the non-anomalous coefficients $\psi\in\{\trch,\eta,\etab,\om,\omb, \trchbt\}$ we use the interpolation
inequality
$$
\|\psi\|_{\Lsc^\infty(S)}\les \|\nab\psi\|_{\Lsc^4(S)}^{\frac 12}  \|\psi\|_{\Lsc^4(S)}^{\frac 12} +
\de^{\frac 14}  \|\psi\|_{\Lsc^4(S)},
$$
which leads to the desired  estimate,
$$
\|\psi\|_{\Lsc^\infty(S)}\les C.
$$
In the anomalous case of $\psi=\{\chih,\chibh\}$ we use the interpolation inequality \eqref{eq:Linf-loc-sc}
$$
\|\psi\|_{\Lsc^\infty(S)}\les \sup_{^\de S}\left (\|\nab\psi\|_{\Lsc^4(S)}+ \|\psi\|_{\Lsc^4(^\de S)}\right),
$$
which gives 
$$
\|\psi\|_{\Lsc^\infty(S)}\les C.
$$
as desired.  We deduce,
\begin{proposition} There exists  a  constant $ C=C(\OO^{(0))},\RR,\, \RRb)$ such that,
for  $\de^{1/2} \De_0$  sufficiently small we have,
\bea
\OS_{0,\infty}&\les& C.
\eea
\end{proposition}
In particular, choosing $\De_0\approx C $,  and $\de>0 $ sufficiently  small,
depending only on $C$ we    dispense of the bootstrap assumption and
 derive the conclusion of Theorem A. 
\section{$\Lsc^4(S)$ estimates for curvature and the first derivatives of the Ricci coefficients}
In this section we establish $\Lsc^4(S)$ estimates for all first derivatives of the Ricci coefficients $\psi$. 
In the previous section we have already established such estimates for $\nab\psi$. The Ricci coefficients 
satisfy the structure equations
\begin{align*}
&\nab_4\psi=\trchb_0\c\psi+\psi\c\psi+\nab\psi+\Psi,\\
&\nab_3\psi=\trchb_0\c\psi+\psi\c\psi+\nab\psi+\Psi.
\end{align*}
We note that the double anomalous terms $\trchb_0\c\chih$ and $\trchb_0\c\chibh$ appear only in the 
$\nab_4 \chibh$, $\nab_3\chih$ and $\nab_3\chibh$ equations. Similarly the anomalous $\a$ curvature
component only appears in the $\nab_4\chih$ equation.

For the remaining equations we estimate
\begin{align*}
\|\nab_4\psi\|_{\Lsc^4(S)} &\les \|\psi\|_{\Lsc^4(S)} + \de^{\frac 12} \|\psi\|_{\Lsc^\infty(S)} \|\psi\|_{\Lsc^4(S)} 
+\|\nab\psi\|_{\Lsc^4(S)} +\|\Psi\|_{\Lsc^4(S)} \\ 
&\les \OO_{0,4} + \de^{\frac 14} \OO_{0,\infty} \OO_{0,4}  + \OO_{1,4} + \|\Psi\|_{\Lsc^4(S)},
\end{align*}
where the $\de^{\frac 14}$ takes into account a potential anomaly of the $\|\psi\|_{\Lsc^4(S)}$ term.
To estimate $ \|\Psi\|_{\Lsc^4(S)}$ we use the interpolation estimates 
\begin{align*}
& \|\Psi\|_{\Lsc^4(S)}\les \left (\de^{\frac 12} \|\Psi\|_{\Lsc^2(H)} +  \|\nab\Psi\|_{\Lsc^2(H)}\right)^{\frac 12}
  \left (\de^{\frac 12} \|\Psi\|_{\Lsc^2(H)} +  \|\nab_4\Psi\|_{\Lsc^2(H)}\right)^{\frac 12},\\
& \|\Psi\|_{\Lsc^4(S)}\les \left (\de^{\frac 12} \|\Psi\|_{\Lsc^2(\Hb)} +  \|\nab\Psi\|_{\Lsc^2(\Hb)}\right)^{\frac 12}
  \left (\de^{\frac 12} \|\Psi\|_{\Lsc^2(\Hb)} +  \|\nab_3\Psi\|_{\Lsc^2(\Hb)}\right)^{\frac 12}  
\end{align*}
Each of the null curvature components $\Psi$ satisfies either $\nab_4$ or $\nab_3$ equation. These
equations can be written schematically in the form
\begin{align*}
&\nab_{4} \Psi^{(s)}=\nab \Psi^{(s+\frac 12)} +\sum\limits_{s_1+s_2=s+1}\psi^{(s_1)}\c\Psi^{(s_2)},\\
&\nab_{3} \Psi^{(s)}=\nab \Psi^{(s-\frac 12)} +\trchb_0\c\Psi^{s}+\sum\limits_{s_1+s_2=s}\psi^{(s_1)}\c\Psi^{(s_2)}
\end{align*}
Let us consider the $\nab_3$ equation since the presence of the $\trchb_0$ makes it more difficult to handle.
We estimate
\begin{align*}
\|\nab_{3} \Psi^{(s)}\|_{\Lsc^2(\Hb)}\les \|\nab \Psi^{(s-\frac 12)}\|_{\Lsc^2(\Hb)}+\|\Psi^{s}\|_{\Lsc^2(\Hb)}+
\de^{\frac 12} \sum\limits_{s_1+s_2=s} \|\psi^{(s_1)}\|_{\Lsc^\infty} \|\Psi^{(s_2)}\|_{\Lsc^2(\Hb)}
\end{align*}
Note that the terms $\|\Psi^{s}\|_{\Lsc^2(\Hb)}$ and $\|\Psi^{s_2}\|_{\Lsc^2(\Hb)}$ are anomalous only 
for $s=s_2=2$, that is in the case of the estimate for $\a$. We summarize these estimates in the following
\begin{lemma}
For a constant $C=C(\II,\OO,\RR,\RRb)$ and $\Psi\in \{\b,\rho,\si,\bb,\alphab\}$
$$
\de^{\frac 14} \|\a\|_{\Lsc^4(S)}+\|\Psi\|_{\Lsc^4(S)}\le C
$$
\end{lemma}
Combining this result with $\nab_4\psi$ and $\nab_3\psi$ equations, as described above, gives us the 
$$
\|\nab_4\psi\|_{\Lsc^4(S)}+\|\nab_3\psi\|_{\Lsc^4(S)}\le C
$$ estimates for those derivatives, with the exception of $\psi=\chih, \chibh$. On the other hand, the anomalies
present in their respective equations lead to the anomalous estimates 
$$
\|\nab_4\chih\|_{\Lsc^4(S)}+\|\nab_3\chih\|_{\Lsc^4(S)}+\|\nab_4\chibh\|_{\Lsc^4(S)}+\|\nab_3\chibh\|_{\Lsc^4(S)}\le 
C\de^{-\frac 14}
$$
It remains to estimate $\nab_3\eta,\nab_4\etab,\nab_3\omegab,\nab_4\om$ which do not satisfy direct equations.
We argue as in sections \ref{sec:nodirect} and \ref{sec:nabnodir}.  Using the interpolation estimates stated in the
beginning of this section and the bounds 
\begin{align*}
&\|\nab\nab_3\eta\|_{\Lsc^2(H)} + \|\nab_4\nab_3\eta\|_{\Lsc^2(H)}\le C,\\
&\|\nab\nab_4\etab\|_{\Lsc^2(H)} + \|\nab_3\nab_4\etab\|_{\Lsc^2(H)}\le C
\end{align*}
of sections \ref{sec:nodirect} and \ref{sec:nabnodir}, we obtain the desired $\Lsc^4(S)$ estimates
for $\nab_3\eta$ and $\nab_4\etab$. However, we can not obtain the corresponding estimates 
for $\nab_4\om$ and $\nab_3\omegab$.
We summarize the second main result of this section.
\begin{lemma}
\begin{align*}
&\|\nab \psi\|_{\Lsc^4(S)}+\|\nab_{3,4}\eta\|_{\Lsc^4(S)}+\|\nab_{3,4}\etab\|_{\Lsc^4(S)}+
\|\nab_4\omegab\|_{\Lsc^4(S)}+\|\nab_3\omega\|_{\Lsc^4(S)}\le C,\\
&\|\nab_4\chih\|_{\Lsc^4(S)}+\|\nab_3\chih\|_{\Lsc^4(S)}+\|\nab_4\chibh\|_{\Lsc^4(S)}+\|\nab_3\chibh\|_{\Lsc^4(S)}\le 
C\de^{-\frac 14}
\end{align*}
\end{lemma}


\section{Renormalized estimates}\label{sec:ren}

\subsection{Trace theorems} 

The results of this  section  rely on sharp trace theorems which we discuss below.
We introduce the following new  norms for an $S$ tangent tensor $\phi$ with scale $\sc(\phi)$ along
 $H=H_u^{(0,\ub)}$, relative to the transported coordinates $(\ub, \th)$ of proposition \ref{prop:gamma}:
\beaa
 \|\phi\|_{\Trsc(H)}&=&\de^{-\sc(\phi) - \frac 1 2 } \big( \,\sup_{\th\in S(u,0)} \int_0^{\ub}|\phi(u,\ub', \th)|^2 d\ub' \big)^{1/2}
 \eeaa
 Also,  along $\Hb=\Hb_{\ub}^{(0,u)} $
  relative to the transported coordinates $(u, \thb)$ of proposition \ref{prop:gamma}
  \beaa
\|\phi\|_{\Trsc(\Hb)}&=&\de^{-\sc(\phi)  } \big(\, \sup_{\thb\in S(\ub,0)} \int_0^{u}|\phi(u',\ub,\thb)|^2 d u' \big)^{1/2} \eeaa
\begin{proposition}
\label{prop.trace}
For any horizontal tensor $\phi$ along $H=H_u^{(0,\ub)}$,
\begin{equation}
\label{eq:*}
\begin{split}
\|\nab_4\phi\|_{\Tr_{(sc)}({H})}&\les \left (\|\nab_4^2\phi\|_{\Lsc^2(H)}+\|\phi\|_{\Lsc^2(H)}+ 
\de^{\frac 12} C(\|\phi\|_{\Lsc^\infty}+
\|\nab_4\phi\|_{\Lsc^4(S)})\right)^{\frac 12}\\ &\times  \left (\|\nab^2\phi\|_{\Lsc^2(H)}+ \de^{\frac 12} C(\|\phi\|_{\Lsc^\infty}+
\|\nab\phi\|_{\Lsc^4(S)})\right)^{\frac 12}\\ &+\|\nab_4\nab\phi\|_{\Lsc^2(H)}+ \de^{\frac 12} C(\|\phi\|_{\Lsc^\infty}+
\|\nab\phi\|_{\Lsc^4(S)})+\|\nab\phi\|_{\Lsc^2(H)}
\end{split}
\end{equation}
 where
 $C$ is a constant which 
depends  on $\OO^{(0)}, \RR, \RRb$.

Also, 
for any horizontal tensor $\phi$ along $\Hb=H_{\ub}^{(u,0)}$, and
a similar constant $C$,
\begin{equation}
\label{eq:*b}
\begin{split}
\|\nab_3\phi\|_{\Tr_{(sc)}({\Hb})}&\les \left (\|\nab_3^2\phi\|_{\Lsc^2(\Hb)}+\|\phi\|_{\Lsc^2(\Hb)}+ 
\de^{\frac 12} C(\|\phi\|_{\Lsc^\infty}+
\|\nab_3\phi\|_{\Lsc^4(S)})\right)^{\frac 12}\\ &\times  \left (\|\nab^2\phi\|_{\Lsc^2(H)}+ \de^{\frac 12} C(\|\phi\|_{\Lsc^\infty}+
\|\nab\phi\|_{\Lsc^4(S)})\right)^{\frac 12}\\ &+\|\nab_3\nab\phi\|_{\Lsc^2(\Hb)}+ \de^{\frac 12} C(\|\phi\|_{\Lsc^\infty}+
\|\nab\phi\|_{\Lsc^4(S)})+\|\nab\phi\|_{\Lsc^2(\Hb)}
\end{split}
\end{equation}

\end{proposition}

The proof relies on the classical (euclidean) trace inequality formulated in $(u,\theta)$ or 
$(\ub,\theta)$ coordinates 
\begin{lemma}
\label{le:trace.classical}
For any    scalar function $\phi$  along $H=H_u^{(0,\ub)} $,  supported in a 
coordinate chart, we have
\bea
\big(\int_0^{\ub} |\pr_{\ub} \phi(u,\ub',\th)|^2 d\ub'\big)^{1/2} &\les&
\big( \|\pr_{\ub}^2\phi\|_{L^2(H)} +\de^2\|\phi\|_{L^2(H)} \big)^{1/2}\c \|\pr^2_\th\phi\|_{L^2(H)}^{1/2}\nn\\
&+&  \|\pr_\theta\pr_{\ub} \phi\|_{L^2(H)}    +      \de \|\pr_\theta\phi\|_{L^2(H)}          \label{eq:traceH1}
  \eea
 For any  scalar function $\phi$ along $\Hb =\Hb_{\ub}^  {(0, u)}    $, 
   supported in a neighborhood patch,
 \bea
(\int_0^u  |\pr_u\phi(u',\ub, \thb)|^2 du')^{1/2} & \les &\big(  \|\pr_u^2\phi\|_{L^2(\Hb)}+ \|\pr_u^2\phi\|_{L^2(\Hb)}\big)^{1/2} \c \|\pr^2_{\thb}\phi\|_{L^2(\Hb)}^{1/2}\nn   \\
&+&     \|\pr_{\thb}\pr_u\phi\|_{L^2(\Hb)} +  \|\pr_{\thb}\phi\|_{L^2(\Hb)} \label{eq:traceHb1}
\eea
In scale invariant norms we have,
\beaa
\|\pr_{\ub}\phi\|_{\Trsc(H)} &\les&\big( \|\pr_{\ub}^2\phi\|_{\Lsc^2(H)}+\|\phi\|_{\Lsc^2(H)}\big)^{1/2}\c \|\pr^2_\th\phi\|_{\Lsc^2(H)}^{1/2}\\
&+&
 \|\pr_\theta\pr_{\ub} \phi\|_{\Lsc^2(H)} +\|\pr_\theta \phi\|_{\Lsc^2(H)}
 \eeaa
 and,
 \beaa
 \|\pr_{u}\phi\|_{\Trsc(\Hb)} &\les&\big( \|\pr_{u}^2\phi\|_{\Lsc^2(\Hb)}+ \|\phi\|_{\Lsc^2(\Hb)}\big)^{1/2}\c \|\pr^2_\th\phi\|_{\Lsc^2(\Hb)}^{1/2}\\
 &+&
 \|\pr_\theta\pr_{\ub} \phi\|_{\Lsc^2(\Hb)}    +  \|\pr_\theta \phi\|_{\Lsc^2(\Hb)}   
 \eeaa
\end{lemma}
\begin{proof} 
We start by making the additional assumption  that $\phi(\ub,\th) $ is compactly  supported 
for $\ub'\in(0,\ub)$.

  Integrating 
by parts in   $\th=(\th^1, \th^2)$,
\beaa
\big|\int_0^{\ub}|\pr_{\ub}\phi(\ub, \th)|^2\big|&=&\big|\int_{\th^1}^\infty\int_{\th^2}^\infty  d\th^1 d\th^2\,\, \pr_{\th^1} \,\pr_{\th^2}\int  \pr_{\ub}\phi(\ub', \th) \c \pr_{\ub}\phi(\ub', \th)  d\ub'   \big|\\
&\les& \int_{D} \big| \pr_{\th^1} \pr_{\th^2}\int_0^{\ub} 
 \pr_{\ub}\phi(\ub', \th) \c \pr_{\ub}\phi(\ub', \th)  d\ub' \big|  \,d\th^1 d\th^2\\
 &\les& \int_{D}\big|\int_0^{\ub} 
 \pr_{\th^1}\pr_{\th^2} \pr_{\ub}\phi(\ub, \th) \c \pr_{\ub}\phi(\ub, \th) d\ub \big| d\th\\
 &+&\int_{D}\int_0^{\ub} 
 \big|\pr_{\th} \pr_{\ub}\phi(\ub, \th)\big|^2 d\ub' d\th
\eeaa
Now, integrating by parts in $\ub$,
\beaa
\int_0^{\ub} 
 \pr_{\th^1}\pr_{\th^2} \pr_{\ub}\phi(\ub', \th) \c \pr_{\ub}\phi(\ub', \th) d\ub'&=&-\int_0^{\ub} 
 \pr_{\th^1}\pr_{\th^2} \phi(\ub', \th) \c \pr^2_{\ub}\phi(\ub', \th) 
 \eeaa
Hence,
\bea
\int_0^{\ub}|\pr_{\ub}\phi(\ub, \th)|^2&\les& \|\pr_\th^2 \phi\|_{L^2(H)}\c \|\pr_{\ub}^2 \phi\|_{L^2(H)}+
\|\pr_{\th} \pr_{\ub}\phi\|_{L^2(H)}^2.\label{eq:traceH1.interm}
\eea
To remove our additional assumption concerning the compact support
in $(0,\ub)$ we simply extend the original $\phi$ to $-\de\le \ub\le 2\de$
such that all norms on the right hand side of \eqref{eq:traceH1}, on the extended
interval,  are bounded by a constant multiple of the  same norms 
restricted to  the original interval $(0,\ub)$.  We then apply a cut-off
 to make the extended $\phi$ compactly supported in   the interval $(-\de,2\de)$
 and finally use     \eqref{eq:traceH1.interm} in the extended interval
  to get the desired result. The proof of \eqref{eq:traceHb1} is exactly the same.
The scale version  of these  estimates is immediate.
\end{proof} 
We now pass to the proof of proposition \ref{prop.trace}.
It suffices to prove \eqref{eq:*}, the proof of 
\eqref{eq:*b} is exactly the same.

One can easily pass from the coordinate dependent  form of the trace inequalities
to  a covariant form  with the help of  the estimates of proposition \ref{prop:gamma}.

According to that proposition we have, for $C=C(\OO^{(0)},\RR,\RRb)$,
\beaa
\|\Ga\|_{\Lsc^2(S)}+     \|\nab\Ga\|_{\Lsc^2(S)}&\les& C
\eeaa
Thus,
\beaa
\nab_4\phi _a&=&\Omega^{-1} \pr_{\ub} \phi_a-\chi_{ab} \phi_b
\eeaa
As a consequence,  along $H=H_u$,
\beaa
 \|\nab_4\phi\|_{\Trsc(H)}& \les&  \|\pr_{\ub} \phi\|_{\Trsc(H)}
 +\de^{1/2} \|\chi\|_{\Lsc^\infty} \|\phi\|_{\Lsc^\infty}\\
 &\les&\|\pr_{\ub} \phi\|_{\Trsc(H)}+C\de^{1/2} \|\phi\|_{\Lsc^\infty}
\eeaa
Also, schematically, ignoring factors of $\Om$ (which are bounded in $L^\infty$), we have with $\psi\in\{\chi,\om\}$,
\beaa
\nab_4^2\phi&=&\pr^2_{\ub} \phi+\psi\c  \pr_{\ub}  \phi+\a \c\phi +\psi\c\psi\c\phi
\eeaa 
Thus, in view of our estimates for the Ricci coefficients $\psi$,
we have
\beaa
\|\pr^2_{\ub} \phi\|_{\Lsc^2( H)}&\les&\| \nab_4^2\phi\|_{\Lsc^2(H)}+\de^{1/2}\|\psi\|_{\Lsc^\infty}\c
 \| \nab_4 \phi\|_{\Lsc^2(H)}\\
 &+& \de^{1/2} \|\phi\|_{\Lsc^\infty}\big( \|\a\|_{\Lsc^2(H)}   +\|\psi\|_{\Lsc^\infty}^2\big)\\
 &\les&\| \nab_4^2\phi\|_{\Lsc^2(H)}+C\de^{1/2} \big( \| \nab_4 \phi\|_{\Lsc^2(H)}+\|\phi\|_{\Lsc^\infty}\big)
\eeaa
We next note that for a horizontal tensor we can convert $\pr_\th$ into a covariant $\nab$ derivative 
according to the formula 
$
\pr_{\th}=\nab+\Gamma.
$
Therefore,
\beaa
\|\pr_{\th}\phi_a\|_{\Lsc^2(S)} &\les &\|\nab\phi\|_{\Lsc^2(S)}+\de^{1/2}\|\Ga\|_{\Lsc^2(S)}\|\phi\|_{\Lsc^\infty}\\
&\les&\|\nab\phi\|_{\Lsc^2(S)}+\de^{1/2}C\|\phi\|_{\Lsc^\infty}
 \eeaa
 and,
 \beaa
\|\pr^2_{\th}\phi_a\|_{\Lsc^2(S)}&\les& \|\nab^2\phi\|_{\Lsc^2(S)}+\de^{1/2}\|\pr \Ga\|_{\Lsc^2(S)}\|\phi\|_{\Lsc^\infty}
+\de^{1/2}\|\Ga\|_{\Lsc^4(S)} \|\nab\phi\|_{\Lsc^4(S)}\\
&\les& \|\nab^2\phi\|_{\Lsc^2(S)}+\de^{1/2}C\big(\|\phi\|_{\Lsc^\infty}+ \|\nab\phi\|_{\Lsc^4(S)}\big)
\eeaa
Also,
\beaa
\|\pr_\th \pr_{\ub} \phi_a\|_{\Lsc^2(S)}&\les& \|\nab\nab_4\phi\|_{L^2(S)}+\de^{1/2}\|\pr \Ga\|_{L^2(S)}\|\phi\|_{L^\infty}
+\de^{1/2}\|\Ga\|_{L^4(S)} \|\nab_4\phi\|_{L^4(S)}\\
&\les& \|\nab\nab_4\phi\|_{\Lsc^2(S)}+\de^{1/2}C\big(\|\phi\|_{\Lsc^\infty}+ \|\nab_4\phi\|_{\Lsc^4(S)}\big)
\eeaa
According , to the  the scale invariant estimate of  lemma \ref{le:trace.classical},
\beaa
\|\pr_{\ub}\phi\|_{\Trsc(H)} &\les&\big( \|\pr_{\ub}^2\phi\|_{\Lsc^2(H)}+\|\phi\|_{\Lsc^2(H)}\big)^{1/2}\c \|\pr^2_\th\phi\|_{\Lsc^2(H)}^{1/2}\\
&+&
 \|\pr_\theta\pr_{\ub} \phi\|_{\Lsc^2(H)} +\|\pr_\theta \phi\|_{\Lsc^2(H)}
 \eeaa
Combining this with the previous estimates we obtain the desired result, which
can be  clearly  extended to any  $\phi$ along $H_u$,
 not necessarily restricted to a coordinate patch, by a simple partition
 of unity argument.
This proves  the desired estimate \eqref{eq:*}.
  Estimate  \eqref{eq:*b} is proved in exactly the same manner.
\subsection{Estimate for   the trace norms of $\nab\chi$, $\nab\chib$ }Our main goal in this  subsection is  to derive estimates for
the trace  norms  $\|\nab\chi\|_{\Trsc(H)}$ and $\|\nab\chib\|_{\Trsc(\Hb)}$.
In view of proposition \ref{prop.trace} we could achieve this goal
if we could write $\nab \chih=\nab_4\phi $ and $\nab\chibh =\nab_3\phib$
where $\phi$, respectively  $\phib$  are such that the norms  on the right hand side
of \eqref{eq:*}, respectively \eqref{eq:*b},
 are finite.  We prove the following proposition.
\begin{proposition}
  \label{le:estimates.phi-phib}
  
  Consider the following  transport equations along $H=H_u$, respectively $\Hb=\Hb_{\ub}$
   \bea
 \nab_4 \phi&=& \nab\chih,            \qquad \phi(0,\ub)=0         \label{def:transp.phi}
 \eea
 and 
  \bea
 \nab_3 \phib&=& \nab\chibh,            \qquad \phib(0,\ub)=0         \label{def:transp.phib}
 \eea
\begin{enumerate}
\item Solution $\phi$ of  \eqref{def:transp.phi} verifies 
 the estimates,
 \bea
  \|\phi\|_{\Lsc^2(S)}+\|\phi\|_{\Lsc^4(S)}+
 \|\nab\phi\|_{\Lsc^2(S)}+\|\nab_4\phi\|_{\Lsc^2(S)}&\les & C\label{estimates.easphi1}\\
 \| \nab\nab_4 \phi\|_{\Lsc^2(H)} +\|\nab_4^2\phi\|_{\Lsc^2(H)}&\les& C\label{estimates.easphi2}
 \eea
 with  a constant $C=C(\OO^{(0)},\RR,\RRb)$.
 Moreover,
 \bea
  \| \nab^2 \phi\|_{\Lsc^2(H)}&\les& \|\nab^3\trch\|_{\Lsc^2(H)}+C
  \eea
  As a consequence (see  calculus inequalities of subsection \ref{subs.calculus.ineq})
  we also have,
  \bea
  \|  \phi\|_{\Lsc^\infty}&\les&  \|\nab^3\trch\|_{\Lsc^2(H)}+C
  \eea
  and as a consequence of  the trace estimate \eqref{eq:*},
   \bea
  \| \nab_4 \phi\|_{\Trsc(H)}&\les&  \|\nab^3\trch\|_{\Lsc^2(H)}+C \label{estim.tracenab4phi}
  \eea
  
  \item 
  Solution $\phib$ of  \eqref{def:transp.phib} verifies 
 the estimates,
 \bea
  \|\phib\|_{\Lsc^2(S)}+\|\phib \|_{\Lsc^4(S)}+
 \|\nab\phib\|_{\Lsc^2(S)}+\|\nab_4\phib\|_{\Lsc^2(S)}&\les & C\label{estimates.easphib1}\\
 \| \nab\nab_3 \phib\|_{\Lsc^2(\Hb)} +\|\nab_3^2\phi\|_{\Lsc^2(\Hb)}&\les& C\label{estimates.easphib2}
 \eea
 with  a constant $C=C(\OO^{(0)},\RR,\RRb)$.
 Moreover,
 \bea
  \| \nab^2 \phib\|_{\Lsc^2(\Hb)}&\les& \|\nab^3\trch\|_{\Lsc^2(\Hb)}+C\label{estim.L2.phib}
  \eea
  As a consequence (see  calculus inequalities of subsection \ref{subs.calculus.ineq})
  we also have,
  \bea
  \|  \phib\|_{\Lsc^\infty}&\les&  \|\nab^3\trchb\|_{\Lsc^2(\Hb)}+C
  \eea
   and as a consequence of  the trace estimate \eqref{eq:*},
   \bea
  \| \nab_3 \phib\|_{\Trsc(\Hb)}&\les&  \|\nab^3\trchb\|_{\Lsc^2(H)}+C\label{estim.tracenab33phib}
  \eea
\end{enumerate}
 \end{proposition} 
 \begin{proof}
 Estimates \eqref{estimates.easphi1}-\eqref{estimates.easphi2}  and respectively \eqref{estimates.easphib1}-\eqref{estimates.easphib2}
 follow easily from  \eqref{def:transp.phi}, respectively   \eqref{def:transp.phib} in view of our estimates for 
  $\chih$, respectively $\chibh$,  and their  first two derivatives derived in the previous sections.
 The second  $\nab$  derivative
 estimates  are subtle; they require       a non-trivial  renormalization procedure, nothing less
  than another  series miracles.  As always we expect
  the estimates for $\phib$ to be somewhat more demanding in view of the presence
  of $\trchb=\trch_0+\trchbt$. We shall thus concentrate on them
  in what follows. No other anomalies occur at this high level of 
  differentiability.
   The idea
  is to derive first a transport  equation  for $\lap\phib$ 
  and  hope somehow that the principal term on the right,
 i.e.   $\nab \lap\chibh$,  can be re-expressed a $\nab_4$ derivative of
  another quantity depending only on two derivatives of a Ricci coefficient. 
  We write,
   \beaa
    \nab_3 \lap \phib&=& \lap \nab\chibh+[\nab_3, \lap] \phi\
   \eeaa
   Now, recalling commutation lemma \ref{le:comm},
   we write schematically (we  eliminate $\bb$ using
   the Codazzi equation)
   \beaa
  \, [\nab_3,\nab]\phib&=&\chib \c\nab\phi+  \nab\psi_3\c\phi+\psi_3\c \nab_3\phi+\chib\c\psi_3\c\phi\\
  \,  [\nab_3,\nab^2]\phib&=&\chib \c\nab^2\phi+\nab\psi_3\c (\nab\phi+\nab_3\phi)
  +\nab^2\psi_3\c \phi+\psi_3\c\nab\nab_3\phi+\nab( \chibh\c \psi_3\c\phi)\nn\\
  &+&          \psi_3\c\nab_3\nab\phi+\chibh\c\psi_3\c \nab\phi 
   \eeaa
   where $\psi_3\in\{\trchbt,\chibh, \eta,\etab\}$.
   
   Hence, using our estimates for $\psi_3$ as well as the 
   estimates \eqref{estimates.easphib1}-\eqref{estimates.easphib2}
   for $\phib$ we can write,
   \bea
   [\nab_3,\Delta]\phib&=&\trchb_0 \nab^2\phib+\chibh \c\nab^2\phib+\Err_\phi \label{eq:com3-nab}\\
   \|\Err_\phi\|_{\Lsc^2(\Hb)}&\les&C \de^{1/2}\big(C+\|\nab^2\phib\|_{\Lsc^2(\Hb)}\big)
   \eea
   Indeed,
   we have, for example,
   \beaa
   \|\nab^2\psi_3\c\phib\|_{\Lsc^2(\Hb)}&\les&\de^{1/2}\|\phib\|_{\Lsc^\infty} \|\nab^2\psi_3\|_{\Lsc^2(\Hb)} \les\de^{1/2} C\|\phib\|_{\Lsc^\infty} \\
   &\les&C\de^{1/2}\big(\|\nab^2\phi\|_{\Lsc^2(\Hb)}+\|\nab\nab_3\phib\|_{\Lsc^2(\Hb)}+
  \|\phi\|_{\Lsc^2(\Hb)}\big)\\
  &\les& C\de^{1/2} \|\nab^2\phi\|_{\Lsc^2(\Hb)}+C^2\de^{1/2}.
   \eeaa
   Consequently,
   \bea
   \nab_3 \lap \phib&=& \lap \nab\chibh+\trchb_0 \nab^2\phib+\chibh \c\nab^2\phib+\Err_\phi
   \label{eq:nab3lapphi1}
   \eea
   Since,
   \beaa
   \, [\lap,\nab] \phi&=& K \nab\phi+\nab K\c\phi
\eeaa
 we have, 
 \beaa
 \| [\lap,\nab] \phi \|_{\Lsc^2(\Hb)}&\les&\|K\|_{\Lsc^4(\Hb)}\c \|\nab\phi\|_{Lsc^4(\Hb)}+
 \|\nab K\|_{\Lsc^2(\Hb)}\|\phi\|_{\Lsc^\infty}\\
 &\les&  C\de^{1/2} \|\nab^2\phi\|_{\Lsc^2(\Hb)}+C^2\de^{1/2}
 \eeaa
 Hence, also,
 \bea
   \nab_3 \lap \phib&=& \lap \nab\chibh+\trchb_0 \nab^2\phib+\chibh \c\nab^2\phib+\Err_\phi
   \label{eq:nab3lapphi2}\\
    \|E\|_{\Lsc^2(\Hb)}&\les&C \de^{1/2}\big(C+\|\nab^2\phib\|_{\Lsc^2(\Hb)}\big)\nn
   \eea
 Now, according to the Codazzi equations,
\beaa
\dcall \chibh=&-\bb-\frac 1 2\nab  \trchb + \trchb\psi_3+ \psi_4\c\psi_3
\eeaa
 Thus,
\beaa
\dcallll\dcall \chibh=&\dcallll\bb-\frac 1 2\dcallll \nab \trchb +\dcallll(  \trchb\psi_3+     \psi_3\c\psi_3)
\eeaa  
or,  making use of \eqref{eq:dcallident},
\beaa
-\frac 1 2 \lap \chibh+K\chibh=\dcallll\bb-\frac 1 2\dcallll \nab\trch +\dcallll( \trchb\psi_3+\psi_3\c\psi_3).
\eeaa
Thus, differentiating once more,
\bea
\nab \lap \chibh&=&\nab^2 \bb +\nab^3\trchb +  K  \nab\chibh +\Err
\label{eq:nablapchibh}\\
\Err&=& \nab K\c \chibh+\trchb \nab^2\psi_3+\nab^2(\psi_3\c\psi_3)\nn
\eea
Here, and in what follows, $\Err$ denotes an error term of the form,
\beaa
\|\Err\|_{\Lsc^2(\Hb)}&\les&C
\eeaa
On the other hand we  recall the structure equation,
\beaa
\nab_3\etab&=\bb +\chib\c(\eta-\etab)
\eeaa
Thus, commuting, and writing as before,
\beaa
\,[\nab_3,\nab  ]\etab&=&\chib \c\nab\etab +  \nab\psi_3\c\etab +\psi_3\c \nab_3\etab+\chib\c\psi_3\c\etab\\
\, [\nab_3,\nab^2]\etab &=&\chib \c\nab^2\etab +\nab\psi_3\c (\nab\etab +\nab_3\etab )
  +\nab^2\psi_3\c \etab  +\psi_3\c\nab\nab_3\etab +\nab( \chibh\c \psi_3\c\etab )\nn\\
  &+&          \psi_3\c\nab_3\nab\etab +\chibh\c\psi_3\c \nab\etab
\eeaa
Observe that,
\beaa
\|[\nab_3,\nab  ]\etab\|_{\Lsc^2(\Hb)}&\les& C
\eeaa
and consequently,
\bea
\nab^2\bb&=& \nab_3(\nab^2\etab)+ \Err\label{eq:nab2bb}\\
\Err&=& \nab^2\big(\chib\c(\eta-\etab)\big)+[\nab_3,\nab^2]\etab\nn
\eea
Clearly,
\bea
\| \Err\|_{\Lsc^2(\Hb)}&\les& C\label{estim:nab2bb}
\eea
Therefore, we deduce,
\beaa
\nab \lap \chibh&=&-\nab_3(\nab^2 \etab)+\nab^3\trchb +  K\nab\chibh +\Err
\eeaa
Commuting $\nab $ with $\lap$ again,
\beaa
\lap\nab\chibh &=&\nab \lap \chibh+K\nab\chibh+\nab K\chibh
\eeaa
Hence, since $\nab\chibh=\nab_3\phib$,
\bea
\lap\nab\chibh &=&\nab_3(\nab^2 \etab)+\nab^3\trchb+   K\nab_3\phib +\Err \label{eq:nablapchih2}
\eea

Back to  \eqref{eq:nab3lapphi2} we rewrite,
\beaa
\nab_3 \lap \phib&=&-\nab_3(\nab^2 \etab)  +\nab^3\trchb +\trchb_0\c\nab^2 \phib+K\c\nab+3\phib+ \Err_\phi\\
\|\Err_\phi\|_{\Lsc^2(\Hb)}&\les&C\big(1+\de^{1/2}\|\nab^2\phib\|_{\Lsc^2(\Hb)}\big)
\eeaa
which we could rewrite in the form,
\bea
\nab_3\big( \lap \phib+\nab^2\etab-K\phib )&=& \nab^3\trchb +
\trchb_0\c\nab^2 \phi- \nab_3K \c \phib+\Err_\phi
\label{renorm.phi}
\eea
Recall that $K=\rho-\frac 1 4 \trch\trchb-\frac 1 2 \chih\c\chibh$.
Hence, we easily find,
\beaa
\|\nab_3K\|_{\Lsc^2(\Hb)}&\les& C
\eeaa
Thus,
 \beaa
 \| \nab_3\big( \lap \phib+ \nab^2\etab-K\chibh)\|_{\Lsc^2(u,\ub )}&\les&\|\nab^3\trchb\|_{\Lsc^2(u,\ub)}+ |\nab^2\phib\|_{\Lsc^2(u,\ub )}\\
 &+&\|\Err_\phi\|_{\Lsc^2(u,\ub)}
\eeaa
i.e.,
\beaa
 \| \lap  \phib\|_{\Lsc^2(u,\ub)}&\les& \| \nab^2\etab\|_{\Lsc^2(u,\ub)}+C\de^{1/2}\|K\|_{\Lsc^2(u,\ub)}+\|\nab^3\trchb\|_{\Lsc^2(\Hb)}\\
 &+&(1+\de^{1/2} C) \int_0^{u}  \|\nab^2\phib\|_{\Lsc^2(u',\ub )} du'
 +\|E_1\|_{\Lsc^2(\Hb)}
 \eeaa
 Now,
  using the elliptic estimates discussed in subsection \ref{section.Hodge},
  we have 
  and our estimates for $K$, we deduce
 \bea
 \| \nab^2  \phib\|_{\Lsc^2(S)}&\les &\| \lap  \phib\|_{\Lsc^2(S)}\label{Hodge.nab2phib}\\
 &+&\de^{1/2}\big( \|\nab K\|_{\Lsc^2(S)}
 \|\phib\|_{\Lsc^\infty(S)}+\|K\|_{\Lsc^4(S)}\|\nab\phib\|_{\Lsc^4(S)}\big)\nn\\
 &\les&\| \lap\phib\|_{\Lsc^2(S)}+\de^{1/2}\big( \|\phib\|_{\Lsc^\infty(S)}+\|\nab\phib\|_{\Lsc^4(S)}\big)\nn\\
 &\les&\| \lap\phib\|_{\Lsc^2(S)}+\de^{1/2}\big(C+ \| \nab^2 \phib\|_{\Lsc^2(\Hb)}\big)\nn
 \eea
Thus,
\beaa
 \| \nab^2  \phib\|_{\Lsc^2(u,\ub)}&\les& \| \nab^2\etab\|_{\Lsc^2(u,\ub)}+C\de^{1/2}\|K\|_{\Lsc^2(u,\ub)}+\|\nab^3\trchb\|_{\Lsc^2(\Hb)}\\
&+&(1+\de^{1/2} C) \int_0^{u}  \|\nab^2\phib\|_{\Lsc^2(u',\ub )} du'\\
&+& C(1+ \de^{1/2} )\| \nab^2 \phib\|_{\Lsc^2(\Hb)} 
\eeaa
Using Gronwall,
\bea
\| \nab^2  \phib\|_{\Lsc^2(u,\ub)}&\les&\| \nab^2\etab\|_{\Lsc^2(u,\ub)}+C\de^{1/2}\|K\|_{\Lsc^2(u,\ub)}+\|\nab^3\trchb\|_{\Lsc^2(\Hb)}\label{nab2phib.100}\\
&+&  C(1+ \de^{1/2} )\| \nab^2 \phib\|_{\Lsc^2(\Hb)} 
\nn
\eea
Integrating we deduce, for  $C\de^{1/2}$ sufficiently small,
\beaa
\| \nab^2 \phib\|_{\Lsc^2(\Hb)} &\les& C+\|\nab^3\trchb\|_{\Lsc^2(\Hb)}
\eeaa
as desired.
\end{proof}
To close the estimates
of proposition \ref{le:estimates.phi-phib}
it remains   to estimate $\|\nab^3\trch\|_{\Lsc^2(H)}$
 and  $\|\nab^3\trchb\|_{\Lsc^2(\Hb)}$. 
To achieve this we
  start with the  transport equation for $\trchb$,
\beaa
\nab_3(\trchb)&=&-\frac 1 2 \trchb^2-|\chibh|^2 -2\omb \chibh
\eeaa
which we rewrite in the form,
\beaa
\nab_4( \trchb')&=&-\frac 1 2\Om^{-1} \trchb^2-\Om^{-1}|\chibh|^2 \\
\trchb'&=&\Om^{-1}\trchb
\eeaa
The plan is to  derive a transport equation for the quantity $\lap\nab\trchb'$.
We   make use of the following commutation  formulae, written
schematically, for an arbitrary  scalar  $f$ verifying
the equation $\nab_3 f=F$,
\beaa
\nab_3(\nab f)&=&\nab F +\chib \c\nab f +\psi_3\c F \\
\nab_3(\nab^2 f)&=&\nab\big(\nab F  +\chib \c\nab f +\psi_3\c F   ) +
\chib\c\nab^2 f +\bb\c\nab f +\psi_3\c\nab_3(\nab f)\\
&=&\nab^2 F  +\psi_3\c \nab F+    \nab \psi_3 \c F        +\chib \c\nab^2 f + 
\nab\chib\c \nab f \\
&+&\psi_3\c\nab_3(\nab f)+\chibh\c\psi_3\c\nab f
\\
\nab_3(\nab^3 f)&=&\nab^3 F+\psi_3\c\nab^2 F+\nab\psi_3\c\nab F+\nab^2\psi_3 \c F\\
&+&\chib\c \nab^3 f+\nab\chib\c \nab^2 f 
+\nab^2\chib\c \nab f\\
& +&\nab\big(\psi_3\nab_3(\nab f)+ \chibh\c\psi_3\c\nab f  \big)+\bb\c\nab^2 f
+\psi_3\c\nab_3(\nab^2 f)
\eeaa
or, 
\beaa
\nab_3(\nab^3 f)&=&\nab^3 F+\psi_3\c\nab^2 F+\nab\psi_3\c\nab F+\nab^2\psi_3 \c F\\
&+&\chib\c \nab^3 f+\nab\chib\c \nab^2 f 
+\nab^2\chib\c \nab f\\
& +&\psi_3\c \nab_3(\nab^2 f)+\nab\psi_3\c\nab_3(\nab f)+\psi_3[\nab,\nab_3](\nab f)
 + \nab(\chibh\c\psi_3\c\nab f )
\eeaa
Applying the calculations above to  $f=\Om^{-1}\trch$, $F= -\frac 1 2\Om^{-1} \trch^2-\Om^{-1}|\chi|^2 $     and using   $\nab(\Om^{-1})= -\Om^{-2} \nab \Om=-\frac 1 2 \Om^{-2}(\eta-\etab)$ we derive, omitting factors of $\Om$ which are bounded in $L^\infty$,
\beaa
\nab_4(\lap\nab\trchb')&=&\chibh\c \lap\nab\chibh+\chib\c\nab^3\trchb+\nab\chibh\c\nab^2\chibh + \nab\chib\c\nab^2\trchb+F\\
F&=&\trchb_0\big(\psi_3\c\nab^2\psi_3+\nab\psi_3\c\nab\psi_3 +  \psi_3\c\psi_3\c\nab\psi_3     )\\
&+&\psi_3\c \psi_3\c\nab^2\psi_3+ \psi_3 \c\nab\psi_3\c\nab\psi_3 +\psi_3\c   \psi_3\c\psi_3\c\nab\psi_3
\eeaa
Making use of our estimates for $\psi_3$ we easily derive,
 with a constant $C=C(\OO^{(0))},\RR,\RRb)$,
  \beaa
  \|F\|_{\Lsc^2(\Hb)}&\les&  \de^{1/2} C
  \eeaa
  Thus,
  \bea
  \nab_3(\lap\nab\trchb')&=&\chibh\c \lap\nab\chibh+\chib\c\nab^3\trchb+\nab\chibh\c\nab^2\chibh + \nab\chib\c\nab^2\trchb+F_1\label{eq:nab3chiv}\\
  \|F_1\|_{\Lsc^2(H)}&\les&  \de^{1/2} C\nn 
  \eea 
   Observe that neither the principal   term $ \chibh\c\nab\lap \chibh$ or  the lower order term 
  $ \nab\chibh  \c\nab^2 \chibh $
 appear to  satisfy an $\Lsc^2(H)$ estimate.  The principal terms seems particularly nasty since
  we can't possible expect to estimate three derivatives of $\chibh$  using norms which involve only
  one derivative of curvature components.  Clearly another renormalization is
  needed. In fact  we make use of equation \eqref{eq:nab3lapphi1} which we write in the form, 
\beaa
\lap \nab\chibh=\nab_3\lap\phib-\trchb_0 \nab^2\phib-\chibh \c\nab^2\phib-E
\eeaa
 We can thus replace the dangerous term $\lap\nab  \chibh$ in \eqref{eq:nab3chiv}  and obtain,
 \beaa
 \nab_3(\lap\nab\trchb')&=&\chibh\c\nab_3\lap \phib+\chib\c\nab^3\trchb+\nab\chibh\c\nab^2\chibh + \nab\chib\c\nab^2\trchb+F_2\\
 F_2&=&F_1-(\trchb_0 \nab^2\phib-\chibh \c\nab^2\phib-E)\c \chibh
 \eeaa
 In view of our estimates for $\phib$ we have,
 \beaa
 \|F_2\|_{\Lsc^2(\Hb)}&\les& C\de^{1/2} (1+\de^{1/2} C)\|\nab^2\phib\|_{\Lsc^2(\Hb) }\\
 \eeaa
 Now,  recalling also the  definition of $\phib$,
 \beaa
 \nab_3(\lap\nab\trchb' -\chibh\c\lap \phib)&=&-\nab_3\chibh\c \lap\phib+\trchb_0 \nab^3\trchb  + \psi_3\c\nab^3\trchb+\nab_3\phib\c\nab^2\chib\\
 & +& \nab\trchb\c\nab^2\trchb+F_2
 \eeaa
 Hence,
 \beaa
\| \lap\nab\trchb' \|_{\Lsc^2(u,\ub)}&\les&\| \lap\nab\trchb' \|_{\Lsc^2(0,\ub)}+ 
 C\de^{1/2}  \|\chibh\|_{\Lsc^\infty}  \c  \|\lap \phib\|_{\Lsc^2(u,\ub)}\\
 &+&(1+C\de^{1/2} )\int_0^u\|\nab^3\trchb\|_{\Lsc^2(u',\ub)} du'\\
 &+&C\de^{1/2} \|\nab_3\chibh\|_{\Trsc(\Hb)}    \c  \|\lap \phib\|_{\Lsc^2(\Hb)}
\\
&+&\de^{1/2}   \|\nab_3\phib\|_{\Trsc(\Hb)}\c\|\nab^2\chib\|_{\Lsc^2(\Hb)}\\
&+&\de^{1/2} \|\nab\trchb\|_{\Lsc^\infty}
\c\|\nab^2\trchb\|_{\Lsc^2(\Hb)}\\
&+& \|F_2\|_{\Lsc^2(\Hb)} 
\eeaa 
 Using the calculus inequalities  of subsection \ref{subs.calculus.ineq}
 and our estimates for $\nab^2\nab_3\trchb$, 
 \beaa
 \|\nab \trchb\|_{\Lsc^\infty}&\les& C+ \|\nab^3\trchb\|_{\Lsc(\Hb)}
 \eeaa
 Also, in view of 
  the trace estimate \eqref{estim.tracenab33phib},
\beaa
 \|\nab_3\phib\|_{\Trsc(\Hb)} &\les&  C+ \|\nab^3\trchb\|_{\Lsc(\Hb)}
\eeaa

Hence,
\beaa
\| \lap\nab\trchb' \|_{\Lsc^2(u,\ub)}&\les&\| \lap\nab\trchb' \|_{\Lsc^2(0,\ub)}+
 C\de^{1/2}  \|\nab^2 \phib\|_{\Lsc^2(u,\ub)}\\
&+&(1+C\de^{1/2} )\int_0^u\|\nab^3\trchb\|_{\Lsc^2(u',\ub)} du'\\
&+&C\de^{1/2} \|\nab_3\chibh\|_{\Trsc(\Hb)}    \c  \|\lap \phib\|_{\Lsc^2(\Hb)}\\
&+&C\de^{1/2} \|\nab^3\trchb\|_{\Lsc^2(\Hb)}+ 
C ^2\de^{1/2}
\eeaa
 Now,
  \beaa
  \|\lap\nab\trchb'\|_{\Lsc^2(u,\ub)}&\les& \|\lap\nab\trchb\|_{\Lsc^2(u,\ub)}+\de^{1/2}C\big(   \|\nab^2\omb\|_{\Lsc^2(u,\ub)}+C\de^{1/2}\big)
  \eeaa
   Now,
  using the elliptic estimates discussed in subsection \ref{section.Hodge},
  we have 
  and our estimates for $K$, we deduce
 \beaa
 \| \nab^3\trchb \|_{\Lsc^2(S)}&\les &\| \lap  \trchb\|_{\Lsc^2(S)}\\
 &+&\de^{1/2}\big( \|\nab K\|_{\Lsc^2(S)}
 \|\nab\trchb \|_{\Lsc^\infty(S)}+\|K\|_{\Lsc^4(S)}\|\nab^2\trchb\|_{\Lsc^4(S)}\big)\\
 &\les&\| \lap\nab\trchb\|_{\Lsc^2(S)}+\de^{1/2}\big( \|\nab \trchb\|_{\Lsc^\infty(S)}+\|\nab^2 \trchb\|_{\Lsc^4(S)}\big)\\
 &\les&\| \lap\nab\trchb \|_{\Lsc^2(S)}+\de^{1/2}\big(C+ \| \nab^3 \trchb\|_{\Lsc^2(\Hb)}\big)
 \eeaa
Hence,
 after using Gronwall, 
\beaa
 \| \nab^3\trchb \|_{\Lsc^2(u,\ub)}&\les & \| \nab^3\trchb \|_{\Lsc^2(0,\ub)}+C\de^{1/2}\big(\|\nab^2\omb\|_{\Lsc^2(u,\ub)}+\|\nab^2\phib\|_{\Lsc^2(u,\ub)}\big)\\
 &+&C\de^{1/2} \|\nab_3\chibh\|_{\Trsc(\Hb)}    \c  \|\lap \phib\|_{\Lsc^2(\Hb)}\\
 &+&C\de^{1/2} \|\nab^3\trchb\|_{\Lsc^2(\Hb)}+ 
C ^2\de^{1/2}
  \eeaa
Thus,
after integration,
\bea
 \| \nab^3\trchb \|_{\Lsc^2(\Hb^{(0,u)}}^2&\les& C^2+C^2\de \int_0^u
 \|\nab_3\chibh\|_{\Trsc(\Hb^{(0,u'})}^2    \c  \|\lap \phib\|_{\Lsc^2(\Hb^{(0,u')})}^2 du'\nn\\
 \label{eq:man3trchb.almost}
\eea
It remains to estimate  the trace norm 
$ \|\nab_3\chibh\|_{\Trsc(\Hb^{(0,u'})}  $.
We claim the following,
\begin{lemma}
\label{le:tracenorm.chibh}
There exists a constant $C$ depending
only on $\OO^{(0)}, \RR, \RRb$ as well
as $\|\nab_3\aa\|_{\Lsc^2(\Hb)}$ such that,
\bea
\|\nab_3\chibh\|_{\Trsc(\Hb)}&\les& C\de^{-1/2}.
\eea
\end{lemma}
\begin{proof}
in view of the trace estimate \eqref{eq:*b}, we have for $\Hb=\Hb^{(0,u')}$,
\beaa
 \|\nab_3\chib\|_{\Trsc(\Hb)} &\les&  \|\nab_3^2\chibh\|_{\Lsc^2(\Hb)}+\|\nab\nab_3\chibh\|_{\Lsc^2(\Hb)}\\
 &+&\|\nab^2\chibh\|_{\Lsc^2(\Hb)}+\|\chibh\|_{\Lsc^2(\Hb)}+C\de^{1/2}\|\chibh\|_{\Lsc^\infty}\nn
\eeaa
Observe that,
\beaa
  \|\nab_3\chibh\|_{\Lsc^2(\Hb)}+ \|\chibh\|_{\Lsc^2(\Hb)}&\les& C\de^{-1/2}
  \eeaa
  We claim also that,
  \beaa
  \|\nab^2_3\chibh\|_{\Lsc^2(\Hb)} &\les& C\de^{-1/2}+\|\nab_3\aa\|_{\Lsc^2(\Hb)}.
  \eeaa
  Indeed, differentiating,
\beaa
\nab_3\chibh&=&-\aa -\trchb\chibh   -2\omb\, \chibh
\eeaa
Thus, 
\beaa
\nab_3^2\chibh&=&-\nab_3\aa-\nab_3\trchb\,\c \chibh-\trchb\,\c \nab_3\chibh -2\nab_3\omb\c\chibh-\omb \c\nab_3\chibh
\eeaa
Hence,
\beaa
 \|\nab_3^2\chibh\|_{\Lsc^2(\Hb)}&\les& \|\nab_3\aa\|_{\Lsc^2(\Hb)}+\|\chibh\|_{\Lsc^2(\Hb)}+C\de^{1/2}\big(\|\nab_3\omb\|_{\Lsc^2(\Hb)}+\nab_3\chibh\|_{\Lsc^2(\Hb)}\big)\\
 &\les& C\de^{-1/2} + \|\nab_3\aa\|_{\Lsc^2(\Hb)}
\eeaa
which completes the proof of our estimate.
\end{proof}
Returning to \eqref{eq:man3trchb.almost}, we have with a constant  $C$ depending on $\OO^{(0)}, \RR, \RRb,$
 as well as   $\|\nab_3\aa\|_{\Lsc^2(\Hb)})$,
\beaa
\| \nab^3\trchb \|_{\Lsc^2(\Hb^{(0,u)}}^2&\les& C^2+C^2 \int_0^u
     \|\nab^2 \phib\|_{\Lsc^2(\Hb^{(0,u')})}^2 du'\\
     &\les& C^2\big(1+ \int_0^u
     \|\nab^3 \trchb \|_{\Lsc^2(\Hb^{(0,u')})}^2 du'\big)
 \eeaa
 Thus,
 applying Gronwall once more we 
 derive,
 \beaa
\| \nab^3\trchb \|_{\Lsc^2(\Hb^{(0,u)}}^2&\les& C^2
\eeaa
This  finishes the proof of  the second part of the following. 
\begin{proposition} 
\label{prop:final.traceestim,nabchibh}
The following estimates hold true with a constant $C$
depending on $\OO^{(0)}, \RR, \RRb$ as well
as $  \sup_{u}  \|\nab_4\a\|_{\Lsc^2(H_{u})  }$ and $\sup_{\ub}\|\nab_3\aa\|_{\Lsc^2(\Hb_{\ub}) }$
\begin{enumerate}
\item We have along $H=H_u$,
\beaa
\|\nab^3\trch\|_{\Lsc^2(H)}+\|\nab\trch\|_{\Lsc^\infty}&\les& C\\
\sup_S\|\nab\chih\|_{\Lsc^4(S)}+\|\nab\chih\|_{\Trsc(H)}&\les& C.
\eeaa

\item We have along $\Hb=\Hb_{\ub}$,
\beaa
\|\nab^3\trchb\|_{\Lsc^2(\Hb)}+\|\nab\trchb\|_{\Lsc^\infty}&\les& C\\
\sup_S \|\nab\chibh\|_{\Lsc^4(S)}+\|\nab\chibh\|_{\Trsc(\Hb)}&\les& C.
\eeaa
\end{enumerate}

\end{proposition}
 
 \subsection{Estimates for the trace norms of $\nab\eta, \nab\etab$}
 As in the previous subsection we need a series of renormalization.
 The proof follows, however, the same outline as  above.
 We first prove the following,
 \begin{proposition}
  \label{le:estimates.phit-phitb}
    
  Consider the following  transport equations along $H=H_u$, respectively $\Hb=\Hb_{\ub}$
   \bea
 \nab_4 \phif&=& \nab\eta,            \qquad \phif(0,\ub)=0         \label{def:transp.phif}\\
 \nab_4\phibf&=&\nab\etab         \qquad \phibf(0,\ub)=0         \label{def:transp.phibf}
 \eea
 and 
  \bea
 \nab_3 \phit&=& \nab\eta,            \qquad \phit(0,\ub)=0         \label{def:transp.phit}\\
 \nab_3 \phibt&=& \nab\etab,            \qquad \phibt(0,\ub)=0         \label{def:transp.phibt}
 \eea
\begin{enumerate}
\item Solutions $\phi=(\phif, \phibf)$ of  \eqref{def:transp.phif} -\eqref{def:transp.phibf} verify
 the estimates,
 \bea
  \|\phi\|_{\Lsc^2(S)}+\|\phi\|_{\Lsc^4(S)}+
 \|\nab\phi\|_{\Lsc^2(S)}+\|\nab_4\phi\|_{\Lsc^2(S)}&\les & C\label{estimates.easphi11}\\
 \| \nab\nab_4 \phi\|_{\Lsc^2(H)} +\|\nab_4^2\phi\|_{\Lsc^2(H)}&\les& C\label{estimates.easphi12}
 \eea
 with  a constant $C=C(\OO^{(0)},\RR,\RRb)$.
  Moreover,
 \bea
  \| \nab^2 \phi\|_{\Lsc^2(H)}&\les& \|\nab^2\mu\|_{\Lsc^2(H)}+C
  \eea
   As a consequence (see  calculus inequalities of subsection \ref{subs.calculus.ineq})
  we also have,
  \bea
  \|  \phi\|_{\Lsc^\infty}&\les&  \|\nab^2\mu\|_{\Lsc^2(H)}+C
  \eea
  and as a consequence of  the trace estimate \eqref{eq:*},
   \bea
  \| \nab_4 \phi\|_{\Trsc(H)}&\les&  \|\nab^2\mu\|_{\Lsc^2(H)}+C     \label{estim.tracenab4phif}
  \eea

  \item 
  Solutions $\phib=(\phit, \phibt)$ of  \eqref{def:transp.phit},         \eqref{def:transp.phibt}     verify
 the estimates,
 \bea
  \|\phib\|_{\Lsc^2(S)}+\|\phib \|_{\Lsc^4(S)}+
 \|\nab\phib\|_{\Lsc^2(S)}+\|\nab_3\phib\|_{\Lsc^2(S)}&\les & C\label{estimates.easphitb1}\\
 \| \nab\nab_3 \phib\|_{\Lsc^2(\Hb)} +\|\nab_3^2\phi\|_{\Lsc^2(\Hb)}&\les& C\label{estimates.easphitb2}
 \eea
 with  a constant $C=C(\OO^{(0)},\RR,\RRb)$.
 Moreover,
 \bea
  \| \nab^2( \phit, \phibt)\|_{\Lsc^2(\Hb)}&\les& \|\nab^2\mub\|_{\Lsc^2(\Hb)}+C\label{estim.L12.phib}
  \eea
  As a consequence (see  calculus inequalities of subsection \ref{subs.calculus.ineq})
  we also have,
  \bea
  \| ( \phit,\phibt)\|_{\Lsc^\infty}&\les&  \|\nab^2\mub\|_{\Lsc^2(\Hb)}+C
  \eea
   and as a consequence of  the trace estimate \eqref{eq:*b},
   \bea
  \| \nab_3 (\phit, \phibt)\|_{\Trsc(\Hb)}&\les&  \|\nab^2\mub\|_{\Lsc^2(\Hb)}+C\label{estim.tracenab3phib}
  \eea

 \end{enumerate}
 \end{proposition}
 \begin{proof}
  We start with 
 \beaa
 \nab_3\phit=\eta,\qquad \nab_3\phibt=\etab
 \eeaa
 Commuting  both equations with $\lap$ and 
 proceeding  exactly as in the derivation
 of \eqref{eq:nab3lapphi2} we derive 
 \bea
   \nab_3\, \lap \phit&=&\nab \lap  \eta+\trchb_0 \nab^2\phit+\chibh \c\nab^2\phit+E   \label{eq.nab3lapphit}\\
  \nab_3 \, \lap \phibt&=&  \nab\lap\etab +\trchb_0 \nab^2\phibt+\chibh \c\nab^2\phibt+\underline{E} \label{eq.nab3lapphibt} \\
    \|E\|_{\Lsc^2(\Hb)}&\les&C \de^{1/2}\big(C+\|\nab^2\phit\|_{\Lsc^2(\Hb)}\big)\nn\\
      \|\underline{E}\|_{\Lsc^2(\Hb)}&\les&C \de^{1/2}\big(C+\|\nab^2\phibt\|_{\Lsc^2(\Hb)}\big)\nn
   \eea 
 Recall that, see \eqref{eq:div.curleta}, \eqref{eq:div.curletab},
 \beaa
 \div \eta&=&-\mu-\rho,\qquad \curl \eta=\si-\frac 1 2 \chih\wedge\chibh\\
 \div \etab&=&-\mub-\rho,\qquad \curl \eta=\si-\frac 1 2 \chih\wedge\chibh\\
 \eeaa
 i.e., schematically,
 \beaa
\dcalll \dcal \eta&=& \dcalll( -\mu-\rho, \si -\chih\wedge\chibh)\\
\dcalll \dcal \etab&=& \dcalll( -\mub-\rho, \si -\chih\wedge\chibh)
 \eeaa
 Prceeding as in the derivation of \eqref{eq:nablapchibh}
 we find, schematically, 
 \beaa
 \nab\lap\eta&=&\nab^2\mu+\nab^2(\rho,\si)+F_1\\
 \nab\lap\etab&=&\nab^2\mub+\nab^2(\rho,\si)+F_1\\
 \|F_1\|_{\Lsc^2(\Hb)}&\les&C
 \eeaa
 We now make use of  the equations,
 see equations \eqref{nab3.omega.second}  and \eqref{eq:omt3}, 
 \beaa
 \nab_3\omega&=&\frac 12 \rho+ 2\omega\omegab+\frac 34 |\eta-\etab|^2+\frac 14 (\eta-\etab)\cdot (\eta+\etab)-
\frac 18 |\eta+\etab|^2\\
 \nab_3\omt&=&\frac 12 \si
 \eeaa
 Proceeding now exactly as in the derivation
 of \eqref{eq:nab2bb} and \eqref{estim:nab2bb},
 we deduce,
 \beaa
 \nab^2(\rho,\si)&=&\nab_3\nab^2(\om,\omt)+F_2\\
 \|F_2\|_{\Lsc^2(\Hb)}&\les& C.
 \eeaa
 Therefore, just as  before for  the derivation of $\nab \lap \chih$, schematically,
 \bea
 \nab\lap \eta &=& \nab_3\nab^2(\om,\omt)+\nab^2\mu+F\\
 \nab\lap\etab &=& \nab_3\nab^2(\om,\omt)+\nab^2\mub+\underline{F}\\
 \|F,\underline{F}\|_{\Lsc(\Hb)}&\les & C.\nn
 \eea
 Thus,
 back to \eqref{eq.nab3lapphit} and \eqref{eq.nab3lapphibt}
  we deduce  (just as in  \eqref{renorm.phi}
  \bea
    \nab_3\big(\, (\lap \phit   - \nab^2(\om,\omt)\,\big)   &=&\nab^2\mu+ \trchb_0 \nab^2\phit+\chibh \c\nab^2\phit+E\\
    \|E\|_{\Lsc^2(\Hb)}&\les&C\big(1+\de^{1/2}\|\nab^2\phit\|_{\Lsc^2(\Hb)}\big)\nn
  \eea
  and,
   \bea
    \nab_3\big(\, (\lap \phibt   - \nab^2(\om,\omt)\,\big)   &=&\nab^2\mub+ \trchb_0 \nab^2\phibt+\chibh \c\nab^2\phibt+\underline{E}\\
    \|\underline{E}\|_{\Lsc^2(\Hb)}&\les&C\big(1+\de^{1/2}\|\nab^2\phibt\|_{\Lsc^2(\Hb)}\big)\nn
  \eea
  We then proceed with elliptic  $\Lsc^2$  estimates, exactly as in \eqref{Hodge.nab2phib} and,
  after using also Gronwall, we find  (as in \eqref{nab2phib.100})
 \bea
\| \nab^2  \phit\|_{\Lsc^2(u,\ub)}&\les&\| \nab^2(\om,\omt)\|_{\Lsc^2(u,\ub)}+\int_0^u \|\nab^2\mu\|_{\Lsc^2(u',\ub)}
du'\label{nab2phit.100}\\
&+&  C(1+ \de^{1/2} )\| \nab^2 \phit\|_{\Lsc^2(\Hb)} 
\nn
\eea
and 
 \bea
\| \nab^2  \phibt\|_{\Lsc^2(u,\ub)}&\les&\| \nab^2(\om,\omt)\|_{\Lsc^2(u,\ub)}+\int_0^u\|\nab^2\mub\|_{\Lsc^2(u',\ub)}du'\label{nab2phit.1001}\\
&+&  C(1+ \de^{1/2} )\| \nab^2 \phibt\|_{\Lsc^2(\Hb)} 
\nn
\eea
Integrating we deduce, for  $C\de^{1/2}$ sufficiently small,
\beaa
\| \nab^2 \phit\|_{\Lsc^2(\Hb)} &\les& C+\|\nab^2\mu\|_{\Lsc^2(\Hb)}\\
\| \nab^2 \phibt\|_{\Lsc^2(\Hb)} &\les& C+\|\nab^2\mub\|_{\Lsc^2(\Hb)}
\eeaa
as desired.
\end{proof}
 It remains to estimate $\|\nab^2\mu\|_{\Lsc^2(H)} $ and $\|\nab^2\mub\|_{\Lsc^2(\Hb)}$.
 As before we treat only the estimate for
 the  slightly more difficult case of  $\mub$.
 In view of the proof of the previous proposition
 we have (neglecting signs and constants, as before),
 \bea
 \nab \lap  \eta&=&  \nab_3\, \lap \phit+\trchb_0 \nab^2\phit+\chibh \c\nab^2\phit+E \\
  \nab\lap\etab&=&\trchb_0 \nab^2\phibt+\chibh \c\nab^2\phibt+\underline{E} \\
    \|E\|_{\Lsc^2(\Hb)}&\les&C \de^{1/2}\big(C+\|\nab^2\phit\|_{\Lsc^2(\Hb)}\big)\nn\\
      \|\underline{E}\|_{\Lsc^2(\Hb)}&\les&C \de^{1/2}\big(C+\|\nab^2\phibt\|_{\Lsc^2(\Hb)}\big)\nn
      \eea

We start with  the transport equation  \eqref{eq:transport-mub},
\beaa
\nab_3\mub+\trchb\mub&=&-\frac 1 2\trchb  \div\eta+ (\etab-\eta)\nab\trchb  \\
&+&\chibh\c \nab  (2\etab-\eta)
+\frac 1 2\,  \chih\c \aa 
 -(\etab-3\eta)\c\bb +\frac 1 2 \trchb \rho\nn\\
&+&\frac 1 2 \trchb(|\etab|^2-\etab\c\eta)+\frac 1 2 (\eta+\etab)\c\chibh\c (\etab-\eta)
\eeaa

Commuting with the laplacean, we derive
\beaa
\nab_3\lap \mub&=&\chibh\c \lap  \nab (\etab+\eta)+\trchb \lap \div \eta+   (\nab\eta+\nab\etab)\c\nab^2\chibh+\trchb\lap \mub\\
&+&(\nab^2\eta+\nab^2\etab)\c\nab\chibh +\frac 1 2\,  \chih\c\lap  \aa -(\etab-3\eta)\c\lap \bb+\frac 1 2 \trchb\lap  \rho\\
&+&\Err
\eeaa
Here, and in what follows, $\Err$ denotes any term which
allows a bound of the form,
\bea
\| \Err\|_{\Lsc^2(\Hb)}&\les&C\label{eq:error}
\eea 
Using  equation,
$
\nab_3\chibh=-\aa-\trchb_0\, \chibh+\psi_3\c\psi_3
$
we  write,
$$
\lap\aa=-\nab_3\lap\chibh+\Err.
$$
Using equation, 
$
\nab_3\etab=\bb +\chib\c(\eta-\etab)
$
we can write
\beaa
\lapp \bb&=&\nab_3\lap\etab+\Err
\eeaa
Using equation
 $
 \nab_3\omega=\frac 12 \rho+\psi\c\psi 
 $
 we can write 
 \beaa
 \lap\rho&=& 2 \nab_3\lap \om+\Err
 \eeaa
Therefore we can write,
\beaa
\nab_3\lap \mub&=&\chibh\c\nab_3\lap(\phit+\phibt)
+\trchb \nab_3\lap(\phit+\phibt)\\
&+&\nab_3(\phit+\phibt)\c\nab^2\chibh+\nab^2(\eta+\etab)\c\nab\chibh\\
&+&\trchb\nab_3\lap \om +(\eta+\etab)\nab_3\lap\etab+ \chibh\c\nab_3\lap\chibh+\Err_{\phi}
\eeaa
with $\Err_{\phi}$ verifying,
\beaa
\|\Err_{\phi}\|_{\Lsc^2(\Hb)}&\les& C\big(1+\|\nab^2(\phit+\phibt)\|_{\Lsc^2(\Hb)}\big)\\
&\les& C\big(1+\|\nab^2\mub\|_{\Lsc^2(\Hb)}\big)
\eeaa
Therefore,
introducing the renormalized quantity
\bea
\mubv=\mub-\chib\c\lap(\phit+\phibt)-\trchb\c\lap\om-(\eta+\etab)\c
\lap \etab-\chibh\c\lap\chibh\label{def:muv}
\eea
we have,
\beaa
\nab_3\mubv&=&-\nab_3\chib\c\lap(\phit+\phibt+\chibh)-
\nab_3\trchb\c \lap(\phit+\phibt)\\
&+&\nab_3(\phit+\phibt)\c\nab^2\chibh+\nab^2(\eta+\etab)\c\nab\chibh\\
&+&\nab_3\trchb\c\lap\om +\nab_3(\eta+\etab)\c\lap\etab+\Err_{\phi}
\eeaa
Consequently,
\beaa
\|\mubv\|_{\Lsc^2(u,\ub)}&\les&
 \de^{1/2} \|\nab_3\chibh\|_{\Trsc(\Hb)}\c\|\nab^2(\phit+\phibt+\chibh)\|_{\Lsc^2(\Hb)}\\
 &+& \de^{1/2} \|\nab_3(\phit+\phibt)\|_{\Trsc(\Hb)}\c\|\nab^2(\phit+\phibt)\|_{\Lsc^2(\Hb)}\\
 &+&\de^{1/2} \|\nab_3(\eta+\etab)\|_{\Trsc(\Hb)}\c\|\nab^2 \etab\|_{\Lsc^2(\Hb)}\\
 &+&\de^{1/2}\|\nab\chibh \|_{\Trsc(\Hb)} \c  
 \|\nab^2(\eta+\etab)\|_ {\Lsc^2(\Hb)}+   \Err_\phi
\eeaa
We recall
from the previous subsection, see lemma \ref{le:tracenorm.chibh}, that
\beaa
\|\nab_3\chibh\|_{\Trsc(\Hb)}&\les& C\de^{-1/2}
\eeaa
with a constant $C$ depending
only on $\OO^{(0)}, \RR, \RRb$ as well
as $\|\nab_3\aa\|_{\Lsc^2(\Hb)}$.
Also, from the previous section, we have (see proposition 
\ref{prop:final.traceestim,nabchibh})
\beaa
\|\nab\chibh\|_{\Trsc(\Hb)}&\les& C
\eeaa
Also, in view of \eqref{estim.tracenab3phib},
\beaa
 \| \nab_3( \phit, \phibt)\|_{\Trsc(\Hb)}&\les&  \|\nab^2\mub\|_{\Lsc^2(\Hb)}+C
\eeaa
Also, we can easily show, with the help of the trace 
estimates of proposition \ref{prop.trace} and  our 
Ricci coefficient estimates, 
\beaa
\|\nab_3(\eta, \etab)\|_{\Lsc^2(\Hb)}&\les&C
\eeaa
Consequently,
\beaa
\|\mubv\|_{\Lsc^2(u,\ub)}&\les&\|\mubv\|_{\Lsc^2(0,\ub)}+(1+C\de^{1/2}\|\nab^2\mub\|_{\Lsc^2(\Hb)}
\eeaa
On the other hand,
\beaa
\|\mubv\|_{\Lsc^2(u,\ub)}&\les&\|\lap\mub\|_{\Lsc^2(u,\ub)}+\|\nab^2\omb\|_{\Lsc^2(u,\ub)}+C\de^{1/2}\|\nab^2\eta\|_{\Lsc^2(u,\ub)}\\
&+&C\de^{1/2}
\|\nab^2\chibh\|_{\Lsc^2(u,\ub)}
\eeaa
Hence,
\beaa
\|\lap\mub\|_{\Lsc^2(u,\ub)}&\les&\|\mubv\|_{\Lsc^2(0,\ub)}+\|\nab^2\omb\|_{\Lsc^2(u,\ub)}+C\de^{1/2}\|\nab^2\eta\|_{\Lsc^2(u,\ub)}\\
&+&\|\nab^2\mub\|_{\Lsc^2(\Hb)}+C\de^{1/2}\|\nab^2\mub\|_{\Lsc^2(\Hb)}
\eeaa
We  can now proceed precisely as in the last part of the  proof
of proposition  \ref{prop:final.traceestim,nabchibh}
to deduce, after applying elliptic estimates and integrating,
\beaa
\|\nab^2\mub\|_{\Lsc^2(\Hb_{\ub}^{(0,u)} )}&\les&\OO^{(0)}+
 (1+C\de^{1/2} )\int_0^u\|\nab^2\mub\|_{\Lsc^2(\Hb_{\ub}^{(0,u')})} du' +C
\eeaa
from which the desired estimate follows.
We have thus proved the second part of the  following:
\begin{proposition} 
\label{prop:final.traceestim,nabeta.etab}
The following estimates hold true with a constant $C$
depending on $\OO^{(0)}, \RR, \RRb$ as well
as $  \sup_{u}  \|\nab_4\a\|_{\Lsc^2(H_{u})  }$ and $\sup_{\ub}\|\nab_3\aa\|_{\Lsc^2(\Hb_{\ub}) }$.
\begin{enumerate}
\item We have along $H=H_u$,
\beaa
\|\nab(\eta, \etab)\|_{\Trsc(H)} &\les& C\\
\eeaa
\item We have along $\Hb=\Hb_{\ub}$,
\beaa
\|\nab(\eta, \etab)\|_{\Trsc(\Hb)}&\les& C\\
\eeaa
\item Also,
\beaa
\sup_S \|\nab(\eta, \etab)\|_{\Lsc^4(S)}&\les& C
\eeaa
\end{enumerate}
\end{proposition}
\subsection{Refined estimate for $^{(3)}\phi$}
We end this section by establishing a more refined estimate on $^{(3)}\phi$. This estimate is needed in the argument for the formation of a trapped surface described in our introduction. We examine  the equation
$$
\nab_3\, ^{(3)} \phi=\nab\eta.
$$
Commuting with $\nab$ we obtain 
$$
\nab_3 \nab \,^{(3)} \phi=(\trchb_0+\psi)\c\nab\, ^{(3)} \phi+(\Psi+\psi\c\psi) \,^{(3)} \phi+\nab^2\eta
$$
Taking into account triviality of the data for $\nab \,^{(3)}\phi$, non-anomalous estimates for $\Psi$ 
appearing in this equation, and Gronwall we obtain
$$
\|\nab \,^{(3)} \phi\|_{\Lsc^2(S)}\les \|\nab^2\eta\|_{\Lsc^2(\Hb_{\ub})} +\de^{\frac 12} C.
$$
Using Proposition \ref{prop:bettereta} we obtain 
$$
\|\nab\, ^{(3)} \phi\|_{\Lsc^2(S)}\les \|\nab\rho\|_{\Lsc^2(\Hb_{\ub})}+ \|\nab\si\|_{\Lsc^2(\Hb_{\ub})} +\de^{\frac 14} C.
$$
Combining with the interpolation estimates 
\begin{align*}
&\|^{(3)}\phi\|_{\Lsc^\infty(S)}\les  \|^{(3)}\phi\|_{\Lsc^4(S)}^{\frac 12}  \|\nab ^{(3)}\phi\|_{\Lsc^4(S)}^{\frac 12} +
\de^{\frac 14}  \|^{(3)}\phi\|_{\Lsc^4(S)},\\ 
& \|\nab ^{(3)}\phi\|_{\Lsc^4(S)}\les  \|\nab ^{(3)}\phi\|_{\Lsc^2(S)}^{\frac 12}  \|\nab^2  (^{(3)}\phi)\|_{\Lsc^2(S)}^{\frac 12} 
+\de^{\frac 14}  \|\nab ^{(3)}\phi\|_{\Lsc^2(S)}
\end{align*}
we conclude
\begin{proposition}\label{prop:betterphi}
The solution $^{(3)} \phi$ of the problem $\nab_3 \, ^{(3)} \phi=\nab\eta$ with trivial initial data satisfies 
$$
\|^{(3)}\phi\|_{\Lsc^\infty(S)}\les C \left (\|\nab\rho\|_{\Lsc^2(\Hb_{\ub})}+ \|\nab\si\|_{\Lsc^2(\Hb_{\ub})} \right)^{\frac 14}
+ C\de^{\frac 18}.
$$
\end{proposition}
 
\section{Trace estimates for curvature}\label{sec:trace}
\begin{proposition}
\label{prop:traceCurv}
Under the assumptions of the finiteness of the norms $\RR$ and $\RRb$, which include 
$\|\nab_3\aa\|_{\Lsc^2(\Hb_{\ub}) }$ and 
the anomalous norm $\|\nab_4\a\|_{\Lsc^2(H_{u})}$ we have 
\begin{align*}
&\|\a\|_{\Tr_{sc}(H)}\le \de^{-\frac 14} C,\\
&\|(\b,\rho,\si)\|_{\Tr_{sc}(H)}\le C,\\
&\|(\rho,\si,\bb)\|_{\Tr_{sc}(\Hb)}\le C,\\
&\|\alphab\|_{\Tr_{sc}(\Hb)}\le \de^{-\frac 14} C
\end{align*}
\end{proposition}
The proof is based on the application of the trace inequalities of Proposition \ref{prop.trace}
and the  null structure equations   \eqref{null.str1}, \eqref{null.str3}-\eqref{null.str5}. According to these the curvature components $\Psi_4=\{\a,\b,\rho,\si\}$
can be expressed in the form
$$
\Psi_4=\nab_4\phi_4+\phi\c\phi,
$$
while $\Psi_3=\{\rho,\si,\bb,\alphab\}$ can be represented as
$$
\Psi_3=\nab_3\phi_3+\trchb_0\c\psi+\phi\c\phi,
$$
with\footnote{Recall that  $\ombtild=(\omb,\ombt)$  and $\omtild=(-\om,\omt)$, see \eqref{eq:ombt4}  and    \eqref{eq:omt3}.    }  $\phi_4\in \{\chih,\eta,<\omegab>\}$ and $\phi_3\in\{\chibh,\etab,<\om>\}$. 

Therefore,
\begin{align*}
&\|\Psi_4\|_{\Tr_{sc}(H)}\les \|\nab_4\phi_4\|_{\Tr_{sc}(H)}+\de^{\frac 12} \|\phi\|^2_{\Lsc^\infty},\\
&\|\Psi_3\|_{\Tr_{sc}(\Hb)}\les \|\nab_3\phi_3\|_{\Tr_{sc}(H)}+(1+\de^{\frac 12} \|\phi\|_{\Lsc^\infty})   \|\phi\|_{\Lsc^\infty}.
\end{align*}
By Proposition \ref{prop.trace} 
\begin{align*}
\|\nab_4\phi_4\|_{\Tr_{(sc)}({H})}&\les \left (\|\nab_4^2\phi_4\|_{\Lsc^2(H)}+\|\phi_4\|_{\Lsc^2(H)}+ 
\de^{\frac 12} C(\|\phi_4\|_{\Lsc^\infty}+
\|\nab_4\phi_4\|_{\Lsc^4(S)})\right)^{\frac 12}\\ &\times  \left (\|\nab^2\phi_4\|_{\Lsc^2(H)}+ \de^{\frac 12} C(\|\phi_4\|_{\Lsc^\infty}+
\|\nab\phi_4\|_{\Lsc^4(S)})\right)^{\frac 12}\\ &+\|\nab_4\nab\phi_4\|_{\Lsc^2(H)}+ \de^{\frac 12} C(\|\phi_4\|_{\Lsc^\infty}+
\|\nab_4\phi_4\|_{\Lsc^4(S)})+\|\nab\phi_4\|_{\Lsc^2(H)}
\end{align*}
\begin{align*}
\|\nab_3\phi_3\|_{\Tr_{(sc)}({\Hb})}&\les \left (\|\nab_3^2\phi_3\|_{\Lsc^2(\Hb)}+\|\phi_3\|_{\Lsc^2(\Hb)}+ \de^{\frac 12} C(\|\phi_3\|_{\Lsc^\infty}+
\|\nab_3\phi_3\|_{\Lsc^4(S)})\right)^{\frac 12}\\ &\times  \left (\|\nab^2\phi_3\|_{\Lsc^2(\Hb)}+ \de^{\frac 12} C(\|\phi_3\|_{\Lsc^\infty}+
\|\nab\phi_3\|_{\Lsc^4(S)})\right)^{\frac 12}\\ &+\|\nab_3\nab\phi_3\|_{\Lsc^2(\Hb)}+ \de^{\frac 12} C(\|\phi_3\|_{\Lsc^\infty}+
\|\nab_3\phi_3\|_{\Lsc^4(S)})+\|\nab\phi_3\|_{\Lsc^2(\Hb)}
\end{align*}
We observe that all the involved norms with the exception of $\|\nab_4^2\phi_4\|_{\Lsc^2(H)}$ and 
$\|\nab_3^2\phi_3\|_{\Lsc^2(\Hb)}$ have been
already estimated. 

Recall that the derivatives with no estimates are 
the $\Lsc^4(S)$ norms of $\nab_4\om, \nab_3\omb$ and 
either $\Lsc^2(H)$ and $\Lsc^2(\Hb)$ norms of 
$\nab\nab_4\om$ and $\nab\nab_3\omb$, while $\nab\nab_4\chih$ and $\nab\nab_3\chibh$ are controlled only along $H$ and $\Hb$ respectively. Finally, the $\Lsc^2(S)$ and 
$\Lsc^4(S)$ estimates for $\chih, \chibh$, 
$\nab_{3,4} \chih$, $\nab_{3,4} \chibh$ are $\de^{-\frac 12}$ and $\de^{-\frac 14}$
anomalous. Therefore, for $\phi_4=\chih$, i.e. $\Psi_4=\a$
$$
\|\nab_4\chih\|_{\Tr_{(sc)}({H})}\les C (\|\nab_4^2\chih\|_{\Lsc^2(H)}+C \de^{-\frac 12})^{\frac 12} + C,
$$
for $\phi_3=\chibh$, i.e. $\Psi_3=\alphab$
$$
\|\nab_3\chibh\|_{\Tr_{(sc)}({\Hb})}\les C (\|\nab_3^2\chibh\|_{\Lsc^2(\Hb)}+C \de^{-\frac 12})^{\frac 12} + C.
$$
The remaining $\phi_4, \phi_3$ satisfy 
\begin{align*}
&\|\nab_4\phi_4\|_{\Tr_{(sc)}({H})}\les C (\|\nab_4^2\phi_4\|_{\Lsc^2(H)}+C)^{\frac 12} + C,\\
&\|\nab_3\phi_3\|_{\Tr_{(sc)}({\Hb})}\les C (\|\nab_4^2\phi_3\|_{\Lsc^2(\Hb)}+C)^{\frac 12} + C.
\end{align*}
We now express
\begin{align*}
&\nab^2_4\phi_4=\nab_4\Psi_4 + \nab_4\phi\c\phi,\\
&\nab^2_3\phi_3=\nab_3\Psi_3 + \nab_3\phi\c\psi.
\end{align*}
Therefore,
\begin{align*}
&\|\nab^2_4\phi_4\|_{\Lsc^2(H)} \les \|\nab_4\Psi_4\|_{\Lsc^2(H)} +\de^{\frac 12}  \|\nab_4\phi\|_{\Lsc^2(H)} 
\|\phi\|_{\Lsc^\infty} \les  \|\nab_4\Psi_4\|_{\Lsc^2(H)} +C,\\
&\|\nab^2_3\phi_3\|_{\Lsc^2(\Hb)} \les \|\nab_3\Psi_3\|_{\Lsc^2(\Hb)} +\de^{\frac 12}  \|\nab_3\phi\|_{\Lsc^2(H)} 
\|\phi\|_{\Lsc^\infty} \les  \|\nab_4\Psi_4\|_{\Lsc^2(\Hb)} +C,
\end{align*}
where we took into account possible $\de^{-\frac 12}$ anomalies of $\|\nab_4\phi\|_{\Lsc^2(H)}$
and $ \|\nab_3\phi\|_{\Lsc^2(\Hb)}$. These immediately yield the desired trace estimates for 
$\a$ and $\alphab$. For the remaining components $\Psi_4, \Psi_3$ we may express from Bianchi
\begin{align*}
&\nab_4\Psi_4=\nab \Psi^4 + \phi\c\Psi,\\
&\nab_3\Psi_3=\nab\Psi^3+\trchb_0\c\Psi+\phi\c\Psi,
\end{align*} 
where $\Psi^4\in \{\a,\b\}$ and $\Psi^3\in\{\alphab,\bb\}$. Therefore,
\begin{align*}
& \|\nab_4\Psi_4\|_{\Lsc^2(H)}\les  \|\nab\Psi^4\|_{\Lsc^2(H)}+\de^{\frac 12} \|\phi\|_{\Lsc^\infty} 
\|\Psi\|_{\Lsc^2(H)}\les \RR+ C,\\
& \|\nab_3\Psi_3\|_{\Lsc^2(\Hb)}\les  \|\nab\Psi^3\|_{\Lsc^2(\Hb)}+(1+\de^{\frac 12} \|\phi\|_{\Lsc^\infty})
\|\Psi\|_{\Lsc^2(\Hb)}\les \RR+ C.
\end{align*}
In the last step we have to be careful to avoid the double anomalous term $\trchb_0\c\a$. Its appearance
is prohibited by the signature considerations, according to which 
$$
1\ge sgn(\nab_3\Psi_3)=sgn(\trchb_0\c\a)=2.
$$

\section{ Estimates for the Rotation Vectorfields }
\label{sect:deformation}
We define the algebra of rotation vectorfields $^{(i)} O$ obeying the commutation relations 
$$
[^{(i)}O,^{(j)}O]=\in_{ijk}\, ^{(k)}O,
$$
obtained by parallel transport of the standard rotation vectorfields on ${\Bbb S}^2=S_{u,0}\subset H_{u,0}$ 
along the integral curves of $e_4$. Suppressing the index $^{(i)}$ we obtain that 
$$
\nab_4 O_b=\chi_{bc} O_c.
$$
Commuting with $\nab$ and $\nab_3$ we obtain
\begin{align*}
&\nab_4 (\nab O)=\chi\c\nab O +\b\c O+\nab\chi\c O +\chi\c\etab\c O,\\ 
&\nab_4(\nab_3 O)=(\etab-\eta)\c\nab O + (\chi+\om) \nab_3 O +\si\c O + (\omb\c\chi+\eta\c\etab)\c O+\nab_3\chi\c O
\end{align*}

The only non-trivial components of the deformation tensor $\pi_{\a\b}=\frac 12 (\nab_\a O_\b+\nab_\b O_\a)$ 
are given below:
\begin{align*}
&\pi_{34}=-2(\eta+\etab)_a O_a,\\
&\pi_{ab}=\frac 12 (\nab_a O_b+\nab_b O_a),\\
&\pi_{3a}=\frac 12 (\nab_3 O_a-\chib_{ab} O_b):=\frac 12 Z_a.
\end{align*}
\subsection{Estimates for $H, Z$}
 The quantity $Z$ verifies the following transport equation\footnote{Note the absence of
 $\chib$ and $\omb$.}, written schematically,
 $$
 \nab_4 Z=\nab(\eta+\etab)\c O+(\etab-\eta)\c\nab O + \om Z+(\si+\rho)\c O + (\eta-\etab)\c(\eta+\etab)\c O
 $$
 Let $H_{ab}=\nab_a O_b$ denote the non-symmetrized derivative of $O$. Then, 
 $$
 \nab_4 H=\chi\c H +\b\c O+\nab\chi\c O +\chi\c\etab\c O
 $$
 We now rewrite these equations schematically in the form
 \begin{equation}
 \label{eq:mainHZ}
 \begin{split}
  &\nab_4 Z=\nab\psi_{34}\c O+\psi_{34}\c H + (\chi+\om) Z+\Psi_g\c O + \psi_{34}\c\psi_{34}\c O,\\
  &\nab_4 H=\psi \c H + (\Th_4+\nab\psi_4)\c O + \psi\c\psi_{34}\c O.
 \end{split}
 \end{equation}
 Here $\psi_{34}\in\{\eta, \etab\}$,  $\Psi_g\in\{\rho, \si\}$. In what follows 
 $\psi_34$   will  be treated  either as a  $\psi_3$ or a  $\psi_4$ quantity,   depending on the situation. The quantities, $H$ and $Z$ can be assigned signature and scaling, (consistent with those for the Ricci coefficients and curvature components) according to.  
 \bea
 sgn(H)-\frac 12 =sc(H)=0,\qquad sgn(Z)-\frac 12=sc(Z)=-\frac 12.
 \eea
 In view of equations \eqref{eq:mainHZ} we derive, by integration,
 \beaa
 \|Z\|_{\Lsc^\infty} &\les&\|\nab\psi_4\|_{\Trsc}+\|\Psi_g\|_{\Trsc}+\de^{\frac 12} \|\psi\|_{\Lsc^\infty}(\|\psi\|_{\Lsc^\infty}+\|H\|_{\Lsc^\infty}+\|Z\|_{\Lsc^\infty})\\
 \eeaa
 Thus, according to the trace  estimates of proposition \ref{prop:final.traceestim,nabeta.etab}
 for $\psi_4\in\{\eta, \etab\}$ and proposition \ref{prop:traceCurv} for  $\Psi_g$ we derive,
  \beaa
 \|Z\|_{\Lsc^\infty} &\les&C+\de^{\frac 12} C(\|H\|_{\Lsc^\infty}+\|Z\|_{\Lsc^\infty})
\eeaa
Similarly,
\beaa
 \|H\|_{\Lsc^\infty} &\les&\|\nab\psi_4\|_{\Trsc}+\|\Th_4\|_{\Lsc^\infty}
 +\de^{\frac 12} \|\psi\|^2_{\Lsc^\infty}(\|\psi\|_{\Lsc^\infty}+\|H\|_{\Lsc^\infty})\\&\les& C+\de^{\frac 12} C(C+\|H\|_{\Lsc^\infty}),
\eeaa
Therefore we  have proved\footnote{Note the triviality of the  data for $Z$ on $\Hb_0$. Otherwise the term
 $\chib\c O$ in the definition of $Z$ might have caused an $\Lsc^\infty$ anomaly. The data for $H$ however
 is not trivial. Initially $\|H\|_{L^\infty} \sim 1$, which means that while it is anomalous in $\Lsc^2(S)$ it is not in 
 $\Lsc^\infty$.}  the following.
 \begin{proposition}
 The quantities $Z$ and $H$ verify the estimates
 $$
 \|H\|_{\Lsc^\infty} +\|Z\|_{\Lsc^\infty} \les C,
 $$
 with a constant $C=C(\II^{(0)} ,\RR_{[1]}, \RRb_{[1]})$.
 \end{proposition}
 We add a small remark concerning the symmetrized $\nab$ derivatives of $O$.
\begin{proposition}
Let $H'_{ab}:=\nab_a O_b+\nab_a O_b=H_{ab} +H_{ba}$. Then in addition to 
all the estimates for $H$, $H'$ also enjoys a non-anomalous $\Lsc^2(S)$ estimate
$$
\|H'\|_{\Lsc^2(S)}\les C.
$$
Similarly,
$$
\|Z\|_{\Lsc^2(S)}\les C.
$$
\end{proposition}
The result follows easily from the transport equation for $H'$, which is virtually the same as
for $H$, and crucially, triviality of the initial data for $H^s$. The claim for $Z$ follows from the
same considerations.

\subsection{$\Lsc^2(S)$ estimates for $\nab H, \nab Z$} 
We prove below the following,
\begin{proposition}The following estimates hold true with $ C=C(\II^{(0)},\RR, \RRb)$,
\begin{align*}
&\|\nab H\|_{\Lsc^2(S)} +\|\nab Z\|_{\Lsc^2(S)}\les C,\\
&\|\nab_4\nab H\|_{\Lsc^2(H)} +\|\nab_4\nab Z\|_{\Lsc^2(H)}\les C
\end{align*}
\end{proposition} 
\begin{proof}
We  first  commute the transport equations for $H$ and $Z$ with $\nab$.
 \begin{align*}
 \nab_4(\nab H)&=\psi\c\nab H+\nab\psi\c H+(\nab\Th_4+\nab^2\psi_4)\c O + (\Th_4+\Psi_g)\c H\\& +
 \psi\c\nab\psi\c O + \psi_{34}\c \nab_4 H+\psi\c\psi_g\c H,\\
 \nab_4(\nab Z)&=\nab^2\psi_{34}\c O+(\nab\psi+\Psi_g)\c (H+Z)+\psi\c (\nab H+\nab Z) +\nab\Psi_g\c O \\&+ \psi\c\nab\psi\c O
+\psi\c\psi\c (H+Z)+\psi_{34} \nab_4 Z
 \end{align*}
 The term $\nab\Psi_g$ is in fact $\nab(\si+\rho)$.
 The estimate for $\nab H$ follows immediately from the following:
 \begin{align*}
&\|\psi\c\nab H\|_{\Lsc^2(H)}\les \de^{\frac 12} \|\psi\|_{\Lsc^\infty} \|\nab H\|_{\Lsc^2(H)}\les \de^{\frac 12} C  \|\nab H\|_{\Lsc^2(H)}\\
&\|\nab\psi\c H\|_{\Lsc^2(H)}\les \de^{\frac 12} \|H\|_{\Lsc^\infty} \|\nab \psi\|_{\Lsc^2(H)}\les \de^{\frac 12} C\\
&\|\nab\Th\c O\|_{\Lsc^2(H)}\les  \|\nab \Th\|_{\Lsc^2(H)}\les C\\
&\|\nab^2\psi\c O\|_{\Lsc^2(H)}\les  \|\nab^2 \psi\|_{\Lsc^2(H)}\les C\\
&\|(\Th+\Psi_g)\c H\|_{\Lsc^2(H)}\les \de^{\frac 12} \|H\|_{\Lsc^\infty} (\|\Th\|_{\Lsc^2(H)}+\|\Psi_g\|_{\Lsc^2(H)})\les \de^{\frac 12} C\\
&\|\psi\c\nab\psi\c O\|_{\Lsc^2(H)}\les \de^{\frac 12} \|\psi\|_{\Lsc^\infty} \|\nab \psi\|_{\Lsc^2(H)}\les \de^{\frac 12} C\\
&\|\psi\c\psi_g\c H\|_{\Lsc^2(H)}\les \de \|\psi\|_{\Lsc^\infty} \|H\|_{\Lsc^\infty} \|\psi_g\|_{\Lsc^2(H)}\les \de C,\\
& \|\psi\c\nab_4 H\|_{\Lsc^2(H)}\les \de^{\frac 12} \|\psi\|_{\Lsc^\infty} \|\nab_4 H\|_{\Lsc^2(H)}\les \de^{\frac 12} C.
\end{align*}
The estimates for $\nab Z$ are proved in exactly the same manner.
\end{proof}
\subsection{$\Lsc^4(S)$ estimates for $\nab H, \nab Z$} 
The results of the previous  proposition can be strengthened to give the following,
\begin{proposition} The following hold true,
\begin{align*}
&\|\nab H\|_{\Lsc^4(S)} +\|\nab Z\|_{\Lsc^4(S)}\les C
\end{align*}
\end{proposition} 
\begin{proof}
The arguments can be followed almost verbatim, as in the last proposition,  with the exception of the analysis of the two terms:
$$
\nab^2\psi_{34}\c O, \quad \nab\Psi_g\c O=\nab(\si+\rho)\c O
$$
We recall that $\psi_{34}=\{\eta, \etab\} $ and  according to Proposition \ref{le:estimates.phit-phitb} we can write,
$$
\nab \psi_{43}=\nab_4 \phi
$$
with $\phi$ satisfying the estimates
\begin{align*}
&\|\nab^2 \phi\|_{\Lsc^2(H)}+\|\nab^2 \phi\|_{\Lsc^2(\Hb)}+
\|\nab \phi\|_{\Lsc^2(S)}+\|\phi\|_{\Lsc^2(S)}
\le C,\\
&\|\nab_4\phi\|_{\Lsc^2(S)}+
\|\nab_4\nab \phi_4\|_{\Lsc^2(S)}+\|\phi\|_{\Lsc^\infty} +\|\nab\phi\|_{\Lsc^4(S)}\les C
\end{align*}
We now  the write
\begin{align*}
\nab^2\psi_{43}\c O&=\nab_4 (\nab\phi\c O) - \nab\phi\c \chi\c O-[\nab_4,\nab] \phi\c O\\ &=
\nab_4 (\nab\phi\c O) +\chi\c \nab\phi\c O + \Psi_g\c\phi\c O+\psi\c \nab\psi\c O+
\psi\c\psi\c\phi\c O
\end{align*}
We estimate
\begin{align*}
&\de^{-1}\int_0^{\ub} \|\nab\phi\c\chi\c O\|_{\Lsc^4(S_{u,\ub})} d\ub \les \de^{\frac 12} \sup_{\ub}
\|\nab\phi\|_{\Lsc^4(S_{u,\ub})} \|\chi\|_{\Lsc^\infty} \les\de^{\frac 12} C,\\
&\de^{-1}\int_0^{\ub} \|\Psi_g\c\phi\c O\|_{\Lsc^4(S_{u,\ub})} d\ub\les \de^{\frac 12} 
\left (\|\nab\Psi_g\|^{\frac 12}_{\Lsc^2(H)} \|\Psi_g\|^{\frac 12}_{\Lsc^2(H)} +\de^{\frac 14} 
\|\Psi_g\|_{\Lsc^2(H)}\right)
\|\phi\|_{\Lsc^\infty} \les\de^{\frac 12} C,\\
&\de^{-1}\int_0^{\ub} \|\nab\psi\c\psi\c O\|_{\Lsc^4(S_{u,\ub})} d\ub \les \de^{\frac 12} \sup_{\ub}
\|\nab\psi\|_{\Lsc^4(S_{u,\ub})} \|\psi\|_{\Lsc^\infty} \les\de^{\frac 12} C,\\
&\de^{-1}\int _0^{\ub}\|\psi\c\psi\c \phi\c O\|_{\Lsc^4(S_{u,\ub})} d\ub \les \de \sup_{\ub}
\|\phi\|_{\Lsc^4(S_{u,\ub})} \|\psi\|^2_{\Lsc^\infty} \les\de C.
\end{align*}
On the other hand, the null structure equations give for $\ombtild=(\omb,\ombt)$
$$
\nab_4\ombtild=(\rho,\si) +\psi_g\c\psi_g. 
$$
As a result,
$$
\nab (\rho,\si)\c O=\nab_4(\nab{\ombtild} \c O)+\left(\psi\c\nab\psi+\chi\c \nab\ombtild+\Psi_g\c \ombtild +
\psi_g\c\Psi_g+\psi\c\psi_g\c (\ombtild+\psi_g)\right)\c O
$$
We can estimate
\beaa
\de^{-1}\int_0^{\ub} \|\nab\ombtild\c\chi\c O\|_{\Lsc^4(S_{u,\ub})} d\ub &\les &\de^{\frac 12} \sup_{\ub}
\|\nab\ombtild\|_{\Lsc^4(S_{u,\ub})} \|\chi\|_{\Lsc^\infty} \les\de^{\frac 12} C,\\
\de^{-1}\int_0^{\ub} \|\Psi_g\c\psi\c O\|_{\Lsc^4(S_{u,\ub})} d\ub &\les& \de^{\frac 12} 
\left (\|\nab\Psi_g\|^{\frac 12}_{\Lsc^2(H)} \|\Psi_g\|^{\frac 12}_{\Lsc^2(H)} +\de^{\frac 14} 
\|\Psi_g\|_{\Lsc^2(H)}\right)
\|\psi\|_{\Lsc^\infty} \les\de^{\frac 12} C,\\
\de^{-1}\int_0^{\ub} \|\nab\psi\c\psi\c O\|_{\Lsc^4(S_{u,\ub})} d\ub &\les& \de^{\frac 12} \sup_{\ub}
\|\nab\psi\|_{\Lsc^4(S_{u,\ub})} \|\psi\|_{\Lsc^\infty} \les\de^{\frac 12} C,\\
\de^{-1}\int_0^{\ub} \|\psi\c\psi\c (\ombtild+\psi_g)\c O\|_{\Lsc^4(S_{u,\ub})} d\ub& \les& \de \sup_{\ub'\le \ub}
\|\ombtild+\psi_g\|_{\Lsc^4(S_{u,\ub'})} \|\psi\|^2_{\Lsc^\infty} \les\de C.
\eeaa
These allow us to conclude that,
\beaa
\de^{-1} \int_0^{\ub} \|\nab_4\left [(\nab H,\nab Z)-\nab\phi\c O-\nab \ombtild\c O\right]\|_{\Lsc^4(S_{u,\ub'})} d\ub''   & \les &
\de^{\frac 12} \sup_{\ub'\le \ub} \|(\nab H,\nab Z)\|_{\Lsc^4(S_{u,\ub'})} +\de^{\frac 12} C.
\eeaa
Making use of the  $\Lsc^4(S)$ bounds on both $\nab\phi$ and $\nab<\omb> \, $ we finally  obtain
 the    estimate    $
\de^{-1} \int_0^u \|\nab_4(\nab H,\nab Z)\|_{\Lsc^4(S_{u,\ub'}) } du' \les \de^{\frac 12} \sup_{\ub'\le \ub } \|(\nab H,\nab Z)\|_{\Lsc^4(S_{u,\ub'})} +C\de^{1/2},$
from which  the conclusion of the proposition easily  follows.
\end{proof}
\subsection{Estimates for $\nab_3 Z$.}
 We now examine the equation for $\nab_3 Z$.
 \begin{align*}
 \nab_4(\nab_3 Z)&=\nab_3\nab\psi_{34}+\nab\psi_{34}\c Z+\nab\psi_{34}\c\chib +\nab_3\psi_{34}\c H 
 +\psi_{34}\c\nab_3 H\\ &
 + (\nab_3\chi+\nab_3\om)\c Z+\om\c\nab_3 Z\\&+\nab_3\Psi_g\c O +(\rho+\si)\c Z+\Psi_g\c\chib + \nab_3\psi_{34}\c\psi_{34}
 +\psi_{34}\c\psi_{34}\c Z+\psi_{34}\c\psi_{34}\c \chib,
 \end{align*}
 To estimate the right hand side of this equation we will need to use the first and second derivative estimates
for $\psi$ of Propositions \ref{prop:1der},\ref{prop:1der-eta},\ref{prop:2der} and \ref{prop:2der-eta}, keeping 
in mind possible anomalies of $\chib$, $\nab_4\chih, \nab_3\chih, \nab_3\chibh$, the relationship 
$$
\nab_3(\rho+\si)=\nab\bb+(\trchb_0+\psi)\c\Psi,
$$
given by the null Bianchi identities and the $\Lsc^2(S)$ curvature estimate\footnote{Note that $\Psi$ in the 
nonlinear term may
contain an $\alphab$ component but not the anomalous $\a$ term.} $\|\Psi\|_{\Lsc^2(S)}\le C$ of Propositions \ref{prop:curv.estim} 
and \ref{prop:aproh}. Thus,
\begin{align*}
&\|\nab_3\nab\psi_{34}\|_{\Lsc^2(H)}\les C,\\
&\|\nab\psi\c Z\|_{\Lsc^2(H)}\les \de^{\frac 12} \|Z\|_{\Lsc^\infty} \|\nab\psi\|_{\Lsc^2(H)}\les \de^{\frac 12} C,\\
&\|\nab\psi\c \chib\|_{\Lsc^2(H)}\les \de^{\frac 12} \|\chib\|_{\Lsc^\infty} \|\nab\psi\|_{\Lsc^2(H)}\les C,\\
&\|\nab_3\psi_{34}\c H\|_{\Lsc^2(H)}\les \de^{\frac 12} \|H\|_{\Lsc^\infty} \|\nab_3\psi_{34}\|_{\Lsc^2(H)}\les 
\de^{\frac 12} C,\\
&\|\psi\c \nab_3 H\|_{\Lsc^2(H)}\les \de^{\frac 12} \|\psi\|_{\Lsc^\infty} \|\nab_3 H\|_{\Lsc^2(H)}\les \de^{\frac 12} C
\|\nab_3 H\|_{\Lsc^2(H)},\\
&\|\nab_3\om\c Z\|_{\Lsc^2(H)}\les \de^{\frac 12} \|Z\|_{\Lsc^\infty} \|\nab_3\om\|_{\Lsc^2(H)}
\les 
\de^{\frac 12} C,\\
&\|\nab_3\chi\c Z\|_{\Lsc^2(H)}\les \de^{\frac 12} \|Z\|_{\Lsc^\infty} \|\nab_3\chi\|_{\Lsc^2(H)}
\les  C,\\
&\|\om\c \nab_3 Z\|_{\Lsc^2(H)}\les \de^{\frac 12} \|\om\|_{\Lsc^\infty} \|\nab_3 Z\|_{\Lsc^2(H)}\les 
\de^{\frac 12} C\|\nab_3 Z\|_{\Lsc^2(H)},\\
&\|\nab_3 (\rho+\si)\|_{\Lsc^2(H)}\les \|\nab\bb\|_{\Lsc^2(H)}+\|(\trchb_0+\psi)\c\Psi\|_{\Lsc^2(H)}\les
\RR_1+C,\\
&\|\Psi_g\c Z\|_{\Lsc^2(H)}\les \de^{\frac 12} \|Z\|_{\Lsc^\infty} \|\Psi_g\|_{\Lsc^2(H)}\les \de^{\frac 12} C,\\
&\|\Psi_g\c \chib\|_{\Lsc^2(H)}\les \de^{\frac 12} \|\chib\|_{\Lsc^\infty} \|\Psi_g\|_{\Lsc^2(H)}\les C,\\
&\|\nab_3\psi_{34}\c \psi\|_{\Lsc^2(H)}\les \de^{\frac 12} \|\psi\|_{\Lsc^\infty} \|\nab_3\psi_{34}\|_{\Lsc^2(H)}\les 
\de^{\frac 12} C,\\
&\|\psi\c \psi\c Z\|_{\Lsc^2(H)}\les \de \|Z\|_{\Lsc^\infty} \|\psi\|_{\Lsc^\infty} \|\psi\|_{\Lsc^2(H)}\les 
\de^{\frac 12} C,\\
&\|\psi\c \psi_{34}\c \chib\|_{\Lsc^2(H)}\les \de \|\chib\|_{\Lsc^\infty} \|\psi\|_{\Lsc^\infty} \|\psi_g\|_{\Lsc^2(H)}\les 
\de^{\frac 12} C
\end{align*}
\subsection{Estimates for   $ \|\nab_3 H\|_{\Lsc^2(H)}$.}
 The only quantity still requiring an estimate is $ \|\nab_3 H\|_{\Lsc^2(H)}$. We use the relation\footnote{Note a
 crucial cancellation of an anomalous term $\chib\c H$.}
 $$
 \nab_3 H=\nab_3\nab O=\nab \nab_3 O+[\nab,\nab_3]O = \nab Z + \nab \chib \c O +
 \bb\c O+\psi_{34}\c Z +\psi_{34}\c\chib\c O
 $$
 Therefore,
 \begin{align*}
  \|\nab_3 H\|_{\Lsc^2(S)}&\les  \|\nab Z\|_{\Lsc^2(S)}+ \|\nab \chib\|_{\Lsc^2(S)}+\|\Psi_g\|_{\Lsc^2(S)}+
  \de^{\frac 12} \|\psi_{34}\|_{\Lsc^2(S)} \|Z\|_{\Lsc^\infty}\\ &+\de^{\frac 12}\|\chib\|_{\Lsc^\infty} \|\psi_{34}\|_{\Lsc^2(S)}
  \les   \|\nab Z\|_{\Lsc^2(S)}+C
 \end{align*}
 This immediately implies the bounds
 $$
 \|\nab_3 H\|_{\Lsc^2(S)} + \|\nab_3 Z\|_{\Lsc^2(S)} + \|\nab_4\nab_3 Z\|_{\Lsc^2(H)}\les C.
 $$
 A similar argument allows us to immediately strengthen the $\|\nab_3 H\|_{\Lsc^2(S)}$ estimate  (unlike the one for   $\nab_3 Z$) to the  $\Lsc^4(S)$ norm
 $$
  \|\nab_3 H\|_{\Lsc^4(S)}\le C
 $$
 
  Furthermore,
\begin{align*}
\nab_4 \nab_3 H &= \nab_4\nab Z + \nab_4\nab \chib \c O +\nab\chib\c \chi\c O+
\nab_4 \bb\c O+\Psi_g\c \chi\c O+\nab_4 \psi_{34}\c Z \\ &+\psi_{34}\c\nab_4 Z+\nab_4\psi_{34}\c\chib\c O+
\psi_{34}\c\nab_4\chib\c O +\psi_{34}\c\chib\c\chi\c O
 \end{align*}
 We once again remind the reader of the possible anomalies for $\chih, \chibh$ in $\Lsc^2(S)$,
 double anomaly for $\trchb$ in $\Lsc^2(S)$ and a simple anomaly in $\Lsc^\infty$, anomalies
 for $\nab_4\chih$ and $\nab_3\chibh$.
 We estimate
 \begin{align*}
 &\|\nab_4\nab Z\|_{\Lsc^2(H)}\les C,\\
 &\|\nab_4\nab\chib\|_{\Lsc^2(H)}\les C,\\
 &\|\nab\chib\c\chi\|_{\Lsc^2(H)}\les \de^{\frac 12} \|\chi\|_{\Lsc^\infty} 
 \|\nab\chib\|_{\Lsc^2(H)}\les \de^{\frac 12} C,\\
 &\|\nab_4\bb\|_{\Lsc^2(H)}\les \|\nab\Psi_g\|_{\Lsc^2(H)}+\|\psi\c\Psi_g\|_{\Lsc^2(H)}\les 
 \RR_1 +\de^{\frac 12} C,\\
 &\|\Psi_g\c\chi\|_{\Lsc^2(H)}\les \de^{\frac 12} \|\chi\|_{\Lsc^\infty} \|\Psi_g\|_{\sc^2(H)}\les \de^{\frac 12}C,\\
 &\|\nab_4\psi_{34}\c\chib\|_{\Lsc^2(H)}\les \de^{\frac 12} \|\chib\|_{\Lsc^\infty} \|\nab_4\psi_{34}\|_{\Lsc^2(H)}\les C,\\
 &\|\psi_{34}\c\nab_4\chib\|_{\Lsc^2(H)}\les \de^{\frac 12} \|\psi\|_{\Lsc^\infty} \|\nab_4\chib\|_{\Lsc^2(H)}\les 
 \de^{\frac 12}C,\\
 &\|\psi_{34}\c\chib\c\chi\|_{\Lsc^2(H)}\les \de \|\chi\|_{\Lsc^\infty} \|\chib\|_{\Lsc^\infty} \|\psi_{34}\|_{\Lsc^2(H)}\les 
 \de^{\frac 12}C.
 \end{align*}
 As a result we now established the following 
 \begin{proposition}
 There exists a constant $C=C(\OO_{[2]},\OO_\infty,\RR_{[1]},\RRb_{[1]})$ such that 
 $$
 \|\nab_3 H\|_{\Lsc^2(S)} + \|\nab_3 Z\|_{\Lsc^2(S)} + \|\nab_4\nab_3 Z\|_{\Lsc^2(H)}+
 \|\nab_4\nab_3 H\|_{\Lsc^2(H)}\les C.
 $$
 \end{proposition}

\subsection{Derivatives of the deformation tensor}
We now compute the derivatives of the deformation tensor $D\pi$. 
\begin{align*}
&D_4 \pi_{44}=0,\quad D_4\pi_{34}=-2\nab_4(\eta+\etab)\c O -2(\eta+\etab)\c\chi\c O,
\\ &D_4 \pi_{33}=\frac 14 \etab\c Z,
\quad D_4 \pi_{3a}=\frac 12 \nab_4 Z+\etab\c(\eta+\etab)-\frac 12 \etab\c H^s,\quad D_4 \pi_{4a}=0,\quad
D_4 \pi_{ab}=\nab_4 H^s,\\
&D_3 \pi_{44}=0,\quad D_3\pi_{34}=-2\nab_3(\eta+\etab)\c O -2(\eta+\etab)\c (Z-\chib\c O)-\frac 14\eta\c Z,
\\ &D_3 \pi_{33}=0,
\quad D_3 \pi_{3a}=\frac 12\nab_3 Z,\quad D_3 \pi_{4a}=-\eta\c(\eta+\etab)-\frac 12 \eta\c H^s,\quad
D_3 \pi_{ab}=\nab_3 H^s+\frac 14 \eta\c Z,\\
&D_c \pi_{44}=0,\quad D_c\pi_{34}=-2\nab(\eta+\etab)\c O -2(\eta+\etab)\c H^s-\frac 12\chi\c Z,
\\ &D_c \pi_{33}=-\frac 12 \chib\c Z,
\quad D_c \pi_{3a}=\frac 12 \nab Z-\chib\c H^s -2 \chib (\eta+\etab)\c O,\\ &D_c\pi_{4a}=-\chi\c H^s-2\chi(\eta+\etab)\c O,\quad
D_c \pi_{ab}=\nab H^s-\chi\c Z,
\end{align*}
Based on the results of the previous section we then easily deduce the following result
\begin{proposition}
There exists a constant $C=C(\OO_{[2]},\OO_\infty,\RR_{[1]},\RRb_{[1]})$ such that 
$$
\|D\pi\|_{\Lsc^2(S)}\les C
$$
\end{proposition}
The only potentially problematic term is $\chib\c H^s$, which can be estimated as follows:
$$
\|\chib\c H^s\|_{\Lsc^2(S)}\les \de^{\frac 12} \|\chib\|_{\Lsc^\infty} \|H^s\|_{\Lsc^2}\les C. 
$$
It is precisely this term that requires a non-anomalous $\Lsc^2(S)$ estimate for $H^s$, 
which incidentally does not hold for the non-symmetrized derivative $H$.

\subsection{Theorem B}
We are now ready to state  the main result of this section, mentioned in the introduction.
\begin{theorem}[Theorem B] 
\label{thmn:thmB}The  deformation tensors  $\piO$ of the angular momentum operators $O$ verify the following estimates, with a constant $C=C(\II^{(0)}, \RR, \RRb)$,
\bea
\|\piO\|_{\Lsc^4(S)} + \|\piO\|_{\Lsc^\infty(S)}\les C
\eea
Also all null components  of the derivatives  $D\piO$, with the exception of  $(D_3\piO)_{3a}$,  verify the estimates, 
\bea
\|D\piO\|_{\Lsc^4(S)  } \les C
\eea
Moreover,
\bea
\|(D_3\piO)_{3a}-\nab_3 Z\|_{L^4(S)} +  \|\sup_{\ub}\nab_3 Z\|_{L^2(S)}  &\les & C
\eea
\end{theorem}

\section{Curvature estimates I.}
\label{sect:CurvI}
In this section, in all the remaining sections of the paper  $C$ denotes a constant which depends on the initial data $\II_0$
all the curvature norms $\RR, \RRb$, including $\|\nab_4\a\|_{\Lsc^2(H^{(0,\ub)}_u)}$
and $\|\nab_3\a\|_{\Lsc^2(\Hb^{(0,u)}_{\ub})}$. Using the results of the previous sections we assume that  the norms $\OO$ of the Ricci 
coefficients are bounded by $C$.
\subsection{Preliminaries}
Let $W$ be a Weyl tensorfield, with $\dual W$ its Hodge dual  verifying the Bianchi equations with sources
\bea
\Div W=J,\qquad \Div \dual W=J^*\label{eq:Bianchi.gen}
\eea
where $J, \dual J $ are   Weyl currents, i.e. 
\beaa
J_{[\a\b\ga]}&=&0,\qquad
J_{\a\b\ga}=-J_{\a\ga\de},\qquad 
g^{\b\ga}J_{\b\ga\de}=0.
\eeaa
and $J^*_{\a\b\ga}=\frac 1 2 J_{\a\mu\nu}\in^{\mu\nu}_{\,\,\,\, \b\ga}$ the
 right Hodge dual of $J$.  Following the definitions of \cite{Chr-Kl} we  let $Q[W]$ be the Bel-Robinson tensor of $W$.  As proved there we have, 
 \begin{proposition}
 \label{prop:mainenergy}
 Assume $W$ verifies \eqref{eq:Bianchi.gen}. Given vectorfields
  $X,Y, Z$  and $P[W]=P[W](X, Y, Z)$ defined by
$
P[W]^\a:=Q[W]_{\a\b\ga\de}X^\b Y^\ga Z^\de
$
we have,
\bea
\Div(P[W])  &=&\Div Q[W](X, Y, Z)+\frac 1 2 (Q[W]\c\pi)(X,Y,Z)
\eea
where,
\beaa
 (Q[W]\c\pi)(X,Y,Z):&=& Q[W](\piX, Y, Z)+Q[W](\piY, X, Z)\\
 &+&Q[W](\piZ, X, Y)
\eeaa
Thus, integrating on our fundamental domain $\DD=\DD(u,\ub)$,
\beaa
&&\int_{H_u} Q[W] (L, X, Y, Z)+\int_{\Hb_{\ub}}   Q[W]  X, Y, Z,(\Lb)\\
\qquad &&=\int_{H_0}Q[W] (L, X, Y, Z)+\int_{\Hb_0}] Q[W] ( X, Y, Z,\Lb)\nn\\
\qquad && +\int\int_{\DD(u, \ub)}\Div Q[W](X, Y, Z)
+\frac 1 2 \int\int_{\DD(u, \ub)}Q[W]\c \pi(X,Y, Z)
\eeaa
 \end{proposition}
  In the particular case when $W$ is the curvature tensor $R$ (and thus $J=J^*=0$), recalling
 that the initial data on $\HHb_0$ vanishes, we have
 \begin{corollary} 
 \label{cor;0energy}
 The following identity holds on our fundamental domain $\DD(u,\ub)$,
 \beaa
 \int_{H_u} Q[R] (L, X, Y, Z)+\int_{\Hb_{\ub}} Q[R] ( X, Y, Z, \Lb)
&=&\int_{H_0}Q[R] (L, X, Y, Z)\\
&+&
\frac 1 2 \int\int_{\DD(u, \ub)}Q[R]\c \pi(X,Y, Z)
\eeaa
\end{corollary}
On the other hand, given a vectorfield $O$, we have
\bea
\Div(\Lieh_O R)&=&J(O, R),\qquad \Div(\dual \Lieh_O R)=  J^*(O, R).\label{eq:Bianchi.LieO}
\eea
 where  $J(O,R)$ is a Weyl current   (calculated below in lemma  \ref{le:calcJ(O,R)})   and $\Lieh_O R$ denotes the  modified Lie derivative of the curvature tensor $R$, i.e. 
 (following \cite{Chr-Kl}),
$\Lieh_O R= \Lie_O R -\frac 1 8 \tr  \piO R-    \frac 1 2 \pihO\c R  $ 
and,
\beaa
(\pihO\c R)_{\a\b\ga\de} &=&\pihO^\mu_\a W_{\mu\b\ga\de}+   \pihO^\mu_\b W_{\a\mu\ga\de} +   \pihO^\mu_\ga W_{\a\b \mu\de}    + \pihO^\mu_\de W_{\a\b \ga \mu}   \eeaa
with $\pihO$ is the traceless  part of  $\piO$,  i.e. 
$\piO=\pihO+\frac 1 4 \tr\piO g$. 
Observe  that  $\Lieh_O R$ is also a Weyl field and that the  modified Lie derivative commutes with the Hodge dual, i.e.,$\Lieh_O(\dual  R\,)=\dual \Lieh_O  R$. The following corollary of proposition
\ref{prop:mainenergy} and   proposition 7.1.1 in \cite{Chr-Kl}.
\begin{corollary}
\label{corr:mainenergy1} Let $O$ be a vectorfield  defined  in   our fundamental domain $\DD(u,\ub)$,  tangent to $\HHb_0$. Then,  with $H_u=H_u([0,\ub])$,
 \beaa
&& \int_{H_u} Q[\Lieh_O R] (L, X, Y, Z)+\int_{\Hb_{\ub} }Q[\Lieh_OR](  X, Y, Z,\Lb)=
\int_{H_0}Q[\Lieh_OR] (L, X, Y, Z)\\
&&+\frac 1 2 \int\int_{\DD(u, \ub)}Q[\Lieh_O R]\c \pih(X,Y, Z)+\int\int_{\DD(u, \ub)}D(R, O)(X,Y, Z)
\eeaa
where,
$ D(O, R):=\Div Q[\Lieh_OR]
$
is given by the formula,
\beaa
 D(O, R)_{\b\ga\de}&=&(\Lieh_O R)_\b\,^\mu\,_\de\,^\nu J(O, R)_{\mu\ga\nu}+
 (\Lieh_O R)_\b\,^\mu\,_\ga\,^\nu J(O, R)_{\mu\de\nu}\\
 &+&\dual ( \Lieh_O R)_\b\,^\mu\,_\ga\,^\nu\,  J^*(O, R)_{\mu\de\nu}+
 \dual(\Lieh_O R)_\b\,^\mu\,_\ga\,^\nu\,  J^*(O, R)_{\mu\de\nu}
\eeaa
\end{corollary}
The Weyl current $J(O,R)$ is given by the following  commutation  formula, 
see proposition  7.1.2 and   in {\cite{Chr-Kl},
\begin{lemma}
\label{le:calcJ(O,R)}
We have,
\bea
\Div (\Lieh_O R)&=&  J(O; R):=J^1(O; R)+J^2(O; R)+J^3(O; R)
\eea
\beaa
J^1(O,R)_{\b\ga\de}&=&\frac 1 2 \pihO^{\mu\nu} D_\nu R_{\mu\b\ga\de}\\
J^2(O,R)_{\b\ga\de}&=&\frac 1 2 \pO_\la R^\la_{\, \, \b\ga\de}\\
J^3(O,R)_{\b\ga\de}&=&\frac 1 2 \big(\qO_{\a\b\la} R^{\a\la}\,_{\ga\de}+\qO_{\a\ga\la} R^{\a}\,_{\b}\,^\la\,_\de+\qO_{\a\de\la} R^\a\,_{\b\ga}\,^\la\big)
\eeaa
where,
$
\pO_\ga= D^\a (\pihO_{\a\ga}).\quad 
\qO=D_\b\, \pihO_{\ga\a}-D_\ga\, \pihO_{\b\a}-\frac 1 3 (\pO_\ga\,  g_{\a\b}-\pO_\b \,g_{\a\ga})$
\end{lemma}

In the remaining part of this section we should establish estimates
for  the norms  $\RR_0$ and $\RRb_0$. We start with $\a$.
\subsection{Estimate for $\a$}
We apply corollary \ref{cor;0energy}  to  $X=Y=Z=e_4$ to derive,
 \bea
 \int_{H^{(0,\ub)}_u}|\a |^2 +\int_{H^{(0,u)}_{\ub}}|\b |^2 &\les&  \int_{H^{(0,\ub)}_0}|\a|^2+\int_{\DD(u,\ub)}
 (Q[   R]\c\piL)(e_4,  e_4, e_4) \label{eq:Lie0aa}
 \eea
Based on  conservation of signature  we write schematically,
\bea
(Q[R] \c\piL)(e_4,  e_4, e_4)&=&\sum_{s_1+s_2+s_3=4}\phi^{(s_1)} \c \Psi^{(s_2)}\c \Psi^{(s_3)}\label{eq:Lie0a.1}\
\eea
with  Ricci coefficients $\phi  \in\{\chi, \om, \eta, \etab,\omb\}$, null curvature components $\Psi$ and labels $s_1, s_2, s_3 $ denoting the signature of the 
corresponding component. In scale invariant norms we have,
\beaa
\|\a\|_{\Lsc^2(H^{(0,\ub)}_u)}^2 +\|\b\|_{\Lsc^2(H^{(0,u)}_{\ub})}^2 &\les& \|\a\|_{\Lsc^2(H^{(0,\ub)}_0)}^2+I
\eeaa
with,
\beaa
I&=&  \de^{1/2} \sum_{s_1+s_2+s_3=4}   \|\phi^{(s_1)}\|_{\Lsc^\infty}   \int_0^u \| \Psi^{(s_2)}\|_{\Lsc^2(H^{(0,\ub)}_{u'})}\c  \| \Psi^{(s_3)}\|_{\Lsc^2(H^{(0,\ub)}_{u'})} du'
\eeaa
By far the worst term occur when $s_2=s_3=2$ and $s_1=0$.  Observe also that, since the signature
of a Ricci coefficient $\phi^{(s_1)}$ may not exceed $s_1=1$, neither 
$s_2$ or $s_3$ can be zero, i.e. $\aa$ cannot occur  among the curvature terms on the right.
Using our estimates,
$   \|\phi^{(s_1)}\|_{\Lsc^\infty} \les C$, with $C=C(\II^0, \RR, \RRb)$ we deduce,
\beaa
\|\a\|_{\Lsc^2(H^{(0,\ub)}_u)}^2 +\|\b\|_{\Lsc^2(H^{(0,u)}_{\ub})}^2 &\les& \|\a\|_{\Lsc^2(H^{(0,\ub)}_0)}^2     +      C\de^{1/2} \|\a\|_{\Lsc^2(H^{(0,\ub)}_u)}^2\\
&+&
C\RR_0\de^{1/2} \|\a\|_{\Lsc^2(H^{(0,\ub)}_u)}+C\de^{1/2}\RR_0^2
\eeaa
Therefore, recalling the anomalous character of    $\RR_0[\a]$, $\RRb_0[\b]$ we deduce,
\bea
\RR_0[\a]+\RRb_0[\b]&\les& \II^0+ C\de^{3/4} \RR_0
\eea
\subsection{Remaining estimates} We follow the procedure outlined 
in the introduction. 
 Define the energy quantities,
\bea
\QQ_0(u,\ub)&=&\de^2 \int_{\HH_u^{(0,\ub)  }} Q[ R](e_4,e_4,e_4,e_4)+ \int_{\HH_u^{(0,\ub)}} Q[ R](e_3 ,e_4,e_4, e_4)\\
&+&\de^{-1}\int_{\HH_u^{(0,\ub)}} Q[ R](e_3 ,e_3,e_4, e_4)+\de^{-2}\int_{\HH_u^{(0,\ub)}  } Q[ R](e_3 ,e_3 ,e_3,e_4)\nn
\eea
\bea
\QQb_0(u,\ub)&=&\de^2\int_{\HHb_{\ub}^{  (0,u)}}
     Q[R](e_4,e_4,e_4, e_3)+\int_{\HHb_{\ub}^{  (0,u)}} Q[ R](e_4,e_4,e_3,e_3 )\\
&+&\de^{-1}\int_{\HHb_{\ub}^{  (0,u)}} Q[ R](e_4,e_3,e_3,e_3) +\de^{-2}\int_{\HHb_{\ub}  ^{(0,u)}} Q[ R](e_3,e_3,e_3,e_3)\nn
\eea
According to   corollary  \eqref{cor;0energy},   for all possible choices
of the vectorfields  $X, Y, Z$ in  the set   $\{e_4,e_3\}$ we are led to the identity,
\bea
\QQ_0(u,\ub) +\QQb_0(u,\ub) &\approx&\QQ_0(0,\ub)+\EE_0(u,\ub)\label{eq:intr.mainid}
\eea
where, 
\beaa
\EE_0(u,\ub) &=&
\de^2 \int\int_{\DD(u,\ub)}Q[ R](\piL, e_4, e_4)\\
&+& \int\int_{\DD(u,\ub)}Q[ R](\piL, e_3,e_4) +\quad  \int\int_{\DD(u,\ub)}Q[ R](  \piLb, e_4, e_4)\\
&+&\de^{-1}\int\int_{\DD(u,\ub)}Q[ R](\piL, e_3 , e_3 ) +\de^{-1}\int\int_{\DD(u,\ub)}Q[ R](  \piLb, e_4, e_3) \\
&+&\de^{-2} \int\int_{\DD(u,\ub)}Q[ R] ( \piLb, e_3, e_3 )
\eeaa
with $\piL, \piLb$ the deformation tensors of $e_4, e_3$.
Every term appearing in the above  integrands  
linear in $\piL$ or $\piLb$ and quadratic with respect to $R$. 
Also all  components of $\piL$ can be expressed in terms of our 
Ricci coefficients $\chi,\om, \eta, \etab,\,\omb$.   In fact one can easily check the following,
    $
\piL_{44}=\piL_{4a}=0, 
\piL_{34}=g(D_3 e_4, e_4)+g(D_4 e_4, e_3)=4\om$, 
$\piL_{33}=2g(D_3 e_4,  e_3)=-8\omb, 
\piL_{ab}=2\chi_{ab}$, 
$\piL_{a3}=g(D_ae_4, e_3)+g(D_3 e_4, e_a)=2\ze_a+2\eta_a$.
A similar formula holds for $\piLb$, with $\chi$ replaced by $\chib$. 
Observe, in particular, that the term $\trchb$  can only occur in connection to $\piLb$.
Thus,  all terms appearing  in the   $\EE$ integrand  are of the form,
$$\phi \c\Psi_1\c \Psi_2$$
 with $\phi$ one of the Ricci coefficients
 and  $\Psi_1,\Psi_2$  null curvature components.
Consider first the contribution to $\QQ_0$ of the anomalous 
terms $ \de^2 \int_{\HH_u^{(0,\ub)  }} Q[ R](e_4,e_4,e_4,e_4)
+\de^2\int_{\HHb_{\ub}^{  (0,u)}}     Q[R](e_4,e_4,e_4, e_3)
$  obtained  in  \eqref{eq:diverg.Bel}  in the case
   $X=Y=Z=e_4$. Since  $Q[ R](e_4,e_4,e_4,e_4)=|\a|^2$
   and $ Q[R](e_4,e_4,e_4, e_3)=|\b|^2$ we derive,
 \beaa
  \|\a\|_{L^2(H^{(0,\ub)}_u)}^2+\|\b\|_{L^2(\Hb^{(0,u)}_{\ub})}^2&\approx&   \|\a\|_{L^2(H^{(0,\ub)}_0)}^2+\EE_{01}(u,\ub)\\
  \EE_{01}(u,\ub)&\approx&\int\int_{\DD(u,\ub)}    Q(\piL, e_4, e_4)
  \eeaa
   Since all  terms of the form $\phi\c\Psi_1\c\Psi_2$ have the same 
   overall signature  $4$. 
    Thus,
   it is easy to derive the scale invariant norms  estimate,   
   \beaa
  \|\a\|_{\Lsc^2(H^{(0,\ub)}_u)}^2+\|\b\|_{\Lsc^2(\Hb^{(0,u)}_{\ub})}^2&\les&   \|\a\|_{\Lsc^2(H^{(0,\ub)}_0)}^2+\EE_{01}
  \eeaa
  and,
 \bea
  \EE_{01}&\les&\de^{1/2} \|\phi\|_{\Lsc^\infty}\int_0^{\ub} \|\Psi_1\|_{\Lsc^2(H^{(0,\ub')}_u) }  \|\Psi_2\|_{\Lsc^2(H^{(0,\ub')}_u) }
  \label{eq:intr.error1}
   \eea
      The gain of $\de^{1/2} $ is  a reflection of  the product estimates 
  of type   \eqref{product.inv.estim}.  Now,
   the only  null curvature component
  which  is anomalous with respect to the scale invariant  norms  $\Lsc^2(H^{(0,\ub)}_u)$ is $\a$. On the other hand the only Ricci coefficient
  which is anomalous in $\Lsc^\infty$ is $\trchb$. Indeed 
    we have to decompose   $\trchb=\trchbt+\trchb_0$,  where $\trchb_0$
    is the flat value of $\trchb_0$  and therefore  independent of $\de$. This
    leads to a loss of $\de^{1/2}$ in the corresponding estimates.
Now, since $\trchb$ cannot appear  among the components of $\piL$,
we can lose at most a power of $\de$   on the right hand side of  \eqref{eq:intr.error1}, which      occurs only   when $\Psi_1=\Psi_2=\a$.
  Fortunately  the terms on the left of our integral inequality  are also anomalous with respect to the same power of $\de$.   Therefore, since 
  $\|\phi\|_{\Lsc^\infty}\les C$, with $C=C(\II^0,\RR, \RRb)$ we derive   
  \beaa
  \RR_0^2[\a]+\RRb_0^2[\b]\les( \II^{(0)})^2+\de^{1/2}\c  C  \RR_0^2.
  \eeaa
  Therefore, for small $\de>0$,
 we derive  the  bound,
  \bea
  \RR_0[\a]+\RRb_0[\b]\les \II^{(0)} + \de^{1/4} C(\RR, \RRb).
  \label{intr:firsten.a}
  \eea
  with $C$ a universal constant depending only on   the curvature norms $\RR,\RRb$. 
  We would like to show that all other error  terms can be estimated
  in the same fashion, i.e. we would like to prove an estimate 
  of the form,
  \bea
  \RR_0+\RRb_0\les \II^{(0)} + \de^{1/4} C(\RR, \RRb).
   \label{intr:firsten}
  \eea
  Assuming that a similar estimate holds for  $\RR_1+\RRb_1$
  we would thus  conclude, for sufficiently small  $\de>0$, 
  \bea
  \RR+\RRb&\les& \II_0.
  \eea
  To prove \eqref{intr:firsten} we observe that 
 all  remaining  terms   in \eqref{eq:intr.mainid}  are scale invariant (i.e.
  they have the correct powers of $\de$).  
    In estimating the corresponding  error terms,
  appearing on the right hand side,  we only  have to be mindful of 
  those which contain $\trchb$ and $\a$. All other terms can be estimated
  by $\de^{1/2}p(\RR, \RRb)$ exactly as above. It is easy to check that all   terms involving $\trchb$  can only appear  through $\pihLb_{34}$.
   Thus,  it is easy to see that  all such terms
are of the form,
\beaa
Q_{3444} \pihLb^{34}&\approx&-|\b|^2 \trchb \\
Q_{3434}  \pihLb^{34}&\approx&-(\rho^2+\si^2)\trchb\\
Q_{3433} \pihLb^{34}&=&-|\bb|^2\trchb
\eeaa
Thus, since $\trchb=\trchbt+\trchb_0$, we easily deduce that 
all  error terms containing $\trchb$  can be estimated by,
\beaa
\de^{-1} \int_0^{\ub} \QQ_0(u,\ub')  d\ub'+
\de^{1/2} C(\RR,\RRb).
\eeaa
It is easy to check that the integral term  can be
 absorbed on the left by a Gronwall type inequality.
 It  thus remains to consider only the  terms linear\footnote{By signature considerations  there can be no terms   quadratic in $\a$} in $ \|\a\|_{\Lsc(H_u^{(0,\ub)})} $  which we have already    estimated above. 
 These lead   to error terms  with no excess powers of $\de $,  which could be potentially dangerous.     
 In fact we have to be a little more careful, because
 we  would   get an estimate of the form,
 \beaa
 \RR_0+\RRb_0&\les& \II^{(0)}+C(\RR,\RRb)
 \eeaa
 which is useless for large  curvature norms $\RR, \RRb$.
 To avoid this problem we need   to  refine our use  
 of   the $\OS_{0,\infty}$  norms. We observe that
  among all  terms  $\phi\c\Psi_1\c\Psi_2$  linear in $\a$  we can get better estimates for  all, except those which contain a Ricci component $\phi$
 which is anomalous  in   $\Lsc^4(S)$. All other terms gain a power 
 of $\de^{1/4}$. Indeed 
 the corresponding error  terms in  $\EE_1$
 can be estimated by\footnote{It follows from the Gagliardo-Nirenberg inequality
$  \|\a\|_{L^4(u,\ub)}^2\les    \|\nab\a\|_{L^2(u,\ub)}   \|\a\|_{L^2(u,\ub)}$ },
 \beaa
 && \de^{1/2}  \|\phi \|_{\Lsc^4(u,\ub)}   \c   \|\Psi\|_{\Lsc^2(H_u^{(0,\ub)}}\c 
 \|\nab\a\|_{\Lsc(H_u^{(0,\ub)})   } ^{1/2}\c  \|\a\|_{\Lsc(H_u^{(0,\ub)})} ^{1/2}\\
\les && \de^{1/4}\, \OS_{0,4}\c\RR_0\c  \RR_0[\a]^{1/2}\c\RR_1[\a]^{1/2}.
 \eeaa
  Denoting
 by $\EE_g$ all such error  terms  we  thus have,
 \beaa
 |\EE_g|&\les& \de^{1/4}C(\RR, \RRb)
 \eeaa
 It remains
 to check  the terms linear in $\a$ for which the Ricci coefficient
 is anomalous in the $\Lsc^4$ norm, i.e. terms for which $\phi$
 is either $\chih$ or $\chibh$.  
 It is easy to check that there are no terms linear in $\a$ which contain $\chih$
 and   thus  we only have to consider terms of the form $\chibh\c\a\c \Psi$, which we denote by  $\EE_b$. 
 Since   $\|\chibh \|_{\Lsc^4(u,\ub)} $  loses a power of $\de^{1/4}$
 we now have,
 \beaa
 && \de^{1/2}  \|\chibh \|_{\Lsc^4(u,\ub)}   \c   \|\Psi\|_{\Lsc^2(H_u^{(0,\ub)}}\c 
 \|\nab\a\|_{\Lsc(H_u^{(0,\ub)})   } ^{1/2}\c  \|\a\|_{\Lsc(H_u^{(0,\ub)})} ^{1/2}\\
\les &&  \OS_{0,4}[\chibh]\c\RR_0\c  \RR_0[\a]^{1/2}\c\RR_1[\a]^{1/2}
 \eeaa
 Since we are left with no positive  power of $\de$ we must now  be
 mindful of the fact  that  the estimates for   $\OS_{0,4}$ depend at least linearly   on  the curvature norms $\RR, \RRb$,  in which case  $\EE_b$  is super-quadratic in $\RR, \RRb$. We can  however trace back the $\de^{1/4}$ loss  of $\|\chibh \|_{\Lsc^4(u,\ub)}$   to   initial data, i.e. upon a careful inspection we find,  see   estimate \eqref{precise.estim.chib} of theorem A,  
 \bea
 \|\chibh \|_{\Lsc^4(u,\ub)}&\les& \de^{-1/4}\II^{(0)}+ C(\RR,\RRb)
 \eea
 Thus,
 \beaa
 \EE_b&\les& \II^{(0)}\c\RR_0\c  \RR_0[\a]^{1/2}\c\RR_1[\a]^{1/2} +
 \de^{1/4}C(\RR, \RRb)
 \eeaa
  The above considerations lead us to conclude, back to 
 \eqref{eq:intr.mainid}, 
 \bea
 \RR_0+\RRb_0&\les& \II^{(0)}+c\, \RR_0[\a]^{1/2}\c\RR_1[\a]^{1/2}+   \de^{1/8}C(\RR,\RRb).   
  \label{intr:firsten.correct}
 \eea
with $c= c(\II^{(0)})$ a constant depending only on the initial data.

 {\bf Remark} In the analysis above we have not considered the possibility that, among the terms in 
 the integrands of $\EE_0$ we  can have terms of the form $\phi \c\Psi_1\c \Psi_2$ with
  at least one of the curvature term being the null component  $\aa$, which cannot be
  estimated along $H_u$.  Among these terms  only those containing $\trchb$
  lead to terms which are $O(1)$ in $\de$. 
   These can be treated by using $\Hb$  which leads to estimates
  of the form,
   \beaa
 \QQ_0(u,\ub) +\QQb_0(u,\ub)&\les& \II_0^2+ \big(\int_0^u\QQ_0(u',\ub)  du'  +
   \de^{-1}\int_0^{\ub} \QQb_0(u,\ub')  d\ub'  \big)+C\de^{1/2}\nn  \\
    \eeaa
    with $C=C(\II^{(0)}, \RR, \RRb)$.
  The final estimate would follow from the following:
   lemma  below(which can be easily proved by the method of continuity).
 \begin{lemma}
 \label{le:Gronwall}
 Let $f(x,y), g(x,y)$ be positive functions  defined in the rectangle,
 $0\le x\le x_0$,  $0\le y\le y_0$ which verify the inequality,
 \beaa
 f(x,y)+g(x,y)&\les &J+ a\int_0^x f(x', y) dx'+b\int_0^y g(x,y') dy'
 \eeaa
 for some  nonnegative  constants  $a, b$ and $J$.
 Then, for all $0\le x\le x_0$,  $0\le y\le y_0$,
 \beaa
 f(x,y), g(x,y)&\les &J e^{ax+by}
 \eeaa
 \end{lemma}
 We summarize the results of this section in the following.
 \begin{proposition}
 \label{prop:curvI} The following estimate hold true with  constants
 $C=C(\II^{(0)}, \RR,\RRb)$, $c=c(\II^{(0))}   )  $  and $\de$ sufficiently small,
 \beaa
 \RR_0[\a]+\RRb_0[\b]&\les& \II^{(0)}+ C\de^{3/4}\\
 \RR_0+\RRb_0&\les& \II^{(0)}+c(\II^{(0)}) \RR^{1/2}+\de^{1/8} C.
 \eeaa
 \end{proposition}
  
\section{ Curvature   estimates II.}
\label{sect:CurvII}
We shall now estimate the first derivative of the null curvature components
appearing in $\RR_1, \RRb_1$. We apply \eqref{corr:mainenergy1} for the
angular  momentum vectorfields $O$ as well as for the vectorfields $L, \Lb$.
We prefer to work here with the vectorfields  $L,\Lb$ instead of $e_4, e_3$,
 as in the previous section, because their deformation tensors do not include
 $\om$, respectively  $\omb$. This  will make a difference in this section because
 we don't have good estimates for $\nab_4 \om $ and $\nab_3\omb$ which
 would appear among the derivatives of $\piL$ and $\piLb$. On the other
 hand, since $e_3$, $e_4$ differ from $L,\Lb$ only by the bounded factor $\Om$  no other  estimates will be affected.
 \subsection{Deformation tensors of the vectorfields $L$ and $\Lb$}
Below we list the components of $^L \pi_{\a\b}$ and $^\Lb \pi_{\a\b}$.
\begin{align*}
&^L \pi_{44}=0,\quad ^L\pi_{43}=0,\quad ^L\pi_{33}=-2\Omega^{-1} \omb,\\
&^L\pi_{4a}=0,\quad ^L\pi_{3a}=\Omega^{-1} (\eta_a+\zeta_a)+\Omega^{-1}\nab_a \log \Omega,\quad
^L \pi_{ab}=\Omega^{-1} \chi_{ab}
\end{align*}
\begin{align*}
&^\Lb \pi_{33}=0,\quad ^\Lb\pi_{43}=0,\quad ^\Lb\pi_{33}=-2\Omega^{-1} \om,\\
&^\Lb\pi_{3a}=0,\quad ^\Lb\pi_{4a}=\Omega^{-1} (\etab_a+\zeta_a)+\Omega^{-1}\nab_a \log \Omega,\quad
^\Lb \pi_{ab}=\Omega^{-1} \chib_{ab}
\end{align*}

 We start first with a sequence of lemmas:
\subsection{Preliminaries} Given  a vectorfield  X  we decompose both 
 $\Lieh_X R$ and $D_X R$  into their  null components $\a(\Lieh_X R), \b(\Lieh_X R),\ldots \aa(\Lieh_X R)$ and $\a(D_X R), \b(D_X R),\ldots \aa(D_X R)$.  We consider these decompositions  fo the vectorfields (note our discussion above concerning   $X=L, \Lb$ and    $ e_a$,\, $ a=1,2$.   In the spirit of our discussion above  we    write $e_4$ and $e_3$ instead of $L, \Lb$.   In the following lemma we estimate the null components 
 of $D_XR$, for $X=e_3, e_4, e_a$, in terms of $\RR$, $\RRb$. 
\begin{lemma}\label{lem:D}
 Denoting 
$\RR_u$ and $\RRb_{\ub}$ the restriction of the norms $\RR$ and $\RRb$ to the interval 
$[0,u]$ and $[0,\ub]$ respectively, we have 
with $C=C(\OO^{(0)},\RR,\RRb)$, the following anomalous estimates,
\beaa
\de^{\frac 12}\|\a(D_3R)\|_{\Lsc^2(H^{(0,\ub)}_{u})}+\de^{\frac 12} 
\|\b(D_a R)\|_{\Lsc^2(H^{(0,\ub)}_{u})}& \les & \II^{(0)}+\de^{\frac 14} C,\\
\eeaa
We also have the regular estimates,
\beaa
&&\|\a(D_aR)\|_{\Lsc^2(H^{(0,\ub)}_{u})}+\|\b(D_3 R)\|_{\Lsc^2(H^{(0,\ub)}_{u})} +\|\b(D_4 R)\|_{\Lsc^2(H^{(0,\ub)}_{u})}  \\
+&&\|(\rho, \si)(D_4R)\|_{\Lsc^2(H^{(0,\ub)}_{u})}+\|(\rho, \si)(D_3R)\|_{\Lsc^2(H^{(0,\ub)}_{u})}+\|(\rho, \si)(D_aR)\|_{\Lsc^2(H^{(0,\ub)}_{u})}\\
+&&\|\bb(D_4R)\|_{\Lsc^2(H^{(0,\ub)}_{u})}+\|\bb(D_aR)\|_{\Lsc^2(H^{(0,\ub)}_{u})}+ \|\alphab(D_4R)\|_{\Lsc^2(H^{(0,\ub)}_{u})}\les\RR_u+\de^{\frac 14} C
\eeaa
and
\beaa
&&\|\b(D_3R)\|_{\Lsc^2(H^{(0,u)}_{\ub})}+
\|(\rho, \si)(D_4R)\|_{\Lsc^2(H^{(0,u)}_{\ub})}+\|(\rho, \si)(D_3R)\|_{\Lsc^2(H^{(0,u)}_{\ub})}\\
+&&\|(\rho, \si)(D_4R)\|_{\Lsc^2(H^{(0,u)}_{\ub})}+\|\bb(D_4R)\|_{\Lsc^2(H^{(0,u)}_{\ub})}+\|\bb(D_3R)\|_{\Lsc^2(H^{(0,u)}_{\ub})}\\
+&&\|\bb(D_aR)\|_{\Lsc^2(H^{(0,u)}_{\ub})}+\|\alphab(D_4R)\|_{\Lsc^2(H^{(0,u)}_{\ub})}
+\|\alphab(D_aR)\|_{\Lsc^2(H^{(0,u)}_{\ub})}\les \RRb_{\ub}+\de^{\frac 14}C
\eeaa
\end{lemma}
\begin{remark}\label{rem:ab}
We note the special nature of the anomalies in $\a(D_3R)$  and $\b(D_a R)$. Specifically,
we can show that both terms can be written in the form
$G+F$ with $G=\trchb_0\c\a$ and $F$ obeying the estimate 
$$
\|F\|_{\Lsc^2(H^{(0,u)}_{\ub})}+\|F\|_{\Lsc^2(\Hb^{(0,\ub)}_{u})}\le C.
$$
\end{remark}
\begin{proof}
Let $\Psi^{(s)}(D_X R)$ denote the null components of $D_XR$ and 
$\phi^{(s)}  $ Ricci curvature components  of signature $s$. Then,  for $X=L, \Lb, e_1, e_2$, recalling
  that $\sgn(X)=1, 1/2, 0$ for $X=L, e_a, \Lb$, we write, 
\bea
\Psi^{(s)}(D_X R)&=&\nab_X \Psi^{(s)}+\sum_{s_1+s_2=s+\sgn(X)}\phi^{(s_1)}\c\Psi^{(s_2)}\label{formula.Psi(DR)}
\eea
Ignoring possible anomalies   we write,
\begin{equation}
\label{eq:regular.curv}
\begin{split}
\|\Psi^{(s)}(D_X R)\|_{\Lsc^2(H^{(0,\ub)}_{u}) }&\les \|\nab_X \Psi^{(s)}( R)\|_{\Lsc^2(H^{(0,\ub)}_{u})}+ \de^{1/2}\OS_{0,\infty} \c\RR_0\\
&\les \|\nab_X \Psi^{(s)}( R)\|_{\Lsc^2(H^{(0,\ub)}_{u})}+C\de^{1/2}\\
\|\Psi^{(s)}(D_X R)\|_{\Lsc^2(\Hb^{(0,u)}_{\ub}) }&\les \|\nab_X \Psi^{(s)}( R)\|_{\Lsc^2(\Hb^{(0,u)}_{\ub})}+\de^{1/2} \OS_{0,\infty} \c\RRb_0\\
&\les \|\nab_X \Psi^{(s)}( R)\|_{\Lsc^2(\Hb^{(0,u)}_{\ub})    } +C\de^{1/2}
\end{split}
\end{equation}
We only have to pay special  attention to the case when $\phi^{(s_1)}=\trchb$
and  $\Psi^{(s_2)}=\a$. If $s_2=2$, i.e.    $\Psi^{(s_2)}=\a$ then    $s_1$ can be  $1,  1/ 2 $ and $0$.
The case $s_1=1$ occur  only if $X=e_4$, which is not covered by the lemma.  The case $s_2=2, s_1= 1/ 2$   is regular.  Indeed,   in that case
$s+\sgn(X)=5/2$. Thus either $s=2, X=e_a$  or $s=3/2$,  $X=L$.   In
 both  cases   we simply estimate  the worst quadratic  term, on the right hand side of  \eqref{formula.Psi(DR)},
  with $s_2=2$, by
 \beaa
\|\phi\c\a\|_{\Lsc^2(H^{(0,\ub)}_{u})}&\les& \de^{\frac 12} \|\phi\|_{{\Lsc^4}_{u,\ub}}  \|\a\|_{{\Lsc^4}_{u,\ub}}\les
 \de^{\frac 12} \OS_{0,4}[\phi]  \|\a\|_{{\Lsc^2}_{u,\ub}}^{\frac 12}  \|\nab \a\|_{{\Lsc^2}_{u,\ub}}^{\frac 12}\\ 
 &\les & \de^{\frac 14} \OS_{0,4}[\phi] \c \RR_0[\a]^{\frac 12}\c  \RR_1[\a]^{\frac 12}\les C\de^{1/4}.
 \eeaa
 The principal term is either   $\nab\a$ in the first case or $\nab_L\b$ in the second. In the second situation, using the null Bianchi
 identities, (proceeding as above with the term of the form
 $\phi\c\a$),
 \beaa
\| \nab_L\b\|_{\Lsc^2(H^{(0,\ub)}_{u})}&\les&\|\nab\a\|_{\Lsc^2(H^{(0,\ub)}_{u})}+C\de^{1/4} 
 \eeaa

 In the  case ($s_2=2$,  $s_1=0$)   $\trchb$ can  appear among the quadratic terms on the right. In that case  $s+\sgn(X)=2$. The  $s=2$ 
 and $X=\Lb$ corresponds to the anomalous estimate for $\a(D_\Lb R)$. 
   In that case the estimate
is,
\beaa
\|\a(D_\Lb R)\|_{\Lsc^2(H^{(0,\ub)}_{u}) }&\les& \|\nab_\Lb \a\|_{\Lsc^2(H^{(0,\ub)}_{u})   } + (1+\de^{1/2} C) \| \a\|_{\Lsc^2(H^{(0,\ub)}} +\de^{1/2} C
\eeaa
Also, in view of the Bianchi identities, \eqref{eq:null.Bianchi},
\beaa
 \|\nab_\Lb \a\|_{\Lsc^2(H^{(0,\ub)}_{u})   } &\les&  \|\nab \b\|_{\Lsc^2(H^{(0,\ub)}_{u})   } + \| \a\|_{\Lsc^2(H^{(0,\ub)}_{u})   }+C\de^{1/2}
\eeaa
Hence, in view of our estimate for  $\a$ in the previous section 
\beaa
\de^{1/2}\|\a(D_L R)\|_{\Lsc^2(H^{(0,\ub)}_{u}) }&\les&\de^{1/2} \|\nab_L \a\|_{\Lsc^2(H^{(0,\ub)}_{u})}+ (1+\de^{1/2} C)\de^{1/2} \| \a\|_{\Lsc^2(H^{(0,\ub)}}\\
&\les& \II^{(0)}+ \de^{1/4} C
\eeaa
as desired. 
We need also to consider the case  $s_2=2, s_1=0$, $s=3/2$ and $X=e_a$.
Then, due to  the term $\trchb_0\c\a$ on the right hand side of
\eqref{formula.Psi(DR)} we have,
\beaa
 \| \b(D_a R) \|_{\Lsc^2(H^{(0,\ub)}_{u})   } &\les&  \|\nab \b\|_{\Lsc^2(H^{(0,\ub)}_{u})   } +\|\a\|_{\Lsc^2(H^{(0,\ub)}_{u})   }+C\de^{1/4}
\eeaa   
Thus,  
\beaa
\de^{1/2}\| \b(D_a R) \|_{\Lsc^2(H^{(0,\ub)}_{u})   } &\les& \II^{(0)}+C\de^{1/4}
\eeaa
as which is the second anomalous estimate.

It remains to consider the cases $s_2<2$, $s_1=0$.  In the worst  case,  when a quadratic  term on the right hand side of  \eqref{formula.Psi(DR)} is of the form $\trchb_0\c\Psi^{(s_2)}$ we  make the following  correction to   estimate  \eqref{eq:regular.curv},
\begin{equation*}
\begin{split}
\|\Psi^{(s)}(D_X R)\|_{\Lsc^2(H^{(0,\ub)}_{u}) }&\les \|\nab_X \Psi^{(s)}( R)\|_{\Lsc^2(H^{(0,\ub)}_{u})}+ \|\Psi^{(s_2)}\|_{\Lsc^2(H^{(0,\ub)}_{u}) }+C\de^{1/4}\\
&\les \|\nab_X \Psi^{(s)}( R)\|_{\Lsc^2(H^{(0,\ub)}_{u})}+\RR_u+C\de^{1/4}\\
\|\Psi^{(s)}(D_X R)\|_{\Lsc^2(\Hb^{(0,u)}_{\ub}) }&\les \|\nab_X \Psi^{(s)}( R)\|_{\Lsc^2(\Hb^{(0,u)}_{\ub})}+    \|\Psi^{(s_2)}\|_{\Lsc^2(\Hb^{(0,u)}_{\ub}) }     +C\de^{1/4}\\
&\les \|\nab_X \Psi^{(s)}( R)\|_{\Lsc^2(\Hb^{(0,u)}_{\ub})}+\RRb_{\ub}+C\de^{1/4}
\end{split}
\end{equation*}

These imply the  regular estimates of the Lemma for the case
$X=e_a$.  For the cases $X=L, \Lb$ we can express $\nab_X\Psi^{(s)}(R)$
using the Bianchi identities,
\beaa
\nab_3\Psi^{(s)}&=&\nab\Psi^{(s-\frac 1 2 )}+\sum_{s_1+s_2=s} \phi^{(s_1)}\c\Psi^{(s_2)}, \qquad   0<s<2 \\
\nab_4\Psi^{(s)}&=&\nab\Psi^{(s+\frac 1 2 )}+\sum_{s_1+s_2=s+1} \phi^{(s_1)}\c\Psi^{(s_2)},\qquad   0\le s<2.
\eeaa
The worst quadratic  terms which can appear  on the right  are of
the form $\trchb\c\Psi^{(s)} $ with  $s<2$ which can be easily estimated.
We thus derive all the regular estimates of the Lemma.
\end{proof}

\begin{lemma}\label{lem:Lie}
The following estimates  for the Lie derivatives $\Lieh_X R$, with respect to 
    hold true $X=\{\Lb ,L,O\}$.
    \bea
\|\a(\Lieh_L R) -\nab_L\a\|_{\Lsc^2(H^{(0,\ub)}_{u})   } &\les& C\\
   \de^{1/2}\|\a(\Lieh _\Lb R)  - \nab_\Lb \a \|_{\Lsc^2(H^{(0,\ub)}_{u})   } &\les &\RR_0+C\de^{3/4}\
\eea
Also,
\bea
    \|\Psi^{(s)}(\Lieh_L R)-(\nab_L \Psi)^{(s)}  \|_{\Lsc^2(H^{(0,\ub)}_u)}&\les&C\de^{1/4},\qquad 1\le s \le 5/2,    \\    
              \|\Psi^{(s)}(\Lieh_\Lb R )  -( \nab_\Lb \Psi)^{(s)} \|_{\Lsc^2(H^{(0,\ub)}_{u})   } &\les&
  \RR_0+C\de^{1/4},\qquad 1\le s\le 3/2 \\\
  \|\Psi^{(s)}(\Lieh_\Lb R )  -( \nab_\Lb \Psi)^{(s)} \|_{\Lsc^2(\Hb^{(0,u)}_{\ub})   } &\les&
  \RR_0+C\de^{1/4},\qquad s\le 1/2.
     \eea

      For $X=O$ we have the estimates.
      \bea
     \|\Psi^{(s)}(\Lieh_O R)-(\nab_O \Psi)^{(s)}  \|_{\Lsc^2(H^{(0,\ub)}_u)}&\les&C\de^{1/4},\qquad  1\le s\le 5/2\\
     \| \Psi^{(s)}(\Lieh_OR)-(\nab_O \Psi)^{(s)}\|_{\Lsc^2(\Hb^{(0,u)}_{\ub})  }&\les&C\de^{1/4} ,\qquad   1/2\le s\le 2.
    \eea

\end{lemma}
\begin{proof}
We  will  make use of the regular $\Lsc^\infty$   estimates 
for  Ricci coefficients $\phi\in \{\chi, \om, \eta, \etab, \chibh, \trchbt, \omb\}$. We also make use of the following estimates for  $\nab O$ and $\piO$.

We  write,   recalling the definition of the Lie derivative and with $E$ denoting the set 
$e_1, e_2, e_3, e_4$, 
\begin{equation}
\begin{split}
\label{eqle:1}
\Psi^{(s)}(\Lie_X R)&=X(\Psi^{(s)})-\sum\limits_{s_1+s_2=s} \sum_{Y\in{E}}([X,Y])^{(s_1)} \Psi^{(s_2)}\\
&=\Llie_X(\Psi^{(s)})-\sum\limits_{s_1+s_2=s} \sum_{Y\in{E}}(  ([X,Y])^{(s_1)})^\perp\c \Psi^{(s_2)} 
\end{split}
\end{equation}
Here  $\Llie_X(\Psi^{(s)})$ denotes the projection of the Lie derivative 
on the $S(u,\ub)$ surfaces and  $[X, Y]^\perp$ the  orthogonal component
of  $[X,Y]$ i.e.,
\beaa
[X, Y]^\perp=-\frac  12g( [X, Y], e_3) e_4-\frac  12g( [X, Y], e_4) e_3
\eeaa
Consider first the case when $X=L, \Lb$.   In that case $ [X, Y]^\perp $ 
depends only on the regular Ricci coefficients $\om,\eta,\etab,\omb$ 
Therefore, taking into account the worst possible case when $\a$ appear
among the quadratic terms (in which case we appeal to $\Lsc^4$ estimates),
we derive, 
\begin{equation}
\label{eqle:2}
\begin{split}
\|\Psi^{(s)}(\Lie_L R)- ( \Llie _L  \Psi)^{(s)}\|_{\Lsc^2(H^{(0,\ub)}_{u})   } &\les C\de^{1/4},\qquad 1\le s\le 3
\\
\|\Psi^{(s)}(\Lie_\Lb R)- ( \Llie _\Lb  \Psi)^{(s)}\|_{\Lsc^2(H^{(0,\ub)}_{u})   } &\les C\de^{1/4},\qquad  1\le s \le 2.\\
\|\Psi^{(s)}(\Lie_\Lb R)- ( \Llie _\Lb \Psi)^{(s)}\|_{\Lsc^2(\Hb^{(0,u)}_{\ub})    } &\les C\de^{1/4}, \qquad  0\le s\le 1/2.
\end{split}
\end{equation}

On the other hand, schematically,
\beaa
\Llie _L  \Psi^{(s)}=\nab_L\Psi^{(s)} +\sum_{s_1+s_2=1+s}\phi^{(s_1)}\c\Psi^{(s_2)}
\eeaa
with $\phi^{(s_1)}\in \{\chi,  \eta, \etab\} $.
In the particular case $s=3$  we can have  a double anomaly  of the form,
$\chi\c\a$. In that case,
\beaa
\|\Llie_L\a-\nab_L\a\|_{\Lsc^2(H^{(0,\ub)}_{u})   } &\les&C\de^{\frac 12} \|\a\|_{\Lsc^2(H^{(0,\ub)}_{u})   } +C\de^{1/2}
\eeaa
Therefore,
$
\|\Llie_L\a-\nab_L\a\|_{\Lsc^2(H^{(0,\ub)}_{u})   } \les C,
$
from which, combining with \eqref{eqle:2},
\beaa
\|\a(\Lie_L R) -\nab_L\a\|_{\Lsc^2(H^{(0,\ub)}_{u})   } &\les& C
\eeaa
Recalling the definition of $\Lieh_L R$ we deduce,
\beaa
\de^{1/2}\|\a(\Lieh_L R) -\nab_L\a\|_{\Lsc^2(H^{(0,\ub)}_{u})   } &\les& C
\eeaa
as desired. 

We now consider all other cases, $1\le s\le 5/2$.  Since there are no double   anomalies,
we deduce,  (using $\Lsc^4(S)$ estimates for the term containing $\a$)
\beaa
\|\Llie _L  \Psi^{(s)}-(\nab_L \Psi)^{(s)}\|_{\Lsc^2(H^{(0,\ub)}_{u})   } &\les&C\de^{1/4}
\eeaa
Hence, combining with \eqref{eqle:2},
\beaa
\|  \Psi^{(s)}(\Lie_L R) -(\nab_L \Psi)^{(s)}\|_{\Lsc^2(H^{(0,\ub)}_{u})   } &\les&C\de^{1/4}
\eeaa
Recalling the definition of $  \Psi^{(s)}(\Lie_L R)$ we deduce,
\beaa
\|  \Psi^{(s)}(\Lieh_L R) -(\nab_L \Psi)^{(s)}\|_{\Lsc^2(H^{(0,\ub)}_{u})   } &\les&C\de^{1/4}, \qquad 1\le s\le 5/2.
\eeaa
as desired.

We now consider the estimates for $\Lb$.  We have,
\beaa
\Llie _{\Lb}  \Psi^{(s)}=\nab_{\Lb} \Psi^{(s)} +\trchb_0 \Psi^{(s)}+
\sum_{s_1+s_2=s}\phi^{(s_1)}\c\Psi^{(s_2)}
\eeaa
with $\phi^{(s_1)}\in\{\eta, \etab, \chibh, \trchbt \}$.
 Observe that the worst  terms 
$\trchb_0\c\a$ can only appear for $s=2$.  In that case,
\beaa
\|\Llie _\Lb \a  - \nab_\Lb \a \|_{\Lsc^2(H^{(0,\ub)}_{u})   } &\les&\|\a\|_{\Lsc^2(H^{(0,\ub)}_{u})   }    +C\de^{1/4}
\les\de^{-1/2} \RR_0   +C\de^{1/4}\
\eeaa
Thus, combining with \eqref{eqle:2},
\beaa
\de^{1/2}\|\a(\Lie _\Lb R)  - \nab_\Lb \a \|_{\Lsc^2(H^{(0,\ub)}_{u})   } &\les&\RR_0+C\de^{3/4}
\eeaa
Finally, recalling the definition of  $ \a(\Lieh _\Lb R)$ we deduce,
$\de^{1/2}\|\a(\Lieh _\Lb R)  - \nab_\Lb \a \|_{\Lsc^2(H^{(0,\ub)}_{u})   } \les \RR_0+C\de^{3/4}$ as desired. 

 In all other cases, $1 \le s\le \frac 3 2 $ we have,
 \beaa
 \|\Llie _\Lb \Psi^{(s)}  -( \nab_\Lb \Psi)^{(s)} \|_{\Lsc^2(H^{(0,\ub)}_{u})   } &\les&\|\Psi^{(s)}\|_{\Lsc^2(H^{(0,\ub)}_{u})   }    +C\de^{1/4}\\
 &\les&\RR_0+C\de^{1/4}
 \eeaa
 Hence, combining  with \eqref{eqle:2} and recalling the definition of  $\Lieh$
 we deduce,
 \beaa
  \|\Psi^{(s)}(\Lieh_\Lb R )  -( \nab_\Lb \Psi)^{(s)} \|_{\Lsc^2(H^{(0,\ub)}_{u})   } &\les&
  \RR_0+C\de^{1/4}
 \eeaa 
 as desired.

We now consider the case when $X=O$.  In view of \eqref{eqle:1},
\beaa
\|\Psi^{(s)}(\Lie_O R)-(\Llie _O  \Psi)^{(s)}\|_{\Lsc^2(H^{(0,\ub)}_{u})   } &\les&C\de^{1/4}
\eeaa
Indeed      the projections of $[O, e_4]$,   $[O,e_3]$  on $e_3, e_4$ depend
only on $O$ and  the Ricci coefficients $\om, \eta,\etab, \omb$  while 
 $[O, e_a]$, $a=1,2$  are tangent to $S(u,\ub)$.   On the other hand,
 $\Llie_O \Psi^{(s)}$ differs from $(\nab_O \Psi)^{(s)}$  by terms quadratic in
 $\nab O$ and $\Psi$.   We recall that we have   $\|\nab O\|_{\Lsc^\infty}\les C$, i.e. they are regular in the supremum norm.
 Thus, as before, 
\beaa
\|\Llie_O \Psi^{(s)}-(\nab_O \Psi)^{(s)}\|_{\Lsc^2(H^{(0,\ub)}_{u})   }&\les&C\de^{1/4}.
\eeaa
 Combining this with the estimate above and recalling the definition
 of $\Lieh_O R$ as well as the estimates $\|\piO\|_{\Lsc^\infty}\les C$
 we derive, for all $s\ge 1/2$.
 \beaa
\| \Psi^{(s)} (\Lieh_OR)-(\nab_O \Psi)^{(s)}\|_{\Lsc^2(H^{(0,\ub)}_{u})  }&\les&C\de^{1/4}
 \eeaa
Similarly we prove, for $s\le 3/2$
\beaa
\| \Psi^{(s)}(\Lieh_OR)-(\nab_O \Psi)^{(s)}\|_{\Lsc^2(\Hb^{(0,u)}_{\ub})  }&\les&C\de^{1/4}
 \eeaa

\end{proof}

\subsection{Estimate for $\|\nab_4\a\|_{\Lsc^2(H)}$.}
It is important to observe throughout this section  that the deformation tensors $   \pil $ of $L$   does not contain $\om$  and   $\pilb$ of $\Lb$  does  not contain  either $\omb$.  

 We apply   corollary \ref{corr:mainenergy1}  to $O=L$ and $X=Y=Z=e_4$. 
 and  derive
 \bea
 \int_{H^{(0,\ub)}_u}|\a(\Lieh_L R)|^2&\les&  \int_{H^{(0,\ub)}_0}|\a(\Lieh_L R)|^2+\int_{\DD(u,\ub)}
 (Q[  \Lieh_L R]\c\piL)(e_4,  e_4, e_4) \nn\\
 &+&\int_{\DD(u,\ub)} D(L, R)(e_4, e_4, e_4)\label{eq:Lie4aa}
 \eea
In view of the conservation of signature we can write schematically,
 \bea
 (Q[  \Lieh_L R]\c\piL)(e_4,  e_4, e_4)&=&\sum_{s_1+s_2+s_3=6}\phi^{(s_1)} \c \Psi^{(s_2)}[\Lieh_4 R]\c \Psi^{(s_3)}[\Lieh_4 R]\label{eq:Lie4a.1}\\
 D(L, R)(e_4, e_4, e_4)&=&\sum_{s_1+s_2+s_3=6}\Psi^{(s_2)}[\Lieh_4 R]\c \big(\psi^{(s_1)} \c   (D\Psi)^{(s_3)   }+
    (D\psi)^{(s_1)}\c \Psi^{(s_3)}\big)    \label{eq:Lie4a.2}
\eea
with  Ricci coefficients $\phi  \in \{\chi, \om, \eta, \etab,\omb\}$,   $\psi\in \{\chi, \eta, \etab, \omb \} $ null curvature components $\Psi$ and labels $s_1, s_2, s_3 $ denoting the signature of the  corresponding component.
Thus,
\beaa
\|\a(\Lieh_L R)\|_{\Lsc^2(H^{(0,\ub)}_u)}^2&\les& \|\a(\Lieh_L R)\|_{\Lsc^2(H^{(0,\ub)}_0)}^2+I_1+I_2+I_3
\eeaa
with
\beaa
I_1&=&  \de^{1/2} \sum   \|\phi^{(s_1)}\|_{\Lsc^\infty}   \int_0^u \| \Psi^{(s_2)}(\Lieh_L R)\|_{\Lsc^2(H^{(0,\ub)}_{u'})}\c  \| \Psi^{(s_3)}(\Lieh_4 R)\|_{\Lsc^2(H^{(0,\ub)}_{u'})} du'\\
I_2&=& \de^{1/2} \sum   \|\psi^{(s_1)}\|_{\Lsc^\infty}    \int_0^u  \| \Psi^{(s_2)}(\Lieh_L R)\|_{\Lsc^2(H^{(0,\ub)}_{u'})}\c  \| (D \Psi)^{(s_3)}\|_{\Lsc^2(H^{(0,\ub)}_{u'})} du'\\
I_3&=& \sum  \int_0^u \| \Psi^{(s_2)}(\Lieh_L R)\|_{\Lsc^2(H^{(0,\ub)}_{u'})}   \|(D\psi)^{(s_1)}\c \Psi^{(s_3)}\|_{\Lsc^2(H^{(0,\ub)}_{u'})} du'
\eeaa
Among the terms $I_1$  the worst are those in which $s_2=s_3=3$, in which case $s_1=0$.  Since $\trchb $ cannot appear among our Ricci coefficients here,
 and $\|\phi\|_{\Lsc^\infty}\les C$,  with $C=C(\II^0, \RR,\RRb)$
\beaa
I_{11}&\les& C\de^{1/2}  \int_0^u \| \a(\Lieh_L R)\|_{\Lsc^2(H^{(0,\ub)}_{u'})}^2 du'
\eeaa
All curvature terms  $ \| \Psi^{(s)}(\Lieh_L R)\|_{\Lsc^2(H^{(0,\ub)}_{u})}$
with $s<3$  can be estimated according to lemmas \ref{lem:D} and \ref{lem:Lie} to derive,
\beaa
\|\Psi^{(s)}[\Lieh_L R]\|_{\Lsc^2(H_u)}&\les& \RR_0+ \de^{1/4} C\les  C,\qquad 
s<3.
\eeaa
Therefore, estimating all remaining terms  in $I_1$ we deduce,
\beaa
I_1(u,\ub)&\les& C\de^{1/2}  \int_0^u \big( \| \a(\Lieh_L R)\|_{\Lsc^2(H^{(0,\ub)}_{u'})}^2 +    \| \a(\Lieh_L R)\|_{\Lsc^2(H^{(0,\ub)}_{u'})} \RR  \big)   du'+\de^{\frac 12} \RR^2
\eeaa
The term $I_2$ can be estimated in exactly the same manner. Since $0\le s_1\le 1$ and $1\le s_2\le 3$ we 
have $2\le s_3\le 3$. This implies that the term $(D\Psi)^{s_3}$ may be estimated along $H_u$. 
With the exception of the term $\a(D_L R)$ these estimates are given in Lemma \ref{lem:D}. Among those there
are two anomalous terms $\a(D_3 R)$ and $\b(D_a R)$. We then obtain
\bea
I_2(u,\ub)&\les& C\de^{1/2}  \int_0^u  \big(    \| \a(\Lieh_L R)\|_{\Lsc^2(H^{(0,\ub)}_{u'})}^2+   
(C\de^{-\frac 14}+\II^{(0)} \de^{-\frac 12})\| \a(\Lieh_L R)\|_{\Lsc^2(H^{(0,\ub)}_{u'})} \big)   du'\nn\\ &+&\II^{(0)} \de^{-\frac 12} 
+ C \de^{-\frac 14}\nn\\
&\les&  C\de^{1/2}  \int_0^u  \| \a(\Lieh_L R)\|_{\Lsc^2(H^{(0,\ub)}_{u'})}^2 du'+\II^{(0)} \de^{-\frac 12} 
+ C \de^{-\frac 14}
\eea
It remains to estimate $I_3$. We note that, in the worst  case,  the term $D\psi$ can be written in the form
$$
(D\psi)^{(s_1)}=({\bf \nab} \psi)^{s_1} +\trchb_0\c \psi^{(s_1)}+
\sum\limits_{s_{11}+s_{12}=s_1}\psi^{(s_{11})}\c\psi^{(s_{12})}. 
$$
Observe that $({\bf \nab} \psi)^{s_1}\ne (\nab_4\om, \nab_3\omb)$.
Indeed $\nab_4 \om$ cannot occur, since $\psi^{(s_1)}\in \{\chi, \eta, \etab, \omb$
  On the other hand  $\nab_3\omegab$ cannot occur by signature considerations.
  Indeed     in that case $s_1=sgn(\nab_3\omegab)=0, $
which is ruled out since $s_1+s_2+s_3=6$ while $s_2\le 3$ and $s_3\le 2$.

Thus, since  $({\bf \nab} \psi)^{s_1}\ne (\nab_4\om, \nab_3\omb)$ (for which we
do not have $\Lsc^4$ estimates !),
we  derive,
\begin{align*}
 \|(D\psi)^{(s_1)}\c \Psi^{(s_3)}\|_{\Lsc^2(H^{(0,\ub)}_{u'})}&\les  \de^{\frac 12}
 \|({\bf \nab}\psi)^{(s_1)}\|_{\Lsc^4(H^{(0,\ub)}_{u'})} \| \Psi^{(s_3)}\|_{\Lsc^4(H^{(0,\ub)}_{u'})}
\\ &+\left(\de \sum\limits_{s_{11}+s_{12}=s_1}\|\psi^{(s_{11})}\|_{\Lsc^\infty}\|\phi^{(s_{12})}\|_{\Lsc^\infty} +
\de^{\frac 12}\|\psi^{(s_1)}\|_{\Lsc^\infty}\right)
\| \Psi^{(s_3)}\|_{\Lsc^2(H^{(0,\ub)}_{u'})}\\ &\les C.
\end{align*}
Observe that  in the last step we  have used the $\Lsc^4$ estimates for the first derivatives of the Ricci coefficients  $\psi\in\{\chi, \eta, \etab\}$
and the null curvature components,  and allowed for the worst possible scenario in which ($\Psi^{(s_3)}=\a$), 
\begin{align*}
 &\|({\bf \nab}\psi)^{(s_1)}\|_{\Lsc^4(H^{(0,\ub)}_{u'})}+ \| \Psi^{(s_3)}\|_{\Lsc^4(H^{(0,\ub)}_{u'})} \le C\de^{-\frac 14},\\
 &\| \Psi^{(s_3)}\|_{\Lsc^2(H^{(0,\ub)}_{u'})}\les C\de^{-\frac 12}
\end{align*}
As a consequence we derive,
\begin{align*}
I_3(u,\ub)\les  C \int_0^u  \| \a(\Lieh_4 R)\|_{\Lsc^2(H^{(0,\ub)}_{u'})} du'+C
\end{align*}
Combining the estimates for $I_1, I_2, I_3$ we derive,
 \beaa
 \|\a(\Lieh_4 R)\|_{\Lsc^2(H^{(0,\ub)}_u)}^2&\les& \|\a(\Lieh_4 R)\|_{\Lsc^2(H^{(0,\ub)}_0)}^2+ C\big(1+\de^{1/2}\big) \int_0^u \|\a(\Lieh_4 R)\|_{\Lsc^2(H^{0,\ub)}_{u'} } du' +C\de^{1/2}
   \eeaa
    Therefore, in view of the anomalous character of  $ \|\a(\Lieh_L R)\|_{\Lsc^2(H_u)}$,
  \beaa
 \de \|\a(\Lieh_L R)\|_{\Lsc^2(H^{(0,\ub)}_u)}^2&\les& \de \|\a(\Lieh_L R)\|_{\Lsc^2(H^{(0,\ub)}_0)}^2+
 C\de^{\frac 32} 
  \eeaa
 from which we infer that, for some $C=C(\II^0,\RR, \RRb)$,
 \beaa
 \de^{1/2}  \|\a(\Lieh_L R)\|_{\Lsc^2(H^{(0,\ub)}_u)}&\les& \de^{1/2}  \|\a(\Lieh_L R)\|_{\Lsc^2(H^{(0,\ub)}_0)}
 +C\de^{1/2}\\
 &\les& \II^0+C\de^{1/2}
\eeaa
On the other hand, in view of the definition of  $\Lieh_LR$ we have,
 \beaa
\a(\Lieh_L R)=\nab_L \a+\sum_{s_1+s_2=3}\phi ^{(s_1)}\c \Psi^{(s_2)}
\eeaa
Hence,
\beaa
\|\nab_4\a\|_{\Lsc^2(H^{(0,\ub)}_u)}&\les&  \|\a(\Lieh_L R)\|_{\Lsc^2(H^{(0,\ub)}_u)}+C\RR_0
\eeaa
Therefore we deduce,
\begin{proposition}\label{prop:nab4a}
The following estimate holds true for sufficiently small $\de>0$, with a constant $C=C(\II^0, \RR, \RRb)$,
\bea
\|\nab_4\a\|_{\Lsc^2(H^{(0,\ub)}_u)}\les \de^{-1/2}\II^0+C.
\eea
\end{proposition}
\subsection{Estimate for $\|\nab_3\aa\|_{\Lsc^2(H)}$.}
 Applying corollary \ref{corr:mainenergy1}  to $O=e_3$ and $X=Y=Z=e_3$
 we derive,
\bea
 \int_{\Hb^{(0,u)}_{\ub}}|\aa(\Lieh_\Lb R)|^2&\les&   \int_{\Hb^{(0,u)}_{0}}|\aa(\Lieh_\Lb R)|^2
 +\int_{\DD(u,\ub)}
 (Q[  \Lieh_\Lb R]\c\piLb)(e_3,  e_3, e_3) \nn\\
 &+&\int_{\DD(u,\ub)} D(\Lb, R)(e_3, e_3, e_3)\label{eq:Lie3aa}
 \eea
 In view of the conservation of signature we can write schematically (we need to take into 
account the  signature associated to the integrals),
 \bea
 (Q[  \Lieh_\Lb R]\c\piL)(e_3,  e_3, e_3)&=&\sum_{s_1+s_2+s_3=1}\psi^{(s_1)} \c \Psi^{(s_2)}[\Lieh_3 R]\c \Psi^{(s_3)}[\Lieh_3 R]\label{eq:Lie3a.1}\\
 D(\Lb, R)(e_3, e_3, e_3)&=&\sum_{s_1+s_2+s_3=1}\Psi^{(s_2)}[\Lieh_\Lb R]\c \big(\psi^{(s_1)} \c   (D\Psi)^{(s_3)   }+
    (D\psi)^{(s_1)}\c \Psi^{(s_3)}\big)    \label{eq:Lie3a.2}
\eea
with  Ricci coefficients $\psi  \in\{ \om, \eta, \etab,  \chib\}$, null curvature components $\Psi$ and labels $s_1, s_2, s_3 $ denoting the signature of the 
corresponding component.  We now need to be careful with 
terms which involve $\trchb$ and $\nab_3\trchb$.  In \eqref{eq:Lie3a.1} the only terms
which  contain $\trchb$ have  the form
 $\trchb\c |\bb (\Lieh_\Lb R)|^2$ which
we write in the form  
 \beaa
 \trchb_0\c |\bb (\Lieh_\Lb R)|^2+ \trchbt\c |\bb (\Lieh_\Lb R)|^2
 \eeaa
In \eqref{eq:Lie3a.2} the only terms which contains $\nab_3\trchb$, must be of the form
\beaa
\nab_3\trchb\c  \Psi^{(s_2)}(\Lieh_\Lb R)\c  \Psi^{(s_3)}, \qquad s_2+s_3=1.
\eeaa
Recall that,
\beaa
\nab_3\trchb&=&-\frac 1 2 \trchb^2-2\omb \trchb-|\chibh|^2
\eeaa
Thus, writing, $\trchb=\trchb_0+\trchbt$, we have schematically,
\beaa
\nab_3\trchb&=&-\frac 1 2 \trchb_0^2 +\trchb_0\psi_g+\psi\c\psi
\eeaa
We have,
\beaa
\|\aa(\Lieh_\Lb R)\|_{\Lsc^2(H^{(0,u)}_{\ub})}^2&\les& \|\aa(\Lieh_\Lb R)\|_{\Lsc^2(H^{(0,u)}_{0})}^2 +P_1+P_2+ P_3+J_1+J_2+J_3
\eeaa
with,
$P_1,P_2, P_3$ the terms corresponding to the   terms   in $\trchb_0$,
\beaa
P_1&=&\sum_{s_2+s_3=1}  \de^{-1} \int_0^{\ub} \| \Psi^{(s_2)}(\Lieh_\Lb R)\|_{\Lsc^2(\Hb^{(0,u)}_{\ub'})}\c  \| \Psi^{(s_3)}(\Lieh_\Lb R)\|_{\Lsc^2(\Hb^{(0,u)}_{\ub'})} d\ub'\\
P_2&=&\sum_{s_2+s_3=1} \de^{-1}   \int_0^{\ub}  \| \Psi^{(s_2)}(\Lieh_\Lb R)\|_{\Lsc^2(\Hb^{(0,u)}_{\ub'})}\c  \| (D \Psi)^{(s_3)}\|_{\Lsc^2(\Hb^{(0,u)}_{\ub'})} d\ub'\\
P_3&=&\sum_{s_2+s_3=1}  \de^{-1}  \int_0^{\ub}  \| \Psi^{(s_2)}(\Lieh_3 R)\|_{\Lsc^2(\Hb^{(0,u)}_{\ub'})}\c
 \|  \Psi^{(s_3)}\|_{\Lsc^2(\Hb^{(0,u)}_{\ub'})} d\ub'
\eeaa
and $J_1, J_2, J_3$ the remaining terms with Ricci terms $\psi\in \{ \eta, \etab, \chibh\}$, 
\beaa
J_1&=&  \de^{-1/2} \sum_{s_1+s_2+s_3=1}   \|\psi^{(s_1)}\|_{\Lsc^\infty}   \int_0^{\ub} \| \Psi^{(s_2)}(\Lieh_3 R)\|_{\Lsc^2(\Hb^{(0,u)}_{\ub'})}\c  \| \Psi^{(s_3)}(\Lieh_3 R)\|_{\Lsc^2(\Hb^{(0,u)}_{\ub'})} d\ub'\\
J_2&=& \de^{-1/2} \sum_{s_1+s_2+s_3=1}   \|\psi^{(s_1)}\|_{\Lsc^\infty}    \int_0^{\ub}  \| \Psi^{(s_2)}(\Lieh_3 R)\|_{\Lsc^2(\Hb^{(0,u)}_{\ub'})}\c  \| (D \Psi)^{(s_3)}\|_{\Lsc^2(\Hb^{(0,u)}_{\ub'})} d\ub'\\
J_3&=&\sum_{s_1+s_2+s_3=1} \de^{-1}  \int_0^u \| \Psi^{(s_2)}(\Lieh_3R)\|_{\Lsc^2(\Hb^{(0,u)}_{\ub'})}   \|(D\psi)^{(s_1)}\c \Psi^{(s_3)}\|_{\Lsc^2(\Hb^{(0,u)}_{\ub'})} d\ub'
 \eeaa
  It clearly suffices to estimate the principal terms $P$. Indeed the $J$ terms can be treated  exactly as in the previous subsection\footnote{Remark that  in  $J_2$     $(D\Psi)^{(3)}$ differ
   from  $\nab_3\omb$, because $\Psi^{(s_3)}\in \{\om, \eta,\etab,\chib\}$, and 
   $\nab_4\om$ by signature considerations.} . We have,
 \beaa
 P_1&\les&
 \de^{-1}\int_0^{\ub} \| \bb(\Lieh_\Lb R)\|^2_{\Lsc^2(\Hb^{(0,u)}_{\ub'})} d\ub'
 \eeaa
 According to Lemma \eqref{lem:Lie} we have,
 \beaa
 \| \bb(\Lieh_\Lb R)\|^2_{\Lsc^2(\Hb^{(0,u)}_{\ub'})}&\les&  \| \nab_3 \bb \|^2_{\Lsc^2(\Hb^{(0,u)}_{\ub'}) }+\RRb_0+\de^{1/4}C
 \eeaa
  In view of the  Bianchi identities, for $\frac 12\le s\le 1$, 
 \beaa
 \nab_3\bb=&\div \aa - 2\trchb \c  \bb -2\omb\c \bb +\etab\c\aa
  \eeaa
 Therefore,
  \beaa
 \| \bb(\Lieh_\Lb R)\|^2_{\Lsc^2(\Hb^{(0,u)}_{\ub'})}&\les& \| \nab\aa \|^2_{\Lsc^2(\Hb^{(0,u)}_{\ub'})}+\RR_0+\de^{1/4} C
 \eeaa
 Consequently,

 \beaa
 P_1(u,\ub)&\les& 
 \de^{-1} \int_0^{\ub} \big(\|\nab \alphab\|^2_{\Lsc^2(\Hb^{(0,u)}_{\ub'})} +\RR_0(u,\ub') \big) d\ub' +C\de^{1/4}\\ &\les &
  \de^{-1} \int_0^{\ub}\RRb^2(u,\ub')  d\ub'+\de^{1/4} C
 \eeaa
 $P_2, P_3$ can be estimated exactly in the same manner. First,  observe that in $P_2$ the terms of the form  $(D\Psi)^{(s_3)}$ obey the
 bounds, 
 $$
 \|(D\Psi)^{(s_3)}\|_{\Lsc^2(\Hb^{(0,u)}_{\ub'})}\les \RRb(u,\ub') +\de^{\frac 14} C.
 $$
 This follows from the restriction $s_3\le 1$. Similarly, for $s_2\le 1$
  $$
 \| \Psi^{(s_2)}(\Lieh_\Lb R)\|_{\Lsc^2(\Hb^{(0,u)}_{\ub'})}\les  \| \aa(\Lieh_\Lb R)\|_{\Lsc^2(\Hb^{(0,u)}_{\ub'})}+
 \RRb(u,\ub') +\de^{\frac 12} C.
 $$
 Therefore,
 \beaa
 P_2(u,\ub)&\les& \de^{-1}\int_0^{\ub} \| \aa(\Lieh_\Lb R)\|^2_{\Lsc^2(\Hb^{(0,u)}_{\ub'})} d\ub'\\
 &+&\de^{-1} \int_0^{\ub} \RRb^2(u,\ub')  d\ub'+\de^{1/2} C
 \eeaa
 Similarly,
 \beaa
   P_3(u,\ub)&\les& \de^{-1}\int_0^{\ub} \| \aa(\Lieh_3 R)\|_{\Lsc^2(\Hb^{(0,u)}_{\ub'})}\RRb(u,\ub')  d\ub'\\
 &+&\de^{-1} \int_0^{\ub} \RRb^2(u,\ub')  d\ub'+\de^{1/2} C.
 \eeaa
  Therefore using Lemma \ref{lem:Lie} we derive,
 \begin{proposition} \label{prop:nab3aa}
 The following estimate holds true for sufficiently small $\de>0$, with a constant $C=C(\II^0, \RR, \RRb)$,
 \bea
\| \nab_3\aa\|_{\Lsc^2(\Hb^{(0,u)}_{\ub})}^2 &\les&\|\nab_3\aa\|_{\Lsc^2(\Hb^{(0,u)}_{0})}^2+\de^{-\frac 12} \RR_0+
\int_0^{\ub} \RR(u,\ub') ^2d\ub' +\de^{\frac 14} C
 \eea
 \end{proposition}
 \subsection{Estimates for the angular derivatives of $R$} 
  Applying corollary \ref{corr:mainenergy1}  to  the angular momentum  vectorfields $O$ and $X,Y,Z\in \{e_3,e_4\}$
 we derive,
 \bea
 \int_{H^{(0,\ub)}_u}|\Psi^{(s)}(\Lieh_O R)|^2 + \int_{H^{(0,u)}_{\ub}}|\Psi^{(s-\frac 12)}(\Lieh_O R)|^2&\les&  \int_{H^{(0,\ub)}_0}|\Psi^{(s)}(\Lieh_O R)|^2+\int_{\DD(u,\ub)}
 (Q[  \Lieh_O R]\c\pi)(X,  Y, Z) \nn\\
 &+&\int_{\DD(u,\ub)} D(O, R)(X, Y, Z)\label{eq:LieOaa}
 \eea
 In view of the conservation of signature we can write schematically,
 \bea
 (Q[  \Lieh_O R]\c \pi)(X,  Y, Z)&=&\trchb_0\c  \sum_{s_2+s_3=2s}\ \Psi^{(s_2)}[\Lieh_O R]\c \Psi^{(s_3)}[\Lieh_O R]\label{eq:LieOa.1}\\
&+& \sum_{s_1+s_2+s_3=2s}\phi^{(s_1)} \c \Psi^{(s_2)}[\Lieh_O R]\c \Psi^{(s_3)}[\Lieh_O R]\nn
 \eea
 with $\phi$ Ricci coefficients in $\{\chi,\om,\eta,\etab, \chibh, \trchbt,  \om\}.$
  Also, recalling that $\pi=\pih+\frac 1 4 \tr(\pi) g$, 
 \bea
 D(O, R)(X, Y, Z)&=&\sum_{s_1+s_2+s_3=2s}\Psi^{(s_2)}[\Lieh_O R]\c \big(\piO^{(s_1)} \c   (D\Psi)^{(s_3)   }+
    (D\piO)^{(s_1)}\c \Psi^{(s_3)}\big)    \label{eq:Lie4O.2}
\eea
with     $\piO^{(s)}$    null components of the deformation tensor of $O$.
 Thus, for all $s> \frac 1 2 $,
\beaa
\|\Psi^{(s)}(\Lieh_O R)\|_{\Lsc^2(H^{(0,\ub)}_u)}^2+\|\Psi^{(s-\frac 12)}(\Lieh_O R)\|_{\Lsc^2(H^{(0,u)}_{\ub})}^2&\les& \|\Psi^{(s)}(\Lieh_O R)\|_{\Lsc^2(H^{(0,\ub)}_0)}^2+I_1+I_2+I_3
\eeaa
\begin{itemize}
\item  $I_1$ is the integral  in $\DD(u,\ub)$ whose integrand is
given by  \eqref{eq:LieOa.1}, 
\item $I_2$ is  the integral  in $\DD(u,\ub)$ whose integrand is
given by $$\sum_{s_1+s_2+s_3=2s}\Psi^{(s_2)}[\Lieh_O R]\c \piO^{(s_1)} \c   (D\Psi)^{(s_3)   }.$$
\item  $I_3$ is  the integral  in $\DD(u,\ub)$ whose integrand is
given by  $$\sum_{s_1+s_2+s_3=2s}\Psi^{(s_2)}[\Lieh_O R]\c (D\piO)^{(s_1)}\c \Psi^{(s_3)}.$$
\end{itemize}
In what follows we make use of the   estimates
for the deformation tensors of the angular momentum vectorfields
established in theorem \ref{thmn:thmB} 
$O$,
\beaa
\|\piO\|_{\Lsc^4(S)} + \|\piO\|_{\Lsc^\infty(S)}\les C
\eeaa
Also all null components  of the derivatives  $D\piO$, with the exception of  $(D_3\piO)_{3a}$,  verify the estimates, 
\bea
\|D\piO\|_{\Lsc^4(S)  } \les C
\eea
Moreover,
\bea
\|(D_3\piO)_{3a}-\nab_3 Z\|_{L^4(S)} +  \|\sup_{\ub}|\nab_3 Z|\|_{L^2(S)}  &\les & C\\
\eea
The term $I_1$ can be easily estimated, since none of the curvature terms are 
anomalous. Indeed,  in view of lemma \ref{lem:Lie} we have,  for all $s>1/2$
\beaa
 \| \Psi^{(s)}(\Lieh_O R)\|_{\Lsc^2(H^{(0,\ub)}_{u})} &\les&  \| \Psi^{(s)}(\Lieh_O R)-\nab_O\Psi^{(s)}\|_{\Lsc^2(H^{(0,\ub)}_{u})} +\|\nab_O\Psi^{(s)}\|_{\Lsc^2(H^{(0,\ub)}_{u})}\\
&\les& \RR(u,\ub)
\eeaa
while,   for $s=\frac 12$,
\beaa
 \| \Psi^{(1/2)}(\Lieh_O R)\|_{\Lsc^2(\Hb^{(0,u)}_{\ub})}&\les& \RRb(u,\ub)
\eeaa
Consequently, for $s>1/2$,
\beaa
I_1&\les&\sum_{s\ge 1 }
\int_0^u \|\Psi^{(s)}(\Lieh_OR)\|_{\Lsc^2(H^{(0,\ub)}_{u'})}^2 du' +\de^{1/2} C
\eeaa
while for $s=1/2$, 
\beaa
I_1&\les&\sum_{s\le 2}\de^{-1} \int_0^{\ub}  \|\Psi^{(s)}(\Lieh_OR)\|_{\Lsc^2(H^{(0,u)}_{\ub'})}^2 d\ub 
+\de^{1/2} C\nn
\eeaa
Therefore,
\bea
I_1&\les&\sum_{s\ge 1 }
\int_0^u\|\Psi^{(s)}(\Lieh_OR)\|_{\Lsc^2(H^{(0,\ub)}_{u'})}^2 du' 
+\sum_{s\le 2} \de^{-1}  \int_0^{\ub}  \|\Psi^{(s)}(\Lieh_OR)\|_{\Lsc^2(H^{(0,u)}_{\ub'})}^2 d\ub' +\de^{1/2} C\nn\\
\label{eq.finalcurv:1}
\eea

Among the terms $I_2$ the only possible anomalies may be due to the case when 
$s_3=3$, i.e. $(D\Psi)^{(s_3)}=\a(D_4R)$ or in the easier cases $(D\Psi)^{(s_3)}=\a(D_3R)$ and $(D\Psi)^{(s_3)}=\b(D_aR)$ (i.e. $s_3=2$).  We denote by 
$I_{21}$ all terms in $I_2$ except those which corresponds 
to these anomalous  cases.  For all other terms we 
have  either  $\|(D\Psi)^{(s_3)}\|_{\Lsc^2(H^{(0,u)}_{\ub'})}  \les C$
or $\|(D\Psi)^{(s_3)}\|_{\Lsc^2(\Hb^{(0,\ub)}_{u'})}\les C$. 
Using also  $\|\piO\|_{\Lsc^\infty}\les C$ and,
\beaa
 \|\Psi^{(s_2)}(\Lieh_OR)\|_{\Lsc(H^{(0,\ub)}_{u})}\les C, \qquad s_2\ge 1\\
  \|\Psi^{(s_2)}(\Lieh_OR)\|_{\Lsc^2(H^{(0,u)}_{\ub})}\les C, \qquad  s_2\le 2
  \eeaa
      we derive,
\beaa
I_{21}&\les&\de^{\frac 14} C
\eeaa
We now consider the terms $I_{22}$ which contain $(D\Psi)^{(s_3)}\a(D_3R)$ and $(D\Psi)^{(s_3)}=\b(D_aR)$ but not   $\a(D_4R)$. In this case write, according to the Remark \ref{rem:ab}, 
\beaa
(D\Psi)^{(s_3)}  &=&G + F^{(s_3)}, \qquad \\
\|F^{(s_3)}\|_{\Lsc^2(H^{(0,\ub)}_{u})   }   &\les& C,\qquad  s_3>1,\\
\|F^{(s_3)}\|_{\Lsc^2(\Hb^{(0,u)}_{\ub})  }  &\les& C, \qquad s_3<2. 
\eeaa
where $G=\trchb_0\c\a$. Clearly,  the terms corresponding to $F^{(s_3)}$  
can be estimated exactly as above.  To estimate the terms corresponding
to $G$  we make use of the $\Lsc^4(S)$ estimate,
$\|G\|_{\Lsc^4(S)}\le C \de^{-\frac 14}$. Using also,
 $
  \|\piO\|_{\Lsc^4(S)}\les C
 $ 
 we obtain, 
\beaa
I_{22}&\les&
\de^{\frac 14} C
\eeaa
It remains to estimate the terms in $I_{23}$ which contain 
$\a(D_4  R)$. The   integrand,   which contain $\a(D_4  R)$ has  the  form,
\beaa
D_{23}&=&\sum_{s_1+s_2=2s-3} \piO^{(s_1)}\c\Psi^{(s_2)}(\Lieh_O R)\c \a (D_4 R)
\eeaa
This term is   potentially  dangerous !  In view of  lemma \eqref{lem:Lie} $\Psi^{(s_2)}(\Lieh_O R)$ differs
from  $(\nab_O\Psi)^{(s_2)}$ by a lower order terms. It thus
suffices to estimate,
\beaa
D_{23}&\equiv&\sum_{s_1+s_2=2s-3} \piO^{(s_1)}\c(\nab_O  \Psi)^{(s_2)}\c \a (D_4 R)
\eeaa
We also decompose 
\beaa
\a(D_4 R) = \nab_4 \a + \sum_{s_3+s_4=3}\phi^{(s_3)}\c\Psi^{(s_4)}
\eeaa
where $\phi^{(s_3)}\in \{\om, \eta,\etab\}$. This forces  $s_4<2$    and thus, since there are no anomalies  we derive,
\beaa
\|\a(D_4 R) - \nab_4 \a\|_{\Lsc^2(H_u^{(0,\ub)})}\les C\de^{1/2}
\eeaa
Therefore we can  safely replace $\a(D_4 R)$ by   $\nab_4 \a$  and thus it remains
to estimate,
\beaa
D_{23}&\equiv&\sum_{s_1+s_2=2s-3} \piO^{(s_1)}\c(\nab_O  \Psi)^{(s_2)}\c \nab_4 \a
\eeaa
Because of the anomaly of  $\nab_4 \a$   the best we can by a straightforward estimate is to derive  an estimate of the form  $I_{23}\les\II^{(0)}+ C$ which is not acceptable.  Because of this we are forced to integrate by parts,
Ignoring the boundary term $\int_{\Hb_{\ub}}\piO^{(s_1)}\c (\nab_O \Psi)^{(s_2)}\c\a$, for a 
moment
\bea
\int_{\DD} \piO^{(s_1)}\c(\nab_O \Psi)^{(s_2)}\c \nab_4\a&=- \int_{\DD}\nab_4  \piO^{(s_1)}\c(\nab_O \Psi)^{(s_2)}\c \a- \int_{\DD}  \piO^{(s_1)}\c \nab_4(\nab_O \Psi)^{(s_2)}\c \a  \nn\\
\label{eq:nab4integration}
\eea
We write schematically,  with $\phi^{(1/2)}\in\{ \eta, \etab\}$
\beaa
 \nab_4(\nab_O \Psi)^{(s_2)}&=&\nab_4\nab_O (\Psi)^{(s_2-\frac 1 2 )}\\
 &=&\nab_O\nab_4 (\Psi)^{(s_2-\frac 1 2 )}+ \sum_{s_3+s_4=s_2+1} \Psi^{(s_3)}\c\Psi^{(s_4)}+\sum_{s_4=s_2+1/2}\phi^{(1/2)}\c\Psi^{(s_4)}.
\eeaa

We can therefore replace the integrand  $D_{23}$
by, 
\beaa
D_{23}&\equiv& -D_{231}-D_{232}-D_{233}-D_{234}\\
 D_{231}&=&\sum_{s_1+s_2=2s-3}\nab_4  \piO^{(s_1)}\c(\nab_O \Psi)^{(s_2)}\c \a\\
 D_{232}&=&\sum_{s_1+s_2=2s-3} \piO^{(s_1)}\c\nab_O(\nab_4 \Psi^{(s_2+1/2)}) \c \a\\
 D_{233} &=&\sum_{s_1+s_2=2s-3}  \piO^{(s_1)}\c     \big(  \sum_{s_3+s_4=s_2+1}   \Psi^{(s_3)}\c\Psi^{(s_4)}\big)\c\a\\
 D_{234}&=&\sum_{s_1+s_2=2s-3}  \piO^{(s_1)}\c \big(\sum_{s_4=s_2+1/2}\phi^{(1/2)}\c\Psi^{(s_4)}\big) \c\a
\eeaa
 Accordingly we   decompose $I_{23}\equiv I_{231}+I_{232}+I_{233}+I_{234}$.
Now, 
\beaa
I_{231}&\les& \de^{\frac 12}
\|\nab_4\piO^{(s_1)}\|_{\Lsc^4(S)}
\|\a\|_{\Lsc^4(S)}  \c
\de^{-1} \int_0^{\ub} \|(\nab_O \Psi)^{(s_2)}\|_{\Lsc(H^{(0,u)}_{\ub'})} d\ub'\\
&\les& \de^{\frac 14} C.
\eeaa
 The  terms $I_{233}$ and $I_{234} $ are clearly lower order in $\de$, we derive
 \beaa
 I_{233}+I_{234}&\les&\de^{\frac 12} C
 \eeaa
 It remains to estimate $I_{232}$ for which we need to perform another integration 
 by parts.   We write
 \beaa
\int_{\DD} \piO^{(s_1)}\c\nab_O(\nab_4 \Psi^{(s_2+1/2)}) \c \a&=&-
\int_{\DD} \nab_O \piO^{(s_1)}\c (\nab_4 \Psi^{(s_2+1/2)}) \c \a\\
&-&\int_{\DD}\piO^{(s_1)}\c (\nab_4 \Psi^{(s_2+1/2)}) \c   \nab_O \a\\
&-&\int_{\DD}\piO^{(s_1)}\c (\nab_4 \Psi^{(s_2+1/2)}) \c   (\nab^a O_a) \a
\eeaa
 By Bianchi, since $s_2+1/2< 3$,  
$$
\|(\nab_4  \Psi)^{(s_2+1/2)}\|_{\Lsc(H^{(0,\ub)}_{u})}\les \|(\nab \Psi)^{(s_2+1/2)}\|_{\Lsc(H^{(0,\ub)}_{u})}+\de^{\frac 12} \|\phi\|_{\Lsc^\infty} \|\Psi\|_{\Lsc^2(S)}\le C. 
$$
Therefore,
\begin{align*}
|\int_{\DD} \nab_O\piO^{(s_1)}\c (\nab_4 \Psi)^{(s_2+\frac 12)}\c\a|&\les \de^{\frac 12}
\|\nab_O\piO^{(s_1)}\|_{\Lsc^4(S)}
\|\a\|_{\Lsc^4(S)}\\&\times \int_0^{u} \|(\nab_4 \Psi)^{(s_2+\frac 12)}\|_{\Lsc(H^{(0,\ub)}_{u'})} du'\\
&\les \de^{\frac 14} C.
\end{align*}
Also,
\begin{align*}
|\int_{\DD} \piO^{(s_1)}\c (\nab_4 \Psi)^{(s_2+\frac 12)}\c\nab_O \a|&\les \de^{\frac 12}
\int_0^{u}\|\nab_O\a\|_{\Lsc(H^{(0,\ub)}_{u'})}  \|(\nab_4 \Psi)^{(s_2+\frac 12)}\|_{\Lsc(H^{(0,\ub)}_{u'})} du'\\&\times \|\piO^{(s_1)}\|_{\Lsc^\infty(S)}\\
&\les \de^{\frac 12} C.
\end{align*}
The remaining integral in $I_{232}$ is clearly lower order in $\de$.
For the boundary term  in \eqref{eq:nab4integration}  we have,
\beaa
|\int_{\Hb_{\ub}}\piO^{(s_1)}\c (\nab_O \Psi)^{(s_2)} \c\a |&\les& \de^{\frac 12} 
\|(\nab_O \Psi)^{(s_2)}\|_{\Lsc^2(\Hb_{\ub})} \c  
 \|\    \piO\|_{\Lsc^4(S)} \|\a  \|_{\Lsc^4(S)}\le \de^{\frac 14} C.
\eeaa
We therefore deduce,
\bea
I_2&\les& C\de^{1/4}
\eea

Consider now $I_3$. Ignoring powers of $\de$,  we have to estimate  the integral 
$\int_{\DD} (D\piO)^{(s_1)}\c\Psi^{(s_2)}(\Lieh_O R)\c \Psi^{(s_3)}$. Recall  the estimates 
$$
\|(D\piO)^{(s_1)}\|_{\Lsc^4(S)}\les C
$$
for all components of $(D\piO)^{(s_1)}$ with the exception of the term $D_3 \piO_{3a}$ which corresponds to the signature  $s_1=0$.  In this latter case we have,
$$
\|D_3 \piO_{3a}-\nab_3 Z\|_{\Lsc^4(S)}\les C, \qquad \|\sup_{\ub} |(D_3\piO)_{3a}|\|_{\Lsc^2(S)}\les C
$$
In the case  $(D\piO)^{(s_1)}\neq D_3 \piO_{3a}$, we have
\begin{align*}
|\int_{\DD} (D\piO)^{(s_1)}\c\Psi^{(s_2)}(\Lieh_O R)\c \Psi^{(s_3)}|&\les 
 \de^{\frac 12}\,
\de^{-1} \int_0^{\ub} \|(\nab_O \Psi)^{(s_2)}\|_{\Lsc(H^{(0,u)}_{\ub'})} d\ub' \c  \\ &\times
 \| (D\piO)^{(s_1)}\|_{\Lsc^4(S)} \|\Psi^{(s_3)}  \|_{\Lsc^4(S)}\le \de^{\frac 14} C,
\end{align*}
where we considered the worst case in which $\Psi^{(s_3)}=\a$ and thus anomalous and 
$(\nab_O \Psi)^{(s_2)}$ has to be estimated along $H^{(0,u)}_{\ub'}$.

For the case we can replace, without loss of generality,     $(D\nab \piO)^{(s_1)}$ by $  \nab_3 Z$. Indeed the remaining  error term can be estimated exactly as above.
In this case, since  
$s_1=0$,  signature considerations dictate that $s_3\ge 1$. It follows from the 
conditions $s_1+s_2+s_3=2s$, $s_2\in\{s,s-\frac 12\}$ and $s\ge 1$. This implies that 
we may use the trace theorem along $H_{u}$
$$
\|\Psi^{(s_3)}\|_{\Tr_{(sc)}(H)}\les \de^{\frac 14} C,
$$
where in fact $\de^{\frac 14}$ only occurs in the case $\Psi^{(s_3)}=\a$, for
all other terms the behavior in $\de$ is better. 
We thus  give the argument only  for $\Psi^{(s_3)}=\a$, other cases are even  easier.
Recalling also lemma \ref{lem:Lie},
\begin{align*}
|\int_{\DD} \nab_3 Z\c\Psi^{(s_2)}(\Lieh_O R)\c \Psi^{(s_3)}|&\les 
 \de^{\frac 12}\,
\de^{-1} \int_0^{\ub} \|(\nab_O \Psi)^{(s_2)}\|_{\Lsc(H^{(0,u)}_{\ub'})} d\ub' \c  \\ &\times
 \| \sup_{\ub} |\nab_3 Z |\|_{\Lsc^2(S)} \sup_u \|\Psi^{(s_3)} \|_{\Tr_{(sc)}(H_u)}\le \de^{\frac 14} C,
\end{align*}
 Finally we observe that the only borderline terms, not resulting in positive powers of the parameter $\de$
and arising from coupling to $\trchb$, involve only $\b,\rho,\si$ and $\bb$ components of 
curvature.

Combining all our estimates for $I, I_2, I_3$ and using lemma \ref{le:Gronwall} 
we derive,
\beaa
\sum_{1 \le s \le 5/2} \big(\|\Psi^{(s)}(\Lieh_O R)\|_{\Lsc^2(H^{(0,\ub)}_u)}+\|\Psi^{(s-\frac 12)}(\Lieh_O R)\|_{\Lsc^2(H^{(0,u)}_{\ub})}\big)&\les&\sum_{1 \le s \le 2} \|\Psi^{(s)}(\Lieh_O R)\|_{\Lsc^2(H^{(0,\ub)}_0)}\\
&+&\de^{1/4} C\label{estim:curvang}
\eeaa
More precisely,  we easily check the following,
\beaa
\|\a(\Lieh_O R)\|_{\Lsc^2(H^{(0,\ub)}_u)}+\|\b(\Lieh_O R)\|_{\Lsc^2(H^{(0,u)}_{\ub})}&\les&
\|\a(\Lieh_O R)\|_{\Lsc^2(H^{(0,\ub)}_u)}+\de^{1/4} C
\eeaa
For $s\le 2$ we have,
\beaa
\sum_{s\le 2} \big (\|\Psi^{(s)}(\Lieh_O R)\|_{\Lsc^2(H^{(0,\ub)}_u)}+\|\Psi^{(s-\frac 12)}(\Lieh_O R)\|_{\Lsc^2(H^{(0,u)}_{\ub})}\big)&\les&\sum_{s\le 2}\|\Psi^{(s)}(\Lieh_O R)\|_{\Lsc^2(H^{(0,\ub)}_u)}+\de^{1/4} C
\eeaa
Using  the estimates  of lemma   \ref{lem:Lie} we  derive,
\bea
\|\nab \a\|_{\Lsc^2(H^{(0,\ub)}_u)}+\|\nab \b\|_{\Lsc^2(H^{(0,u)}_{\ub})}&\les&
\|\nab\a \|_{\Lsc^2(H^{(0,\ub)}_u)}+\de^{1/4} C
\eea
For $s\le 2$ we have,
\bea
\sum_{s\le 2} \big( \|(  \nab  \Psi)^{(s)}\|_{\Lsc^2(H^{(0,\ub)}_u)}+\|(\nab\Psi)^{(s-\frac 12)}\|_{\Lsc^2(H^{(0,u)}_{\ub})}\big)&\les&\sum_{s\le 2}\|(\nab \Psi)^{(s)}\|_{\Lsc^2(H^{(0,\ub)}_u)}+\de^{1/4} C\label{eq:curvs<2.5}
\eea

We  summarize the result above  in the following.
\begin{proposition} The following estimates hold for $\de$ sufficiently small
and $C=C(\II^{(0)}, \RR, \RRb)$.
\bea
\sum_{1 \le s \le 5/2} \big(\nab \|\Psi^{(s)} \|_{\Lsc^2(H^{(0,\ub)}_u)}+\|\nab \Psi^{(s-\frac 12)}\|_{\Lsc^2(H^{(0,u)}_{\ub})}\big)&\les&\II_0+C\de^{1/4}
\eea

\end{proposition}
Combining this proposition with propositions \ref{prop:nab3aa} and 
\ref{prop:nab4a} we  derive.
\bea
\RR_1+\RRb_1&\les& \II_0+C\de^{1/4}
\eea
Finally combining this with proposition \ref{prop:curvI} we derive,
\bea
\RR+\RRb&\le& \II_0+C\de^{1/4}
\eea
This ends the proof of our main theorem. 

\subsection{ Proof  of  propositions \ref{prop:improved.curv} and   \ref{prop:better}. } 
\label{subs:last} The proof of 
 proposition \ref{prop:improved.curv} is an immediate consequence 
of   estimate  \eqref{eq:curvs<2.5}  together  with
 the initial  assumptions  derived in  proposition \ref{prop:improved.in.curv}.
Indeed, under  initial assumptions \eqref{eq:improved.in.curv} we derive,
\beaa
\sum_{s\le 2} \left (\|(\nab\Psi)^{(s)}\|_{\Lsc^2(H^{(0,\ub)}_u)}+\|(\nab\Psi)^{(s-\frac 12)}\|_{\Lsc^2(H^{(0,u)}_{\ub})}\right)&\les&\ep+\de^{1/4} C
\eeaa
which gives, for sufficiently small $\de$, estimate \eqref{eq:improved.curv}.

We combine this result with proposition \ref{prop:betterphi} to prove
the following scale invariant  version of proposition  \ref{prop:better} of the introduction.
\begin{proposition}\label{prop:better'}
The solution $^{(3)} \phi$ of the problem $\nab_3 ^{(3)} \phi=\nab\eta$ with trivial initial data satisfies 
$$
\|^{(3)}\phi\|_{\Lsc^\infty(S)}\le C \ep^{\frac 14}
+ C\de^{\frac 18}.
$$
\end{proposition}

\end{document}